\documentclass[letterpaper,11pt]{article}
\usepackage[margin=0.965in]{geometry}

\usepackage{cmap} 
\usepackage[utf8]{inputenc}
\usepackage[english]{babel}
\usepackage[T1]{fontenc}
\usepackage{amsmath}
\usepackage{amsfonts}
\usepackage{bm}
\usepackage{xcolor}
\definecolor{blueviolet}{rgb}{0.2, 0.2, 0.6}
\definecolor{webgreen}{rgb}{0,.5,0}
\definecolor{webbrown}{rgb}{.6,0,0}
\usepackage{setspace}
\usepackage[pdftex,
	bookmarks=false,
	colorlinks=true, 
	urlcolor=webbrown, 
	linkcolor=blueviolet, 
	citecolor=webgreen,
	pdfstartpage=1,
	pdfstartview={FitH},  
	bookmarksopen=false
	]{hyperref}
\allowdisplaybreaks
\usepackage{tikz}
\usepackage{braket}
\usepackage[numbers,sort&compress]{natbib}
\usepackage{amsthm}
\usepackage{dsfont}
\usepackage{listings}
\usepackage[capitalize]{cleveref}
\usepackage{appendix}

\usepackage{authblk}

\numberwithin{equation}{section}

\newcounter{lemc}
\newcounter{propc}
\newcounter{coroc}
\numberwithin{lemc}{section}
\numberwithin{propc}{section}
\numberwithin{coroc}{section}

\newtheorem{theorem}{Theorem}
\newtheorem{corollary}[coroc]{Corollary}

\newtheorem{definition}{Definition}

\newtheorem{lemma}[lemc]{Lemma}
\newtheorem{proposition}[propc]{Proposition}

\theoremstyle{definition}
\newtheorem{remark}{Remark}

\newcommand{\bx}{{\bm{x}}}
\newcommand{\be}{{\bm{e}}}
\newcommand{\balpha}{{\bm{\alpha}}}

\newcommand{\BC}{\mathbb{C}}
\newcommand{\CD}{\mathcal{D}}

\newcommand{\CI}{\mathcal{I}}

\newcommand{\CL}{\mathcal{L}}

\newcommand{\CO}{\mathcal{O}}

\newcommand{\BR}{\mathbb{R}}
\newcommand{\CS}{\mathcal{S}}

\newcommand{\BZ}{\mathbb{Z}}
\newcommand*{\verr}{E}

\newcommand{\vA}{\bm{A}}

\newcommand{\vB}{\bm{B}}

\newcommand{\vD}{\bm{D}}
\newcommand{\vE}{\bm{E}}
\newcommand{\vF}{\bm{F}}
\newcommand{\vG}{\bm{G}}
\newcommand{\vH}{\bm{H}}
\newcommand{\vh}{\bm{h}}
\newcommand{\vI}{\bm{I}}

\newcommand{\vL}{\bm{L}}

\newcommand{\vO}{\bm{O}}
\newcommand{\vP}{\bm{P}}

\newcommand{\vQ}{\bm{Q}}
\newcommand{\vR}{\bm{R}}
\newcommand{\vS}{\bm{S}}

\newcommand{\bt}{\bar{t}}
\newcommand{\vU}{\bm{U}}
\newcommand{\vV}{\bm{V}}
\newcommand{\vW}{\bm{W}}
\newcommand{\vX}{\bm{X}}
\newcommand{\vY }{\bm{Y }}
\newcommand{\vZ}{\bm{Z}}
\newcommand{\vsigma}{\bm{ \sigma}}

\newcommand{\vrho}{{\bm{\rho}}}
\renewcommand{\L}{\left}
\newcommand{\R}{\right}

\newcommand{\bE}{\bar{E}}
\newcommand{\bomega}{\bar{\omega}}
\newcommand{\tOmega}{\tilde{\Omega}}
\newcommand{\tCO}{\tilde{\CO}}

\newcommand{\vertiii}[1]{{\left\vert\kern-0.25ex\left\vert\kern-0.25ex\left\vert #1 \right\vert\kern-0.25ex\right\vert\kern-0.25ex\right\vert}}

\newcommand{\normp}[2]{\norm{#1}_{#2}}
\newcommand{\lnormp}[2]{\lnorm{#1}_{#2}}

\newcommand{\labs}[1]{\left\vert {#1} \right\vert}
\newcommand{\lnorm}[1]{\left\Vert {#1} \right\Vert}
\newcommand{\e}{\mathrm{e}}
\newcommand{\ri}{\mathrm{i}}
\newcommand{\rd}{\mathrm{d}}

\newcommand{\vect}[1]{\boldsymbol{#1}}
\newcommand{\xv}{{\vect{x}}}
\newcommand{\yv}{{\vect{y}}}
\newcommand{\sv}{{\vect{s}}}

\newcommand{\indicator}{\mathds{1}}
\newcommand{\undersetbrace}[2]{ \underset{#1}{\underbrace{#2}}}
\newcommand{\rom}[1]{\mathtt{\uppercase\expandafter{\romannumeral #1\relax}}}

\newcommand{\inp}{\textnormal{in}}
\newcommand{\circuit}{C}
\newcommand{\prop}{\textnormal{prop}}
\newcommand{\clock}{\textnormal{clock}}
\newcommand{\eff}{\textnormal{eff}}

\DeclareMathOperator*{\argmax}{arg\,max}
\DeclareMathOperator*{\argmin}{arg\,min}

\DeclareMathOperator{\wt}{wt}
\DeclareMathOperator{\ad}{ad}

\DeclareMathOperator{\Tr}{tr}

\DeclareMathOperator{\poly}{poly}

\newcommand{\ketbra}[2]{\lvert #1 \rangle \! \langle #2 \rvert}
\newcommand{\ketbrat}[1]{\lvert #1 \rangle \! \langle #1 \rvert}
\newcommand{\norm}[1]{\left\lVert#1\right\rVert}

\begin{document}

\title{Local minima in quantum systems}
\author[1,4]{Chi-Fang Anthony Chen}
\author[1,2,3]{Hsin-Yuan Huang}
\author[1,4]{John Preskill}
\author[1]{Leo Zhou}
\affil[1]{California Institute of Technology}
\affil[2]{Google Quantum AI}
\affil[3]{Massachusetts Institute of Technology}
\affil[4]{AWS Center for Quantum Computing}

\date{September 28, 2023}

\maketitle

\begin{abstract}
\normalsize
Finding ground states of quantum many-body systems is known to be hard for both classical and quantum computers. As a result, when Nature cools a quantum system in a low-temperature thermal bath, the ground state cannot always be found efficiently. Instead, Nature finds a local minimum of the energy. In this work, we study the problem of finding local minima in quantum systems under thermal perturbations. While local minima are much easier to find than ground states, we show that finding a local minimum is computationally hard for classical computers, even when the task is to output a single-qubit observable at any local minimum. In contrast, we prove that a quantum computer can always find a local minimum efficiently using a \emph{thermal gradient descent} algorithm that mimics the cooling process in Nature. To establish the classical hardness of finding local minima, we consider a family of two-dimensional Hamiltonians such that any problem solvable by polynomial-time quantum algorithms can be reduced to finding ground states of these Hamiltonians. We prove that for such Hamiltonians, all local minima are global minima. Therefore, assuming quantum computation is more powerful than classical computation, finding local minima is classically hard and quantumly easy.
\end{abstract}

\thispagestyle{empty}
\clearpage
\thispagestyle{empty}
\tableofcontents
\thispagestyle{empty}
\clearpage
\pagenumbering{arabic} 
\setcounter{page}{1}

\section{Introduction}

Finding ground states and other low-energy states of quantum many-body systems is a central problem in physics, materials science, and chemistry.
To address this problem, many powerful computational methods, such as density functional theory (DFT) \cite{HohenbergKohn, NobelKohn}, quantum Monte Carlo (QMC)~\cite{CEPERLEY555,SandvikSSE, becca_sorella_2017},
variational optimization with tensor network ansatzes~\cite{DMRG1,DMRG2, Garcia07,Vestraete2008,SCHOLLWOCK201196, Haghshenas2019, Hyatt2019} or neural network ansatzes \cite{Carleo_2017, hibat2020recurrent, Deng2017}, and data-driven machine learning approaches~\cite{gilmer2017neural,qiao2020orbnet, huang2021provably,lewis2023improved}, have been developed.
These methods work well for many physically relevant problem instances but fail badly in other cases.
One hopes that scalable fault-tolerant quantum computers will be able to solve a broader array of problem instances, but finding ground states of local Hamiltonians is known to be $\mathsf{QMA}$-hard~\cite{KSV02, kempe2006complexity}, and therefore is expected to be intractable even for quantum computers in some instances.
Indeed, the efficacy of existing quantum algorithms requires additional assumptions that are yet to be justified~\cite{lee2022there}, such as the presence of a trial state with sufficient ground state overlap~\cite{lin2022heisenberg,gharibian2022dequantizing} or a parameterized adiabatic path whose spectral gap remains open~\cite{farhi2000quantum}.

Under the widely accepted conjecture that Nature can be efficiently simulated on a quantum computer, the hardness of finding ground states on quantum computers implies that Nature cannot find ground states in general.
When a quantum system with Hamiltonian $\vH$ is placed in a low-temperature thermal bath, the system seeks a local minimum of the energy, which may not be the ground state of $\vH$.
For some physical systems, such as spin glasses \cite{edwards1975theory, kirkpatrick1978infinite, binder1986spin, mydosh1993spin}, finding a ground state is indeed known to be computationally hard; such systems, when cooled, almost always find a local minimum instead of the ground state. In these cases, the ground state of the Hamiltonian is physically irrelevant in that it is never observed in experiments.

Motivated by this perspective, in this work we study the problem of finding local minima in quantum many-body systems.
For concreteness, we consider an $n$-qubit system governed by a local Hamiltonian $\vH$.
The central question we are interested in is:
\begin{center}
    \emph{How tractable is the problem of finding local minima of the energy}\\
    \emph{in quantum systems using classical and quantum computers?}
\end{center}
To begin to answer this question, we need a mathematical definition of local minima in quantum systems.
Based on the standard definition in mathematical optimization \cite{pardalos1991quadratic, boyd2004convex,  agarwal2017finding, jin2018accelerated, ahmadi2022complexity}, we consider a local minimum in a quantum system governed by Hamiltonian $\vH$ to be a quantum state such that the expectation value of $\vH$ does not decrease under any small perturbation applied to the state.
The local minima of $\vH$ form a subset of the entire quantum state space, which contains the global minima, the ground states of $\vH$.
We will consider two definitions of perturbations for defining local minima. The first one is, in a sense, mathematically natural but turns out to be inadequate for reasons we will explain. The second one is well-motivated physically and turns out to have interesting properties which we will explore.

The first definition of perturbations we study in this work is local unitary perturbations, which can be viewed as short-time unitary evolution governed by a sum of few-body Hermitian operators, as might arise in an adaptive variational quantum eigensolver (VQE) \cite{o2016scalable, grimsley2019adaptive, cerezo2021variational}.
A drawback of this definition is that finding a local minimum becomes so easy that even a classical computer can solve it efficiently.
We prove that a random $n$-qubit pure state is almost always a local minimum of $\vH$ under local unitary perturbations. Hence, there are $\exp(\exp(\Omega(n)))$ many local minima that are not global minima in the energy landscape. Because the number of local minima is enormous, finding a local minimum under this definition is \emph{classically easy}.
While local unitary perturbations are natural from a mathematical perspective, they are not physically motivated since the evolution of a quantum system interacting with a low-temperature thermal bath is governed by \emph{quantum thermodynamics} and is inherently nonunitary.

\begin{figure}
    \centering
    \includegraphics[width=0.95\textwidth]{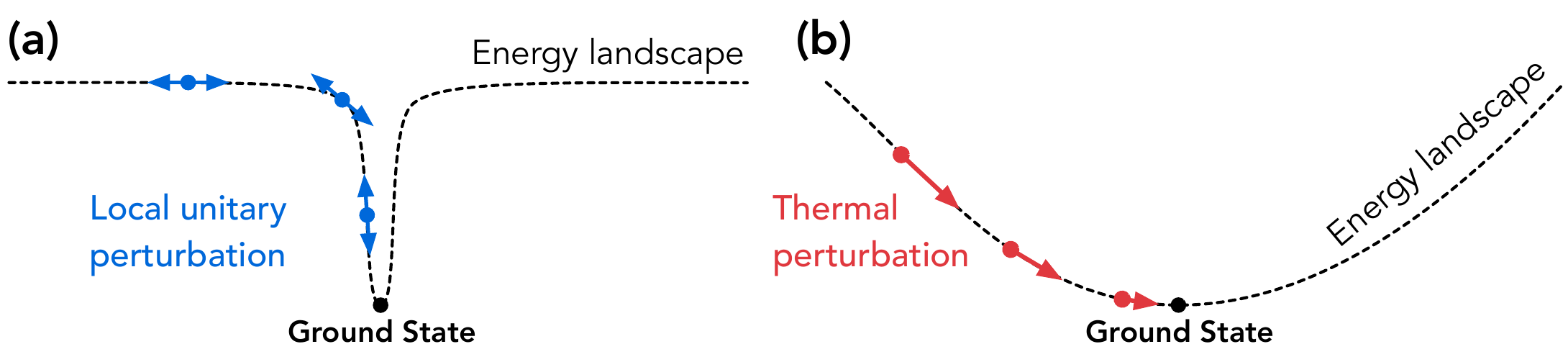}
    \caption{\emph{(a) Energy landscape under local unitary perturbations.} For any local Hamiltonian $\vH$, there will be doubly exponentially many local minima within the $n$-qubit state space that stems from a large barren plateau.
    \emph{(b) Energy landscape under thermal perturbations.} For some local Hamiltonians, such as a family of $\mathsf{BQP}$-hard Hamiltonians, the energy landscape over the entire $n$-qubit state space has a nice bowl shape, and the only local minimum is the global minimum. However, for $\mathsf{QMA}$-hard Hamiltonians, the energy landscape necessarily contains many suboptimal local minima.
    Local unitary perturbations are reversible, while thermal perturbations are irreversible.
    }
    \label{fig:landscape}
\end{figure}

Our second definition is inspired by how quantum systems actually seek out local minima in Nature.
Under suitable physical assumptions\footnote{Typical assumptions are that the system-bath coupling is weak and the thermal bath is memoryless.}, perturbations induced by a thermal bath are represented by a master equation defined by a linear combination of \textit{thermal Lindbladians} $\CL_a$, each associated with a local system-bath interaction $\vA^a$ \cite{lindblad1976generators, davies1979generators, breuer2002theory}.
In its modern formulation~\cite{mozgunov2020completely, Chen2023quantumthermal}, the thermal Lindbladian $\CL_a$ depends on the system Hamiltonian $\vH$ and two macroscopic bath quantities: the inverse temperature $\beta$ and a characteristic time scale $\tau$.
We prove two fundamental results concerning the problem of finding local minima under thermal perturbations.
We prove that a quantum computer can efficiently find a local minimum under thermal perturbations using a proposed \emph{quantum thermal gradient descent algorithm} that mimics Nature's cooling process.
And in stark contrast to the definition of a local minimum based on local unitary perturbations, we prove that finding local minima under thermal perturbations is universal for quantum computation and, hence, is \emph{classically hard} under the standard assumption $\mathsf{BPP} \neq \mathsf{BQP}$.

To establish the classical hardness of finding local minima under thermal perturbations, we consider geometrically local Hamiltonians on a 2D lattice, such that the ground state encodes the outcome of any efficient quantum computation using a modified version of Kitaev's circuit-to-Hamiltonian construction \cite{KSV02, OliveiraTerhal,AharonovAQCUniversal}.
The most technically involved result of this work is a theorem stating that for these 2D Hamiltonians, all local minima under low-temperature thermal perturbations are global minima, i.e., ground states. That is, the energy landscape for these Hamiltonians has a nice bowl shape over the entire $n$-qubit state space such that quantum gradient descent efficiently finds the ground state. Meanwhile, if a classical computer can efficiently find any local minima under thermal perturbations, then the classical computer can efficiently simulate quantum computation, which is widely believed to be impossible. To prove the theorem, we develop a set of techniques for establishing that a Hamiltonian $\vH$ has \emph{no suboptimal local minima}, i.e., all local minima of $\vH$ are global minima.

We conclude that local minima under thermal perturbations are, in general, hard to find classically but easy to find on a quantum computer. Hence, the local minima problem provides a quantumly tractable alternative to the ground state problem, which is believed to be hard for both classical and quantum computers.
Since ground states of quantum systems are frequently encountered in the laboratory, one wonders whether generic quantum many-body systems relax to their ground states efficiently when cooled because these systems have no suboptimal local minima, similar to the situation in convex optimization \cite{boyd2004convex}.
Exploring the shape of the energy landscape of Hamiltonians arising in physics, chemistry, and materials science may suggest new opportunities for solving classically intractable and physically relevant problems using quantum computers.

\section{Results}
\label{sec:results-main}
    
We now present our main results concerning the tractability of finding local minima in quantum systems.
The results are organized into the complexity of finding local minima under local unitary perturbations and under thermal perturbations.
A collection of notational conventions and some background on thermal Lindbladians can be found in Appendix~\ref{sec:recap_notation}.

We define local minima in quantum systems by generalizing a definition commonly used in classical optimization; see a brief review of classical optimization in Appendix~\ref{sec:local-minima-classical}.
Let $\mathcal{P}_{\balpha}$ be a perturbation parameterized by a small vector $\balpha$ that maps quantum states to quantum states.
An $\epsilon$-approximate local minimum of an $n$-qubit Hamiltonian $\vH$ under perturbation $\mathcal{P}$ is a state $\vrho$ with an energy $\Tr(\vH \vrho)$ that is an approximate minimum under perturbations, i.e.,
\begin{equation}\label{eq:epsilon-approximate}
    \Tr(\vH \vrho) \leq \Tr(\vH \mathcal{P}_{\balpha}(\vrho)) + \epsilon \norm{\balpha}
\end{equation}
for all small enough $\balpha$.
The formal definition is given in Appendix~\ref{sec:local-minima-quantum}.
We say an algorithm $\mathcal{A}$ has solved the problem of \textit{finding} local minima under perturbation $\mathcal{P}$ if given any $n$-qubit Hamiltonian $\vH$, written as a sum of few-qubit Hermitian operators, and any few-qubit observable $\vO$, the algorithm $\mathcal{A}$ can output a real value $\Tr(\vO \vrho)$ corresponding to any approximate local minimum $\vrho$ of $\vH$ under perturbations $\mathcal{P}$ up to a small error.\footnote{Since there could be multiple local minima and we consider finding one instance to be sufficient, this problem is closer to a \textit{relational} problem than to a \textit{decision} problem. }

\subsection{Local minima under local unitary perturbations}

We first study local minima under local unitary perturbations.
Local unitary perturbations are short-time unitary evolutions under a sum of few-body Hermitian operators.
A quantum circuit consisting of near-identity two-qubit gates induces a local unitary perturbation.
Consider an $n$-qubit pure state $\ket{\psi}$. A local unitary perturbation of $\ket{\psi}$ is given by
\begin{equation}
\text{(local unitary perturbation):} \quad \quad \ket{\psi} \rightarrow \exp\left( -\ri \sum_{a=1}^m \alpha_a \vh^a \right) \ket{\psi},
\end{equation}
where $\vh^a$ is a Hermitian operator acting on a few qubits, $m = \mathrm{poly}(n)$ is the number of such Hermitian operators, and $\balpha = \sum_{a} \alpha_a \hat{\be}_a \in \BR^{m}$ is a vector close to zero.
This definition is inspired by adaptive variational quantum eigensolvers \cite{o2016scalable, grimsley2019adaptive, cerezo2021variational}, and is the state version of the Riemannian geometry of quantum computation defined in \cite{nielsen2006quantum}.
When one variationally minimizes the energy by applying unitary gates, one finds a local minimum under local unitary perturbations.

To understand how easy the problem of finding local minima under local unitary perturbations is, we need to characterize the energy landscape.
The following lemma provides a universal characterization of the structure of the energy landscape under the geometry defined by local unitary perturbations for any local Hamiltonian $\vH$.
The formal statement is given in Lemma~\ref{lem:Haar-concentration}, and the proof is given in Appendix~\ref{sec:barren-plateau-energy-unitary}.
\begin{lemma}[Barren plateau; informal] \label{lem:barren-plateau-random-state-informal}
Given any $n$-qubit local Hamiltonian $\vH$. A random pure $n$-qubit state $\ket{\psi}$ is an approximate local minimum of $\vH$ under local unitary perturbations.
\end{lemma}
\noindent Furthermore, the proof of the above lemma illustrates the following physical picture: the energy landscape in the pure state space defined based on local unitary perturbations consists of a large barren plateau \cite{mcclean2018barren} with doubly-exponentially many approximate local minima having exponentially small energy gradient.
Additionally, almost all of the local minima have local properties that are exponentially close to that of the maximally mixed state.
As a result, while finding ground states is classically hard, finding local minima under local unitary perturbations is classically trivial.

\begin{theorem}[Classically easy to find local minima under local unitary perturbations; informal] \label{thm:classical-easy-local-unitary-informal}
The problem of finding approximate local minima of $n$-qubit local Hamiltonian $\vH$ under local unitary perturbations is classically easy.
\end{theorem}
\noindent
See Theorem~\ref{thm:classical-easy-local-unitary} for a more detailed statement, and Appendix~\ref{sec:main-unitary} for its proof.

The presence of barren plateaus in the energy landscape under local unitary perturbations causes the problem of finding local minima to be classically easy.
However, a definition of local minima based on local unitary perturbation is not physically well motivated since Nature cools a physical system via open-system dynamics by coupling to a thermal bath rather than by unitary dynamics.

\subsection{Local minima under thermal perturbations}

In this section, we consider local minima under thermal perturbations induced by a heat bath, formally defined in Appendix~\ref{sec:local-minima-quantum}.
We will show that the classical hardness of finding local minima under thermal perturbations is much different than the classical hardness under local unitary perturbations.

When the coupling between an $n$-qubit system and a thermal bath is weak, and the bath is memoryless, the complicated joint system-bath Hamiltonian dynamics reduces to a Markovian Lindbladian evolution of the system alone, $\vrho(t) = \e^{\CL t}[\vrho]$.
Remarkably, this continuous time generator $\CL$ can be defined by merely the system Hamiltonian $\vH$, the \textit{jump operators} $\vA^a$ through which the bath interacts with the system, and thermodynamic quantities of the bath: inverse temperature $\beta$ and a characteristic time-scale $\tau$. See Appendix~\ref{sec:thermo-lindblad} for an introduction and Appendix~\ref{sec:thermo-lindblad-detail} for an in-depth discussion.
Under these assumptions, we may effectively consider a thermal perturbation of $n$-qubit state $\vrho$ to be
\begin{equation}
\text{(thermal perturbation):} \quad \quad \vrho \rightarrow \exp\left( \sum_{a=1}^m \alpha_a \CL^{\beta, \tau, \vH}_{a} \right) (\vrho), \label{eq:thermal-perturb}
\end{equation}
where $\CL^{\beta, \tau, \vH}_{a}$ is the thermal Lindbladian associated with each jump operator $\vA^a$ acting on a few qubits, $m = \mathrm{poly}(n)$ is the number of jump operators, and $\balpha = \sum_{a} \alpha_a \hat{\be}_a \in \BR^{m}_{\geq 0}$ is a \emph{nonnegative} vector close to zero.
Here, the vector is nonnegative because thermodynamic processes are generally irreversible.
The irreversibility in thermal perturbations is crucial to ensure that there are fewer than doubly-exponentially many local minima in the energy landscape; see a discussion in Appendix~\ref{sec:irreversible-perturbations}.

We may define a local minimum under thermal perturbations to be a state $\vrho$ with the minimum energy $\Tr(\vH \vrho)$ under thermal perturbations given in Eq.~\eqref{eq:thermal-perturb}. More precisely, we will consider $\epsilon$-approximate local minima as in  Eq.~\eqref{eq:epsilon-approximate}.
A central concept that enables us to understand the energy landscape and establishes the computational complexity of finding local minima under thermal perturbations is the \textit{energy gradient operator},
\begin{equation}
	\text{(energy gradient operator):} \quad\quad \sum_{a=1}^m \CL^{\dag \beta, \tau, \vH}_{a}(\vH) \hat{\be}_a,
\end{equation}
where the adjoint $\CL^{\dag}$ is the Heisenberg-picture Lindbladian, i.e., $\Tr( \CL^{\dag}[\vO] \vrho ) = \Tr( \vO \CL[\vrho])$.
The energy gradient operator is a vector of individual gradient operators\footnote{This is similar to the spin operator $\vec{\sigma} = \sigma^x \hat{x} + \sigma^y \hat{y} + \sigma^z \hat{z}$, which is a vector of Hermitian observables.} associated with each jump operator $\vA^a$.
Indeed, the energy gradient operator naturally emerges by taking an infinitesimal perturbation, i.e., the gradient of the energy $\Tr(\vH \vrho)$,
\begin{equation}
    \Tr\left(\vH \exp\left(\sum_{a=1}^m \alpha_a \CL^{\beta, \tau, \vH}_{a} \right) (\vrho) \right) = \Tr(\vH \vrho) + \balpha \cdot \sum_{a=1}^m \Tr\left(\CL^{\dag \beta, \tau, \vH}_{a}(\vH) \vrho\right) \hat{\be}_a + \mathcal{O}(\norm{\balpha}^2).
\end{equation}
In Appendix~\ref{sec:local-minima-thermal-prop}, we describe the formal definition and some properties of the energy gradient.
We provide a concrete example showing the sets of local minima for ferromagnetic Ising chains under different longitudinal field strengths in Appendix~\ref{sec:example-Ising}.

Given the definition of thermal perturbations, we next study how tractable is the problem of finding a local minimum under thermal perturbations.
In stark contrast to finding local minima under local unitary perturbations, which is classically easy, our complexity-theoretic results show that finding local minima under thermal perturbations is both quantumly easy (Section~\ref{sec:quantumly_easy}) and classically hard (Section~\ref{sec:classically_hard}) if we assume the well-accepted conjecture that not all quantum circuits can be efficiently simulated on classical computers ($\mathsf{BPP} \neq \mathsf{BQP}$).

\subsubsection{Finding local minima is easy for quantum computers}\label{sec:quantumly_easy}

In practice, quantum systems find local minima easily when coupled to a cold thermal bath. Therefore, if our definition of a local minimum properly captures how a quantum system behaves in a cold environment, we expect finding local minima to be quantumly easy.
Indeed, in the following theorem, we prove that a quantum computer can always efficiently find a local minimum of $\vH$ under thermal perturbations starting from any initial state.

\begin{theorem}[Quantumly easy to find local minima under thermal perturbations; informal] \label{thm:quantum-easy-thermal-informal}
The problem of finding an $\epsilon$-approximate local minimum of an $n$-qubit local Hamiltonian $\vH$ under thermal perturbations with inverse temperature $\beta$ and time scale $\tau$ can be solved in $\mathrm{poly}(n, 1/\epsilon, \beta, \tau)$ quantum computational time.
\end{theorem}
The formal restatement is given in Theorem~\ref{thm:quantum-easy-thermal} and is proven in Appendix~\ref{sec:proof-thm-quantum-easy-thermal}. To establish the theorem, we propose a \emph{quantum thermal gradient descent algorithm} based on the energy gradient operator.
Gradient descent is necessary when the inverse temperature $\beta$ and time scale $\tau$ are not infinite.
When $\beta = \tau = \infty$, the energy gradient $\CL_a^{\dag \infty, \infty, \vH} (\vH)\preceq 0$ is nonpositive.
In this case, the algorithm can just perform a random walk along random directions because no perturbations increase energy.
But when $\beta$ and $\tau$ are finite, the energy gradient can be positive.
To find a local minimum that is a minimum under all thermal perturbations, the algorithm needs to carefully walk in directions with negative energy gradients.

To prove the convergence of quantum thermal gradient descent, we show that every small gradient step decreases the energy.
To establish this claim, we derive analytic properties of thermal Lindbladians based on a smoothness bound on the second derivatives in~\cite{Chen2023quantumthermal}.
To implement a gradient step based on thermal perturbations, we build on a recently developed efficient quantum algorithm that simulates thermal Lindbladian evolution using a quantum circuit augmented by mid-circuit measurements~\cite{Chen2023quantumthermal}.

\subsubsection{Finding local minima is hard for classical computers}\label{sec:classically_hard}

Given that finding local minima under local unitary perturbations is classically trivial, it is natural to wonder whether finding local minima under thermal perturbations is also classically easy. What does the corresponding energy landscape look like? And what computational problems can be solved using quantum thermal gradient descent?
As our second main result, we address these questions for a class of geometrically local Hamiltonians $\{ \vH_C \}$ on two-dimensional lattices, where the ground state encodes the output of quantum circuit $C$.

\begin{theorem}[No suboptimal local minimum in $\mathsf{BQP}$-hard Hamiltonians; informal] \label{thm:no-suboptimal-local-minima-informal}
For any quantum circuit $C$ with size $\labs{C}$, all approximate local minima of the geometrically local 2D Hamiltonian $\vH_C$ under thermal perturbations with inverse temperature $\beta=\poly(|C|)$ and time scale $\tau=\poly(|C|)$ are close to the ground state.
\end{theorem}
\noindent This theorem is the most technically involved contribution of this work. The formal statement is given in Theorem~\ref{thm:no-suboptimal-local-minima-informal2} and is proven in Appendix~\ref{sec:universal-quantum-computation}. Conceptually, the landscape of these 2D Hamiltonians has a nice \emph{bowl shape}, like in convex optimization \cite{boyd2004convex}.
Therefore, performing thermal gradient descent (Theorem~\ref{thm:quantum-easy-thermal-informal}) allows us to prepare the ground state starting from an \textit{arbitrary} initial state.
For a choice of inverse temperature that grows polynomially with $|C|$, thermal fluctuations in the cooling process do not kill the power of quantum computation.

As a consequence of this energy landscape characterization, we can show that finding a local minimum under thermal perturbations is classically intractable, assuming quantum computation is more powerful than classical computation.
See Theorem~\ref{thm:classical-hard-LM-thermal} for a formal restatement and the proof.

\begin{theorem}[Classically hard to find local minima under thermal perturbations; informal] \label{thm:classical-hard-LM-thermal-informal}
Assume the widely believed conjecture that $\mathsf{BPP} \neq \mathsf{BQP}$.
The problem of finding an approximate local minimum of an $n$-qubit local Hamiltonian $\vH$ under thermal perturbations is universal for quantum computation and is thus classically hard.
\end{theorem}

There have been other proposals for solving $\mathsf{BQP}$-hard problems by finding suitable quantum states, such as designing a gapped adiabatic path for Hamiltonians to find ground states \cite{AharonovAQCUniversal}, engineering Lindbladians to have rapid dissipative evolution towards steady states \cite{verstraete2009quantum} and performing quantum phase estimation on an initial state with high ground-state overlap \cite{gharibian2022dequantizing}. 
These approaches draw inspiration from physics to motivate algorithms for solving problems on analog and digital quantum devices but do not emulate naturally occurring physical processes. In contrast, the problem of finding a local minimum is motivated by ubiquitous physical processes in Nature that produce the low-energy states studied in physics, chemistry, and materials science.
Furthermore, the local minima problem enjoys the robustness of thermodynamics: one merely needs to specify macroscopic bath quantities $\beta$ and $\tau$ without worrying about microscopic details, and the choice of jump operators can be flexible since adding more jumps (even unwanted ones) only \textit{improves} the gradient and \textit{removes} suboptimal local minima.

We now highlight the proof idea for Theorem~\ref{thm:no-suboptimal-local-minima-informal} as follows. We consider a family of geometrically local $n$-qubit Hamiltonians
$\{\vH_C\}$ in a 2D lattice defined by modifying Kitaev's circuit-to-Hamiltonian construction \cite{OliveiraTerhal, KSV02} where the ground state encodes the computation of a quantum circuit $\vU_C = \vU_T \ldots \vU_1$. In particular, we design the ground state of $\vH_C$ to be
\begin{equation}
    \sum_{t=0}^T \sqrt{\xi_t}\big(\vU_t \cdots \vU_1 \ket{0^n}\big) \otimes \ket{0^t 1^{T-t}},
    \qquad
    \text{where} \quad
    \xi_t := \frac{1}{2^T}  \binom{T}{t}.
\end{equation}
The binomial coefficient $\xi_t$ is our modification of Kitaev's construction and is chosen to ensure that desired properties hold for the spectrum and the energy gradients.\footnote{The binomial distribution ensures the Bohr-frequency gap is sufficiently large, which is central to the robustness of energy gradients under errors due to finite temperature and small perturbations.
We believe that the standard circuit-to-Hamiltonian construction also has a large Bohr-frequency gap, but the proof seems more difficult.}
Therefore, estimating local properties of the ground state of $\vH_C$ is equivalent to simulating the quantum circuit $C$, which is $\mathsf{BQP}$-hard.

Given the Hamiltonian $\vH_C$, the central challenge is to show that all of its approximate local minima under thermal perturbations are also approximate global minima.
This seems daunting to study due to the complex expression for the thermal Lindbladian $\CL_a^{\beta,\tau,\vH}$ and the doubly exponentially large space of possible quantum states.
Previous studies on circuit-to-Hamiltonian mappings mainly focused on the lowest energy states.
Here, we need to worry about potential local minima in all excited states in any superposition.
To make progress, we propose a sufficient condition in Appendix~\ref{sec:Hamil-suboptimal-local-minima} that captures the nice landscape of $\vH_C$ and rules out the presence of \textit{any} suboptimal local minimum.
Let $\vP_G(\vH)$ be the projector onto the ground state space of $\vH$. Assume there exists a unit vector $\hat{\balpha} \in \BR^m_{\geq 0}$ and $r > 0$ with
\vspace{-5pt}
\begin{equation}
\text{(negative gradient condition):} \quad\quad - \sum_{a=1}^m \hat{\alpha}_a \CL^{\dag \beta, \tau, \vH}_{a}(\vH) \succeq r(\vI - \vP_G(\vH)).
\end{equation}
This negative gradient condition implies that any state with a small ground state overlap must experience a substantially negative energy gradient, i.e., it must not be a local minimum.

To prove that $\vH_C$ satisfies the negative gradient condition, we propose a series of lemmas and mathematical techniques for characterizing energy gradients in few-qubit systems, in commuting Hamiltonians, and in subspaces of the Hamiltonian, which are stated in Appendix~\ref{sec:characterize-neg-grad-condition} and proven in Appendix~\ref{sec:monotone_gradient}.
These new techniques build on the operator Fourier transform, and the secular approximation given in \cite{Chen2023quantumthermal} for systematically handling energy uncertainty in thermal Lindbladians, which we review and adapt for our purpose in Appendix~\ref{sec:OFT}.
Using these techniques, we analyze the energy gradient of the entire system perturbatively by considering a sequence of Hamiltonians
\begin{align*}
    \vH_1\rightarrow \vH_2\rightarrow \vH_3 = \vH_C \quad &\text{with refining ground spaces}\quad \vP_1 \supset \vP_2 \supset \vP_3,\\
    &\text{where}\quad \norm{\vH_1} \gg \norm{\vH_2-\vH_1}\gg \norm{\vH_3-\vH_2}.
\end{align*}
Through these perturbations, we sequentially rule out local minima in excited states of the Hamiltonian $\vH_1, \vH_2$ and, finally, $\vH_3 = \vH_C$.
For example, we show the first Hamiltonian $\vH_1$ satisfies the negative gradient condition and that the gradient is \textit{stable} under perturbation going from $\vH_1 \rightarrow \vH_2 \rightarrow \vH_3$.
Controlling perturbations of the energy gradient is surprisingly challenging, and it is not \emph{a priori} clear why this stability property should hold due to multiple (possibly competing) energy scales, including $\beta^{-1},\tau^{-1}$, the spectral gap, and the Bohr-frequency gap.\footnote{Recall that spectral gap is the minimum non-zero difference between energy eigenvalues. Bohr-frequency gap is the minimum non-zero difference between the difference of energy eigenvalues.} The perturbative errors are not suppressed by the spectral gap of the Hamiltonian as seen in standard settings, but instead by the Bohr-frequency gap, which can be much smaller (see \cref{thm:mono_gradient_full}).
These techniques allow us to establish the robustness of energy gradients when perturbing a degenerate Hamiltonian with a sufficiently large Bohr-frequency gap.

We emphasize that while we proved that $\vH_C$ has no suboptimal local minima when $C$ is a polynomial-size quantum circuit, the same is not true for general local Hamiltonians. Finding the ground state of a local Hamiltonian is a $\mathsf{QMA}$-hard problem; hence, we do not expect it to be solved efficiently by the quantum thermal gradient descent algorithm or by any other quantum algorithm.
In the case of a quantum circuit that verifies the witness for a problem in $\mathsf{QMA}$, Kitaev's corresponding local Hamiltonian contains a term, often denoted $\vH_{\rm in}$, which specifies some of the input qubits and leaves the input qubits corresponding to the witness unspecified, and a term, often denoted $\vH_{\rm out}$, which checks whether the witness is accepted.
Due to the unspecified witness qubits in $\vH_{\rm in}$, the energy landscape contains a significant number of local minima corresponding to all possible witnesses.
Furthermore, most of these local minima correspond to rejected witnesses and are suboptimal because of the energy penalty from $\vH_{\rm out}$.
For these $\mathsf{QMA}$-complete Hamiltonians, quantum thermal gradient descent is likely to remain stuck for a long time at a suboptimal local minimum.
In $\vH_C$, the term $\vH_{\rm in}$ specifies all input qubits, and the term $\vH_{\rm out}$ is absent, which greatly simplifies the energy landscape, enabling quantum thermal gradient descent to find the global minimum efficiently.

\vspace{-5pt}
\section{Discussion}
We have good reasons for believing that scalable fault-tolerant quantum computers will be more powerful than classical computers, but for what problems of practical interest should we expect a superpolynomial quantum advantage? Quantum computers might substantially speed up the task of characterizing properties of ground states for some local Hamiltonians that arise in physics, chemistry, and materials science, but it is not clear how to identify particular problems for which such speedups occur~\cite{lee2023evaluating}. In some cases, classical methods provide good solutions, while in other cases, the problem is hard even for quantum computers. 

Here we have focused on an easier problem, namely finding local minima rather than global minima of a Hamiltonian. This problem is very well motivated physically because the task of finding a local minimum under thermal perturbations is routinely carried out by actual physical systems when in contact with a cold thermal bath. We showed that this problem is solved efficiently by a proposed quantum optimization algorithm, the \emph{quantum thermal gradient descent algorithm}. Furthermore, we showed that finding a local minimum is classically hard in general (assuming that $\mathsf{BPP}\neq \mathsf{BQP}$). Hence, the local minimum problem is a quantumly tractable alternative to the ground state problem for which superpolynomial quantum advantage can be achieved for some problem instances.  

Our main results pertain to perturbations that arise in quantum thermodynamics~\cite{lindblad1976generators, davies1979generators, breuer2002theory, mozgunov2020completely, Chen2023quantumthermal}. We noted that the energy landscape under such thermal perturbations is much nicer than the energy landscape encountered by quantum optimization algorithms relying on local unitary perturbations such as VQE~\cite{o2016scalable, grimsley2019adaptive, cerezo2021variational}; see Theorems~\ref{thm:classical-easy-local-unitary-informal}~and~\ref{thm:no-suboptimal-local-minima-informal}.
From an algorithmic design perspective, we are free to choose any perturbation. Indeed, we may modify the thermal Lindbladians to have nicer analytic properties or algorithmic costs~\cite{Chen2023quantumthermal}. While these synthetic Lindbladians may not simulate Nature, they constitute a broader class of \textit{Monte Carlo} quantum algorithms~\cite{temme2009QuantumMetropolis, Rall22thermal, Chen2023quantumthermal,ding2023single} that may improve upon Nature. Apart from Lindbladians, other families of perturbations, such as unitary perturbations accompanied by mid-circuit measurements and/or qubit resets, may also yield nice bowl-shaped energy landscapes without suboptimal local minima.
Progress on this question could lead to more efficient quantum optimization algorithms for finding low-energy states or for other applications.

There are a plethora of classical algorithms for minimizing energies of quantum systems based on classical variational ansatzes for quantum states, such as tensor networks \cite{DMRG1,DMRG2,Garcia07,Vestraete2008,SCHOLLWOCK201196,Jordan2008, Vanderstraeten2016, Corboz2016,landau2015polynomial, arad2017rigorous, abrahamsen2020polynomialtime,Stoudenmire12,Wu2019,Moore2020, Zalatel2020, Haghshenas2019, Hyatt2019} and neural network quantum states \cite{Carleo_2017, hibat2020recurrent, Deng2017, dassarma2017, Nomura2017, choo_fermionicnqs2020, Ferrari2019, Glasser2018, Choo2018, rbm_su2, morawetz2020, Luo2021}.
These classical algorithms find a local minimum within a family of states defined by the classical variational ansatz.
However, a local minimum of the energy among the set of states subject to the classical ansatz might not be a local minimum under thermal perturbations. If not, we can load the state found by the classical algorithm into a quantum computer and find a lower energy state by running the quantum thermal gradient descent algorithm.
A corollary of our main results states the following. 
A formal statement is given as Corollary~\ref{cor:quantum-adv-classical-ansatz}, and its proof is in Appendix~\ref{sec:main-thermal}.

\begin{corollary}[Quantum advantage in finding lower-energy state; informal]
    Assume that not all polynomial-size quantum circuits can be efficiently simulated on a classical computer.
    Then there are 2D geometrically local Hamiltonians such that given any classical ansatz that allows efficient estimation of single-qubit observables and an output state $\vrho^\#$ of any efficient classical algorithm that optimizes the classical ansatz, running quantum thermal gradient descent starting at $\vrho^\#$ will strictly lower the energy.
\end{corollary}

\noindent The point is that we have proved the existence of local Hamiltonians for which finding a local minimum is quantumly easy and classically hard. For any such Hamiltonian, any quantum state $\vrho^\#$ found by the efficient classical algorithm will not be a local minimum; therefore, quantum thermal gradient descent will be able to descend to a state with strictly lower energy, even with just one gradient step. Furthermore, in many cases, we can evaluate the energy gradient at the classically optimized state $\vrho^\#$ by executing an efficient classical computation. A negative energy gradient confirms that a quantum algorithm starting from $\vrho^\#$ could outperform the classical algorithm. 

Many other interesting and challenging questions remain open.
Theorem~\ref{thm:no-suboptimal-local-minima-informal} shows that there are no suboptimal local minima in $\mathsf{BQP}$-hard $n$-qubit Hamiltonians for inverse temperature $\beta = \mathrm{poly}(n)$.
Do there exist $\mathsf{BQP}$-hard Hamiltonians with no suboptimal local minimum even for constant temperature, i.e., $\beta = \mathcal{O}(1)$?
If so, quantum advantage can be achieved by simply coupling a quantum system to a heat bath at a sufficiently low but constant temperature.
Our conclusion that finding local minima under thermal perturbations is classically hard relied on the complexity-theoretic conjecture that $\mathsf{BPP} \neq \mathsf{BQP}$.
Can we prove unconditionally that finding local minima is hard for classical algorithms, perhaps within a black-box oracle model?
Sometimes, when a system performs a random walk over a large plateau of suboptimal local minima for a sufficiently long time, the system escapes the plateau and reaches the true ground state (see e.g., Case 1 in \cref{sec:example-Ising}). Could we characterize when ground states can be found efficiently despite having many suboptimal local minima?
We have shown that there is a quantum advantage in finding local minima of quantum systems. Might there also be a quantum advantage in finding better local minima in classical optimization problems under some variant of quantum thermal gradient descents?

While ground state problems are hard to solve in general, many experimentally observed quantum systems efficiently relax to their ground states when cooled.
This physical phenomenon suggests that perhaps many Hamiltonians of interest in physics, chemistry, and materials science have no suboptimal local minima.
We have shown in \cref{thm:no-suboptimal-local-minima-informal} that a particular family of \textsf{BQP}-hard Hamiltonians has no suboptimal local minima under thermal perturbation.
An important future goal is to characterize broader classes of Hamiltonians that have a similarly good energy landscape.
Our proposed negative gradient condition suffices to rule out suboptimal local minima (\cref{lem:suff-cond-no-suboptimal}), but checking this condition for a general Hamilton involves highly complex calculations. It would be helpful to develop more general-purpose and efficient methods to verify this property for specified physical Hamiltonians over spins, fermions, or bosons.
We hope the ideas and techniques presented here will yield a deeper understanding of the energy landscapes of quantum systems and point toward promising opportunities for achieving quantum advantage for physically relevant problems.

\vspace{0.5em}
\subsection*{Acknowledgments:}

{ The authors thank Anurag Anshu, Ryan Babbush, Fernando Brandao, Garnet Chan, Sitan Chen, Soonwon Choi, Jordan Cotler, Jarrod R. McClean, and Mehdi Soleimanifar for valuable input and inspiring discussions.
CFC is supported by the AWS Center for Quantum Computing internship.
HH is supported by a Google PhD fellowship and a MediaTek Research Young Scholarship.
HH acknowledges the visiting associate position at Massachusetts Institute of Technology.
LZ acknowledges funding from the Walter Burke Institute for Theoretical Physics at Caltech.
JP acknowledges support from the U.S. Department of Energy Office of Science, Office of Advanced Scientific Computing Research (DE-NA0003525, DE-SC0020290), the U.S. Department of Energy, Office of Science, National Quantum Information Science Research Centers, Quantum Systems Accelerator, and the National Science Foundation (PHY-1733907). The Institute for Quantum Information and Matter is an NSF Physics Frontiers Center.  }

\newpage
\vspace{4em}
\appendix

\addappheadtotoc

\noindent 
\textbf{\LARGE{}Appendices}
\vspace{0.5em}

\section{Notations and Preliminaries}\label{sec:recap_notation}

Before we begin stating and proving our results formally in the rest of the appendices, we present some notations  used throughout the paper.
We also give a brief review of key concepts in quantum information theory that we utilize in this work.

\subsection{Notations}

This section recapitulates notations, and the reader may skim through this and return as needed.
\begin{align*}
\vH &:= \sum_i E_i \ket{\psi_i}\bra{\psi_i} &\text{Hamiltonian and the eigendecomposition}\\
\text{Spec}(\vH) &:= \{ E_i \} & \text{the spectrum of the Hamiltonian}\\
\nu \in B(\vH)&:= \{ E_i - E_j \, | \, E_i, E_j \in \mathrm{Spec}(\vH) \} &\text{the set of Bohr frequencies}\\
\Delta_\nu(\vH) &:= \min \{|\nu_1-\nu_2|: \nu_1\neq \nu_2 \in B(\vH)\} &\text{the Bohr-frequency gap}\\
\vA(t)&:= \e^{\ri \vH t} \vA \e^{-\ri \vH t} & \text{Heisenberg evolution for operator $\vA$}\\
m: & & \text{the number of jump operators}\\
\{\vA^a\}_{a=1}^m: & &\text{the set of jump operators}\\
\vrho:& &\text{the density matrix}\\
\CL: && \text{a Lindbladian in the Schrodinger Picture}\\
\beta:& &\text{the inverse temperature}\\
\hat{\vA}(\omega)\equiv \hat{\vA}_{f}(\omega) &:= \frac{1}{\sqrt{2\pi}}\int_{-\infty}^{\infty} f(t) \e^{-\ri \omega t}\vA(t)\rd t& \text{Operator Fourier transform of $\vA$ under $f$}\\
f_{\tau}(t) &:= \frac{1}{\sqrt{\tau}} \cdot \indicator(\labs{t}\le \tau/2 ) & \text{the normalized window function with width $\tau$}\\
\hat{f}(\omega)&=\frac{1}{\sqrt{2\pi}}\int_{-\infty}^{\infty}\e^{-\ri\omega t} f(t)\mathrm{d}t & \text{Fourier transform of a scalar function $f(t)$}\\
\vA_\nu&:=\sum_{E_2 - E_1 = \nu } \vP_{E_2} \vA \vP_{E_1} &\text{operator $\vA$ at exact Bohr frequency $\nu$}\\
\vI: & &\text{the identity operator}\\
\norm{f}_p&:= (\int_{-\infty}^{\infty} \labs{f(t)}^p\rd t)^{1/p} &\text{the $p$-norm of a function}\\
	\norm{\vO}&:= \sup_{\ket{\psi},\ket{\phi}} \frac{\bra{\phi} \vO \ket{\psi}}{\norm{\ket{\psi}}\cdot \norm{\ket{\phi}}} \quad &\text{the operator norm of a matrix $\vO$}\\
  	\norm{\vO}_p&:= (\Tr \labs{\vO}^p)^{1/p}\quad&\text{the Schatten p-norm of a matrix $\vO$}\\
  \norm{\CL}_{p-p} &:= \sup_{\vO} \frac{\normp{\CL[\vO]}{p}}{\normp{\vO}{p}}\quad&\text{the induced $p-p$ norm of a superoperator $\CL$}
\end{align*}
We write scalars, functions and vectors in normal font, and natural constants $\e, \ri, \pi$ are particularly in Roman font. We write matrices in bold font $\vO$ and super-operators in curly font~$\CL$.

Furthermore, we define the indicator function $\indicator(S)$ which is 1 if the statement $S$ is true and 0 otherwise. 
For any orthogonal projector $\vP$, we denote $\vP^\perp=\vI-\vP$.
We say $\vA\stackrel{\verr}{\approx}\vB$ when $\|\vA-\vB\| \le \verr$.

To simplify the notation, we often drop $f$ as a subscript $\hat{\vA}_{f}(\omega) \equiv \hat{\vA}(\omega)$, by which we have chosen the window function $f(t)=f_{\tau}(t)$.

\begin{figure}[t]
	\centering
	\includegraphics[width=0.2\textwidth]{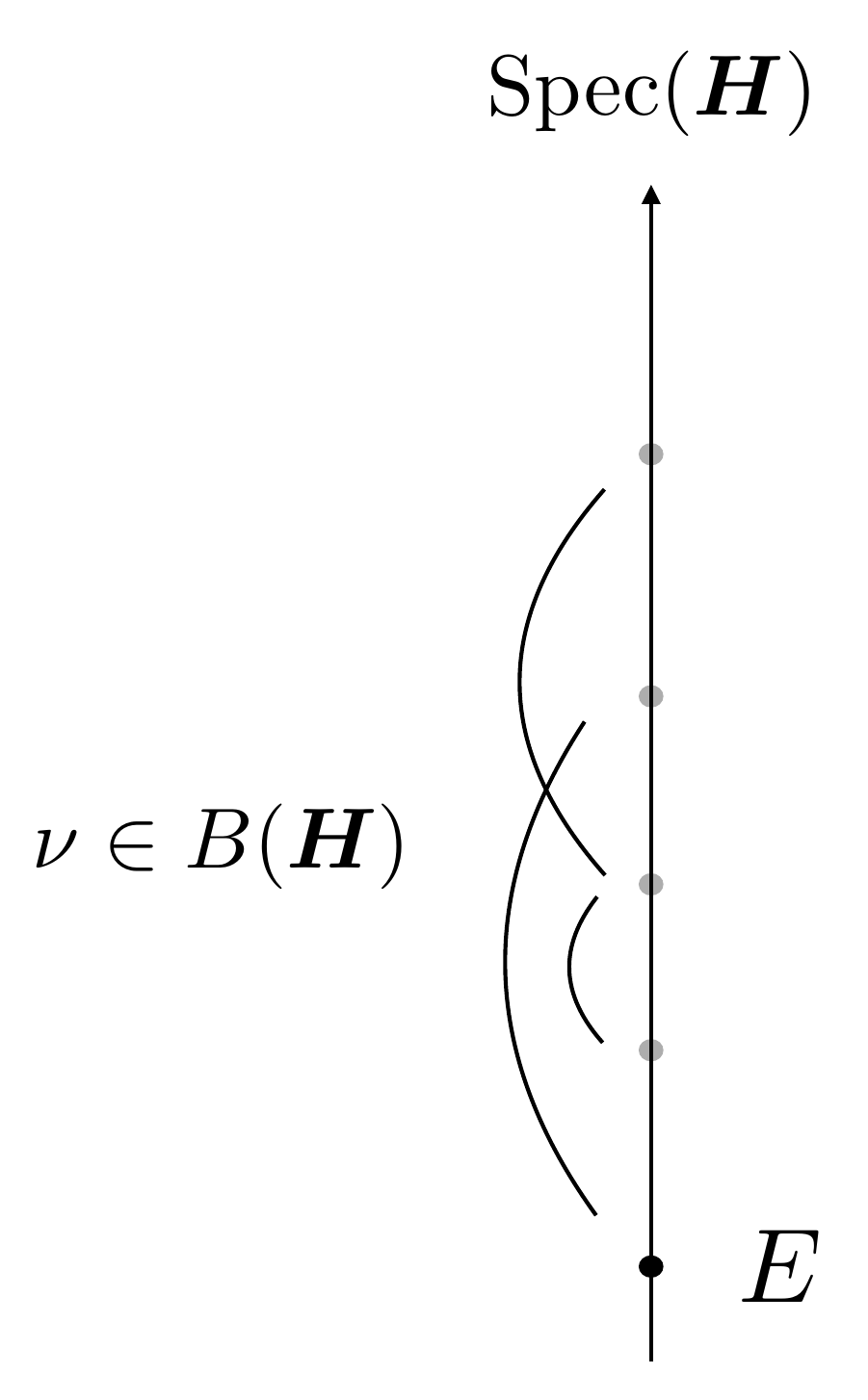}
	\caption{ The Bohr frequencies $\nu \in B(\vH) = \{ E_i - E_j \, | \, E_i, E_j \in \mathrm{Spec}(\vH) \}$ are the differences of energy (eigenvalues of the Hamiltonian $\vH$).
	}
	\label{fig:energybohr}
\end{figure}

\subsection{Lindbladians}

Completely Positive Trace-Preserving (CPTP) maps, also called quantum channels and quantum processes in the literature, correspond to all possible physical operations that could transform quantum states into other quantum states.
\textit{Lindbladian}s are infinitesimal generators of CPTP maps. That is, they map density operators to density operators (even if the map is tensored with the identity)
\begin{align}
    \CI\otimes \e^{\CL t}[\cdot] :\CS \rightarrow \CS \quad \text{for each} \quad t \ge 0.
\end{align}
In the Schrodinger Picture, a Lindbladian always has the following structure
\begin{align}
    \CL[\vrho] = \undersetbrace{\text{coherent term}}{-\ri [\vH, \vrho ]} + \sum_{j \in J} \bigg(\undersetbrace{\text{transition rate}}{\vL_j\vrho \vL_j^{\dag}} - \undersetbrace{\text{decay rate}}{\frac{1}{2} \{\vL_j^{\dag}\vL_j,\vrho\}}\bigg)
\end{align}
where the commutator is shorthanded by $[ \vA,\vB ] = \vA\vB - \vB\vA$ and the anti-commutator by $\{ \vA,\vB\} = \vA\vB+\vB\vA$. The operator $\vH$ can be any Hermitian matrix, and the set of \textit{Lindblad operators} $\{\vL_j\}_{j \in J}$ can be arbitrary as the second term always ensures trace-preserving.

\subsection{Thermal Lindbladians}
\label{sec:thermo-lindblad}
In this section, we describe the basic parameters that define a \textit{thermal Lindbladian}, i.e., Lindbladian originating from generic system-bath interactions under a Markovian, weak-coupling assumption~\cite{mozgunov2020completely}. 
Consider an $n$-qubit quantum system governed by a Hamiltonian $\vH$ and a heat bath with inverse temperature $\beta$ and time scale $\tau$. The bath interacts with the system via a set of local interaction terms
\textit{acting on the system} $\{\vA^1, \ldots, \vA^m\} = \{\vA^a\}_{a=1}^m$, where each operator $\vA^a$ acts on a constant number of qubits.
Each operator $\vA^a$ can be arbitrary ($\vA^a$ does not need to be Hermitian nor unitary), but  the set should be closed under Hermitian conjugate,
\begin{equation}
    \{ \vA^a \}_{a = 1}^m = \{ \vA^{a \dag} \}_{a = 1}^m.
\end{equation}
Each $\vA^a$ is referred to as a \emph{jump operator} and induces changes in energy (in the $n$-qubit system). For simplicity, we will enforce the following normalization for the interaction strengths,
\begin{align} \label{eq:normalization-vAa}
    \norm{\vA^{a \dag} \vA^a}_\infty \le 1 \quad \text{for each} \quad a = 1, \ldots, m.
\end{align}
For example, we may consider $m = 3n$ and $\vA^1, \ldots, \vA^m$ to be all single-qubit Pauli observables $\vX_i, \vY_i, \vZ_i$ for $i = 1, \ldots, n$, which have an interaction strength $\norm{\vA^{a \dag} \vA^a}_\infty = 1$.

The above parameters determine the thermal Lindbladian governing the equation of motion for the density operator, also referred to as the \emph{coarse-grained master equation}~\cite{mozgunov2020completely}
\begin{align} \label{eq:thermodynamics-rho}
    \frac{\rd\vrho}{\rd t} = -\ri [\vH, \vrho ] + \sum_{a=1}^m \alpha_a \mathcal{L}_{a}^{\beta, \tau, \vH}(\vrho),
\end{align}
The term $-\ri [\vH, \vrho ]$ corresponds to the Hamiltonian dynamics governed by the system Hamiltonian $\vH$, the (closed system) Schrodinger's equation. The effects of system bath interaction are captured by a weighted average of the thermal Lindbladian $\mathcal{L}_{a}^{\beta, \tau, \vH}$, defined by each local jump operator $\vA^a$, the Hamiltonian $\vH$, and parameters of the bath $\beta, \tau$.
The weighting is captured by the nonnegative vector $\balpha \in \mathbb{R}^m_{\geq 0}$.

By varying the $m$-dimensional nonnegative vector $\balpha \in \mathbb{R}^m_{\geq 0}$, the open system dynamics in Eq.~\eqref{eq:thermodynamics-rho} have the freedom to tune the interaction strengths for the jump operators.
Each $\alpha_a$ corresponds to the interaction strength of a jump operator $\vA^a$ and can be effectively absorbed into the set of jump operators by considering
\begin{equation}
    \left\{ \sqrt{\alpha_a} \vA^a \right\}_a.
\end{equation}
The interaction strength $\alpha_a \geq 0$ determines how much contribution each thermal Lindbladian $\mathcal{L}_{a}^{\beta, \tau, \vH}$ provides, and can be regarded as a \textit{probabilistic} mixture.
For example, if $\alpha_2$ is set to $0$, one removes the jump operator $\vA^2$ from the system-bath interaction.
This flexibility lets us study a (convex) set of thermal perturbations due to system-bath interaction by considering all $\balpha \in \mathbb{R}^m_{\geq 0}$.
As the system is weakly coupled to the bath, $\balpha$ is considered to be a vector with a small $\norm{\balpha}_1 = \sum_a \alpha_a$.

For each local interaction term $\vA^a$, the corresponding thermal Lindbladian $\mathcal{L}_{a}^{\beta, \tau, \vH}$ is an open system evolution with Lindblad jump operators $\{\hat{\vA}^a(\omega)\}_\omega$ for all possible energy differences $\omega \in (-\infty, \infty)$.
Each Lindblad jump operator $\hat{\vA}^a(\omega)$ is a restricted version of the system-bath interaction term $\vA^a$ that only contains transitions between eigenstates of $\vH$ whose associated eigenvalues, i.e., energies, differ by approximately $\omega$.
The inverse temperature $\beta$ sets the transition weight $\gamma_{\beta}(\omega)$, which determines the probability of occurrence for each Lindblad operator $\hat{\vA}^a(\omega)$.
For $\beta > 0$, the transition weight $\gamma_{\beta}(\omega)$ favors cooling $(\omega<0)$ over heating ($\omega>0$) transitions.
The timescale $\tau$ sets the resolution ($1/\tau$) at which $\hat{\vA}(\omega)$ identifies the energy differences between the eigenstates.
The exact form of thermal Lindbladians is relatively complex, so we defer further discussion to Appendix~\ref{sec:thermo-lindblad-detail} when needed for the full technical proof.

\section{Local minima in quantum systems}

In this appendix, we will introduce local minima in classical optimization, extend the definition to quantum systems, and formalize the problem of finding a local minimum in quantum systems.

\subsection{Local minima in classical optimization}
\label{sec:local-minima-classical}

In this subsection, we describe the definition of local minima in finite-dimensional Euclidean spaces, introduce a direct generalization to geometries with tangent spaces and exponential maps (such as circles and spheres), and discuss the concept of approximate local minima.

\subsubsection{Local minima in Euclidean space}

In classical optimization, one considers a real-valued function $h(\bx): \mathcal{X} \rightarrow \mathbb{R}$ over a domain $\mathcal{X} \subseteq \mathbb{R}^n$ consisting of $n$-dimensional vectors, and the goal is to find the global minimum of $h(\bx)$,
\begin{equation}
    \bx^* = \argmin_{\bx \in \mathcal{X}} h(\bx).
\end{equation}
Finding the global minimum is already NP-hard even when $h(\bx)$ is a quadratic function \cite{pardalos1991quadratic}.
Instead of finding a global minimum, one typically resorts to finding a local minimum $\bx^\#$, which is the minimum in a neighborhood around $\bx^\#$.
The definition of a local minimum $\bx^\#$ is that there exists a distance $\delta > 0$, such that
\begin{equation} \label{eq:exact-local-minima}
    h(\bx^\# + \balpha) \geq h(\bx^\#), \,\,\quad\text{for all}\quad \norm{\balpha} \leq \delta \quad \text{and}\quad \bx^\# + \balpha \in \mathcal{X}.
\end{equation}
Here the vector $\balpha$ is of the same dimension as $\bx^\#$.
We will refer to the above as an \emph{exact local minimum} because all points in the neighborhood have to be \textit{at least} $h(\bx^\#)$.
When there is an $\balpha$ such that $h(\bx^\# + \balpha)$ is only lower than $h(\bx^\#)$ by an extremely small value, $\bx^\#$ is still not an exact local minimum.
We will also define the approximate local minimum that relaxes this in Appendix~\ref{sec:approx-local-minima}.

\subsubsection{Local minima in general geometrical spaces}

The concept of a local minimum can be directly generalized to any geometry with tangent spaces and exponential maps, such as spheres, density matrices, unitaries, and more general Riemannian manifolds.
Consider the tangent space $T_{\bx}$ and the exponential map $\exp_\bx$ of a point $\bx$.
In a physical picture, the tangent space $T_{\bx}$ is the space consisting of all vectors $\balpha$ that describe the direction $\hat{\balpha}$ and magnitude $\norm{\balpha}$ for a particle moving at point $\bx$ on a manifold, and the exponential map $\exp_\bx$ is a function that takes in the vector $\balpha \in T_{\bx}$ encompassing the direction and magnitude and outputs the point after moving $\bx$ in the direction $\hat{\balpha}$ with a magnitude $\norm{\balpha}$.\footnote{Strictly speaking, to define the exponential map, we need to know how to ``transport'' the vector $\alpha$ along itself. Fortunately, this is natural for all cases we consider.}
To visualize these concepts, we give two warm-up examples in the following.

\paragraph{Euclidean space:}
In an $m$-dimensional Euclidean space, $\forall \bx \in \mathcal{X} = \mathbb{R}^m$, the tangent space is
\begin{equation}
    T_{\bx} = \{ \balpha \in \mathbb{R}^m \}.
\end{equation}
Given $\balpha \in T_{\bx}$, when we move $\bx$ in the direction $\hat{\balpha}$ with a magnitude $\norm{\balpha}$, we obtain
\begin{equation}
    \exp_\bx(\balpha) = \bx + \balpha.
\end{equation}
We can see that this matches our physical picture.

\paragraph{Particle moving counter-clockwise on a circle:}
As another warm-up, let us consider a unit circle $\mathcal{X} = \{ \bx \in \mathbb{R}^2 \, | \, \norm{\bx} = 1\}$ where a particle can only move counter-clockwise.
In this example, the tangent space $T_{\bx}$ of a unit vector $\bx \in \mathbb{R}^2$ with $\norm{\bx} = 1$ is the set of one-dimensional rays,
\begin{equation}
    T_{\bx} = \{ \alpha \in \mathbb{R} \,\, | \,\, \alpha \geq 0 \}.
\end{equation}
The condition $\alpha \geq 0$ comes from the constraint that the particle can only move counter-clockwise (unidirectional rather than bidirectional).
When we move $\bx$ according to $\alpha \in T_{\bx}$, we obtain
\begin{equation}
    \exp_\bx(\alpha) = \exp\left(\begin{pmatrix}
        0 & -\alpha \\ \alpha & 0
    \end{pmatrix}  \right) \bx = \begin{pmatrix}
        \cos \alpha & -\sin \alpha\\ \sin \alpha & \cos \alpha
    \end{pmatrix} \bx.
\end{equation}
The larger $\alpha$ is, the bigger the rotation is.

\vspace{1em}

Using the language of tangent spaces and exponential maps, 
an exact local minimum~$\bx^\# \in \mathcal{X}$ of a function $h$ is equivalent to the statement that there exists $\delta > 0$, such that
\begin{equation} \label{eq:exact-local-minima-general}
    h(\exp_{\bx^\#}(\balpha)) \geq h(\bx^\#), \quad \text{for all} \quad \balpha \in T_{\bx^\#}, ~ \norm{\balpha} \leq \delta.
\end{equation}
For the case of optimizing over $m$-dimensional Euclidean space, the condition of Eq.~\eqref{eq:exact-local-minima-general} becomes the same as Eq.~\eqref{eq:exact-local-minima} noting $\exp_{\bx^\#}(\balpha) = \bx^\# + \balpha$.
However, the condition can be quite different when the tangent space changes.
For example, consider a $2$-dimensional Euclidean space and the function $h(\bx) = \norm{\bx}^2$.
In general, there is only one exact local minimum $\bx^\# = 0$.
However, if the particle can only move to the right, the tangent space becomes $T_\bx = \{ \balpha \in \mathbb{R}^2 \,\, | \,\, \balpha_1 \geq 0 \}$ and every point $\bx$ with $\bx_1 \geq 0$ and $\bx_2 = 0$ is an exact local minimum.
Modifying the tangent space changes the definition of neighborhood. Hence, the set of local minima would be changed accordingly.
We will consider the most suitable norm $\norm{\balpha}$ for each context.

\subsubsection{Approximate local minima}
\label{sec:approx-local-minima}

While global minima are computationally hard to find, exact local minima are not much easier.
If there is an $\balpha$ such that $h(\exp_{\bx^\#}(\balpha))$ is lower than $h(\bx^\#)$ by an extremely small value, $\bx^\#$ is not consider to be an exact local minimum.
The requirement to resolve an extremely small value in exact local minima leads to the fact that finding an exact local minimum is still computationally hard~\cite{ahmadi2022complexity}.
Furthermore, exact local minima are very sensitive to small perturbations to the function~$h$.
Therefore, it is desirable to define approximate local minima to promote computational efficiency and robustness to small perturbations.
We consider the following principle for defining $\epsilon$-approximate local minima: \emph{if a function $h^*$ is very close to $h$, then an exact local minimum of $h^*$ is an approximate local minimum of $\tilde{h}$.}
The formal definition is given below.

\begin{definition}($\epsilon$-approximate local minima)
    Given a space $\mathcal{X}$ with tangent spaces $T_x$ and exponential maps $\exp_x$ for all $x \in \mathcal{X}$, and a function $h$: $\mathcal{X} \to \mathbb{R}$.
    $\bx^\#$ is an $\epsilon$-approximate local minimum of $h$ if $\bx^\#$ is the exact local minimum of some function $h^*$, where $\Delta(\bx) := h^*(\bx) - h(\bx)$ satisfies
    \begin{align}
        \left|\Delta(\bx)\right| \leq \epsilon &\quad \text{for each}\quad \bx \in \mathcal{X}, & \text{($\epsilon$-bounded)},\\
        \left|\Delta(\exp_{\bx^\#}(\balpha))\right| \leq \epsilon \norm{\balpha} &\quad \text{for each}\quad \balpha \in T_{\bx^\#}, & \text{($\epsilon$-Lipschitz around $x^\#$)}.
    \end{align}
    A $(\epsilon=0)$-approximate local minimum of $h$ is an exact local minimum of $h$.
\end{definition}

\noindent Under this definition, $\bx^\#$ is an approximate local minimum of $h$ if there is an $\bx$ in the neighborhood of $\bx^\#$ such that $h(\bx)$ is lower than $h(\bx^\#)$ by an extremely small value.
We also give the following equivalent characterization based on looking at the local neighborhood.

\begin{proposition}(An equivalent characterization of $\epsilon$-approximate local minima)
    $\bx^\# \in \mathcal{X}$ is an $\epsilon$-approximate local minimum of the function $h$ if and only if there exists a distance $\delta > 0$,
    \begin{equation} \label{eq:local-minima-general}
        h(\exp_{\bx^\#}(\balpha)) \geq h(\bx^\#) - \epsilon \norm{\balpha}\quad\text{for each}\quad \balpha \in T_{\bx^\#}, \norm{\balpha} \leq \delta,
    \end{equation}
    i.e., all the neighboring points can at most be $\epsilon \norm{\balpha}$ lower than the point $\bx^\#$.
\end{proposition}
\begin{proof}
    For the ``only if'' statement, we recall the definition of an exact local minimum that there exists $\delta > 0$, such that $h^*(\exp_{\bx^\#}(\balpha)) - h^*(\bx^\#) \geq 0$ for all $\balpha \in T_{\bx^\#}$ and $\norm{\balpha} \leq \delta$.
    From the $\epsilon$-Lipschitz condition around $x^\#$ for the function $\Delta(\bx)$, we have
    \begin{align}
        0 &\leq h^*(\exp_{\bx^\#}(\balpha)) - h^*(\bx^\#) = h(\exp_{\bx^\#}(\balpha)) - h(\bx^\#) + \Delta(\exp_{\bx^\#}(\balpha))\\
        &\leq h(\exp_{\bx^\#}(\balpha)) - h(\bx^\#) - \epsilon \norm{\balpha}.
    \end{align}
    This concludes the ``only if'' statement.

    For the ``if'' statement, consider $\delta$ to be of at most $1$ and let
    \begin{equation}
        \Delta(\bx) := \begin{cases}
            h(\bx^\#) - h(\bx), & \text{if $\bx = \exp_{\bx^\#}(\balpha)$ for some $\balpha \in T_{\bx^\#}, \norm{\balpha} \leq \delta$},\\
            0, & \text{otherwise}.
        \end{cases}
    \end{equation}
    We have $\bx^\#$ is an exact local minimum for $h^*(\bx) := h(\bx) + \Delta(\bx)$.
    Furthermore, because $h(\bx^\#) - h(\exp_{\bx^\#}(\balpha)) \leq \epsilon \norm{\balpha} \leq \epsilon$, both $\epsilon$-bounded and $\epsilon$-Lipschitz around $x^\#$ are satisfied by $\Delta(\bx)$.
\end{proof}

\subsection{Defining local minima in quantum systems}
\label{sec:local-minima-quantum}

To define local minima, we need to consider the domain $\mathcal{X}$ of elements $\bx \in \mathcal{X}$, the optimization function $h(\bx)$, the tangent space $T_\bx$ consisting of all possible directions and magnitudes to move an element $\bx$, where $\balpha \in T_\bx$ encompass the direction $\hat{\balpha}$ and the magnitude $\norm{\balpha}$, and the exponential map $\exp_{\bx}(\balpha)$ that describes the resulting element after moving $\bx$ under $\balpha$.

In the following, we present two settings.
The first setting in Appendix~\ref{sec:def-low-temp-thermal} considers general quantum states that can evolve under thermodynamic processes induced by interacting with a low-temperature heat bath. This setting defines local minima under thermal perturbations.
The second setting in Appendix~\ref{sec:def-local-unitary} considers pure quantum states that can move under any unitary generated by a set of local Hermitian operators (e.g., all two-qubit Pauli observables $\vP_i \otimes \vQ_j,$ where $\vP, \vQ \in \{\vX, \vY, \vZ\}$).
This setting defines local minima under local unitary perturbations.

\subsubsection{Definition based on thermal perturbations}
\label{sec:def-low-temp-thermal}

In quantum mechanics, the central optimization problem considers a function $h$ defined by the Hamiltonian $\vH$ of an $n$-qubit quantum system,
\begin{equation}
    h(\vrho) = \Tr(\vH \vrho),
\end{equation}
which is the average energy of an $n$-qubit quantum state $\vrho$.
The ground states $\vrho^{(g)}$ of $\vH$ are the global minima of the optimization over $h(\vrho) = \Tr(\vH \vrho)$ in the quantum state space, i.e., the set of density operators (trace-one positive semidefinite matrices),
\begin{equation}
    \mathcal{S}_{2^n} := \{\vrho \in \mathbb{C}^{2^n \times 2^n} \,\, | \,\, \vrho^\dag = \vrho, \, \vrho \succeq 0, \, \Tr(\vrho) = 1 \}.
\end{equation}
When the quantum system is placed in a heat bath with inverse temperature $\beta \in [0, \infty]$, time scale $\tau \in [0, \infty]$, and system-bath interactions based on $m$ local jump operators\footnote{A local operator $\vA^a$ acts on $\mathcal{O}(1)$ qubits, but the set of qubits that $\vA^a$ acts on may not be geometrically close.} $\vA^1, \ldots, \vA^m$, the system dynamics are effectively described by the \textit{thermal Lindbladian}s $\mathcal{L}^{\beta, \tau,\vH}_{a}$,
\begin{align} \label{eq:thermo-lindblad-rho-again}
    \frac{\rd\vrho(t)}{\rd t} = -\ri [\vH, \vrho] + \sum_{a = 1}^m \alpha_a \mathcal{L}^{\beta, \tau,\vH}_{a}[\vrho],
\end{align}
where $\alpha_a \geq 0$ for each $a$.
After time $t$, the initial quantum state $\vrho$ will evolve to
\begin{equation}
    \vrho(t) = \exp\left( -\ri t [\vH, \cdot] + \sum_{a = 1}^m t \alpha_a \mathcal{L}^{\beta, \tau,\vH}_{a} \right)(\vrho).
\end{equation}
Each term $\mathcal{L}^{\beta, \tau,\vH}_{a}$ is the thermal Lindbladian associated with a local jump operator $\vA^a$ (recall that local operator $\vA^a$ acts on a constant number of qubits). See Appendix~\ref{sec:thermo-lindblad} for a brief review of thermal Lindbladians, and Appendix~\ref{sec:thermo-lindblad-detail} for the exact form of thermal Lindbladians.

The coefficient $\alpha_a \geq 0$ corresponds to the interaction strength of each jump operator $\vA^a$. 
As $\alpha_a < 0$ is equivalent to reversing time, we cannot have $\alpha_a < 0$ since thermodynamic processes are irreversible in general.
Different interaction strength vector $\balpha$ corresponds to a different system-bath interaction, and the thermodynamics could be different.
Because $\balpha$ describes the probability of each jump occurring, the natural norm $\norm{\balpha}$ for the interaction strength vector $\balpha$ is $\norm{\balpha}_1$.
We denote $\hat{\balpha} = \balpha / \norm{\balpha}_1$ as the unit vector.

The thermodynamics equation in Eq.~\eqref{eq:thermo-lindblad-rho-again} consists of a fast-rotating term $-\ri [\vH, \cdot]$ due to the system Hamiltonian $\vH$ that keeps the energy $\Tr(\vH \vrho)$ invariant and the thermal perturbation term $\sum_{a} \alpha_a \CL_a^{\beta, \tau, \vH}$ due to the heat bath that cools the system.
Because $-\ri [\vH, \cdot]$ keeps the energy constant, only the thermal perturbation term $\sum_{a} \alpha_a \CL_a^{\beta, \tau, \vH}$ is relevant for minimizing the energy $h(\vrho) = \Tr(\vH \vrho)$.
For notational simplicity, we will only consider contributions from the thermal perturbations and absorb the $t$ dependence in $t \alpha_a$ into $\alpha_a$ since $\balpha$ is an arbitrary nonnegative vector.
Together, the thermal perturbation on $\vrho$ due to a heat bath with inverse temperature $\beta$, time scale~$\tau$, and system-bath interactions generated by $\{\vA^a\}_a$ can be written as
\begin{equation}
    \vrho \rightarrow \exp\left( \sum_{a = 1}^m \alpha_a \mathcal{L}^{\beta, \tau,\vH}_{a} \right)(\vrho)
\end{equation}
for a nonnegative vector $\balpha \in \mathbb{R}^m_{\geq 0}$ that combines the interaction strength vector and time $t$.

A dictionary between all the relevant functions and variables for optimizing $\Tr(\vH \vrho)$ in $n$-qubit quantum systems under a heat bath with inverse temperature $\beta$ and time scale $\tau$ and optimizing $h(\bx)$ in an $n$-dimensional Euclidean space is given as follows.
\begin{align}
    \mathcal{X} = \mathbb{R}^{n} & \quad \leftrightarrow \quad \mathcal{X} = \mathcal{S}_{2^n}, & \text{(domain)}\\
    \bx \in \mathbb{R}^{n} & \quad \leftrightarrow \quad \vrho \in \mathcal{S}_{2^n}, & \text{(an element)}\\
    h(\bx) &\quad \leftrightarrow \quad h(\vrho) = \Tr(\vH \vrho),& \text{(optimization function)}\\
    T_\bx = \{\balpha \in \mathbb{R}^n\} &\quad \leftrightarrow \quad \{ \balpha \in \mathbb{R}^{m}_{\geq 0} \},& \text{(tangent space)}\\
    \exp_{\bx}(\balpha) = \bx + \balpha &\quad \leftrightarrow \quad \exp\left( \sum_{a = 1}^m \alpha_a \mathcal{L}^{\beta, \tau,\vH}_{a} \right)(\vrho) & \text{(exponential map)}.
\end{align}
The formal definition of tangent spaces and exponential maps via Lindbladians is given below.

\begin{definition}[Tangent spaces of quantum states in a heat bath]
    Consider an $n$-qubit quantum state~$\vrho$, an $n$-qubit Hamiltonian $\vH$, $m$ local jump operators $\{\vA^a\}_{a = 1}^m$, and parameters $\beta, \tau \ge 0$. The tangent space $T^{\beta, \tau, \vH, \{\vA^a\}_{a = 1}^m}_{\vrho}$ under a heat bath with an inverse temperature $\beta$, a time scale $\tau$, and system-bath interactions generated by $\{\vA^a\}_a$ is defined as
    \begin{equation}
        T^{\beta, \tau, \vH, \{\vA^a\}_{a = 1}^m}_{\vrho} := \left\{ \balpha \in \mathbb{R}^{m}_{\geq 0} \right\},
    \end{equation}
    which is independent of $\beta, \tau, \vH, \{\vA^a\}_{a = 1}^m$.
    The exponential map $\exp^{\beta, \tau, \vH, \{\vA^a\}_{a}}_{\vrho}$ is defined as
    \begin{equation}\label{eq:expmap-thermal}
        \exp^{\beta, \tau, \vH, \{\vA^a\}_{a}}_{\vrho}(\balpha) := \exp\left( \sum_{a = 1}^m \alpha_a \mathcal{L}^{\beta, \tau,\vH}_{a} \right)(\vrho).
    \end{equation}
\end{definition}
With the definition of tangent spaces and exponential maps, we can define $\epsilon$-approximate local minimum similar to the classical case in Eq.~\eqref{eq:local-minima-general}.
We consider the natural choice of $\norm{\cdot}_1$ for the nonnegative vector $\balpha$ encompassing the probability of each jump. Our results remain qualitatively the same for other reasonable vector norms, such as Euclidean norm $\norm{\cdot}_2$ or $\ell_p$ norm $\norm{\cdot}_p$.

\begin{definition}[Local minima under thermal perturbations]\label{defn:quantum_local}
    Given an $n$-qubit Hamiltonian $\vH$, $m$ local jump operators $\{\vA^a\}_{a = 1}^m$, and parameters $\beta, \tau \ge 0$, an $n$-qubit state ${\vrho} \in \mathcal{S}_{2^n}$ is an $\epsilon$-approximate local minimum of $\vH$ under thermal perturbations with an inverse temperature $\beta$, a time scale $\tau$, and system-bath interactions generated by $\{\vA^a\}_a$ if there is a $\delta > 0$ such that
    \begin{equation} \label{eq:local-minima-mixed}
        \Tr\left(\vH \exp^{\beta, \tau, \vH, \{\vA^a\}_{a}}_\vrho(\balpha)\right) \geq \Tr(\vH \vrho) - \epsilon \norm{\balpha}_1 \quad \text{for each}\quad \balpha \in \mathbb{R}^{m}_{\geq 0}, \norm{\balpha}_1 \leq \delta,
    \end{equation}
    i.e., all the neighboring points can at most be $\epsilon \norm{\balpha}_1$ lower than the point $\vrho$.
\end{definition}

A central concept we will be using for characterizing local minima under thermal perturbations is the energy gradient. 
The energy gradient at an $n$-qubit state $\vrho$ under thermal perturbation is determined by the following state-independent operator,
\begin{equation}
	\text{(energy gradient operator):} \quad\quad \sum_{a=1}^m \CL^{\dag \beta, \tau, \vH}_{a}(\vH) \hat{\be}_a,
\end{equation}
where we denote $\CL^{\dag \beta, \tau, \vH}_{a}$ to be the Hermitian conjugate of $\CL^{\beta, \tau, \vH}_{a}$.
The energy gradient operator is a vector of Hermitian observables. The terminology stems from the fact that evaluating the energy gradient operator on a state $\vrho$ gives the energy gradient at the state $\vrho$,
\begin{equation}
    \Tr\left(\vH \exp\left( -\ri \sum_{a=1}^m \alpha_a \CL^{\beta, \tau, \vH}_{a} \right) (\vrho) \right) = \Tr(\vH \vrho) + \balpha \cdot \sum_{a=1}^m \Tr\left(\CL^{\dag \beta, \tau, \vH}_{a}(\vH) \vrho\right) \hat{\be}_a + \mathcal{O}(\norm{\alpha}^2).
\end{equation}
In Appendix~\ref{sec:local-minima-thermal-prop}, we provide more discussions about the energy gradient.

Thermal perturbations depend on how the quantum system is interacting with the heat bath.
Local minima defined above are local minima of the Hamiltonian $\vH$ under thermal perturbations induced by all system-bath interactions generated by the jump operators $\{\vA^a\}_a$.

\begin{remark}[Thermodynamics at local minima]
Given a specific system-bath interaction, inverse temperature $\beta$, and time scale $\tau$, there could still be thermodynamics at a local minimum.
For example, when $\beta$ is not infinitely large, a local minimum could still move to other higher-energy states due to thermal fluctuations. Another example is when the local minimum is on a large and flat plateau, then the local minimum can still perform a random walk on the plateau.
\end{remark}

\subsubsection{Definition based on local unitary perturbations}
\label{sec:def-local-unitary}

Inspired by variation quantum eigensolvers \cite{grimsley2019adaptive, cerezo2021variational}, another natural definition for tangent spaces, exponential maps, and local minima considers pure states and local unitary perturbation.
Given $m$ local Hermitian operators $\vh_1, \ldots, \vh_m$ with $\norm{\vh_a}_\infty = 1$.
Here, local means that each operator $\vh_a$ only acts on a constant number of qubits.
We can consider all possible local unitary perturbations formed by performing time evolution under a Hamiltonian generated by the set $\{\vh_a\}_a$ of local Hermitian operators,
\begin{equation}
    \sum_{a=1}^m \alpha_a \vh_a,
\end{equation}
for any $\balpha \in \BR^m$.
Since the time evolution under a Hamiltonian is always reversible, there is no additional requirement that $\balpha$ must be in the nonnegative orthant.
Similar to thermal perturbations, we will absorb the contribution of evolution time $t$ into the arbitrary vector $\balpha$.
Consider the following dictionary between all the relevant functions and variables for optimizing $\bra{\psi} \vH \ket{\psi}$ in $n$-qubit pure state $\ket{\psi}$ under local unitary perturbation and optimizing $h(\bx)$ in an $n$-dimensional Euclidean space.
\begin{align}
    \mathcal{X} = \mathbb{R}^{n} & \quad \leftrightarrow \quad \mathcal{X} = \{ \ket{\psi} \in \mathbb{C}^{2^n} | \braket{\psi | \psi} = 1 \}, & \text{(domain)}\\
    \bx \in \mathbb{R}^{n} & \quad \leftrightarrow \quad \ket{\psi} \in \mathbb{C}^{2^n}, \braket{\psi | \psi} = 1, & \text{(an element)}\\
    h(\bx) &\quad \leftrightarrow \quad h(\ket{\psi}) = \bra{\psi} \vH \ket{\psi},& \text{(optimization function)}\\
    T_\bx = \{\balpha \in \mathbb{R}^n\} &\quad \leftrightarrow \quad \{ \balpha \in \mathbb{R}^m \},& \text{(tangent space)}\\
    \exp_{\bx}(\balpha) = \bx + \balpha &\quad \leftrightarrow \quad \exp\left( \sum_{a=1}^m \alpha_a \vh_a \right)\ket{\psi} & \text{(exponential map)}.
\end{align}
The tangent space and the exponential map can be formally defined as follows.

\begin{definition}[Tangent spaces of pure quantum states under local unitaries]
    Given an $n$-qubit pure quantum state $\ket{\psi}$ and $m$ local Hermitian operators $\{\vh_a\}_a$. The tangent space $T_\psi$ is defined as
    \begin{equation}
        T^{\{\vh_a\}_a}_\psi := \mathbb{R}^{m},
    \end{equation}
    and the exponential map $\exp_\psi$ is defined as
    \begin{equation}
        \exp^{\{\vh_a\}_a}_\psi(\balpha) := \exp\left( - i \sum_{a} \alpha_a \vh_a \right) \ket{\psi}.
    \end{equation}
\end{definition}

When the set $\{\vh_a\}_a$ is the set of all two-qubit Pauli observables, the tangent space $T_\psi$ and exponential map $\exp_\psi$ define a Riemannian manifold that connects all $n$-qubit pure states through unitary evolutions.
This Riemannian manifold is the state version of the manifold over quantum unitaries defined in a seminal work on the geometry of quantum computation \cite{nielsen2006quantum}.

The optimization function is $h(\ket{\psi}) = \bra{\psi} \vH \ket{\psi}$, the average energy of the Hamiltonian $\vH$ for the pure state $\ket{\psi}$.
Performing gradient descent on this pure state Riemannian manifold to minimize $\bra{\psi} \vH \ket{\psi}$ is equivalent to performing adaptive variational quantum optimization \cite{grimsley2019adaptive} to minimize the Hamiltonian $\vH$.
The local minima can be defined similarly as before.
To be consistent with local minima under thermal perturbations, we consider the $\ell_1$-norm $\norm{\balpha}_1$.
All of our results remain qualitatively the same for other reasonable vector norms, such as the Euclidean norm or $\ell_p$ norm.

\begin{definition}[local minima under local unitary perturbations] \label{def:local-minima-pure}
    Given an $n$-qubit Hamiltonian $\vH$, and $m$ local Hermitian operators $\{\vh_a\}_a$. A pure state $\ket{\psi}$ is an $\epsilon$-approximate local minimum of $\vH$ under local unitary perturbations generated by $\{\vh_a\}_a$ if
    \begin{equation} \label{eq:local-minima-pure}
         \exp^{\{\vh_a\}_a}_\psi(\balpha)^\dag \vH \exp^{\{\vh_a\}_a}_\psi(\balpha) \geq \bra{\psi} \vH \ket{\psi} - \epsilon \norm{\balpha}_1, \,\, \text{for each}\quad \balpha \in T^{\{\vh_a\}_a}_{\psi}, \norm{\balpha}_1 \leq \delta.
    \end{equation}
    for some $\delta > 0$.
\end{definition}

\noindent This is also a valid definition of local minima in quantum systems.
However, we will later show that the optimization landscape defined in this way always has a very large barren plateau.
Hence, the problem of finding a local minimum defined in this way will be a trivial problem.

\subsection{The problem of finding a local minimum in quantum systems}

With these definitions of local minima, we can define the task of finding a local minimum in a straightforward manner.
To formulate the problem to have purely classical output, we focus on outputting a simple property, such as the expectation value of a local observable $\vO$, of an approximate local minimum $\vrho$.
Furthermore, we only consider Hamiltonians $\vH$ that can be written as a sum of local observables, commonly referred to as local Hamiltonians in the literature.

While there can be many approximate local minima, we consider the algorithm to be successful if the algorithm outputs the property of any one of the local minima.

\begin{definition}[Finding a local minimum under low-temperature thermal perturbations] \label{def:thermal-LM}
    Given error $\epsilon > 0$, inverse temperature $\beta \geq 0$, time scale $\tau \geq 0$, an $n$-qubit local Hamiltonian $\vH$, $m$ local jump operators $\{\vA^a\}_a$, and a local observable $\vO$ with $\norm{\vO}_\infty \leq 1$.
    Output a real value $v \in [-1, 1]$, such that $v$ is $\epsilon$-close to $\Tr(\vO \vrho)$ for an $\epsilon$-approximate local minimum $\vrho$ of $\vH$ under thermal perturbations with an inverse temperature $\beta$, a time scale $\tau$, and system-bath interactions generated by $\{\vA^a\}_a$.
\end{definition}

\begin{definition}[Finding a local minimum under local unitary perturbations] \label{def:unitary-LU}
    Given error $\epsilon > 0$, an $n$-qubit local Hamiltonian~$\vH$, $m$ local Hermitian operators $\{\vh_a\}_a$, and a local observable $\vO$ with $\norm{\vO}_\infty \leq 1$.
    Output a real value $v \in [-1, 1]$, such that $v$ is $\epsilon$-close to $\bra{\psi} \vO \ket{\psi}$ for an $\epsilon$-approximate local minimum $\ket{\psi}$ of the Hamiltonian $\vH$ under local unitary perturbations generated by $\{\vh_a\}_a$.
\end{definition}

\noindent Ideally, we would like the two problems to be quantumly easy and classically hard.
However, we will show that only the first problem based on thermal perturbations is both quantumly easy and classically hard.
The second problem based on local unitary perturbation is classically trivial due to the presence of too many local minima in an exponentially large barren plateau.

\subsection{The importance of irreversible perturbations}
\label{sec:irreversible-perturbations}

Suppose that the perturbations $\mathcal{P}_\balpha$ parameterized by a polynomial-size vector $\balpha$ are reversible $\mathcal{P}_{- \balpha} = \mathcal{P}_{\balpha}^{-1}$ and are smooth.
The following argument shows that the energy landscape must have doubly-exponentially many approximate local minima.

Given any $n$-qubit state $\vrho$ and any $n$-qubit Hamiltonian $\vH$ with $\norm{\vH}_\infty = \mathrm{poly}(n)$.
Consider a gradient descent algorithm that starts at $\vrho$.
Because $\norm{\vH}_\infty = \mathrm{poly}(n)$, after a polynomial number of steps $T$, the gradient descent algorithm can find an approximate local minimum $\vrho^\#$ of $\vH$,
\begin{equation}
    \vrho^\# = \mathcal{P}_{\balpha_T} \ldots \mathcal{P}_{\balpha_1} (\vrho).
\end{equation}
From the reversibility of the perturbations, we have
\begin{equation} 
    \vrho = \mathcal{P}_{-\balpha_1} \ldots \mathcal{P}_{-\balpha_T} (\vrho^\#).
\end{equation}
Consider a covering net $\mathcal{N}$ for the set of approximate local minima of $\vH$.
Because the packing net for all $n$-qubit states is of size
\begin{equation}
    \exp(\exp(\Omega(n))),
\end{equation}
and the covering net for the perturbations is of size
\begin{equation}
    \exp(\mathrm{poly}(n)),
\end{equation}
we have the following relationship,
\begin{equation}
    \exp(T \cdot \mathrm{poly}(n)) |\mathcal{N}| = \exp(\exp(\Omega(n)))
\end{equation}
Hence, we can see that
\begin{equation}
    |\mathcal{N}| = \exp(\exp(\Omega(n)) - \mathrm{poly}(n)) = \exp(\exp(\Omega(n)))
\end{equation}
since $\exp(\Omega(n))$ grows much faster than $\mathrm{poly}(n)$.

\section{Characterizing local minima under local unitary perturbations}
\label{sec:barren-plateau-energy-unitary}

Now that we have defined local minima in quantum systems, we present a set of results characterizing properties of local minima in quantum systems in this and the next appendix.
These results provide a further understanding of local minima in quantum systems and are essential to establishing the main theorems given in Appendix~\ref{sec:complexity-main}.

We begin by looking at the energy landscape defined by local unitary perturbations.
We will prove a central lemma portraying the energy landscapes defined by local unitary perturbations for pure quantum systems.
The lemma states that most pure quantum states $\ket{\psi}$ are local minima under local unitary perturbations with an expectation value close to $\Tr(\vO) / 2^n = \Tr(\vO (\vI / 2^n))$ for any local observable~$\vO$.
Furthermore, the proof shows that the gradient at a randomly sampled local minimum $\ket{\psi}$ is exponentially close to zero.

Lemma~\ref{lem:Haar-concentration} and its proof provide the following physical picture.
In the energy landscape defined by local unitary perturbations, there is an overwhelmingly large barren plateau consisting of local minima with almost equal energy as their neighbors.
Furthermore, these local minima behave like the maximally mixed state $\vI / 2^n$, which makes the task of predicting properties for a local minimum under local unitary perturbations classically trivial to solve.

Given $m$ local Hermitian operators $\{\vh_a\}_a$ and $\balpha \in \mathbb{R}^m$.
By applying Taylor's theorem in Prop.~\ref{prop:Taylor} to the one-dimensional function
\begin{equation}
g(t) = \exp^{\{\vh_a\}_a}_\psi(t \hat{\balpha})^\dag \vH \exp^{\{\vh_a\}_a}_\psi(t \hat{\balpha})
\end{equation}
for $0 \leq t \leq \norm{\balpha}_1$ and $\hat{\balpha} = \balpha / \norm{\balpha}_1$, we can obtain the following proposition.

\begin{proposition}[Taylor's theorem for local unitary perturbations] \label{prop:taylor-unitary}
    Given an $n$-qubit Hamiltonian $\vH$, $\balpha \in \mathbb{R}^n$, $m$ local Hermitian operators $\{\vh_a\}_a$, and an $n$-qubit pure state $\ket{\psi}$. We have
    \begin{align}
        \exp_\psi(\balpha)^\dag \vH \exp_\psi(\balpha) &= \bra{\psi} \vH \ket{\psi} - \ri \bra{\psi} \left[\vH, \sum_{a=1}^m \alpha_a \vh_a \right] \ket{\psi} \nonumber \\
        &\qquad - \frac{1}{2} \sum_{a=1}^m \sum_{a'=1}^m \alpha_a \alpha_{a'} \exp_\psi(\eta \hat\balpha)^\dag [[\vH, \vh_a], \vh_{a'}] \exp_\psi(\eta \hat\balpha),
    \end{align}
    for some $0 \leq \eta \leq \norm{\balpha}_1$.
\end{proposition}

\begin{lemma}[A random state is a local minimum under local unitary perturbations; Restatement of Lemma~\ref{lem:barren-plateau-random-state-informal}] \label{lem:Haar-concentration}
    Consider a large problem size $n$.
    Given error $\epsilon \geq 1 / 2^{n/4}$, an $n$-qubit local Hamiltonian~$\vH$ with $\norm{\vH}_\infty = \mathrm{poly}(n)$, $m$ local Hermitian operators $\{\vh_a\}_a$ with $m = \mathrm{poly}(n)$ and $\norm{\vh_a}_\infty = 1$, and a local observable $\vO$ with $\norm{\vO}_\infty \leq 1$.
    With probability at least $1 - 1 / 2^{2^{n/4}}$, an $n$-qubit state $\ket{\psi}$ sampled uniformly at random is an $\epsilon$-approximate local minimum of $\vH$ under local unitary perturbations generated by $\{\vh_a\}_a$ and $\bra{\psi} \vO \ket{\psi}$ is $\epsilon$-close to $\Tr(\vO) / 2^n$.
\end{lemma}
\begin{proof}
    From Lemma III.5 in \cite{hayden2006aspects}, for any Pauli operator $\vQ \in \{\vI, \vX, \vY, \vZ\}^{\otimes n} \setminus \{\vI^{\otimes n}\}$ and a random $n$-qubit pure state $\ket{\psi}$ sampled uniformly, we have
    \begin{equation}
        \Pr_{\ket{\psi}}\left[ \left| \bra{\psi} \vQ \ket{\psi} \right| > \delta \right] \leq 2 \exp\left(- \frac{2^n \delta^2}{10} \right),
    \end{equation}
    for any $0 \leq \delta \leq 1$. Let $\delta = 1 / 2^{n/3}$. Then, we have
    \begin{equation}
        \Pr_{\ket{\psi}}\left[ \left| \bra{\psi} \vQ \ket{\psi} \right| > \frac{1}{2^{n/3}} \right] \leq 2 \exp\left(- \frac{2^{n/3}}{10} \right).
    \end{equation}
    Recall that any Hermitian operator has a unique Pauli decomposition:
    \begin{align}
        \vH &= \sum_{\vP \in \{\vI, \vX, \vY, \vZ\}^{\otimes n} } \alpha_{\vP}(\vH) \vP, \\
        \vO &= \sum_{\vP \in \{\vI, \vX, \vY, \vZ\}^{\otimes n} } \alpha_{\vP}(\vO) \vP, \\
        \vh^a &= \sum_{\vP \in \{\vI, \vX, \vY, \vZ\}^{\otimes n} } \alpha_{\vP}(\vh^a) \vP,
    \end{align}
    where the Pauli coefficients $\alpha_{\vP}(\cdot)$ satisfy
    \begin{align}
        \sum_{\vP \in \{\vI, \vX, \vY, \vZ\}^{\otimes n} } \alpha_{\vP}^2(\vH) &\leq \norm{\vH}_\infty^2 = \mathrm{poly}(n), \label{eq:alphavHbd} \\
        \sum_{\vP \in \{\vI, \vX, \vY, \vZ\}^{\otimes n} } \alpha_{\vP}^2(\vO) &\leq \norm{\vO}_\infty^2 = 1. \label{eq:alphavObd} \\
        \sum_{\vP \in \{\vI, \vX, \vY, \vZ\}^{\otimes n} } \alpha_{\vP}^2(\vh_a)  &\leq \norm{\vh_a}_\infty^2 = 1. \label{eq:alphavhabd}
    \end{align}
    Let $S_0$ be the set of Pauli operator $\vP$ with non-zero Pauli coefficients $\alpha_{\vP}$ in the Pauli decompositions of either $\vH$ or $\vO$,
    \begin{equation}
        S_0 = \left\{ \vP \in \{\vI, \vX, \vY, \vZ\}^{\otimes n} \setminus \{\vI^{\otimes n}\} \,\, | \,\, \alpha_{\vP}(\vH) \neq 0 \,\, \text{or} \,\, \alpha_{\vP}(\vO) \neq 0\right\},
    \end{equation} 
    and $S_E$ be the set of Pauli operator $\vP$ with non-zero Pauli coefficients $\alpha_{\vP}$ in the Pauli decompositions of $\vh_a$ for some $a$,
    \begin{equation}
        S_E = \left\{ \vP \in \{\vI, \vX, \vY, \vZ\}^{\otimes n} \setminus \{\vI^{\otimes n}\} \,\, | \,\, \exists 1 \leq a \leq m, \alpha_{\vP}(\vh_a) \neq 0\right\}.
    \end{equation} 
    Because $\vH$ is a local Hamiltonian and $\vO$ is a local observable, we have $|S_0| = \mathrm{poly}(n)$.
    Because $\vh_a$ is a local observable, we have $|S_E| = \mathcal{O}(m) = \mathrm{poly}(n)$.
    We then define,
    \begin{align}
        S &= \Big\{ \vP' \in \{\vI, \vX, \vY, \vZ\}^{\otimes n} \setminus \{\vI^{\otimes n}\} \,\, | \,\, \exists \vQ \in S_0, \vP \in S_E, \Tr(\vP' [\vQ, \vP]) \neq 0 \Big\} \cup S_0.
    \end{align}
    Because $|S_0| = \mathrm{poly}(n)$ and $[\vQ, \vP]$ is another Pauli observable up to a phase, we have $|S| \leq |S_E| |S_1| = \mathrm{poly}(n)$.
    The union bound yields the following probabilistic statement,
    \begin{equation}
        1-\Pr_{\ket{\psi}}\left[ \left| \bra{\psi} \vQ \ket{\psi} \right| < \frac{1}{2^{n/3}}, \,\, \forall \vQ \in S \right] \leq 2 |S| \exp\left(- \frac{2^{n/3}}{10} \right) \leq \frac{\mathrm{poly}(n)}{2^{2^{n/3} / 10}} < \frac{1}{2^{2^{n/4}}},
    \end{equation}
    where the last inequality holds for any large $n$ since $2^{2^{n/3} / 10 - 2^{n/4}}$ grows much faster than any polynomial of $n$.
    We condition on the event for the random state $\ket{\psi}$ that
    \begin{equation}
        \left| \bra{\psi} \vQ \ket{\psi} \right| < \frac{1}{2^{n/3}} \quad \text{for all}\quad \vQ \in S,
    \end{equation}
    referred to as event $E^*$. We can obtain the following from Cauchy-Schwarz inequality,
    \begin{align}
        \left| \bra{\psi} [\vH, \vh_a] \ket{\psi} \right| &\leq \sum_{\vQ, \vP \in \{\vI, \vX, \vY, \vZ\}^{\otimes n} } \left| \alpha_{\vQ}(\vH) \right| \left| \alpha_{\vP}(\vh_a) \right| \left| \bra{\psi} [\vQ, \vP] \ket{\psi} \right| \nonumber\\
        &\leq \sqrt{\sum_{\vQ, \vP \in \{\vI, \vX, \vY, \vZ\}^{\otimes n}} \alpha_{\vQ}^2(\vH) \alpha_{\vP}^2(\vh_a)} \sqrt{ \frac{|S|}{2^{2n/3}} } \leq \frac{\mathrm{poly}(n) }{2^{n/3}}, \label{eq:vHvhi-commutator}
    \end{align}
    where the second inequality uses the conditioning on event $E^*$ and $[\vQ, \vP] \neq 0 \implies \vQ, \vP \neq \vI^{\otimes n}$, and the third inequality uses $|S| = \mathrm{poly}(n)$ and Eq.~\eqref{eq:alphavHbd},~\eqref{eq:alphavhabd}. 
    Similarly, we also have 
    \begin{equation} \label{eq:vHvhi-double-commutator}
        \left| \bra{\psi} [[\vH, \vh_a], \vh_{a'}] \ket{\psi} \right| \leq \frac{\mathrm{poly}(n)}{2^{n/3}}.
    \end{equation}
    Using Eq.~\eqref{eq:alphavObd} instead of Eq.~\eqref{eq:alphavHbd}, we can similarly obtain
    \begin{equation} \label{eq:psivOpsi-err}
        \left| \bra{\psi} \vO \ket{\psi} - \alpha_{I^{\otimes n}}(\vO) \right| = \left| \bra{\psi} \vO \ket{\psi} - \frac{\Tr(O)}{2^n} \right| \leq \frac{\mathrm{poly}(n) }{2^{n/3}} < \frac{1}{2^{n/4}} \leq \epsilon
    \end{equation}
    for any large problem size $n$ since $2^{n/3}$ grows much faster than any polynomial in $n$.
    We now show that $\ket{\psi}$ is an $\epsilon$-approximate local minimum of $\vH$ under local unitary perturbations.
    To establish this claim, from Def.~\ref{def:local-minima-pure}, we need to prove that
    \begin{equation}
        \exp_\psi(\balpha)^\dag \vH \exp_\psi(\balpha) \geq \bra{\psi} \vH \ket{\psi} - \epsilon \norm{\balpha}_1, \,\, \text{for each}\quad \balpha \in T_{\psi}, \norm{\balpha}_1 \leq \delta
    \end{equation}
    for some $\delta > 0$.
    Recall from Lemma~\ref{prop:taylor-unitary} based on Taylor's theorem (Prop.~\ref{prop:Taylor}), we have
    \begin{align}
        \exp_\psi(\balpha)^\dag \vH \exp_\psi(\balpha) &= \bra{\psi} \vH \ket{\psi} - \ri \bra{\psi} \left[\vH, \sum_{a=1}^m \alpha_a \vh_a \right] \ket{\psi} \nonumber\\
        &\qquad - \frac{1}{2} \sum_{a=1}^m \sum_{a'=1}^m \alpha_a \alpha_{a'} \exp_\psi(\eta \hat\balpha)^\dag [[\vH, \vh_a], \vh_{a'}] \exp_\psi(\eta \hat\balpha),
    \end{align}
    for some $0 \leq \eta \leq \norm{\balpha}_1$.
    For the linear term, from Eq.~\eqref{eq:vHvhi-commutator} bounding $|\bra{\psi}[\vH, \vh_a] \ket{\psi}|$, we have
    \begin{equation}
        \left| -\ri \bra{\psi} \left[\vH, \sum_{a=1}^m \alpha_a \vh_a \right] \ket{\psi} \right| \leq \sum_{a=1}^m |\alpha_a| \frac{\mathrm{poly}(n) }{2^{n/3}} \leq \frac{ \mathrm{poly}(n) }{2^{n/3}} \norm{\balpha}_1,
    \end{equation}
    where the last inequality uses $m = \mathrm{poly}(n)$.
    For the quadratic residual term, we have
    \begin{equation}
        \left| \frac{1}{2} \sum_{a=1}^m \sum_{a'=1}^m \alpha_a \alpha_{a'} \exp_\psi(\eta \hat\balpha)^\dag [[\vH, \vh_a], \vh_{a'}] \exp_\psi(\eta \hat\balpha) \right| \leq  2 \norm{\balpha}_1^2 \norm{\vH}_\infty = \mathrm{poly}(n) \norm{\balpha}^2_1.
    \end{equation}
    Together, we can combine the inequalities to get
    \begin{align}
        \exp_\psi(\balpha)^\dag \vH \exp_\psi(\balpha) &\geq \bra{\psi} \vH \ket{\psi} - \frac{ \mathrm{poly}(n) }{2^{n/3}} \norm{\balpha}_1 - \mathrm{poly}(n) \norm{\balpha}_1^2 \nonumber\\
        &\geq \bra{\psi} \vH \ket{\psi} - 0.5 \epsilon \norm{\balpha}_1 - \mathrm{poly}(n) \norm{\balpha}_1^2,
    \end{align}
    where the second inequality holds for any large problem size $n$ since $\epsilon \geq 1 / 2^{n / 4}$ decays much slower than $\mathrm{poly}(n) / 2^{n/3}$.
    For any $\norm{\balpha} < \delta := 0.5 \epsilon / \mathrm{poly}(n)$, we have
    \begin{equation}
        \exp_\psi(\balpha)^\dag \vH \exp_\psi(\balpha) \geq \bra{\psi} \vH \ket{\psi} - 0.5 \epsilon \norm{\balpha}_1 - 0.5 \epsilon \norm{\balpha}_1 = \bra{\psi} \vH \ket{\psi} - \epsilon \norm{\balpha}_1,
    \end{equation}
    which shows that $\ket{\psi}$ is an $\epsilon$-approximate local minimum of $\vH$ under local unitary perturbations.
    Finally, because the event $E^*$ occurs with probability at least $0.99$, by combining Eq.~\eqref{eq:psivOpsi-err} and the above, we establish the claim that, with high probability, a random $n$-qubit state $\ket{\psi}$ sampled uniformly is an $\epsilon$-approximate local minimum of $\vH$ under local unitary perturbations and $\bra{\psi} \vO \ket{\psi}$ is $\epsilon$-close to $\Tr(O) / 2^n$.
\end{proof}

\section{Characterizing local minima under thermal perturbations}
\label{sec:local-minima-thermal-prop}

In this appendix, we characterize local minima under thermal perturbations.
In particular, we will focus on the gradients of the energy landscape, conditions of local minima, and conditions on the Hamiltonian $\vH$ that ensure approximate local minima are approximate global minima, i.e., there are no suboptimal local minima in the energy landscape.

\subsection{Energy gradients}

The energy landscape is much more nontrivial when defined under thermal perturbations.
We can study the energy landscape by looking at the energy gradients.
Recall the exponential map $\exp^{\beta, \tau, \vH, \{\vA^a\}_a}_\vrho$ in Eq.~\eqref{eq:expmap-thermal} and consider the one-dimensional function
\begin{equation}
g(t) = \Tr\left(\vH \exp^{\beta, \tau, \vH, \{\vA^a\}_a}_\vrho(t \hat{\balpha}) \right) = \Tr\left(\vH \exp\left( \sum_{a = 1}^m t \hat{\alpha}_a \mathcal{L}^{\beta, \tau,\vH}_{a} \right)(\vrho) \right)
\end{equation}
for $0 \leq t \leq \norm{\balpha}_1, \hat{\balpha} = \balpha / \norm{\balpha}_1$.
We have the following derivatives,
\begin{align}
    \frac{dg}{dt} (t) &= \Tr\left(\vH \sum_a \hat{\alpha}_a \CL_a^{\beta, \tau, \vH}\left[ \exp\left( \sum_{a = 1}^m t \hat{\alpha}_a \mathcal{L}^{\beta, \tau,\vH}_{a} \right)(\vrho) \right] \right),\\
    \frac{d^2g}{dt^2} (t) &= \Tr\left(\vH \sum_a \sum_{a'} \hat{\alpha}_a \hat{\alpha}_{a'} \CL_{a'}^{\beta, \tau, \vH} \left[ \CL_a^{\beta, \tau, \vH}\left[ \exp\left( \sum_{a = 1}^m t \hat{\alpha}_a \mathcal{L}^{\beta, \tau,\vH}_{a} \right)(\vrho) \right] \right] \right).
\end{align}

Recall Taylor's theorem with the Lagrange form of the remainder from standard single-variate calculus.
By applying Taylor's theorem in Prop.~\ref{prop:Taylor} to $g(t)$, we can obtain Prop.~\ref{prop:taylor-thermal}.

\begin{proposition}[Taylor's theorem] \label{prop:Taylor}
Let $g(t): \mathbb{R} \rightarrow \mathbb{R}$ be twice differentiable on the open interval between $0$ and $t$ and $g'(t)$ continuous on the closed interval between $0$ and $t$. Then
\begin{equation}
    g(t) = g(0) + g'(0) t + \frac{1}{2} g''(\eta) t^2,
\end{equation}
for some real number $\eta$ between $0$ and $t$.
\end{proposition}

\begin{proposition}[Taylor's theorem for thermal perturbations] \label{prop:taylor-thermal}
    Given an $n$-qubit Hamiltonian $\vH$, $m$ local jump operators $\{\vA^a\}_a$, parameters $\beta, \tau \ge 0$, $\balpha \in \mathbb{R}^m_{\geq 0}$, and an $n$-qubit state ${\vrho} \in \mathcal{S}_{2^n}$.
    \begin{align}
        \Tr\left(\vH \exp^{\beta, \tau, \vH, \{\vA^a\}_a}_\vrho(\balpha) \right) &= \Tr(\vH \vrho) + \sum_a \alpha_a \Tr(\vH \CL^{\beta, \tau, \vH}_a[\vrho]) \nonumber \\
        &\quad + \frac{1}{2} \sum_{a} \sum_{a'} \alpha_a \alpha_{a'} \Tr\left(\vH \CL^{\beta, \tau, \vH}_{a'} [\CL^{\beta, \tau, \vH}_a[  \exp^{\beta, \tau, \vH, \{\vA^a\}_a}_\vrho(\eta \hat{\balpha}) ]] \right)
    \end{align}
    for some $0 \leq \eta \leq \norm{\balpha}_1$.
\end{proposition}

We define the energy gradients as follows. We separately consider a positive and a negative energy gradient. The motivation of the definition is that the positive (negative) energy gradient should determine the direction of the thermodynamics that causes the energy of the state to increase (decrease).
Because our goal is to understand local minima, we will focus on the negative energy gradient. When one studies local maxima, one will focus on the positive energy gradient.

\begin{definition}[Energy gradients of a state under thermal perturbations] \label{def:energy-grad}
    Given an $n$-qubit Hamiltonian $\vH$, $m$ local jump operators $\{\vA^a\}_a$, and parameters $\beta, \tau \ge 0$, the gradients of an $n$-qubit state ${\vrho} \in \mathcal{S}_{2^n}$ under thermal perturbations with inverse temperature $\beta$, time scale $\tau$, and system-bath interactions generated by $\{\vA^a\}_a$ is defined as
    \begin{align}
        {\bm{\nabla}}^+_{\beta, \tau, \{\vA^a\}_a}(\vH, \vrho) &:= \sum_{a=1}^m \max\left(+ \Tr\left( \vH \CL^{\beta, \tau, \vH}_a[\vrho] \right), 0 \right) \hat{\be}_a, & \text{(positive energy gradient)}\\
        {\bm{\nabla}}^-_{\beta, \tau, \{\vA^a\}_a}(\vH, \vrho) &:= \sum_{a=1}^m \max\left(- \Tr\left( \vH \CL^{\beta, \tau, \vH}_a[\vrho] \right), 0 \right) \hat{\be}_a, & \text{(negative energy gradient)}\\
        {\bm{\nabla}}_{\beta, \tau, \{\vA^a\}_a}(\vH, \vrho) &:= \sum_{a=1}^m \Tr\left( \vH \CL^{\beta, \tau, \vH}_a[\vrho] \right) \hat{\be}_a, & \text{(energy gradient)}
    \end{align}
    where $\hat{\be}_a$ is the unit vector along the $a$-th coordinate.
\end{definition}

Since the set of jump operators $\{\vA^a\}_a$ will be fixed, we will sometimes drop the dependence on $\{\vA^a\}_a$ for notational simplicity.
The positive/negative energy gradient belongs to the tangent space $\mathbb{R}_{\geq 0}^m$, but the energy gradient
\begin{equation}
    {\bm{\nabla}}_{\beta, \tau, \{\vA^a\}_a}(\vH, \vrho) = {\bm{\nabla}}^+_{\beta, \tau, \{\vA^a\}_a}(\vH, \vrho) - {\bm{\nabla}}^-_{\beta, \tau, \{\vA^a\}_a}(\vH, \vrho)
\end{equation}
may not be in the tangent space due to negative values. So, in general, one could not move in the direction of the energy gradient.
However, one could move in the direction of the positive or negative energy gradient. It is instructive to think about the Heisenberg picture and define the \textit{energy gradient operator}.

\begin{definition}[Energy gradient operator]
    Given an $n$-qubit Hamiltonian $\vH$, $m$ local jump operators $\{\vA^a\}_a$, inverse temperature $\beta \geq 0,$ and time scale $\tau \ge 0$, the energy gradient operators under thermal perturbations is
    \begin{equation}
        \sum_{a=1}^m \CL^{\dag\beta, \tau, \vH}_a[\vH] \,\, \hat{\be}_a,
    \end{equation}
    which is a vector of $n$-qubit Hermitian operators.
\end{definition}

We can provide an upper and lower bound to the energy gradients by combining Prop.~\ref{prop:norm-diss-part} and Prop.~\ref{prop:norm-lamb-shift-part} to obtain the following proposition.

\begin{proposition}[Bound on the energy gradients] \label{prop:bound-ene-grad}
    Given an $n$-qubit Hamiltonian $\vH$, $m$ local jump operators $\{\vA^a\}_a$, inverse temperature $\beta \geq 0,$ and time scale $\tau \ge 0$,
    \begin{align}
        \norm{\CL^{\dag \beta, \tau, \vH}_a(\vH)}_\infty \leq 3 \norm{\vH}_\infty
    \end{align}
    for all $a = 1, \ldots, m$.
\end{proposition}

The $\beta, \tau \rightarrow\infty$ limit (zero temperature heat bath with an infinite time scale) recovers the Davies' generator $\CL^{\infty,\infty,\vH}_a$.
The Davies' generator takes an energy eigenvector $\ketbra{\psi_j}{\psi_j}$ of $\vH$ to energy eigenvectors with equal or lower energy, i.e., for any $t \geq 0$,
\begin{equation}
    \bra{\psi_k} \exp\left(t \CL^{\infty,\infty,\vH}_a \right)\left(\ketbra{\psi_j}{\psi_j}\right) \ket{\psi_k} = 0 \quad \text{for any}\quad j, k\quad\text{such that}\quad  E_k > E_j.
\end{equation}
We can use the above to obtain the following proposition.
\begin{proposition}[Vanishing positive energy gradient]
\label{prop:L-nonpositive}
For $\beta = \tau = \infty$, we have
\begin{align}
    \CL^{\infty,\infty,\vH \dag}_a[\vH] \preceq 0, \quad \text{for each}\quad a.
\end{align}
Hence, the positive energy gradient vanishes ${\bm{\nabla}}^+_{\infty, \infty}(\vH, \vrho) = 0$ and ${\bm{\nabla}}_{\infty, \infty}(\vH, \vrho) = -{\bm{\nabla}}^-_{\infty, \infty}(\vH, \vrho)$ for all Hamiltonian $\vH$ and state $\vrho$.
\end{proposition}

This proposition illustrates that thermal perturbations induced by a zero-temperature heat bath with an infinite time scale should only absorb energy from the quantum system and not cause the energy to increase. Hence, the positive energy gradient must vanish.

\subsection{A sufficient condition and a necessary condition of local minima}

Using the negative gradient, we can show a necessary condition and a sufficient condition for local minima under thermal perturbations.
They differ only slightly ($<$ vs $\leq$).
From the conditions, we can see that local minima are well characterized by the negative energy gradient.
Recall that $\norm{\bx}_\infty = \max_i |x_i| $ is the $\ell_\infty$ norm for a finite-dimensional vector $\bx$.

\begin{lemma}[A sufficient condition for local minima under thermal perturbations]
\label{lem:suff-QLM}
Given $\epsilon > 0$, an $n$-qubit Hamiltonian $\vH$, $m$ local jump operators $\{\vA^a\}_a$, and parameters $\beta, \tau \ge 0$, an $n$-qubit state $\vrho$ with a small negative energy gradient,
\begin{equation} \label{eq:QLM-1st-order-strict}
    \norm{{\bm{\nabla}}^-_{\beta, \tau, \{\vA^a\}_a}(\vH, \vrho)}_\infty < \epsilon,
\end{equation}
is an $\epsilon$-approximate local minimum ${\vrho}$ of the $n$-qubit Hamiltonian $\vH$ under thermal perturbations with inverse temperature $\beta$, time scale $\tau$, and system-bath interactions generated by $\{\vA^a\}_a$.
\end{lemma}
\begin{proof}
    Consider $C_L = \max_{a} \norm{\mathcal{L}_a^{\beta, \tau, \vH}}_{1 - 1} > 0$ and $C_H = \norm{\vH}_\infty$.
    Given $\balpha \in \mathbb{R}^m_{\geq 0}$, we have
    \begin{align}
        \left| \sum_{a} \sum_{a'} \alpha_a \alpha_{a'} \Tr(\vH \CL^{\beta, \tau, \vH}_{a'} [\CL^{\beta, \tau, \vH}_a[\vsigma]]) \right| &\leq C_L^2 C_H \norm{\balpha}_1^2,
    \end{align}
    for any state $\vsigma$.
    Let $\epsilon_0 := \norm{{\bm{\nabla}}^-_{\beta, \tau, \{\vA^a\}_a}(\vH, \vrho)}_\infty < \epsilon$.
    From $\alpha_a \geq 0$ and Cauchy-Schwarz inequality,
    \begin{equation}
        \sum_a \alpha_a \Tr(\vH \CL^{\beta, \tau, \vH}_a[\vrho]) \geq - \sum_a \left|\alpha_a\right| \max\left(- \Tr\left( \vH \CL^{\beta, \tau, \vH}_i[\vrho] \right), 0 \right) \geq -\norm{\balpha}_1 \epsilon_0.
    \end{equation}
    Together, Taylor's theorem for thermal perturbations (Prop.~\ref{prop:taylor-thermal}) implies
    \begin{equation}
        \Tr\left(\vH \exp^{\beta, \tau, \vH, \{\vA^a\}_a}_\vrho(\balpha) \right) \geq \Tr(\vH \vrho) - \norm{\balpha}_1 \epsilon_0 - \frac{\norm{\balpha}^2_1}{2} C_L^2 C_H,
    \end{equation}
    for any $\balpha \in \mathbb{R}^m_{\geq 0}$.
    From the above, we see that for any $\norm{\balpha}_1^2 < \delta := \frac{2 (\epsilon - \epsilon_0)}{C_L^2 C_H}$,
    \begin{align}
        \Tr\left(\vH \exp^{\beta, \tau, \vH, \{\vA^a\}_a}_\vrho(\balpha) \right) &\geq \Tr(\vH \vrho) - \epsilon \norm{\balpha}_1 + \norm{\balpha}_1\left( (\epsilon - \epsilon_0) - \frac{C_L^2 C_H}{2} \norm{\balpha}_1^2 \right) \nonumber \\
        &\geq \Tr(\vH \vrho) - \epsilon \norm{\balpha}_1.
    \end{align}
    So, $\vrho$ is an $\epsilon$-approximate local minimum of $\vH$ under thermal perturbations with inverse temperature $\beta$, time scale $\tau$, and system-bath interactions generated by $\{\vA^a\}_a$.
\end{proof}

\begin{lemma}[A necessary condition for local minima under thermal perturbations]
\label{lem:nece-QLM}
Given $\epsilon > 0$, an $n$-qubit Hamiltonian $\vH$, $m$ local jump operators $\{\vA^a\}_a$, and parameters $\beta, \tau \ge 0$, an $\epsilon$-approximate local minimum ${\vrho}$ of $\vH$ under thermal perturbations with inverse temperature $\beta$, time scale $\tau$, and system-bath interactions generated by $\{\vA^a\}_a$ satisfies
\begin{equation} \label{eq:QLM-1st-order}
    \norm{{\bm{\nabla}}^-_{\beta, \tau, \{\vA^a\}_a}(\vH, \vrho)}_\infty \leq \epsilon,
\end{equation}
which differs only slightly from the condition in Eq.~\eqref{eq:QLM-1st-order-strict}.
\end{lemma}
\begin{proof}
    Recall that ${\bm{\nabla}}^-_{\beta, \tau, \{\vA^a\}_a}(\vH, \vrho) \in \mathbb{R}^m_{\geq 0}$.
    Let $a^* = \argmax_a \left( {\nabla}^-_{\beta, \tau, \{\vA^a\}_a}(\vH, \vrho)_a \right)$.
    If the negative energy gradient vector $\Tr(\vH \mathcal{L}_{a^*}^{\beta, \tau, \vH}[\vrho])$ is zero, then the claim holds.
    Hence, we only need to consider the case when the negative energy gradient vector is nonzero.
    In this case,
    \begin{equation}
        0 < \norm{{\bm{\nabla}}^-_{\beta, \tau, \{\vA^a\}_a}(\vH, \vrho)}_\infty = {\nabla}^-_{\beta, \tau, \{\vA^a\}_a}(\vH, \vrho)_{a^*} = - \Tr(\vH \mathcal{L}_{a^*}^{\beta, \tau, \vH}[\vrho]).
    \end{equation}
    Consider $\hat{\balpha} := \hat{\be}_{a^*} \in \mathbb{R}^m_{\geq 0}$, which satisfies $\norm{\hat{\balpha}}_1 = 1$. We have
    \begin{align}
        \lim_{t \rightarrow 0^+} \frac{\Tr(\vH \exp^{\beta, \tau, \vH, \{\vA^a\}_a}_\vrho(t\hat{\balpha})) - \Tr(\vH \vrho)}{t} &= \Tr(\vH \mathcal{L}_{a^*}^{\beta, \tau, \vH}[\vrho]) \nonumber \\
        &= - \norm{{\bm{\nabla}}^-_{\beta, \tau, \{\vA^a\}_a}(\vH, \vrho)}_\infty.
    \end{align}
    At the same time, for any $t > 0$, we also have
    \begin{equation}
        \frac{\Tr(\vH \exp^{\beta, \tau, \vH, \{\vA^a\}_a}_\vrho(t\hat{\balpha})) - \Tr(\vH \vrho)}{t} \geq -\epsilon \norm{\hat{\balpha}}_1 = -\epsilon.
    \end{equation}
    Together, we obtain the desired claim.
\end{proof}

\subsection{Hamiltonians without suboptimal local minima}
\label{sec:Hamil-suboptimal-local-minima}

An important concept in classical optimization is to understand when all local minima are global minima.
For example, in convex optimization, checking the convexity of the objective function $h(\bx)$ ensures that all local minima are global minima.
When all local minima are global minima, it is commonly referred to as having no suboptimal local minima in the optimization landscape.
For optimizing quantum Hamiltonians, we can define a similar concept.
Let us begin with a definition of approximate global minimum.

\begin{definition}[Approximate global minimum of Hamiltonians]
    Given $\epsilon, \delta > 0$ and an $n$-qubit Hamiltonian $\vH$ with minimum energy $E_0$.
    Let $\vP_{G+\epsilon}(\vH)$ be the projector to the subspace of energy eigenstates of $\vH$ with energy at most $E_0 + \epsilon$.
    An $n$-qubit state $\vrho$ is an $\epsilon$-approximate global minimum of $\vH$ with failure probability $\leq \delta$ if $\Tr(\vP_{G+\epsilon}(\vH) \vrho) \geq 1 - \delta$.
\end{definition}
\begin{definition}[No suboptimal local minima]
    Given $\epsilon > 0$.
    We say an $n$-qubit Hamiltonian $\vH$ has no suboptimal $\epsilon$-approximate local minima with failure probability $\delta$ if any $\epsilon$-approximate local minimum~$\vrho$ of $\vH$ is an $\epsilon$-approximate global minimum of $\vH$ with failure probability $\leq \delta$, i.e., $\Tr(\vP_{G+\epsilon}(\vH) \vrho) \geq 1-\delta$.
\end{definition}

While the above definitions apply to any Hamiltonian $\vH$, in this work, we will focus on Hamiltonians with a gap $\Delta > 0$ between the minimum energy and the second minimum energy, also known as the spectral gap.
By definition of $\vP_{G+\epsilon}(\vH)$ and spectral gap $\Delta$, we have
\begin{equation}
    \epsilon < \Delta \implies \vP_{G+\epsilon}(\vH) = \vP_{G}(\vH).
\end{equation}
As we will almost always consider $\epsilon < \Delta$, any $\epsilon$-approximate global minimum is an $0$-approximate global minimum or \emph{exact} global minimum.

Convexity implies all local minima are global in classical optimization. In the following lemma, we present a sufficient condition for ensuring that all local minima are global in quantum systems.
As we can see, all we need is to check the negative gradient operator is sufficiently positive in the non-ground-state space $\vI - \vP_G$.
We will refer to this as the \emph{negative gradient condition}.

\begin{lemma}[A sufficient condition ensuring all local minima are global]
\label{lem:suff-cond-no-suboptimal}
    Given $\epsilon, \delta > 0$, an $n$-qubit Hamiltonian $\vH$, $m$ local jump operators $\{\vA^a\}_a$, and parameters $\beta, \tau \ge 0$.
    Let $\vP_G(\vH)$ be the projection onto the ground state space of $\vH$.
    If there exists $\balpha \in \BR^m_{\geq 0}$ with $\norm{\balpha}_1 = 1$, such that the negative gradient operator satisfies
    \begin{equation} \label{eq:suff-cond-no-suboptimal}
        \text{(negative gradient condition):}\quad -\sum_{a} \alpha_a \CL_a^{\dag \beta, \tau, \vH} [\vH] \succeq \frac{2 \epsilon}{\delta} (\vI - \vP_G(\vH)) - \epsilon \vI,
    \end{equation}
    then any $\epsilon$-approximate local minimum~$\vrho$ of the $n$-qubit Hamiltonian $\vH$ under thermal perturbations with inverse temperature~$\beta$, time scale $\tau$, and system-bath interactions generated by $\{\vA^a\}_a$ is an exact global minimum with failure probability $\leq \delta$. That is, $\Tr(\vP_G(\vH) \vrho) \geq 1 - \delta$.
\end{lemma}
\begin{proof}
    From the necessary condition for local minima in Lemma~\ref{lem:nece-QLM}, any $\epsilon$-approximate local minimum~$\vrho$ of the $n$-qubit Hamiltonian $\vH$ under thermal perturbations with inverse temperature $\beta$, time scale $\tau$, and system-bath interactions generated by $\{\vA^a\}_a$ satisfies
    \begin{equation}
        -\Tr(\CL_a^{\dag \beta, \tau, \vH} [\vH]\vrho) \leq \epsilon \quad \text{for each}\quad a = 1, \ldots, m.
    \end{equation}
    Hence, from Eq.~\eqref{eq:suff-cond-no-suboptimal}, we have
    \begin{equation}
        \epsilon \geq -\sum_{a} \alpha_a \Tr(\CL_a^{\dag \beta, \tau, \vH}[\vH] \vrho) \geq \frac{2 \epsilon}{\delta} (1 - \Tr(\vP_G(\vH) \vrho)) - \epsilon.
    \end{equation}
    This immediately implies that $\Tr(\vP_G(\vH) \vrho) \geq 1 - \delta$.
\end{proof}

\section{Complexity of finding a local minimum in quantum systems}
\label{sec:complexity-main}

In this appendix, we formally present the main results of this paper shown earlier in \cref{sec:results-main} regarding the computational complexity of finding a local minimum in quantum systems.
We separate the results into two parts.
First, we look at the problem of finding a local minimum under local unitary perturbations (Def.~\ref{def:unitary-LU}), showing that the problem is classically trivial to solve.
Next, we look at the problem of finding a local minimum under low-temperature thermal perturbations (Def.~\ref{def:thermal-LM}). We will see that this problem is quantumly easy but classically hard to solve, establishing a promising candidate problem for quantum advantage.

\subsection{Finding a local minimum under local unitary perturbations}
\label{sec:main-unitary}

We begin with the first main result stating the problem of a local minimum under local unitary perturbations is classically trivial.
The main issue is that there is a large barren plateau (which consists of many local minima with high energy) in the quantum optimization landscape.
Hence, a classical algorithm can efficiently estimate the properties of a single local minimum.

\begin{theorem}[Classically easy to find a local minimum under local unitary perturbations; Restatement of Theorem~\ref{thm:classical-easy-local-unitary-informal}] \label{thm:classical-easy-local-unitary}
    Consider a large problem size $n$.
    There is a trivial classical algorithm that guarantees the following.
    Given error $\epsilon = 1 / \mathrm{poly}(n)$, an $n$-qubit local Hamiltonian~$\vH$ with $\norm{\vH}_\infty = \mathrm{poly}(n)$, $m$ local Hermitian operators $\{\vh^a\}_{a=1}^m$ with $m = \mathrm{poly}(n)$ and $\norm{\vh_a}_\infty = 1$, and a local observable $\vO$ with $\norm{\vO}_\infty \leq 1$.
    
    The classical algorithm runs in time $\mathcal{O}(1)$ and outputs a real value $v \in [-1, 1]$, such that $v$ is $\epsilon$-close to $\bra{\psi} \vO \ket{\psi}$ for an $\epsilon$-approximate local minimum $\ket{\psi}$ of the Hamiltonian $\vH$ under local unitary perturbations generated by $\{\vh^a\}_a$.
\end{theorem}
\begin{proof}
    From Lemma~\ref{lem:Haar-concentration} given in Appendix~\ref{sec:barren-plateau-energy-unitary} characterizing local minima of $\vH$ under local unitary perturbations, with high probability, a state $\ket{\psi}$ sampled uniformly at random from the space of pure states is an $\epsilon$-approximate local minimum $\ket{\psi}$ of the local Hamiltonian $\vH$ under local unitary perturbations, and $\bra{\psi} \vO \ket{\psi}$ is $\epsilon$-close to $\Tr(\vO) / 2^n$.
    Hence, there exists an $\epsilon$-approximate local minimum $\ket{\psi}$ of $\vH$ under local unitary perturbations, and $\bra{\psi} \vO \ket{\psi}$ is $\epsilon$-close to $\Tr(\vO) / 2^n$.

    This characterization of local minima gives rise to the following trivial classical algorithm.
    Given a local observable $\vO$, represented by the subset $S$ of qubits $\vO$ acts on with $|S| = \mathcal{O}(1)$ and a $2^{|S| \times |S|}$ Hermitian matrix $\vO^*$.
    A classical algorithm can compute $\Tr(\vO) / 2^n$ efficiently by computing the trace of $\vO^*$ and dividing by $2^{|S|}$. This trivial classical algorithm runs in time $\mathcal{O}(1)$.
\end{proof}

\subsection{Finding a local minimum under thermal perturbations}
\label{sec:main-thermal}

We now turn to the second main result of this work, which shows that finding a local minimum under low-temperature thermal perturbations is easy with a quantum computer.
This is in contrast to the task of finding the ground state (global minimum), which is hard on quantum computers.
The formal statement is given below in Theorem~\ref{thm:quantum-easy-thermal}.

\begin{theorem}[Quantumly easy to find a local minimum under thermal perturbations; Restatement of Theorem~\ref{thm:quantum-easy-thermal-informal}] \label{thm:quantum-easy-thermal}
    Let $n$ be the problem size.
    There is a $\mathrm{poly}(n)$-time quantum algorithm that guarantees the following.
    Given error $\epsilon = 1 / \mathrm{poly}(n)$, inverse temperature $0 \leq \beta \leq \mathrm{poly}(n)$, time scale $\tau = \mathrm{poly}(n)$, an $n$-qubit local Hamiltonian~$\vH$ with $\norm{\vH}_\infty = \mathrm{poly}(n)$, $m$ local jump operators $\{\vA^a\}_{a=1}^m$ with $m = \mathrm{poly}(n)$, and a local observable $\vO$ with $\norm{\vO}_\infty \leq 1$.

    The quantum algorithm outputs a real value $v \in [-1, 1]$, such that $v$ is $\epsilon$-close to $\Tr(\vO \vrho)$ for an $\epsilon$-approximate local minimum $\vrho$ of $\vH$ under thermal perturbations with an inverse temperature $\beta$, a time scale $\tau$, and system-bath interactions generated by $\{\vA^a\}_{a}$.
\end{theorem}
\begin{proof}[Proof idea]
We consider a version of gradient descent, which we refer to as \emph{Quantum thermal gradient descent}, that mimics how Nature cools the quantum system when the system is interacting locally and weakly with a low-temperature heat bath.
The algorithm starts with an arbitrary initial state $\vrho_0$.
For each step $t = 0, 1, 2, \ldots$, the algorithm considers the current state $\vrho_t$ and proposes the next state $\vrho_{t+1}$.
The tangent space $T_{\vrho_t}^{\beta, \tau, \vH}$ at $\vrho_t$ is high dimensional with many possible directions/dynamics depending on the system-bath interaction.
The algorithm chooses a direction that lowers the energy as fast as possible by computing the gradient of the energy and proposes $\vrho_{t+1}$ by performing gradient descent.
As long as the current state $\vrho_t$ is not an $\epsilon$-approximate local minimum of $\vH$ under thermal perturbations, the energy will decrease by a sufficiently large amount 
\begin{equation}
    \Tr(\vH \vrho_{t+1}) < \Tr(\vH \vrho_{t}) - \frac{1}{\mathrm{poly}(n)}.
\end{equation}
Because the energy is bounded from below, there are, at most, a polynomial number of steps $t \leq \mathrm{poly}(n)$ until the algorithm arrives at an $\epsilon$-approximate local minimum of $\vH$ under thermal perturbations.
The detailed proof of Theorem~\ref{thm:quantum-easy-thermal} is given in Appendix~\ref{sec:proof-thm-quantum-easy-thermal}.
\end{proof}

Finally, we turn to the third main result establishing the difficulty of finding a local minimum under thermal perturbations using a classical computer.
To establish this result, we consider a class of geometrically local Hamiltonians $\{\vH_C\}_C$ on 2D lattices.
Each Hamiltonian $\vH_C$ corresponds to a 2D circuit $\vU_{C}=\vU_T\cdots\vU_2 \vU_1$ acting on $n$ qubits with $T = 2 t_0 + L = \mathrm{poly}(n)$ gates as constructed in Fig. 1 of \cite{OliveiraTerhal} with the additional padding to the construction in \cite{OliveiraTerhal} such that the first and last $t_0 = cL^2$ gates being the identity gates for a constant $c=\CO(1)$.
The construction in \cite{OliveiraTerhal} has the property that each gate of the 2D circuit $\vU_C$ is geometrically adjacent to the subsequent gate.

Given the 2D circuit $\vU_{C}$ on $n$ qubits with $T$ gates.
The geometrically local Hamiltonian $\vH_C$ acts on $n + T$ qubits on a 2D lattice and has a highly-entangled unique ground state that encodes the quantum computation based on the 2D circuit $\vU_C$,
\begin{equation}
    \ket{\eta_{\bm0}} = \sum_{t=0}^T \sqrt{\xi_t}\big(\vU_t \cdots \vU_1 \ket{0^n}\big) \otimes \ket{0^t 1^{T-t}},
    \qquad
    \text{where} \quad
    \xi_t := \frac{1}{2^T}  \binom{T}{t}.
\end{equation}
We present the detailed construction of the 2D Hamiltonian $\vH_C$ in \cref{def:circuit-H} in Appendix~\ref{sec:universal-quantum-computation}.
We have the following proposition for estimating single-qubit observables on the ground state of $\vH_C$.
\begin{proposition}[$\mathsf{BQP}$-hardness for estimating properties of the ground state of $\vH_C$] \label{prop:BQP-hardness-estimate-prop-VHC}
    If there is a classical algorithm that can estimate any single-qubit observable on the unique ground state of the geometrically local Hamiltonian $\vH_C$ in time polynomial in the number of qubits in $\vH_C$ to error $1 / 4$ for any $\vH_C$ in the class, then $\mathsf{BPP} = \mathsf{BQP}$.
\end{proposition}
\begin{proof}
Consider the single-qubit observable $\vZ_{j}$ and let $T_j$ be the last time that qubit $j$ is acted by a gate in the circuit $C$.
The ground state expectation of $\vZ_j$ is
\begin{align}
\braket{\eta_{\bm0}| \vZ_j |\eta_{\bm0}} &= \sum_{t = T_j + 1}^T \xi_t \braket{0^n|\vU_1^\dag \cdots \vU_{t}^\dag \vZ_j \vU_t \cdots \vU_1|0^n} + \sum_{t = 0}^{T_j} \xi_t \braket{0^n|\vU_1^\dag \cdots \vU_{t}^\dag \vZ_j \vU_t \cdots \vU_1|0^n} \nonumber\\
&= \braket{0^n|\vU_1^\dag \cdots \vU_T^\dag \vZ_j \vU_T \cdots \vU_1|0^n} P_{t > T_j} + \epsilon_j,\\
\text{where} &\quad P_{t > T_j} := \sum_{t > T_j} \xi_t, \quad \epsilon_{j} := \sum_{t \le T_j} \xi_t \braket{0^n|\vU_1^\dag \cdots \vU_t^\dag \vZ_j \vU_t \cdots \vU_1|0^n}.
\end{align}
We have used the fact that $U_t$ for $t > T_j$ acts like identity on the $j$-th qubit.
Note that
\begin{equation}
    |\epsilon_{j}| \le 1 - P_{t > T_j} =: P_{t \leq T_j}.
\end{equation}
We can make $\epsilon_j$ arbitrarily small using a tail bound on the binomial distribution.
Given any circuit, one could always pad more identity gates to form an $L$-gate circuit, such that the last 3/4 of the $L$ gates are identity.
Recall that $T = 2 t_0 + L = 2 c L^2+L$. Then $T_j \le cL^2+L/4$ and we have $|\epsilon_j| \le P_{t \leq T_j} \le P_{t \leq c L^2+L/4} $.

Using Hoeffding's inequality, we can bound the probability of sampling a time $t$, such that $t \le c L^2+L/4$, according to the Binomial distribution $\{\xi_t\}_{t=0}^T$. This yields,
\begin{equation}
    |\epsilon_j| \le \exp\left[-2T\left(\frac12 - \frac{cL^2+L/4}{2cL^2+L}\right)^2\right] = \e^{- \frac{L}{8+16 cL}}.
\end{equation}
By choosing a small constant $c\le 1/(16\ln18) - 1/(2L)$, we have $|\epsilon_j| \le 1 / 18$ and $P_{t > T_j} \ge 17 / 18$.
Because of the bounds on the error $|\epsilon_j|$ and the probability $P_{t > T_j}$, a classical algorithm satisfying the assumption of the proposition can determine whether
\begin{equation}
    \bra{0^n} \vU_1^\dag \cdots \vU_T^\dag Z_j \vU_T \cdots \vU_1 \ket{0^n} > 1 / 3 \quad \mathrm{or} \quad \bra{0^n} \vU_1^\dag \cdots \vU_T^\dag Z_j \vU_T \cdots \vU_1 \ket{0^n} < -1 / 3,
\end{equation}
for any 2D circuit $\vU_C$ with $T = 2 t_0 + L = 2 c L^2 + L$ gates, where the first $t_0$ and and the last $(3/4)L + t_0$ gates are identity.
Because one could think of the circuit $\vU_C$ as having $L/4$ gates for any $L = \mathrm{poly}(n)$, this immediately implies that a polynomial-time classical algorithm can decide whether the expectation value of $Z_j$ on the output state is $> 1/3$ or $< -1/3$ for any polynomial-size 2D circuit where all consecutive gates are adjacent in the 2D geometry.

A 2D circuit $\vU_C$ such that any gate is adjacent to the subsequent gate can be constructed from any quantum circuit without the 2D constraint, such that a single-qubit observable $Z_i$ on the output of the original quantum circuit corresponds to a single-qubit observable $Z_j$ on the output of the 2D circuit.
As a result, any polynomial-time classical algorithm that can determine whether the expectation value of $Z_j$ on $\vU_t \cdots \vU_1 \ket{0^n}$ is greater than $1/3$ or smaller than $-1/3$ can be used to simulate any polynomial-time quantum algorithm for solving decision problems in classical polynomial time. Hence, $\mathsf{BPP} = \mathsf{BQP}$.
\end{proof}

Using a series of mathematical techniques presented in Appendix~\ref{sec:characterize-neg-grad-condition} for characterizing whether all local minima are global minima in a many-body Hamiltonian, we prove that all local minima $\vH_C$ are close to the unique ground state $\ket{\eta_{\bm0}}$ in Theorem~\ref{thm:no-suboptimal-local-minima-informal2}.
This theorem is the most involved technical contribution of this work.
Intuitively, one can think of the energy landscape of the 2D Hamiltonian $\vH_C$ over the space of $n$-qubit density matrices under low-temperature thermal perturbations to have a good bowl shape.
This is in stark contrast to the energy landscape under local unitary perturbations, where the landscape always contains an overwhelmingly large barren plateau causing the problem of finding local minima to be classically easy.
Furthermore, this theorem shows that a low-temperature cooling can always find a state close to the ground state irrespective of where we initialize the state in the exponentially large quantum state space.

\begin{theorem}[All local minima are global in $\mathsf{BQP}$-hard Hamiltonians; Restatement of Theorem~\ref{thm:no-suboptimal-local-minima-informal}] \label{thm:no-suboptimal-local-minima-informal2}
Let $\vP_G(\vH_C) = \ketbrat{\eta_{\bm0}}$ be the ground state of the 2D Hamiltonian $\vH_{C}$ acting on $n + T = \mathrm{poly}(n)$ qubits.
There is a choice of $m = \poly(n)$ two-qubit jump operators $\{\vA^a\}_a$ satisfying the following.

Given $0 < \delta < 1$. For any small error $\epsilon = 1 / \poly(n,1/\delta)$, any $\epsilon$-approximate local minimum $\vrho$ of $\vH_{C}$ under thermal perturbations with a large inverse temperature $\beta = \poly(n, 1/\delta)$, a large time scale $\tau = \poly(n, 1/\delta)$, and system-bath interactions generated by $\{\vA^a\}_a$ is an exact global minimum with high probability, i.e., we have $\Tr(\vP_G(\vH_C) \vrho) \geq 1 - \delta$.
\end{theorem}

The proof of Theorem~\ref{thm:no-suboptimal-local-minima-informal2} is given in Appendix~\ref{sec:universal-quantum-computation}.
To show that the landscape has a good bowl shape, we utilize the negative gradient condition given in Appendix~\ref{sec:Hamil-suboptimal-local-minima}.
However, the negative energy gradient operator is not easy to study.
To establish this strong claim, we give a series of techniques in Appendix~\ref{sec:characterize-neg-grad-condition} for characterizing negative energy gradient operator in few-qubit systems, in commuting Hamiltonians, and in perturbed Hamiltonians.
These technical tools can also be used to understand the energy landscape in other interacting many-body Hamiltonians.

While finding a local minimum under local unitary perturbations is classically easy,
the characterization of the energy landscape in these $\mathsf{BQP}$-hard Hamiltonians $\vH_C$ implies that finding a local minimum under thermal perturbations is \emph{universal for quantum computation} and is hence classically hard if $\mathsf{BPP} \neq \mathsf{BQP}$.
Recall that $\mathsf{BPP} = \mathsf{BQP}$ implies that all single-qubit measurements of all polynomial-size quantum circuits can be simulated in polynomial time on a classical computer.
Since one expects some quantum circuits to be hard to simulate on a classical computer, Theorem~\ref{thm:classical-hard-LM-thermal} implies that finding a local minimum under thermal perturbations is classically hard.

\begin{theorem}[Classically hard to find a local minimum under thermal perturbations; Restatement of Theorem~\ref{thm:classical-hard-LM-thermal-informal}] \label{thm:classical-hard-LM-thermal}
    Let $n$ be the problem size.
    Suppose there is a $\mathrm{poly}(n)$-time classical algorithm guaranteeing the following.
    Given error $\epsilon = 1 / \mathrm{poly}(n)$, inverse temperature $0 \leq \beta \leq \mathrm{poly}(n)$, time scale $0 \leq \tau \leq \mathrm{poly}(n)$, an $n$-qubit local Hamiltonian~$\vH$ with $\norm{\vH}_\infty = \mathrm{poly}(n)$, $m$ local jump operators $\{\vA^a\}_{a=1}^m$ with $m = \mathrm{poly}(n)$, and a single-qubit observable $\vO$ with $\norm{\vO}_\infty \leq 1$.
    
    The classical algorithm outputs a real value $v \in [-1, 1]$, such that $v$ is $\epsilon$-close to $\Tr(\vO \vrho)$ for an $\epsilon$-approximate local minimum $\vrho$ of the Hamiltonian $\vH$ under thermal perturbations with an inverse temperature $\beta$, a time scale $\tau$, and system-bath interactions generated by $\{\vA^a\}_{a}$.
    Then $\mathsf{BPP} = \mathsf{BQP}$.
\end{theorem}
\begin{proof}
    Assuming the existence of a polynomial-time classical algorithm that satisfies the properties stated in the theorem.
    Apply this classical algorithm to the 2D Hamiltonian $\vH_C$ considered in Theorem~\ref{thm:no-suboptimal-local-minima-informal2} with a sufficiently small approximation error $\epsilon$, such that any $\epsilon$-approximate local minimum $\vrho$ of $\vH_{C}$ under thermal perturbations with polynomially-large $\beta$, $\tau$ and system-bath interactions generated by $\{\vA^a\}_a$ is an exact global minimum with high probability, i.e.,
    \begin{equation}
        \bra{\eta_{\bm0}} \vrho \ket{\eta_{\bm0}} = \Tr(\vP_G(\vH_C) \vrho) \geq 1 - \frac{1}{16^2},
    \end{equation}
    where $\ket{\eta_{\bm0}}$ is the unique ground state of $\vH_C$.
    We further consider $\epsilon$ to be small enough such that
    \begin{equation} \label{eq:epsilon-constraint-polyn}
        \epsilon < \frac{1}{8}.
    \end{equation}
    Let $\vrho$ be an $\epsilon$-approximate local minimum of the Hamiltonian $\vH$ under thermal perturbations.
    Consider the observable $\vO_j = \vZ_j$ from the proof of Proposition~\ref{prop:BQP-hardness-estimate-prop-VHC}.
    Using the Fuchs–van de Graaf inequalities, we have
    \begin{equation} \label{eq:Fuchs-van-de-graaf}
        \norm{\vrho - \ketbrat{\eta_{\bm0}}}_{1} \leq \frac{1}{8}.
    \end{equation}
    Because the classical algorithm can estimate $\Tr(\vO_j \vrho)$ to error $\epsilon$, from Eq.~\eqref{eq:epsilon-constraint-polyn}~and~\eqref{eq:Fuchs-van-de-graaf}, the classical algorithm can estimate $\braket{\eta_{\bm0}| \vO_j |\eta_{\bm0}}$ to error $1 / 4$ in time polynomial in the number of qubits in $\vH_C$.
    From Prop.~\ref{prop:BQP-hardness-estimate-prop-VHC}, this implies that $\mathsf{BPP} = \mathsf{BQP}$.
\end{proof}

In the following, we use the previous theorem to show that quantum machines can improve over any efficient classical algorithm that variationally optimizes a classical ansatz that can efficiently predict local properties by performing low-temperature cooling.
Examples of the classical ansatz include tensor networks with efficient tensor contraction algorithms and neural-network quantum states with fast sampling algorithms.
This result provides a physically-relevant problem that yields an advantage in minimizing the energy of a geometrically-local Hamiltonian.

\begin{corollary}[Quantum advantage over variationally optimized classical ansatz] \label{cor:quantum-adv-classical-ansatz}
    Under the conjecture that $\mathsf{BPP} \neq \mathsf{BQP}$, there exists a class of $n$-qubit geometrically-local Hamiltonian~$\vH$ on a two-dimensional lattice with $\norm{\vH}_\infty = \mathcal{O}(n)$ that satisfies the following.
    Given any classical ansatz of $n$-qubit state~$\vrho$ that can estimate the expectation value of single-qubit observables to $1 / \mathrm{poly}(n)$ error in $\mathrm{poly}(n)$-time on classical computers, any $\mathrm{poly}(n)$-time classical algorithm for minimizing the energy $\Tr(\vH \vrho)$ using the classical ansatz, and samples of the state $\vrho$ represented by the optimized classical ansatz.
    A quantum machine can find a state $\vrho^\#$ with strictly lower energy than $\vrho$ in $\mathrm{poly}(n)$ time by running a quantum thermal gradient descent based on low-temperature cooling.
\end{corollary}
\begin{proof}[Proof]
    The central claim is that the state $\vrho$ found by an efficient classical algorithm cannot be an $\epsilon$-approximate local minimum under low-temperature thermal perturbations.
    We establish this claim by contradiction. Suppose that the classical ansatz for $\vrho$ found by the efficient classical algorithm is an $\epsilon$-approximate local minimum.
    Then the classical algorithm can use the classical ansatz to predict the expectation values of single-qubit observables of an $\epsilon$-approximate local minimum $\vrho$ of $\vH$ to $\epsilon$ error. From Theorem~\ref{thm:classical-hard-LM-thermal}, this implies that $\mathsf{BPP} = \mathsf{BQP}$, which is a contradiction.

    Because $\vrho$ is not an $\epsilon$-approximate local minimum under low-temperature thermal perturbations, a quantum machine can use samples of $\vrho$ to initialize at the state $\vrho$ and perform one gradient descent step based on low-temperature cooling.
    From Lemma~\ref{lem:nece-QLM} on the necessary condition for local minima, there exists $a \in \{1, \ldots, m\}$ such that $\Tr(\vH \CL^{\beta, \tau, \vH}_a[\vrho]) < -\epsilon$.
    From Lemma~\ref{lem:cool-GD} on cooling by gradient descent, one can show that a single gradient descent step yields a state $\vrho^{\mathrm{(next)}}$ with a strictly lower energy than the state $\vrho$.
    Hence, one establishes the desired claim.
\end{proof}

\section{Details of thermal Lindbladians}
\label{sec:thermo-lindblad-detail}

In the rest of the appendices, we give the full detailed proofs of Theorems~\ref{thm:quantum-easy-thermal} and \ref{thm:no-suboptimal-local-minima-informal2} that are central to establishing the computational complexity of finding local minima under thermal perturbations in the previous appendix.
To that end, we need to provide the techincal details of thermal Lindbladians that generate such perturbations.

We have previously presented a high-level introduction to thermal Lindbladians in Appendix~\ref{sec:thermo-lindblad}. This has been sufficient for defining local minima and analyzing some basic properties, but not enough for proving Theorems~\ref{thm:quantum-easy-thermal} and \ref{thm:no-suboptimal-local-minima-informal2}.
In this appendix, we present the exact form of thermal Lindbladians, their properties, and the algorithmic primitives for simulating quantum thermodynamics.

\subsection{Exact form}
\label{sec:physics-thermo-Lind}

The exact form of the thermal Lindbladian depends on a few physical concepts due to the microscopic derivation from a system-bath interaction~\cite{mozgunov2020completely}. For each jump $\vA^a$, we have
\begin{equation} \label{eq:dissipative-Lind}
    \mathcal{L}_{a}^{\beta, \tau, \vH}(\vrho) := -\ri [ \vH^{\beta, \tau, \vH}_{LS,a}, \vrho ] + \undersetbrace{:= \mathcal{D}^{\beta, \tau, \vH}_{a}[\vrho]}{\int_{-\infty}^{\infty} \gamma_\beta(\omega) \Big[\hat{\vA}^{a}(\omega) \vrho \hat{\vA}^{a}(\omega)^{\dag} - \frac12 \{\hat{\vA}^{a}(\omega)^{\dag} \hat{\vA}^{a}(\omega), \vrho\}\Big] \rd\omega},
\end{equation}
where $\mathcal{D}^{\beta, \tau, \vH}_{a}$ is the purely dissipative part of the thermal Lindbladian. Implicitly, the operator $\hat{\vA}^{a}(\omega)$ also depends on the Hamiltonian $\vH$ and the time scale $\tau.$ We now unpack the physical concepts that form the building blocks of this expression.

\paragraph{Transition weight.}

At a fixed inverse temperature $\beta$, the \emph{transition weight} $\gamma_{\beta}(\omega)$ tells us how strong the rate of a transition/jump should be, depending on the energy difference $\omega$.
In particular, the transition weight satisfies the following \textit{Kubo-Martin-Schwinger (KMS) condition} and convenient normalization
\begin{align}
    \gamma_{\beta}(\omega)/\gamma_{\beta}(-\omega)=\e^{-\beta\omega} \quad \text{and} \quad 0\le \gamma_{\beta}(\omega)\le 1 \quad \text{for any}\quad \beta\ge 0  \quad \text{and any}\quad \omega \in \BR, \label{eq:gamma_KMS}
\end{align}
which is reminiscent of how detailed balance is enforced in classical Markov chains. We remark that any $\gamma_\beta(\omega)$ obeying the above KMS condition and normalization also satisfies the following tail bound 
\begin{equation}
\label{eq:gamma-tailbound}
    \max_{\omega\ge \Delta} \omega \gamma_\beta(\omega)  \le \max_{\omega \ge \Delta } \omega \e^{-\beta \omega} = \frac{1}{\beta} \max_{x \ge \beta\Delta } x \e^{-x} \le \frac{1}{\beta} \max_{x \ge \beta\Delta } \e^{-x/2} = \frac{\e^{-\beta\Delta/2}}{\beta} .
\end{equation}

For concreteness, we usually adopt the common choice of $\gamma_\beta$ corresponding to \textit{Glauber dynamics}, with a cut-off frequency $\Lambda_0$ to regulate the inverse Fourier transform:
\begin{equation}
   \gamma_{\beta}(\omega) = \frac{1}{2 +  \ln(1+\beta \Lambda_0)}  \frac{\e^{-\omega^2/2\Lambda_0^2}}{1+\e^{\beta \omega}}.
   \label{eq:glauber-dyn}
\end{equation}
In the zero temperature regime ($\beta = \infty$), the function $(1+\e^{\beta \omega})^{-1}$ gives a step function (one for negative $\omega$ and zero for positive $\omega$).
Based on the choice of bath phenomenology, there are plenty of options for the transition 
weight, such as ohmic heating $\gamma_{\beta}(\omega) = \frac{1}{\omega_0} \frac{\omega \e^{-\omega^2/2\Lambda_0^2}}{1-\e^{-\beta \omega}}$.
However, for simplicity, we will stick to the Glauber dynamics.
Unless otherwise stated, we will also choose the cut-off frequency to be
\begin{equation}
\label{eq:Lambda_0}
    \Lambda_0 = 1,
\end{equation}
since each local jumps $\vA^a$ changes the energy 
by at most $\CO(1)$ for our usage (and is generally true for local Hamiltonians with bounded degree interaction graph and bounded-norm terms).
We do not expect our main conclusion to change under other reasonable choices of $\gamma_{\beta}(\omega)$. 

\paragraph{Operator Fourier transform.}
{~}
Given a jump operator $\vA^a$, we consider the \textit{operator Fourier Transform}~\cite{Chen2023quantumthermal} for the Heisenberg-evolved jump operator $\vA^a(t)$ characterized by a time scale $\tau \in \BR$ of the heat bath
\begin{align}
\label{eq:Aomega}
    \hat{\vA}^a(\omega) &:= \frac{1}{\sqrt{2\pi \tau}}\int_{-\tau/2}^{\tau/2} \underset{=:\vA^a(t)}{\underbrace{\e^{\ri \vH t}\vA^a \e^{-\ri \vH t}}} \e^{-\ri\omega t} \rd t.
\end{align}
The operator $\hat{\vA}^a(\omega)$ corresponds to matrix elements in $\vA^a$ that induce jumps between energy eigenstates with an energy difference approximately $\omega \pm \CO(\frac{1}{\tau})$.
The bigger $\tau$ is, the more precise $\omega$ corresponds to the true energy difference; see Appendix~\ref{sec:OFT} for further details. Physically, $\tau$ is related to microscopic parameters of the bath (the bath correlation time and the weak-coupling strength~\cite{mozgunov2020completely}), but our discussion only requires the single time scale $\tau$ that sets the Fourier transform energy uncertainty.

\paragraph{Lamb-shift.}

The interaction with the heat bath induces an additional correction term in the coherent Hamiltonian dynamics of the $n$-qubit system, known as the \textit{Lamb-shift}.
Given a jump operator $\vA^a$, we have the following Lamb-shift Hamiltonian that depends on the bath correlation function $c_\beta(t)$ and the time scale $\tau$
\begin{align}
\vH^{\beta, \tau, \vH}_{LS,a} := \frac{\ri}{2\sqrt{2\pi}\tau} \int_{-\tau/2}^{\tau/2} \int_{-\tau/2}^{\tau/2} \textrm{sgn}(t_1-t_2)c_{\beta}(t_2-t_1)\vA^a(t_2)\vA^a(t_1) \rd t_2\rd t_1. \label{eq:Lamb_shift}
\end{align}
While the Lamb-shift term is physically important, for our purposes, it is largely treated as a source of error; the energy gradient contribution comes from the dissipative part $\mathcal{D}^{\beta, \tau, \vH}_{a}$.

\paragraph{Bath correlation function.}

In the Lamb-shift term, the bath correlation function $c_{\beta}(t)$ is the Fourier transform of the transition weight $\gamma_{\beta}(\omega)$,
\begin{equation}
    c_{\beta}(t) = \frac{1}{\sqrt{2\pi}}\int_{-\infty}^{\infty} \gamma_{\beta}(\omega) \e^{+\ri \omega t} \rd t. \label{eq:C_from_gamma}
\end{equation}
The prefactor in Eq.~\eqref{eq:glauber-dyn} is chosen such that (see \cref{prop:cbeta_1norm})
\begin{equation}
\frac{1}{\sqrt{2\pi}}\int_{-\infty}^{\infty} \labs{c_{\beta}(t)}\rd t \le 1\label{eq:|C|}.
\end{equation}
This normalization sets the strength of $\norm{\vH^{\beta, \tau, \vH}_{LS,a}}$ to be bounded by $\CO(1)$.

\paragraph{Absolute zero $\beta = \infty$.}

It is instructive to consider the case of zero temperature $\beta = \infty$ and infinite time scale $\tau = \infty$.
In this case, the transition weight $\gamma_\beta(\omega)$ is a step function ($1$ for $\omega < 0$ and $0$ for $\omega > 0$) and $\hat{\vA}^a(\omega)$ measures the energy difference perfectly. Thus, all heating transitions ($\ket{E}\rightarrow \ket{E+\omega}$ for $\omega >0$) are forbidden, and all cooling transitions ($\ket{E}\rightarrow \ket{E+\omega}$ for $\omega <0$) will remain. 
Hence, in the case when $\beta = \tau = \infty$, the thermal Lindbladian only lowers the energy. This matches our physical intuition that a zero-temperature bath only absorbs energy from the system.

\paragraph{Multiple jumps.}

The thermal Lindbladian $\mathcal{L}_{a}^{\beta, \tau, \vH}$ considers merely a single jump operator $\vA^a$ in the system-bath interaction.
When there are multiple jump operators, the total thermal Lindbladian is a weighted sum of the individual thermal Lindbladian $\mathcal{L}_{a}^{\beta, \tau, \vH}$,
\begin{equation}
    \mathcal{L}^{\beta, \tau, \vH} = \sum_{a = 1}^m \alpha_a \mathcal{L}_{a}^{\beta, \tau, \vH},
\end{equation}
where $\alpha_a \geq 0$ is a nonnegative weight.

Again, the interaction strength vector $\balpha \in \BR^m_{\geq 0}$ weights the contribution of each thermal Lindbladian. Thus, the total equation of motion under multiple jumps reads 
\begin{align}
    \frac{d\vrho}{dt} &= -\ri [\vH, \vrho] + \mathcal{L}^{\beta, \tau, \vH}(\vrho)\\
    &= {-\ri \left[\vH + \sum_{a=1}^m \alpha_a \vH^{\beta, \tau, \vH}_{LS,a}, \vrho \right]} + \sum_{a=1}^m \alpha_a \mathcal{D}^{\beta, \tau, \vH}_{a}(\vrho), \label{eq:mainL}
\end{align}
which consists of a coherent part and a purely dissipative part.

\paragraph{Calculation for normalization of $c_{\beta}(t)$.}
We now give a supplemental calculation that shows our choice of $\gamma_\beta(\omega)$ in \cref{eq:glauber-dyn} satisfies the condition in \cref{eq:|C|}.
\begin{proposition}\label{prop:cbeta_1norm}
For
\begin{align}
\hat{f}(\omega) := \frac{\e^{-\omega^2/2\Lambda_0^2}}{1+\e^{\beta \omega}},
\end{align}
we have that
\begin{align}
    \frac{1}{\sqrt{2\pi}}\norm{ f }_1 \le {2+\ln(1+\beta\Lambda_0)}.
\end{align}
\end{proposition}
\begin{proof}
We want to bound the 1-norm of $f(t)$ in the time domain. To do so, we bound the moments in the time domains
\begin{align}
    \sqrt{2\pi} \norm{f}_{\infty} &\le \norm{\hat{f}}_{1} \le \int_{-\infty}^{\infty} \e^{-\omega^2/2\Lambda_0^2} \rd \omega = \Lambda_0 \sqrt{2\pi}.\\
    \sqrt{2\pi}\norm{t f(t)}_{\infty} &\le \norm{\frac{\rd}{\rd \omega}\hat{f}}_{1} = 2 \cdot \norm{\hat{f}}_{\infty} \le 2.\\
    \sqrt{2\pi}\norm{t^2 f(t)}_{\infty} &\le \norm{\frac{\rd^2}{\rd \omega^2}\hat{f}}_{1}  \le 4 \cdot \norm{ \frac{\rd}{\rd \omega}\hat{f}}_{\infty} \le T.
\end{align}
The second line uses the fact that $\hat{f}$ is increasing and then decreasing (from $-\infty$ to $\infty$). The third line evaluates the derivative 
\begin{align}
         4\labs{\frac{\rd}{\rd\omega}\L( \frac{\e^{-\omega^2/2\Lambda_0^2}}{1+\e^{\beta \omega}} \R)} & = 4\labs{\frac{-\e^{-\omega^2/2\Lambda_0^2}\omega/\Lambda_0^2 }{1+\e^{\beta \omega}} -  \frac{ \e^{-\omega^2/2\Lambda_0^2} \beta\e^{\beta \omega} }{(1+\e^{\beta \omega})^2}}\\
         & \le 4 \L(\frac{1}{\sqrt{\e}\Lambda_0} + \beta\R)=: T \tag{since $\e^{-x^2/2}x \le \frac{1}{\sqrt{\e}}$}.
\end{align}
Thus, we may partition into three integrals to optimize the bound
\begin{align*}
    \norm{f}_1 &= \L(\int_{\labs{t} \le \Lambda_0^{-1}} + \int_{ T \ge \labs{t} \ge \Lambda_0^{-1}}+ \int_{ \labs{t} \ge T} \R) \labs{f(t)} \rd t\\
    &\le \int_{\labs{t} \le \Lambda_0^{-1}} \Lambda_0 \rd t+ \frac{1}{\sqrt{2\pi}}\int_{ T \ge \labs{t} \ge \Lambda_0^{-1}} \frac{2}{\labs{t}} \rd t+ \frac{1}{\sqrt{2\pi}}\int_{ \labs{t} \ge T} \frac{T}{t^2} \rd t\\
    &\le 2+ \frac{4}{\sqrt{2\pi}} \ln(\Lambda_0T) + \frac{2}{\sqrt{2\pi}} \\
    &\le \frac{2 + 2\sqrt{2\pi} +  4 \ln(\frac{4}{\sqrt{\e}}+4\beta\Lambda_0)}{\sqrt{2\pi} }\\
    &\le \frac{2 + 2\sqrt{2\pi} +8\ln(2) +  4 \ln(1+\beta \Lambda_0)}{\sqrt{2\pi}} \le \sqrt{2\pi} (2+\ln(1+\beta\Lambda_0)).
\end{align*}
where in the last line, we used  $1/\sqrt{\e}\le 1$ among other numerical bounds.
\end{proof}

\subsection{Properties of thermal Lindbladians}

\noindent From the exact forms of the thermal Lindbladians, we have the following propositions.

\begin{proposition}[Norm for the dissipative part~\cite{Chen2023quantumthermal}] \label{prop:norm-diss-part}
Any purely dissipative Lindbladian $\sum_{a} \CD^{\beta, \tau, \vH}_a$ defined in Eq.~\eqref{eq:dissipative-Lind} for any set of jump operators $\{\vA^a\}_{a=1}^m$ and any transition weight satisfying Eq.~\eqref{eq:gamma_KMS} have bounded superoperator norms
\begin{align}
\lnormp{\sum_{a=1}^m \alpha_a \CD^{\dag\beta, \tau, \vH}_{a}}{\infty-\infty}= \lnormp{\sum_{a=1}^m\alpha_a\CD^{\beta, \tau, \vH}_{a}}{1-1} \le 2 \norm{\sum_{a=1}^m\alpha_a \vA^{a\dag}\vA^a}.
\end{align}
The first equality is the duality between the $1-1$ and $\infty-\infty$ superoperator norms.
\end{proposition}

\begin{proposition}[Properties of the Lamb-shift term~\cite{Chen2023quantumthermal}] \label{prop:norm-lamb-shift-part}
The sum of Lamb-shift term~\eqref{eq:Lamb_shift} for any set of jump operators $\{\vA^a\}_{a=1}^m$ under a normalized bath correlation function $c_\beta(t)$ given by Eq.~\eqref{eq:|C|} satisfies that\footnote{Implicitly, the Lamb-shift term has units of energy yet do not scale with $\norm{\vH}$.}
\begin{align}
    \lnorm{\sum_{a=1}^m \alpha_a\vH^{\beta, \tau, \vH}_{LS,a}} &\le \frac{1}{2}\norm{\sum_{a=1}^m \alpha_a\vA^{a\dag}\vA^a}\\
    \norm{\sum_{a=1}^m\alpha_a[\vH^{\beta, \tau, \vH}_{LS,a},\vH]} &\le \CO\left( \frac{\norm{\vH}^{3/4}}{\tau^{1/4}}\norm{\sum_{a=1}^m \alpha_a\vA^{a\dag}\vA^a} \right).
\end{align}
For large enough $\tau$, the Lamb-shift term almost commutes with the Hamiltonian.
\end{proposition}
From Prop.~\ref{prop:norm-diss-part} and Prop.~\ref{prop:norm-lamb-shift-part}, we have the following norm bound for thermal Lindbladians.

\begin{proposition}[Norm of thermal Lindbladians] \label{prop:norm-thermal-lindblad}
Given a Hamiltonian $\vH$, an inverse temperature $\beta \geq 0$, a time scale $\tau \geq 0$, $m$ local jump operators $\{\vA^a\}_{a=1}^m$, a transition weight $\gamma_\beta(\omega)$ satisfying Eq.~\eqref{eq:gamma_KMS}, a normalized bath correlation function $c_\beta(t)$ satisfying Eq.~\eqref{eq:|C|}. The associated thermal Lindbladian $\sum_{a=1}^m \alpha_a \mathcal{L}_{a}^{\beta, \tau, \vH}$ has bounded superoperator norms
\begin{align}
\lnormp{\sum_{a=1}^m \alpha_a \mathcal{L}_{a}^{\dag\beta, \tau, \vH }}{\infty-\infty}= \lnormp{\sum_{a=1}^m \alpha_a \mathcal{L}_{a}^{\beta, \tau, \vH}}{1-1} \le 3 \norm{\sum_{a=1}^m \alpha_a \vA^{a\dag}\vA^a},
\end{align}
which is controlled by the interaction strength vector $\balpha$ under the normalization of $\vA^a$ in Eq.~\eqref{eq:normalization-vAa}.
\end{proposition}

\subsection{Algorithmic primitives for simulating thermal Lindbladians}

In this subsection, we review existing algorithmic primitives for simulating thermal Lindbladians~\cite{Chen2023quantumthermal}, estimating energy and expectation value of observables using block-encoding and quantum singular value transform (QSVT). See \cite{martyn2021grand} for a tutorial on block encoding and QSVT.
We begin with a definition of a block-encoding for Hermitian matrices, i.e., observables.

\begin{definition}[Block-encoding for Hermitian matrices]\label{def:blockObservable} 
We say that a unitary $\vU$ is a block-encoding for a Hermitian matrix $\vO$ if  
 \begin{align}
     (\bra{0^d}\otimes \vI)\cdot \vU\cdot(\ket{0^{d}} \otimes \vI)=\vO \quad \text{for}\quad d \in \mathbb{Z}^{+}.
 \end{align}
\end{definition}

Recall the following result stating that expectation values can be estimated using block-encoding.

\begin{proposition}[Measuring observable using block-encoding] \label{prop:measure-obs-block-enc}
Given a block-encoding $\vU_{\vO}$ for a Hermitian matrix $\vO$ and samples of a state $\vrho$.
One could estimate $\Tr(\vO \vrho)$ to small error $0 < \epsilon < 0.5$ using only $\tilde{\mathcal{O}}(1 / \epsilon^2)$ queries to the unitary $\vU_{\vO}$.
\end{proposition}
\begin{proof}
    Consider $\ket{0}\bra{0}\otimes \vrho$ and apply the Hadamard test to sample $ \Tr[ \vU \ket{0}\bra{0}\otimes \vrho ] = \Tr[\vO \vrho]$\footnote{We thank Yu Tong for discussions on this argument.}.
\end{proof}
Linear combinations of unitaries allow us to make efficient block-encodings of Hamiltonians presented as a sum of local terms. This fact results in the following proposition.

\begin{proposition}[Block-encoding for Hamiltonian; see \cite{berry2015hamiltonian, childs2017quantum, martyn2021grand}]\label{def:blockHamiltonian} 
Any $n$-qubit Hamiltonian $\vH$ has an efficient block-encoding $\vU_{\vH/\lambda_1}$ for some scalar $\lambda_1$ being the 1-norm of Pauli expansion coefficients.
\end{proposition}

From \cite{Chen2023quantumthermal}, we have the following for the Lamb-shift term $\vH_{LS}$ from Eq.~\eqref{eq:Lamb_shift}. Conveniently, the Lamb-shift term is already normalized (Proposition~\ref{prop:norm-lamb-shift-part}).
\begin{proposition}[Block-encoding for Lamb-shift term; see~\cite{Chen2023quantumthermal}]\label{prop:blockLS}
The Lamb-shift term $\vH_{LS}$ has an efficient block-encoding $\vU_{LS}$.
\end{proposition}

We define the block-encoding for a Lindbladian without the coherent commutator term $-\mathrm{i} [\vH, \vrho]$.

\begin{definition}[Block-encoding for Lindblad operators~{\cite{Chen2023quantumthermal}}]\label{def:blockLindladian}
	Given a purely irreversible Lindbladian
\begin{align}
     \CL[\vrho]:=\sum_{j \in J} \L(\vL_j\vrho \vL_j^{\dag} - \frac{1}{2} \vL_j^{\dag}\vL_j\vrho - \frac{1}{2}\vrho\vL_j^{\dag}\vL_j\R), 
\end{align} 
we say that a unitary $\vU$ is a block-encoding for Lindblad operators $\{\vL_j\}_{j\in J}$ if \footnote{In the first register, we could use any orthonormal basis, sticking to computational basis elements $\ket{j}$ is just for ease of presentation. Intuitively one can think about $b$ as the number of ancilla qubits used for implementing the operators $\vL_j$, while typically $a-b\approx \log|J|$.} 
 \begin{align}
     (\bra{0^b}\otimes \vI)\cdot \vU\cdot(\ket{0^{c}} \otimes \vI)=\sum_{j\in J} \ket{j} \otimes \vL_j \quad \text{for}\quad b\le c \in \mathbb{Z}^{+}.
 \end{align}
\end{definition}

\begin{theorem}[Linear-time Lindbladian simulation~\cite{Chen2023quantumthermal}]\label{thm:LCUSim_main}
Suppose the jumps $\vA^a$ can be block-encoded by a unitary $\vV_{jump}$ using $c\in \BZ$ ancillas qubits. Then, we can simulate the map $\e^{t \CL}$ for~\eqref{eq:mainL} to $\epsilon \le 1/2$ precision in the diamond norm using 
 \begin{align}
	&\tCO (( c+1)) \quad &\text{resettable ancilla},\\ 
	&\tCO (( t+1) \tau ) \quad &\text{controlled Hamiltonian simulation time},\\
	\quad &\tCO((t+1)(c+1))\quad &\text{other two-qubit gates},\\
    \text{and}\quad &\tCO(t+1) \quad &\text{queries to $\vW$, $\bm{prep}_{c_\beta(t)}$,$\bm{Prep}'_{c_\beta(\bt)}$, and $\vV_{jump}$}
	\end{align}
 where $\tilde{\CO}(\cdot)$ absorbed poly-logarithmic dependences on $t, \norm{\vH},\epsilon, \tau, \beta$.
 Furthermore, a block-encoding of the purely irreversible Lindbladian $\mathcal{D}^{\beta, \tau,\vH}$ with discretized frequency labels can be implemented efficiently. 
\end{theorem}
\noindent The above uses the following circuit components required for implementation: the controlled Hamiltonian simulation
\begin{align}
	\sum_{\bt \in S_{t_0}}\ketbra{\bt}{\bt}\otimes \e^{\pm \ri \bt \vH},
\end{align}
the unitary gates for preparing the bath correlation function in superposition 
\begin{align}
    \bm{Prep}_{c_\beta(\bt)}: \ket{\bar{0}} \rightarrow \sum_{\bar{t}\in S_{t_0}} \sqrt{\labs{c_\beta(\bt)}} \ket{\bt} \quad \text{and}\quad   \bm{Prep}'_{c_\beta(\bt)}:\ket{\bar{0}} \rightarrow \sum_{\bar{t}\in S_{t_0}} \frac{c_\beta(\bt)}{\sqrt{\labs{c_\beta(\bt)}}} \ket{\bt},
\end{align}
and the controlled rotation for transition weights
\begin{align}
    \vW := \sum_{\bomega\in S_{\omega_0}}  \begin{pmatrix} \sqrt{\gamma(\bomega)} & -\sqrt{1-\gamma(\bomega)}\\ \sqrt{1-\gamma(\bomega)} &  \sqrt{\gamma(\bomega)} \end{pmatrix} \otimes \ketbra{\bomega}{\bomega}.
\end{align}
Indeed, the above implementation uses discrete labels for the time $\bt \in S_{t_0}$ and frequencies $\bomega\in S_{\omega_0}$ corresponding to $\rd t$ and $\rd \omega$; these dominate the ancilla use. For conceptual simplicity, we focus on the continuous integral everywhere else and emphasize that the discretization is merely for implementation and introduces a negligible error; see~\cite{Chen2023quantumthermal}.

The controlled Hamiltonian simulation can be implemented efficiently for any $n$-qubit local Hamiltonian $\vH$ \cite{berry2015hamiltonian, childs2017quantum, martyn2021grand}.
The other operations $\vW$, $\bm{prep}_{c_\beta(t)}$,$\bm{Prep}'_{c_\beta(\bt)}$ can all be implemented efficiently~\cite{Chen2023quantumthermal} with the physically-motivated choice considered in Appendix~\ref{sec:physics-thermo-Lind}.
\begin{proposition}[Gradient of an observable under Lindbladian evolution; adapted from~\cite{Chen2023quantumthermal}]\label{prop:LOL}
	Given block-encodings $\vU$ for a purely irreversible Lindbladian (Def.~\ref{def:blockLindladian}) and $\vU_{\vO}$ for a Hermitian observable $\vO$, we get a block-encoding of 
 \begin{align}
  \sum_{j\in J}\vL_j^\dag \vO \vL_j \quad \text{via}\quad \vV:= (\vY_{\frac{1}{2}}\otimes \vU^\dag\otimes \vI_d) \cdot \L(2\ketbra{0^{b+1}}{0^{b+1}}\otimes\vI-\vI\R)\otimes  \vU_{\vO}\cdot(\vY_{\frac{1}{2}}\otimes \vU\otimes \vI_d),
 \end{align}
where $\ket{\pm} := (\ket{0}\pm\ket{1})/\sqrt{2}$ and $\vY_{\frac{1}{2}} := \frac{1}{\sqrt{2}}\begin{pmatrix} 1 & -1\\ 1 &  1 \end{pmatrix}$. 
\end{proposition}
\begin{proof}
We calculate
\begin{align}
	&(\bra{0^{c+1}}\otimes \vI\otimes \bra{0^d})\cdot \vV\cdot (\ket{0^{c+1}}\otimes \vI\otimes \ket{0^d}) \nonumber \\
	&=\bigg(\bra{-}\otimes (\bra{0^{c}}\otimes \vI)\vU^\dag \otimes \bra{0^d}\bigg)\cdot\L(2\ketbra{0^{b+1}}{0^{b+1}}\otimes\vI-\vI\R)\otimes \vU_{\vO} \cdot\bigg(\ket{+}\otimes \vU(\ket{0^{c}}\otimes \vI)\otimes \ket{0^d}\bigg) \nonumber\\
 &=\bigg(\bra{-}\otimes (\bra{0^{c}}\otimes \vI)\vU^\dag \otimes \bra{0^d}\bigg)\cdot\L(2\ketbra{0^{b+1}}{0^{b+1}}\otimes\vI\R)\otimes \vU_{\vO} \cdot\bigg(\ket{+}\otimes \vU(\ket{0^{c}}\otimes \vI)\otimes \ket{0^d}\bigg) \nonumber\\&
	=(\bra{0^{c}}\otimes \vI)\cdot \vU^\dag\cdot (\ketbra{0^{b}}{0^{b}}\otimes\vI)\otimes \vO\cdot (\ket{0^{c}}\otimes \vI) \\&
	=\L(\sum_{j\in J}\bra{j}\otimes \vL_j^\dag \R)\vI\otimes \vO  \L(\sum_{j'\in J}\ket{j'}\otimes \vL_{j'}\R)
	=\sum_{j\in J}\vL_j^\dag \vO \vL_j.\tag*{\qedhere}
\end{align}
\end{proof}

\begin{corollary}[Block-encoding the gradient of the Hamiltonian] \label{cor:grad-Hamil}
Given a block-encoding for a purely irreversible Lindbladian $\CL$ and a Hamiltonian $\vH$, there is an efficient block-encoding for
\begin{align}
    \frac{1}{2}\CL^\dag[\vH],
\end{align}
which is a Hermitian operator corresponding to the gradient of $\vH$ under $\CL$.
\end{corollary}
\begin{proof}
    Apply Proposition~\ref{prop:LOL} for Lindbladian $\CL$, Hermitian observables $\vH$ and $\vI$ to obtain block-encodings for $\sum_{j\in J}\vL_j^\dag \vH \vL_j$ and $\sum_{j\in J}\vL_j^\dag\vL_j$. Then, use quantum singular value transform (QSVT) for products and sums of block-encoding to obtain the block-encoding for
    \begin{equation}
        \frac{1}{2}\CL^\dag[\vH] = \frac{1}{2}\sum_{j\in J}\vL_j^\dag \vH \vL_j - \frac{1}{4} \sum_{j\in J}\vL_j^\dag\vL_j \vH - \frac{1}{4} \vH \sum_{j\in J}\vL_j^\dag\vL_j
    \end{equation}
    at high precision.
\end{proof}

From all of the above propositions, corollaries, and theorems, we can obtain the following.

\begin{lemma}[Measuring energy gradient] \label{lem:measure-energy-grad}
Given an $n$-qubit Hamiltonian $\vH$, inverse temperature $\beta \geq 0,$ time scale $\tau \ge 0$, samples of an $n$-qubit state $\vrho$, a thermal Lindbladian $\mathcal{L}_{\beta, \tau, \{\vA^a\}_a}$ from Eq.~\eqref{eq:mainL}.
The energy gradient
\begin{equation}
\Tr(\vH \mathcal{L}^{\beta, \tau, \vH}(\vrho) ) = \Tr( \mathcal{L}^{\dag \beta, \tau, \vH} (\vH) \vrho )    
\end{equation}
can be estimated to error $\epsilon$ using time and samples of $\vrho$ polynomial in $n, 1/\epsilon, \norm{\vH}, \beta, \tau$.
\end{lemma}
\begin{proof}
    From the form of thermal Lindbladians~\eqref{eq:mainL} and dropping the scripts $\mathcal{L}^{\beta, \tau, \vH} =\CL, \vH_{LS}^{\beta, \tau, \vH} = \vH_{LS}, \mathcal{D}^{\beta, \tau,\vH} = \CD$ we have
    \begin{equation}
        \mathcal{L}^{\dag} (\vH) = \ri [\vH_{LS}, \vH] + \mathcal{D}^{\dag}(\vH).
    \end{equation}
    Our goal is to create the block-encoding for $\mathcal{L}^{\dag} (\vH)$.
    First, we use quantum singular value transform (QSVT) for products and sums of block-encoding to obtain the block-encoding for $\ri [\vH_{LS}, \vH]$ from block-encoding for $\vH$ and $\vH_{LS}$ in Propositions \ref{def:blockHamiltonian} and \ref{prop:blockLS}.
    Next, using the block-encoding for the purely irreversible Lindbladian $\mathcal{D}$ from Theorem~\ref{thm:LCUSim_main} and the block-encoding for $\vH$, we can apply Corollary~\ref{cor:grad-Hamil} to obtain efficient block-encoding for $\mathcal{D}^{\dag}(\vH)$.
    To obtain the block-encoding for $\mathcal{L}^{\dag} (\vH)$, we use QSVT for sums of block-encoding to add $\ri [\vH_{LS}, \vH]$ and $\mathcal{D}^{\dag}(\vH)$.
    Finally, using Prop.~\ref{prop:measure-obs-block-enc}, we can estimate $\Tr(\mathcal{L}^{\dag} (\vH) \vrho)$ efficiently. All the above QSVT manipulations operate at high precision, and the discrete Fourier transform well-approximates the continuum at poly-logarithmic costs~\cite{Chen2023quantumthermal}.
\end{proof}

\section[A polynomial-time quantum algorithm for finding a local minimum under thermal perturbations (Proof of Theorem~\ref{thm:quantum-easy-thermal})]{A polynomial-time quantum algorithm for finding a local minimum under thermal perturbations (Proof of Theorem~\ref{thm:quantum-easy-thermal})} \label{sec:proof-thm-quantum-easy-thermal}

In this appendix, we present the proof of Theorem~\ref{thm:quantum-easy-thermal} by giving a polynomial-time quantum algorithm for finding local minima under thermal perturbations.
We refer to the efficient quantum algorithm as \emph{Quantum thermal gradient descent} as the algorithm performs gradient descent using thermal Lindbladians induced by a heat bath.
The algorithm uses the properties of thermal Lindbladians presented in Appendix~\ref{sec:thermo-lindblad-detail}.

\subsection{Cooling by gradient descent}

The central idea of quantum thermal gradient descent is the following.
When we are not at a local minimum under thermal perturbations, the negative energy gradient will be sufficiently large, and we can decrease the energy by following a direction with a negative energy gradient.
This is characterized by the following lemma.
We will use this lemma to design the gradient descent algorithm for finding a local minimum.

\begin{lemma}[Cooling by gradient descent] \label{lem:cool-GD}
Given parameters $0 < \tilde{\epsilon} < 0.5, B \geq 1$, $\beta, \tau \ge 0$, an $n$-qubit Hamiltonian $\vH$ with $\norm{\vH}_\infty \leq B$, and $m$ local jump operators $\{\vA^a\}_a$.
Consider $a = 1, \ldots, m$ with an approximate energy gradient $g_{a}$ satisfying
\begin{equation}
    \left| g_a - \Tr\left(\vH \mathcal{L}^{\beta, \tau, \vH}_{a} [\vrho^{(t-1)}] \right) \right| < 0.01 \tilde{\epsilon}
\end{equation}
Suppose there exist $a^* \in \{1, \ldots, m\}$ with sufficiently negative approximate energy gradient,
\begin{equation}
    g_{a^*} < -0.99 \tilde{\epsilon}.
\end{equation}
The state after evolving $\vrho$ along the direction $\hat{\be}_{a^*}$ for a small step $s = |g_{a^*}| / (9 B^2) > 0$,
\begin{equation}
    \vrho^{(\mathrm{next})} := \exp_{\vrho}^{\beta, \tau, \vH, \{\vA^a\}_a}\left( s \hat{\be}_{a^*} \right)
\end{equation}
guarantees the following energy decrease,
\begin{equation}
    \Tr\left(\vH \vrho^{(\mathrm{next})} \right) < \Tr(\vH \vrho) - \frac{\tilde{\epsilon}^2}{20 B^2}.
\end{equation}
\end{lemma}
\begin{proof}
From Prop.~\ref{prop:taylor-thermal} on Taylor's theorem, we have the following identity
\begin{equation}
    \Tr\left(\vH \vrho^{(\mathrm{next})} \right) = \Tr(\vH \vrho) + s \Tr(\vH \CL^{\beta, \tau, \vH}_{a^*} [\vrho]) + \frac{s^2}{2} \Tr(\vH \CL^{\beta, \tau, \vH}_{a^*} [\CL^{\beta, \tau, \vH}_{a^*}[\vsigma]])
\end{equation}
for some $n$-qubit state $\vsigma$.
We will separately control the linear term and the quadratic term.

\paragraph{Linear term.}
From the definition of the energy gradient vector ${\bm{\nabla}}_{\beta, \tau, \{\vA^a\}_a}(\vH, \vrho)$, we have
\begin{equation}
    \Tr(\vH \CL^{\beta, \tau, \vH}_{a^*} [\vrho]) < g_{a^*} + 0.01 \tilde{\epsilon} < \frac{98}{99} g_{a^*} = -\frac{98}{99} |g_{a^*}|.
\end{equation}
The second inequality follows from $g_{a^*} < -0.99 \tilde{\epsilon}$, hence $0.01 \tilde{\epsilon} < -(1/99) g_{a^*}$.

\paragraph{Quadratic term.}
We can bound the quadratic term as follows,
\begin{align}
    \frac{1}{2} \Tr(\vH \CL^{\beta, \tau, \vH}_{a^*} [\CL^{\beta, \tau, \vH}_{a^*}[\vsigma]]) \leq \frac{1}{2} \lnormp{\CL^{\dag\beta, \tau, \vH }_{a^*}}{\infty-\infty} \norm{\CL_{a^*}^{\dag \beta, \tau, \vH}(\vH)}_\infty.
\end{align}
From Prop.~\ref{prop:bound-ene-grad} and Prop.~\ref{prop:norm-thermal-lindblad} that bounds the norm of these objects, we have
\begin{equation}
    \frac{1}{2}\lnormp{\mathcal{L}^{\dag\beta, \tau, \vH }_a}{\infty-\infty} \norm{\CL_{a'}^{\dag \beta, \tau, \vH}(\vH)}_\infty \leq 4.5 \norm{\vH}_\infty^2 \leq 4.5 B^2.
\end{equation}
Combining the linear and quadratic terms with $s = |g_{a^*}| / (9 B^2) > 0$, we have
\begin{equation}
    \Tr\left(\vH \vrho^{(\mathrm{next})} \right) \leq \Tr(\vH \vrho) - \left(\frac{98}{99 \times 9} - \frac{1}{18} \right) \frac{|g_{a^*}|^2}{B^2} < \Tr(\vH \vrho) - 0.054 \frac{|g_{a^*}|^2}{B^2}.
\end{equation}
We can use $|g_{a^*}|^2 > 0.99^2 \tilde{\epsilon}^2$ to obtain the desired claim.
\end{proof}

\subsection{Quantum thermal gradient descent}

Given error $\epsilon = 1 / \mathrm{poly}(n)$, norm bound $B = \mathrm{poly}(n)$, inverse temperature $0 \leq \beta \leq \mathrm{poly}(n)$, time scale $\tau = \mathrm{poly}(n)$, an $n$-qubit local Hamiltonian~$\vH$ with $\norm{\vH}_\infty \leq B$, $m$ local jump operators $\{\vA^a\}_a$ with $m = \mathrm{poly}(n)$, and a local observable $\vO$ with $\norm{\vO}_\infty \leq 1$.

We consider a coordinate-wise gradient descent algorithm that implements the following.
The initial state $\vrho^{(0)}$ is arbitrary as long as copies of the state can be prepared on the quantum computer. For example, we can set $\vrho^{(0)}$ to be the maximally mixed state $\frac{\vI}{2^n}$.
The total number of steps is
\begin{equation} \label{eq:T-size-def}
    T := \frac{42 B^3}{\epsilon^2}.
\end{equation}
For each time step $t$ from $1$ to $T$, the algorithm does the following.
\begin{enumerate}
    \item For each direction $a = 1, \ldots, m$, estimate an approximate energy gradient $g_a^{(t)}$ satisfying
    \begin{equation} \label{eq:est-err-grad-vec}
    \left| g_a^{(t)} - \Tr\left(\mathcal{L}^{\dag \beta, \tau, \vH}_{a}(\vH) \vrho^{(t-1)} \right) \right| < 0.0099 \epsilon.
    \end{equation}
    The energy gradient can be estimated efficiently using Lemma~\ref{lem:measure-energy-grad} given copies of $\vrho^{(t-1)}$ prepared through Eq.~\eqref{eq:vrhot-GD} and Theorem~\ref{thm:LCUSim_main}.
    From the bound on energy gradients in Prop.~\ref{prop:bound-ene-grad}, we have $|g_a^{(t)}| \leq 3 B + 0.0099 \epsilon$.
    If $g_a^{(t)} < -0.99 \epsilon$, set $a^{(t)} := a$ and \emph{terminate} the for-loop over $a$.
    \item If $a^{(t)}$ is not found, set $\vrho^{(T)} := \vrho^{(t-1)}$ and \emph{terminate} the for-loop over $t$. Otherwise evolve $\vrho^{(t-1)}$ under the direction $\hat{\be}_{a^{(t)}}$ for a small step $s^{(t)} := |g_a^{(t)}| / (9 B^2)$,
    \begin{equation} \label{eq:vrhot-GD}
        \vrho^{(t)} := \exp\left( s^{(t)} \mathcal{L}^{\beta, \tau, \vH}_{a^{(t)}} \right)\left(\vrho^{(t-1)} \right) = \prod_{t' = 1}^{t} \exp\left( s^{(t)} \mathcal{L}^{\beta, \tau, \vH}_{a^{(t)}} \right) \left(\vrho^{(0)}\right).
    \end{equation}
    Because $0 \leq s^{(t)} \leq 1 / (2 B)$, a single copy of $\vrho^{(t)}$ can be prepared in polynomial-time using the thermal Lindbladian simulation algorithm in \cite{Chen2023quantumthermal}; see Theorem~\ref{thm:LCUSim_main}.
\end{enumerate}
We will show that the state $\vrho^{(T)}$ created by the gradient descent algorithm is an $\epsilon$-approximate local minimum of $\vH$ under thermal perturbations.
Furthermore, using the thermal Lindbladian simulation algorithm, a quantum machine can efficiently create many copies of $\vrho^{(T)}$.

\subsection{Proof of Theorem~\ref{thm:quantum-easy-thermal}}

The central idea in the proof of Theorem~\ref{thm:quantum-easy-thermal} is the following lemma.
The lemma combines the key results characterizing local minima in Appendix~\ref{sec:local-minima-thermal-prop}.

\begin{lemma}[Gradient descent finds a local minimum]
    $\vrho^{(T)}$ from Eq.~\eqref{eq:vrhot-GD} is an $\epsilon$-approximate local minimum of $\vH$ under thermal perturbations with inverse temperature $\beta$, time scale $\tau$, and system-bath interactions generated by $\{\vA^a\}_a$.
\end{lemma}
\begin{proof}
    Suppose the algorithm terminates at some time step $t<T$, then $g^{(t)}_a \geq -0.99 \epsilon$.
    From Eq.~\eqref{eq:est-err-grad-vec}, we have $\Tr\left(\mathcal{L}^{\dag \beta, \tau, \vH}_{a}(\vH) \vrho^{(t-1)} \right) \geq -0.9999 \epsilon$.
    Hence,
    \begin{equation}
        \norm{{\bm{\nabla}}_{\beta, \tau, \{\vA^a\}_a}^-(\vH, \vrho^{(t-1)})}_\infty \leq 0.9999 \epsilon < \epsilon.
    \end{equation}
    From the sufficient condition for local minima given in Lemma~\ref{lem:suff-QLM}, we have $\vrho^{(T)} = \vrho^{(t-1)}$ is an $\epsilon$-approximate local minimum ${\vrho}$ of the $n$-qubit Hamiltonian $\vH$ under thermal perturbations.

    We now show by contradiction that the algorithm must terminate early.
    Assume that the algorithm did not terminate early.
    Then, we can use Lemma~\ref{lem:cool-GD} with $\tilde{\epsilon} = 0.99 \epsilon$ for cooling by gradient descent to obtain
    \begin{equation}
        \Tr(\vH \vrho^{(T)}) \leq \Tr(\vH \vrho^{(T-1)}) - \frac{0.99^2 \epsilon^2}{20 B^2} \leq \ldots \leq  \Tr(\vH \vrho^{(0)}) - \frac{0.99^2 \epsilon^2}{20 B^2} T \leq  \norm{\vH}_\infty - \frac{0.99^2 \epsilon^2}{20 B^2} T.
    \end{equation}
    From the definition of $T$ in Eq.~\eqref{eq:T-size-def} and $\norm{\vH}_\infty \leq B$, we have
    \begin{equation}
        \Tr(\vH \vrho^{(T)}) \leq \norm{\vH}_\infty - \frac{0.99^2 \epsilon^2}{20 B^2} \frac{42 B^3}{\epsilon^2} \leq \norm{\vH}_\infty - 2.05 B \leq -1.05 B.
    \end{equation}
    At the same time, because $\norm{\vrho^{(T)}}_1 = 1$,
    \begin{equation}
        \Tr(\vH \vrho^{(T)}) \geq - \norm{\vH}_\infty \geq -B.
    \end{equation}
    This is a contradiction. Hence the algorithm must terminate early.
\end{proof}

The polynomial-time quantum algorithm for establishing Theorem~\ref{thm:quantum-easy-thermal} is as follows.
The algorithm runs quantum thermal gradient descent to find a local minimum $\vrho^{(T)}$ of $\vH$ under thermal perturbations.
Recall that $B$ is the upper bound on $\norm{\vH}_\infty$, and is equal to $\mathrm{poly}(n)$, and $1/\epsilon = \mathrm{poly}(n)$.
Because every step can be done in polynomial time, and there are at most $T = 42 B^3 / \epsilon^2 = \mathrm{poly}(n)$ time steps, quantum thermal gradient descent runs in time polynomial in $n$.

Now, given any observable $\vO$. The quantum algorithm prepares $\mathcal{O}(1 / \epsilon^2) = \mathrm{poly}(n)$ copies of $\vrho^{(T)}$ in $\mathrm{poly}(n)$ time, then measures $\vO$ on the $\mathcal{O}(1 / \epsilon^2)$ copies of $\vrho^{(T)}$ to estimate $\Tr(\vO \vrho^{(T)})$ to $\epsilon$ error.
This concludes the proof of Theorem~\ref{thm:quantum-easy-thermal}.

\section{Characterizing energy gradients in low-temperature heat bath}
\label{sec:characterize-neg-grad-condition}

Recall from Appendix~\ref{sec:Hamil-suboptimal-local-minima} on certifying Hamiltonians without suboptimal local minima, if there exists $\balpha \in \BR^m_{\geq 0}$ with $\norm{\balpha}_1 = 1$, such that the negative gradient condition holds,
\begin{equation}
    -\sum_{a} \alpha_a \CL_a^{\dag \beta, \tau, \vH}[\vH] \succeq \frac{2 \epsilon}{\delta} (\vI - \vP_G(\vH) ) - \epsilon \vI,
\end{equation}
then any $\epsilon$-approximate local minimum~$\vrho$ of the $n$-qubit Hamiltonian $\vH$ under thermal perturbations is an exact global minimum of $\vH$ with failure probability $\leq \delta$, i.e., $\Tr(\vP_G(\vH) \vrho) \geq 1 - \delta$, where $\vP_G(\vH)$ is the projection onto the ground state space.
To understand when the above condition holds, it is imperative to characterize the energy gradients, $\CL_a^{\dag\beta,\tau,\vH}[\vH]$.

In this appendix, we present various lemmas and theorems characterizing the energy gradients, which will be used in our proof of \cref{thm:no-suboptimal-local-minima-informal2} in \cref{sec:universal-quantum-computation} for showing that a certain family of Hamiltonians has no suboptimal local minima.
We remark that the proofs of many formal statements in this appendix require concepts and results that won't be shown till later in Appendices~\ref{sec:OFT} and \ref{sec:monotone_gradient}, and we recommend the first-time reader to freely skip the proofs and return later.

For simplicity, we will focus on the nonnegative vector $\balpha$ being uniform over a subset $S$ for the remaining appendices. We will show that this is sufficient for our purposes even though having the ability to choose $\balpha$ is more powerful. We define the following Lindbladian with uniform weights over a subset $S \subseteq \{1, \ldots, m\}$,
\begin{equation} 
\label{eq:thermal-L-unif}
    \CL := \sum_{a \in S} \CL_a^{\beta, \tau, \vH}.
\end{equation}
Recall from Appendix~\ref{sec:thermo-lindblad-detail} that each $\CL_a^{\dag \beta, \tau, \vH}$ corresponds to a jump operator $\vA^a$ satisfying the normalization condition $\norm{\vA^a}_\infty \leq 1$. If we let $r := 2 m \epsilon / \delta, \epsilon' = m \epsilon$ and $S = \{1, \ldots, m\}$, then the negative gradient condition becomes
\begin{equation} \label{eq:NGC-full}
    \text{(negative gradient condition)}: \quad - \CL^{\dag}[\vH] \succeq r (\vI - \vP_G) - \epsilon' \vI,
\end{equation}
This will be the central inequality we would like to establish for the remaining appendices.
Throughout the proofs, we will consider different subsets $S$ and show a relation similar to Eq.~\eqref{eq:NGC-full} for subset $S$. 

\subsection{Basic properties of energy gradients in low-temperature bath}
\label{sec:gradient-properties}

We show a few basic properties of energy gradients under a low-temperature, long-time-scale bath. First, we show that the energy gradient, at large $\beta$ (i.e., low temperatures), is negative semi-definite up to controllable error. Intuitively, this can be seen from the KMS condition in Eq.~\eqref{eq:gamma_KMS}, $\gamma_{\beta}(\omega) = \gamma_{\beta}(-\omega) \e^{-\beta \omega}$: the heating transition is suppressed by the Boltzmann weight, allowing energy to increase by $\omega\sim \beta^{-1}$. Another source of error is the uncertainty in energy $\tau^{-1}$.

\begin{lemma}[Almost negative gradients]\label{lem:gradient_upperbound}
Consider the thermal Lindbladian $\CL = \sum_{a\in S}\CL_a^{\beta,\tau, \vH}$ with jump operators $\{\vA^a\}_{a\in S}$ where $\norm{\vA^a}\le 1$, and $\gamma_{\beta}(\omega)$ satisfying Eq.~\eqref{eq:gamma_KMS}.  Then,
    \begin{align}
         \CL^{\dag}[\vH] \preceq  \CO\left( \labs{S}\L(\frac{\norm{\vH}^{3/4}}{\tau^{1/4}}+\frac{1}{\tau} + \frac{1}{\beta}\R) \right) \cdot \vI. 
    \end{align}
\end{lemma}
\begin{proof}
Rewrite the energy gradient with an error controlled by \cref{prop:norm-lamb-shift-part} and Lemma~\ref{lemma:expr_energy_gradient} gives
\begin{align}
\CL^{\dag}[\vH] &\approx \sum_{a\in S} \int_{-\infty}^{\infty} \gamma_{\beta}(\omega)\omega
 \hat{\vA}^{a}(\omega)^{\dag}\hat{\vA}^a(\omega) \rd \omega \nonumber \\
 & = \sum_{a\in S} \int_{0}^{\infty} \gamma_{\beta}(\omega)\omega
 \hat{\vA}^{a}(\omega)^{\dag}\hat{\vA}^a(\omega) \rd \omega + \sum_{a\in S} \int_{-\infty}^{0} \gamma_{\beta}(\omega)\omega
 \hat{\vA}^{a}(\omega)^{\dag}\hat{\vA}^a(\omega) \rd \omega
\end{align}
and bound the positive operator
\begin{align}
    \norm{\sum_{a\in S} \int_{0}^{\infty} \gamma_{\beta}(\omega)\omega
 \hat{\vA}^{a}(\omega)^{\dag}\hat{\vA}^a(\omega) \rd \omega}_{\infty} \le \labs{S}\max_{\omega \ge 0} \gamma_{\beta}(\omega) \omega \le \frac{|S|}{\beta}.
\end{align}
The second inequality uses the tail bound in Eq.~\eqref{eq:gamma-tailbound} with $\Delta=0$.
\end{proof}
Second, we show that the energy gradient operator is nearly diagonal in the energy basis. The intuition is that for any operator $\vA$, the product
\begin{align}
   \hat{\vA}^{\dag}(\omega)\hat{\vA}(\omega) 
\end{align}
is nearly diagonal in the energy basis for large $\tau$.

\begin{lemma}[Energy gradient is almost diagonal]
\label{lem:gradient-diagonal}
In the setting of Lemma~\ref{lem:gradient_upperbound}, assume that for any two well-isolated energy eigensubspaces $\vP_1$ and $\vP_2$ such that the two sets of eigenvalues have at least distance $\delta$. Then,
    \begin{align}
    \norm{\vP_{1} \CL^{\dag}[\vH] \vP_{2}} \le \CO\L(\labs{S}\L(\frac{\norm{\vH}^{3/4}}{\tau^{1/4}}+\frac{1}{\tau} + \frac{\norm{\theta_\beta}_{\infty}}{\sqrt{\delta\tau }}\R)\R).
    \end{align}
    where $\theta_{\beta}(\omega) := \gamma_\beta(\omega)\omega$.
\end{lemma}

\begin{proof}
Formally, approximate the energy gradient by dropping the Lamb-shift term (\cref{prop:norm-lamb-shift-part}) and applying Lemma~\ref{lemma:expr_energy_gradient},
\begin{align}
\CL^{\dag}[\vH] \approx \sum_{a\in S} \int_{-\infty}^{\infty} \gamma_{\beta}(\omega)\omega
 \hat{\vA}^{a}(\omega)^{\dag}\hat{\vA}^a(\omega) \rd \omega .
\end{align}
We then  apply the secular approximation for $\mu = \delta/2$ (Corollary~\ref{cor:AA-SS}) such that the transition amplitudes vanishes between the subspaces
\begin{align}
    \vP_1\hat{\vS}^a_{\mu}(\omega)^{\dag}\hat{\vS}^a_{\mu}(\omega) \vP_2 = 0\quad \text{for each}\quad \omega \in \BR\quad \text{and}\quad a \in A.
\end{align}
Combining the errors in each of the approximations leads to the claimed result.
\end{proof}

Next, we show that the finite-$\tau$ Lindbladian can be approximated by the infinite-$\tau$ version under certain conditions. The latter, known as the Davies' generator~\cite{davies1976quantum}, has a simpler form that is more amenable for analysis in some situations.

\begin{lemma}[Recovering Davies' generator]
\label{lem:recover-Davies}
Consider the dissipative part of the thermal Lindbladian $\CD_a^{\beta,\tau, \vH}$ with the jump operators $\vA^a$ where $\norm{\vA^a}\le 1$, and any $\gamma_{\beta}$ such that $\norm{\gamma_{\beta}}_{\infty}\le 1$. Suppose the Bohr-frequency gap is $\Delta_\nu(\vH)$, then
\begin{equation}
    \norm{\CD_a^{\dag\beta,\tau, \vH}-\CD_a^{\dag\beta,\infty, \vH}}_{\infty-\infty} \le \CO\L(  \max_{\nu} \labs{\gamma_\beta(\nu)- \int_{-\infty}^\infty \gamma_\beta(\omega) \left|\hat{f}(\omega-\nu)\right|^2\rd \omega}+\frac{1}{\sqrt{\Delta_\nu(\vH) \tau}}\R).
\end{equation}
\end{lemma}
Therefore, the Bohr-frequency gap sets a timescale $\sim\Delta_{\nu}^{-1}$ such that the map $\CD_a^{\dag\beta,\tau, \vH}$ stabilized.

\begin{proof}
It suffices to consider $\CD_a^{\dag\beta,\tau,\vH} [\vO] $ acting on arbitrary operator $\vO$ such that $\norm{\vO} =1$:
\begin{align}
\CD_a^{\dag\beta,\tau,\vH} [\vO] &= \int_{-\infty}^\infty \gamma_\beta(\omega) \Big[\hat\vA^a(\omega)^\dag \vO \hat\vA^a(\omega) - \frac12 \{\hat\vA^a(\omega)^\dag \hat\vA^a(\omega), \vO\}\Big]\rd \omega\\
&\stackrel{\verr_1}{\approx} \int_{-\infty}^\infty \gamma_\beta(\omega) \Big[\hat\vS^a(\omega)^\dag \vO \hat\vS^a(\omega) - \frac12 \{\hat\vS^a(\omega)^\dag \hat\vS^a(\omega), \vO\}\Big]\rd \omega\tag{secular approximation: Corollary~\ref{cor:S-A}}\\
&=\sum_{\nu,\nu'\in B(\vH)} \Big(\vA^{a\dag}_{\nu'} \vO \vA^a_\nu - \frac12 \{\vA^{a\dag}_{\nu'} \vA^a_\nu, \vO\} \Big) 
    \int_{-\infty}^\infty \gamma_\beta(\omega)\hat{f}_\mu^*(\omega-\nu') \hat{f}_\mu(\omega-\nu)\rd \omega \tag{truncated at frequency $\mu = \frac{\Delta_{\nu}}{2}$}\\
&=\sum_{\nu\in B(\vH)} \Big(\vA^{a\dag}_{\nu} \vO \vA^a_\nu - \frac12 \{\vA^{a\dag}_{\nu} \vA^a_\nu, \vO\} \Big) 
    \int_{-\infty}^\infty \gamma_\beta(\omega) \left|\hat{f}_\mu(\omega-\nu)\right|^2\rd \omega \tag{different blocks $\nu\ne \nu'$ decohere}\\ 
&\stackrel{\verr_2}{\approx}\sum_{\nu\in B(\vH)} \Big(\vA^{a\dag}_{\nu} \vO \vA^a_\nu - \frac12 \{\vA^{a\dag}_{\nu} \vA^a_\nu, \vO\} \Big) \gamma_\beta(\nu) \tag{Lemma~\ref{lem:Op_purification_bounds}}.
\end{align}
The approximation errors are bounded by
\begin{align}
    \verr_1 &\le 2\norm{\vO} \norm{\gamma_{\beta}}_{\infty}\norm{\vA^a} \norm{f_{\tau}\cdot (1-\hat{s}_{\mu})}_2\norm{f_{\tau}}_2= \CO\L( \frac{1}{\sqrt{\mu\tau}}\R), \\
    \verr_2 & \le 2\norm{\vO} \norm{\vA^a} \max_{\nu} \labs{\gamma_\beta(\nu)- \int_{-\infty}^\infty \gamma_\beta(\omega) \left|\hat{f}_\mu(\omega-\nu)\right|^2\rd \omega} \nonumber \\
    &\le \CO\L(\max_{\nu} \labs{\gamma_\beta(\nu)- \int_{-\infty}^\infty \gamma_\beta(\omega) \left|\hat{f}(\omega-\nu)\right|^2\rd \omega}+\frac{\norm{\gamma_{\beta}}_{\infty}}{\mu \tau}\R).
\end{align}
Combine the error bounds to conclude the proof. 
Note that since the bound becomes vacuous at $\mu\tau  = \Omega(1)$, we have that $\CO( 1/\sqrt{\mu\tau} + 1/(\mu\tau) ) = \CO( 1/\sqrt{\mu\tau})$.
\end{proof}

\begin{lemma}
\label{lem:Davies-error-bound}
For $\gamma_{\beta}(\omega)$ defined in Eq.~\eqref{eq:glauber-dyn} with $\Lambda_0=\Theta(1)$, we have that
\begin{align}
    \max_{\nu\in \BR} \labs{\gamma_\beta(\nu)- \int_{-\infty}^\infty \gamma_\beta(\omega) \left|\hat{f}(\omega-\nu)\right|^2\rd \omega} \le  \CO\L(\frac{1+\beta}{\tau} \ln\tau\R). 
\end{align}
\end{lemma}
\begin{proof}
Recall the integration-by-part trick for expectation integral
\begin{align}
    \int_{0}^{\infty} f(x)p(x)\rd x = -[f(x)P(x)]^{x=\infty}_{x=0} + \int_{0}^{\infty} f'(x) P(x) \rd x\quad \text{where}\quad P(x) := \int_{x}^{\infty} p(y) \rd y.
\end{align}
Then, 
\begin{align}
\int_{-\infty}^\infty \gamma_\beta(\omega) \left|\hat{f}(\omega-\nu)\right|^2\rd \omega
&=\int_{-\infty}^\infty \gamma_\beta(\nu+x) \left|\hat{f}(x)\right|^2\rd x= \int_{0}^\infty \frac{\gamma_\beta(\nu+x)+\gamma_\beta(\nu-x)}{2} \undersetbrace{=:p(x)}{2\left|\hat{f}(x)\right|^2}\rd x \nonumber \\
&= \gamma_{\beta}(\nu)+\int_{0}^\infty \frac{\gamma'_\beta(\nu+x)+\gamma'_\beta(x-\nu)}{2} P(x)\rd x,
\end{align}
where in the last line we applied the integration-by-part and used $\gamma(\pm\infty)=0$ and $P(0)=1$.
The error term can be bounded as follows
\begin{align}
&\labs{\int_{0}^\infty \frac{\gamma'_\beta(\nu+x)+\gamma'_\beta(x-\nu)}{2} P(x)\rd x} \nonumber\\
&=\labs{\int_{0}^{1/\tau} \frac{\gamma'_\beta(\nu+x)+\gamma'_\beta(\nu-x)}{2} P(x)\rd x + \int_{1/\tau}^\infty \frac{\gamma'_\beta(\nu+x)+\gamma'_\beta(\nu-x)}{2} P(x)\rd x} \nonumber \\
& = \CO(\frac{1+\beta}{\tau}) + \CO(\frac{1+\beta}{\tau} \ln\tau). 
\end{align}
In the last line, we control the first term by $P(x) \le 1$ and noting by the product rule we have
\begin{align}
         \labs{\gamma'_\beta(\omega)} &= \CO\L(\labs{\frac{\rd}{\rd\omega} \L( \frac{\e^{-\omega^2/2\Lambda_0^2}}{1+\e^{\beta \omega}} \R)}\R) 
 = \CO\L (\labs{-\frac{\e^{-\omega^2/2\Lambda_0^2}\omega/\Lambda_0^2 }{1+\e^{\beta \omega}} -  \frac{ \e^{-\omega^2/2\Lambda_0^2} \beta\e^{\beta \omega} }{(1+\e^{\beta \omega})^2}}\R)\\
         &\le \CO(\frac{1}{\Lambda_0}+\beta)\tag{with change of variable $x = \beta \omega$ and $y = \omega/ \Lambda_0$ }.
\end{align}
The second term uses the tail bound $P(x) \le \frac{4}{\pi x \tau}$ from Eq.~\eqref{eq:secular_tailbound} and that $\gamma'_\beta(\nu\pm x)$ 
are each rapidly decaying outside an $x\in [\mp\nu-\Lambda_0,\mp\nu+\Lambda_0]$ window so that the integral over $\frac{1}{x} \rd x $ only contributes at most $\CO(\int^{\Lambda_0}_{1/\tau}\frac{1}{x}\rd x)=\CO(\log(\tau\Lambda_0))$.
\end{proof}

\subsection{Relating subspace and local gradients to global gradients}
As a method of proof, we will often analyze a Lindbladian by its constituents, and here we present a few useful relations. 
First, when studying gradients, the gradient acting on a subspace is often conceptually simpler. The following lemma relates the energy gradient in a subspace and the full energy gradient. This is a direct consequence of Lemma~\ref{lem:gradient_upperbound} and Lemma~\ref{lem:gradient-diagonal} above.

\begin{lemma}[Subspace gradient and global gradient]
\label{cor:PLP_and_L}
In the setting of Lemma~\ref{lem:gradient-diagonal}, suppose $\vP$ projects onto a set of eigenstates of $\vH$ separated by the rest by a gap of at least $\delta$.
Then,
\begin{equation}
    -\CL^{\dag}[\vH] \succeq -\vP \CL^{\dag}[\vH] \vP - \CO\left(\labs{S}\L(\frac{\norm{\vH}^{3/4}}{\tau^{1/4}}+\frac{1}{\beta}+ \frac{1}{\tau}+ \frac{\norm{\theta_\beta}_{\infty}}{\sqrt{\delta\tau }} \R)\right)\cdot \vI.
\end{equation}
\end{lemma}
\begin{proof}
Let $\vL = \CL^{\dag}[\vH]$. We have, 
\begin{align}
        \vL &= \vP \vL \vP + \vP^\perp \vL \vP^\perp + \vP \vL \vP^\perp + \vP^\perp \vL \vP.
\end{align}
Using \cref{lem:gradient_upperbound} establishing the almost negativity of the energy gradient,
\begin{align}
     -\vL &\succeq - \CO\left( \labs{S}\L(\frac{\norm{\vH}^{3/4}}{\tau^{1/4}}+\frac{1}{\tau} +\frac{1}{\beta} \R)\right) \vI,
\end{align}
we have
\begin{align}
    - \vP^{\perp}\vL\vP^{\perp} &\succeq - \CO\left(\labs{S}\L(\frac{\norm{\vH}^{3/4}}{\tau^{1/4}}+\frac{1}{\tau} +\frac{1}{\beta}\right)\R)\vP^{\perp}. \tag*{(Lemma~\ref{lem:gradient_upperbound})}\\
     \|\vP \vL \vP^\perp + \vP^\perp \vL \vP \| &\le \CO\left(\labs{S}\L(\frac{\norm{\vH}^{3/4}}{\tau^{1/4}}+\frac{1}{\tau} + \frac{\norm{\theta_\beta}_{\infty}}{\sqrt{\delta\tau }} \right)\R).\tag*{(Lemma~\ref{lem:gradient-diagonal})}
\end{align}
Putting the bounds together yields the advertised result.
\end{proof}

Next, we provide a lemma that gives a simplified expression of the energy gradient operator when restricted to a subspace of low-energy eigenstates.

\begin{lemma}[Gradient in a subspace]
\label{lem:subspace-gradient-expr}
In the setting of Lemma~\ref{lem:gradient_upperbound}, suppose $\vH$ has a subspace of low-energy eigenstates with corresponding projector $\vQ$ that is separated from the higher energy eigenstates by an excitation gap $\Delta_{\vQ}$.
Let $\Delta_\nu=\min_{\nu_1\neq \nu_2\in B(\vH|_{\vQ})} |\nu_1-\nu_2|$ be the Bohr-frequency gap of $\vH$ restricted to the subspace. Assuming $\Delta_\nu/2 < \Delta_{\vQ}$, then the energy gradient operator in the subspace can be approximated using
\begin{align}
\Bigg\|\vQ\CL^{\dag}[\vH]\vQ - \sum_{a\in S} \sum_{\nu\in B(\vH|_{\vQ})} \vQ \vA^{a\dag}_\nu \vQ \vA^a_{\nu} \vQ \int_{-\infty}^0 
\gamma_\beta(\omega) \omega  |\hat{f}_\mu(\omega-\nu)|^2 \rd\omega \Bigg\| \le \epsilon
\end{align}
where $\mu=\Delta_\nu/2$ and 
\begin{align}
\epsilon \le |S| \CO\L(\frac{\norm{\vH}^{3/4}}{\tau^{1/4}} + \frac{1}{\tau} + \frac{1}{\beta} + \frac{\|\omega\gamma_\beta(\omega)\|_{\infty}}{\sqrt{\Delta_\nu \tau}}\R).    
\end{align}
\end{lemma}
\begin{proof}
We invoke a series of approximations to rewrite in terms of the exact Bohr frequencies on the subspace $\vP_{\rom1}$.
\begin{align}
   \CL^\dag[\vH] 
   & \stackrel{\verr_1}{\approx} \CD^\dag[\vH]\tag{Proposition~\ref{prop:norm-lamb-shift-part}}\\
   & \stackrel{\verr_2}{\approx}  
   \sum_{a\in S}\int_{-\infty}^{\infty} \gamma_{\beta}(\omega)\omega  \vA^{a}(\omega)^\dag\vA^{a}(\omega) \rd \omega \tag{Lemma~\ref{lemma:expr_energy_gradient}}\\
   & \stackrel{\verr_3}{\approx} \sum_{a\in S} \int_{-\infty}^{0} \gamma_{\beta}(\omega) \omega  \vA^{a}(\omega)^\dag\vA^{a}(\omega) \rd \omega\tag{Operator norms: Corollary~\ref{cor:iso_OFT}}\\
   &\stackrel{\verr_4}{\approx} \sum_{a\in S} \int_{-\infty}^{0} \gamma_{\beta}(\omega) \omega  \vS^{a}(\omega)^\dag\vS^{a}(\omega) \rd \omega\tag{secular approximation: Corollary~\ref{cor:AA-SS}}\\
   &= \sum_{a\in S} \sum_{\nu',\nu\in B(\vH)} \int_{-\infty}^{0} \gamma_{\beta}(\omega) \omega \vA^{a\dag}_{\nu'}\vA^{a}_{\nu} \hat{f_{\mu}}^*(\omega-\nu') \hat{f_{\mu}}(\omega-\nu) \rd \omega =: \vX.
   \label{eq:La-before-subspace}
\end{align}
The errors are $\verr_1 = \CO(\labs{S}\norm{\vH}^{3/4}/\tau^{1/4})$, $\verr_2 = \CO(\labs{S}/\tau)$, $\verr_3 = \CO(\labs{S}/\beta)$, and $\verr_4 = \CO(\labs{S}\times \|\omega\gamma_\beta(\omega)\|_{\infty}/\sqrt{\mu\tau})$.
In particular, ${\verr_3}$ arises from dropping the positive integral range, with error bounded by $\max_{\omega\ge0} \omega\gamma_\beta(\omega) \le 1/\beta$. 

Sandwiching Eq.~\eqref{eq:La-before-subspace} with $\vQ$ further simplifies the expression as it restricts to transitions in the subspace. 
Specifically, we have
\begin{align}
\vQ \vX \vQ  &= \sum_{a\in S}\sum_{\nu',\nu\in B(\vH)} \int_{-\infty}^{0} \gamma_{\beta}(\omega) \omega \vQ \vA^{a\dag}_{\nu'} \vQ  \vA^{a}_{\nu} \vQ  \hat{f_{\mu}}^*(\omega-\nu') \hat{f_{\mu}}(\omega-\nu) \rd \omega\tag{no heating transitions}\\
&=  \sum_{a\in S} \sum_{\nu\in B(\vH|_{\vQ})} \int_{-\infty}^{0} \gamma_{\beta}(\omega) \omega  \vQ \vA^{a\dag}_{\nu}\vQ\vA^{a}_{\nu}\vQ \labs{\hat{f}_\mu(\omega-\nu)}^2 \rd \omega\tag{different Bohr-frequencies decohere}.
\end{align}
The first line inserts an additional projector $\vQ$ between $\vA^{a\dag}_{\nu'}$ and $\vA^a_{\nu}$ because any transition to excited states require $\nu,\nu' > \Delta_{\vQ}$, but this is forbidden by the restrictions that $\omega \le 0$ (from the integral) and that $|\nu-\omega|, |\nu'-\omega| < \mu < \Delta_{\vQ}$ (from  the secular approximation).
In the second line, since the Bohr frequencies in $B(\vH|_{\vQ})$ are at least $\Delta_\nu=2\mu$ apart, we must have that 
\begin{align}
\hat{f}^*_{\mu}(\omega-\nu') \hat{f}_{\mu}(\omega-\nu)=0 \quad \text{for all}\quad \omega\in \BR, \quad \text{unless}\quad \nu'=\nu.
\end{align}
Combining the above with Eq.~\eqref{eq:La-before-subspace} to conclude the proof.
\end{proof}

When the Hamiltonian is local, thinking about the gradient ``locally'' is sometimes useful. The following lemma gives a sufficient condition that guarantees a global gradient. Since the consequence is strong, the premise is also more stringent; it is only helpful when the Hamiltonian is frustration-free.
\begin{lemma}[Local-to-global gradient condition]
\label{lem:localconvex}
Suppose $\vH=\sum_i \vh_i$, where each term $\vh_i\succeq 0$. Then for any (not necessarily thermal) Lindbladian $\CL$,
\begin{equation}
    -\CL^{\dag}[\vh_i] \succeq r_i \vh_i
    \qquad \Longrightarrow \qquad
    -\CL^{\dag}[\vH] \succeq r \vH,
\end{equation}
where $r = \min_i r_i$.
\end{lemma}
\begin{proof}
By linearity, we have $-\CL^\dag[\vH] = \sum_i -\CL^\dag[\vh_i]\succeq \sum_i r_i \vh_i $.
Since $r_i \vh_i \succeq r \vh_i$, we have $-\CL^{\dag}[\vH] \succeq r\sum_i \vh_i = r \vH$,
concluding the proof.
\end{proof}

\subsection{Gradients for commuting Hamiltonians}

When we are given a commuting Hamiltonian, the energy gradient induced by any local jump operator can be understood by restricting the system to its neighborhood.
In this situation, the negative gradient condition for the overall Hamiltonian can be decomposed into conditions that can be checked locally.
This gives an efficient method to show a commuting Hamiltonian has a negative gradient for all its excited states, which we elucidate in this section of the appendix.

Recall the thermal Lindbladian $\CL:=\CL^{\beta,\tau,\vH}$ defined in Eq.~\eqref{eq:dissipative-Lind} for a local jump operator $\vA^a$, whose Heisenberg picture is
\begin{equation}
\CL_a^{\dag\beta,\tau,\vH}[\vO]  = \ri[\vH^{\beta,\tau,\vH}_{LS,a}, \vO] + \CD_a^{\dag\beta,\tau,\vH}[\vO],
\end{equation}
where
\begin{equation}  
\CD_a^{\dag\beta,\tau,\vH}[\vO] =
\int_{-\infty}^\infty \gamma_\beta(\omega) \Big[\hat\vA^a(\omega)^\dag \vO \hat\vA^a(\omega) - \frac12 \{\hat\vA^a(\omega)^\dag \hat\vA^a(\omega), \vO\}\Big]\rd \omega.
\end{equation}
Note $\hat{\vA}^a(\omega)$ is the operator Fourier transform of $\vA^a(t)=\e^{\ri\vH t} \vA^a \e^{-\ri \vH t}$, and $\vH^{\beta, \tau, \vH}_{LS,a}$ is a Lamb-shift term defined in Eq.~\eqref{eq:Lamb_shift}.

When $\vH$ is a commuting Hamiltonian (e.g.,~\cite{kastoryano2016commuting,capel2021modified}), an important observation is that $\vA^a(t)$ only depends on the part of $\vH$ that does not commute with $\vA^a$. In particular, the energy gradient for each jump operator only depends on the neighborhood of $\vA^a$.

\begin{lemma}[Commuting Hamiltonian and localized Lindblad operators]
\label{lem:commuting-decomp}
Suppose $\vH=\sum_e \vh_e$ is a commuting Hamiltonian. For any jump operator $\vA^a$, the associated energy gradient simplifies to
\begin{equation}
\CL_a^{\dag\beta,\tau,\vH}[\vH] = \CL_a^{\dag\beta,\tau,\vH_{\ni a}}[\vH_{\ni a}]
\end{equation}
where $\vH_{\ni a} = \sum_{e:\,[\vh_e, \vA^a] \neq 0} \vh_e$ is the part of $\vH$ does not commute with $\vA^a$.
\end{lemma}
\begin{proof}
When $\vH$ is commuting, we have $\vA^a(t) = \e^{\ri \vH t} \vA^a \e^{-\ri \vH t} = \e^{\ri \vH_{\ni a} t} \vA^a \e^{-\ri \vH_{\ni a} t}$, so the Lindbladian superoperator only depends on $\vH_{\ni a}$, i.e., $\CL_a^{\dag\beta,\tau,\vH} = \CL_a^{\dag\beta,\tau,\vH_{\ni a}}$.

Let $\vH_{\not\ni a} = \vH - \vH_{\ni a}$ be the part of $\vH$ that commutes with $\vA^a$.
Since $[\vH_{\not\ni a}, \vH_{\ni a}]=0$, we have $[\vA^a(t), \vH_{\not\ni a}] = 0$ for each $t$, which implies $[\vH^{\beta, \tau, \vH}_{LS,a}, \vH_{\not\ni a}] = [\hat{\vA}^a(\omega), \vH_{\not\ni a}]=0$.
Thus we have $\CL_a^{\dag\beta,\tau,\vH}[\vH_{\not\ni a}] = 0$.
\end{proof}

\subsection{Negative gradient condition under perturbations to Hamiltonians}

We next look at how the negative energy gradient condition changes under perturbations to the $n$-qubit Hamiltonian $\vH$.
See Appendix~\ref{sec:monotone_gradient} for the proof of the following theorem.
\begin{theorem}[Monotonicity of gradient under level splitting]\label{thm:mono_gradient}
    Consider a highly degenerate Hamiltonian $\vH = \sum_{\bar{E}} \bar{E}\vP_{\Bar{E}}$ with Bohr-frequency gap $\Delta_{\nu} := \min_{\nu_1\neq \nu_2\in B(\vH)} |\nu_1-\nu_2|$ of $\vH$, and add a perturbation $\vH':=\vH+\vV$. Let $\vP = \vP_{\bar{E}}$ be a projector to an energy subspace and $\vP'$ the corresponding perturbed subspace.
    Suppose the perturbation is weaker than the Bohr-frequency gap, 
        $\norm{\vV} \le \frac18\Delta_{\nu}$.
    For any $\beta, \tau >0$, let $\CL = \sum_{a\in S} \CL^{\beta,\tau, \vH}_a, \CL' = \sum_{a\in S} \CL^{\beta,\tau, \vH'}_a$ be thermal Lindbladians with jumps $\{\vA^a\}_{a\in S}$, where $\norm{\vA^a}\le 1$ and the transition weight $\gamma_{\beta}(\omega)$ is given by Eq.~\eqref{eq:glauber-dyn}.
    Then we have the monotone property that
\begin{align}
    -\CL^\dag[\vH] \succeq r(\vI-\vP) -\epsilon \vI \quad \text{implies}\quad -\CL^{'\dag}[\vH'] \succeq r(\vI-\vP') - \epsilon' \vI
\end{align}
where
\begin{align} \label{eq:mono_err_bound}
    \epsilon' \le \epsilon+ \labs{S} \cdot \CO\L(\frac{1}{\tau} +\frac{\norm{\vH}^{3/4}}{\tau^{1/4}}+\frac{\Lambda_0^{2/3}}{\tau^{1/3}}+\frac{\Lambda_0}{\sqrt{\Delta_{\nu}\tau}}  + \frac{\e^{-\beta\Delta_{\nu}/4}}{\beta} +\big(1+\frac{\Lambda_0+r}{\Delta_\nu}\big)\norm{\vV}\R).
\end{align}
\end{theorem}

Finally, we look at how the negative energy gradient condition changes when restricted to a subspace.
See Appendix~\ref{sec:monotone_grad_subspace} for proofs of the following two corollaries.

\begin{corollary}[Monotonicity of gradient on a subspace]\label{cor:monotonicity_subspace}
Consider a Hamiltonian $\vH = \sum_{\bar{E}} \bar{E}\vP_{\Bar{E}}$ and its perturbation $\vH':=\vH+\vV$. Let $\vP$ be the ground space projector for $\vH$ and $\vP'$ be the corresponding perturbed eigensubspace of $\vH'$.
Let $\vQ$ be a low-energy eigensubspace projector of $\vH$ (i.e., $\vQ  = \sum_{E \le E_{\vQ}} \vP_E$ for $E_{\vQ} \in \text{Spec}(\vH)$) with excitation gap $\Delta_{\vQ}$.
Assume $\frac{\|\vV\|\norm{\vH}}{\Delta_{\vQ}} \le \frac{1}{144}\Delta_\nu$ where $\Delta_\nu := \min_{\nu_1\neq \nu_2\in B(\vH|_{\vQ})} |\nu_1-\nu_2|$ is the Bohr-frequency gap of $\vH$ within the subspace $\vQ$.  
For any $\beta, \tau >0$, let $\CL = \sum_{a\in S} \CL^{\beta,\tau, \vH}_a, \CL' = \sum_{a\in S} \CL^{\beta,\tau, \vH'}_a$ be thermal Lindbladians with jumps $\{\vA^a\}_{a\in S}$, where $\norm{\vA^a}\le 1$ and the transition weight $\gamma_{\beta}(\omega)$ is given by Eq.~\eqref{eq:glauber-dyn}.
    Then we have the monotone property that
\begin{align}
    - \vQ \CL^{\dag}[\vH]\vQ \succeq r \vQ(\vI-\vP) - \epsilon \vI
    \quad \text{implies} \quad - \vQ' \CL'^{\dag}[\vH']\vQ' \succeq r \vQ'(\vI-\vP') - \epsilon' \vI
\end{align}
where $\vQ'$ projects onto the perturbed eigensubspace of $\vH'$ identified with $\vQ$, and
\begin{align}
    \epsilon' \le \epsilon + \labs{S}\cdot \CO\bigg(\frac{1}{\tau} +\frac{\norm{\vH}^{3/4}}{\tau^{1/4}}+\frac{\Lambda_0^{2/3}}{\tau^{1/3}}+\frac{\Lambda_0}{\sqrt{\Delta_{\nu}\tau}} +\frac{\Lambda_0}{\sqrt{\Delta_{\vQ}\tau}} + \frac{\e^{-\beta\Delta_{\nu}/4}}{\beta} + \frac{\e^{-\beta\Delta_{\vQ}/4}}{\beta}  \nonumber \\
    + \L(1+\frac{\Lambda_0}{\Delta_\nu}\R)\frac{\norm{\vV}\norm{\vH}}{\Delta_{\vQ}} + r \Big(\frac{\norm{\vV}}{\Delta_{\vQ}}+\frac{\norm{\vV}}{\Delta_{\nu}}\Big)\bigg).
\label{eq:mono-subspace-err}
\end{align}
\end{corollary}

\begin{corollary}[Monotonicity of gradient on a subspace under off-block-diagonal perturbation]\label{cor:off_diag_mono}
In the setting of Corollary~\ref{cor:monotonicity_subspace}, instead assume $\frac{\norm{\vV}}{\Delta_\nu},\frac{\norm{\vV}}{\Delta_{\vQ}} \le (const.)$, and that the perturbation is off-block-diagonal, i.e.,  $\vQ\vV\vQ=(\vI-\vQ)\vV(\vI-\vQ) = 0$. Then,
\begin{align}
    - \vQ \CL^{\dag}[\vH]\vQ \succeq r \vQ(\vI-\vP) - \epsilon \vI
    \quad \text{implies} \quad - \vQ' \CL'^{\dag}[\vH']\vQ' \succeq r \vQ'(\vI-\vP') - \epsilon' \vI
\end{align}
where
\begin{align}
    \epsilon' &\le \epsilon + \labs{S}\cdot \CO\bigg(\frac{1}{\tau} +\frac{\norm{\vH}^{3/4}}{\tau^{1/4}}+\frac{\Lambda_0^{2/3}}{\tau^{1/3}}+\frac{\Lambda_0}{\sqrt{\Delta_{\nu}\tau}} +\frac{\Lambda_0}{\sqrt{\Delta_{\vQ}\tau}} + \frac{\e^{-\beta\Delta_{\nu}/4}}{\beta} + \frac{\e^{-\beta\Delta_{\vQ}/4}}{\beta}   \nonumber\\
    &\qquad\qquad\qquad\qquad +\frac{\norm{\vV}^2}{\Delta_{\vQ}} +  \norm{\vH_{\vQ}} \cdot \Big( \frac{\norm{\vH_{\vQ}}\norm{\vV}}{\Delta_{\vQ}\Delta_\nu}+\frac{\norm{\vV}^2}{\Delta_{\vQ}\Delta_\nu}\Big) 
    + r \Big(\frac{\norm{\vV}}{\Delta_{\vQ}}+\frac{\norm{\vV}^2}{\Delta_{\vQ}\Delta_{\nu}}\Big)\bigg).
\label{eq:mono-subspace-offdiag-err-cor}
\end{align}
\end{corollary}

\clearpage
\section{Energy landscape of an Ising chain}
\label{sec:example-Ising}

In this appendix, we take a brief aside to characterize the energy landscape of the one-dimensional ferromagnetic Ising chain under thermal perturbations.
This provides a basic example on how the definition of local minima under thermal perturbations is related to the physical picture.
We will see that this system has many suboptimal local minima in the absence of an external field with a lifetime polynomial in the system size.
Once an external field is added, however, the system essentially has no suboptimal local minima and can quickly cool to the ground state where all spins are aligned.
This observation corresponds to the following physical phenomena: When there is no external magnetic field, a ferromagnetic system will often be stuck in a configuration with many domain walls, and an externally applied magnetic field can quickly magnetize the system.

The Hamiltonian for the ferromagnetic Ising chain on a periodic boundary condition is
\begin{equation}
\vH = -\sum_{j=1}^n \vZ_j \vZ_{j+1} - h\sum_{j=1}^n \vZ_j,
\end{equation}
where we identify $\vZ_{n+1}\equiv \vZ_1$.
Intuitively, this system energetically favors configurations where adjacent spins are aligned.
When $h=0$, we have two degenerate ground states, $\ket{00\cdots0}$ and $\ket{11\cdots1}$, which are the global minima.
This degeneracy is broken when $h\neq 0$, and these two states split by energy $2nh$.
The system also has many excited states with \emph{domain walls}, i.e., locations where adjacent spins are anti-aligned such as $\ket{01}$ and $\ket{10}$.
In what follows, we study the energy landscape of the above system under thermal perturbations with jump operator $\{\vA^j = \vX_j\}_{j=1}^n$, setting $\tau=\infty$ for simplicity. 
We analyze three cases.

\paragraph{Case 1: no external field ($h=0$).}
In this case, we will see that any bit string state with domain walls sufficiently far from each other, e.g. $\ket{\cdots0001111000\cdots}$ is a suboptimal local minimum.
Indeed, there is no local operation to strictly decrease the energy of such states; the jump operators $\{\vX_j\}$ can only displace the domain walls by one site, which does not change the energy.

We can see this more formally by computing the energy gradient operator.
Since $\vH$ is a commuting Hamiltonian, we may apply \cref{lem:commuting-decomp} and study the gradient induced by a single jump operator $\vX_j$ by restricting the Hamiltonian to its neighborhood, i.e.,
\begin{equation}
\vH_{\ni j} = -\vZ_{j-1} \vZ_{j} - \vZ_j \vZ_{j+1}.
\end{equation}
Observe $\vH_{\ni j}$ has three degenerate eigenspaces $\vP^{E}_j$ with energy $E$ as follows:
\begin{gather}
\vP^{-2}_j = \sum_{\sv\in \{000,111\}}\ketbrat{\sv}_{j-1,j,j+1}, \qquad \qquad 
\vP^{2}_j = \sum_{\sv\in \{010,101\}} \ketbrat{\sv}_{j-1,j,j+1}, \nonumber \\
\text{and}\qquad
\vP^0_j = \sum_{\sv\in \{001,100,011,110\}}\ketbrat{\sv}_{j-1,j,j+1}.
\end{gather}
Then the negative Bohr-frequencies and the associated jumps are
\begin{align}
    \vA^{j}_{\nu_1} &= \vP_j^{-2} \vX_{j} \vP_j^2 = (\ketbra{000}{010}+\ketbra{111}{101})_{j-1,j,j+1}, &\nu_1&=-4, \nonumber\\
    \vA^{j}_{\nu_2} &= \vP_j^0 \vX_{j} \vP_j^2 + \vP_j^{-2} \vX_{j} \vP_j^0 = 0, &\nu_2 &= -2.
\end{align}
Hence, the energy gradient operator associated with jump $\vX_j$ is
\begin{align}
\CD_j^{\dag\beta,\infty,\vH}[\vH] &= \CD_j^{\dag\beta,\infty,\vH_{\ni j}}[\vH_{\ni j}] = {\textstyle \sum_{\nu\in B(\vH_{\ni j})} \nu \gamma_\beta(\nu) \vA_{\nu}^{j\dag} \vA_{\nu}^j}
\nonumber \\
&= \theta_0\cdot (\ketbrat{010}+\ketbrat{101})_{j-1,j,j+1} + \CO(e^{-4\beta})
\end{align}
where $\theta_0 = -4\gamma_\beta(-4)=-\Omega(1)$.
As we can see, the energy gradient is essentially 0 when the domain walls are more than distance 1 apart, and only becomes significant when two domain walls are next to each other, as in $\ket{\cdots010\cdots}$ or $\ket{\cdots101\cdots}$.
This implies the presence of exponentially many suboptimal local minima; for example, choose whether or not to have a domain wall every 2 sites.

Despite the presence of many suboptimal local minima, we now argue that they have a lifetime polynomial in the system size $n$ when the system evolves under thermal perturbations.
We may understand the dynamics of the system as a random walk of domain walls, and two domain walls annihilate each other when they meet.
Since two domain walls at distance $\ell$ apart moving under diffusive dynamics take $\CO(\ell^2)$ time to meet, a suboptimal local minimum with $k$ domain walls decays to a lower energy state after approximately $\CO(n^2/k^2)$ time.

\paragraph{Case 2: weak external field ($0<h<2$).}
In this case, the ground state of $\vH$ is uniquely $\ket{0^n}$, as all spins are slightly favored to be in the $\ket0$ state instead of the $\ket1$ state.
When the domain walls are far apart, e.g. $\ket{\cdots0001111000\cdots}$, the applied external field causes an attraction across the domain of $1$'s, which energetically favors the domain walls to move closer together.
The presence of the field $h$ removes all the suboptimal local minima that were in the previous case.
The state $\ket{1^n}$, which was a ground state in the previous case, becomes now the only suboptimal local minimum.

We now more formally characterize the energy landscape of $\vH$ using the energy gradient operator.
Again applying \cref{lem:commuting-decomp}, we may consider the gradient induced by a single jump operator by focusing on its neighborhood.
The relevant neighborhood Hamiltonian is
\begin{equation}
\vH_{\ni j} = -\vZ_{j-1} \vZ_{j} - \vZ_j \vZ_{j+1} - h\vZ_j.
\end{equation}
Turning the crank, we see that the negative Bohr-frequencies and the associated jumps are:
\begin{align}
    \vA^{j}_{\nu_1} &= \ketbra{000}{010}_{j-1,j,j+1}, &\nu_1&=-4-2h, \nonumber\\
    \vA^{j}_{\nu_2} &= \ketbra{111}{101}_{j-1,j,j+1}, &\nu_2&=-4+2h, \nonumber\\
    \vA^{j}_{\nu_3} &= (\ketbra{001}{011} + \ketbra{100}{110})_{j-1,j,j+1}, &\nu_3 &= -2h.
\label{eq:Ising-jumps}
\end{align}
Then the energy gradient operator associated with jump $\vX_j$ is
\begin{align}
&\CD_j^{\dag\beta,\infty,\vH}[\vH] = {\textstyle \sum_{\nu\in B(\vH_{\ni j})} \nu \gamma_\beta(\nu) \vA_{\nu}^{j\dag} \vA_{\nu}^j}
\nonumber \\
&\qquad = (\theta_1 \ketbrat{010}+\theta_2 \ketbrat{101} + \theta_3\ketbrat{011} + \theta_3\ketbrat{110})_{j-1,j,j+1} + \CO(e^{-2\beta h}),
\end{align}
where $\theta_j = \nu_j \gamma_\beta(\nu_j)$.
As we can see, any configuration with a domain wall now has a significant gradient from at least one of the jumps.
The only configurations without a significant energy gradient are $\ket{0^n}$, the ground state, and $\ket{1^n}$, a metastable local minimum.

\paragraph{Case 3: strong external field ($h>2$).}
In this case, the external field is sufficiently strong that the state $\ket{1^n}$ is no longer a local minimum, and $\vH$ has no suboptimal local minima.
To see this, we note that $h>2$ implies that $\nu_2 > 0$ in Eq.~\eqref{eq:Ising-jumps}, which means the energetically favored jump operator is actually $\vA^{j}_{-\nu_2} = \vA^{j\dag}_{\nu_2}$.
This implies the energy gradient operator induced by the jump $\vX_j$ in this case is
\begin{equation}
\CD_j^{\dag\beta,\infty,\vH}[\vH] = (\theta_1 \ketbrat{010}+\theta'_2 \ketbrat{111} + \theta_3\ketbrat{011} + \theta_3\ketbrat{110})_{j-1,j,j+1} + \CO(e^{-2\beta (h-2)}),
\end{equation}
where $\theta'_2 = -\nu_2\gamma_\beta(-\nu_2)$.
This gives the state $\ket{1^n}$ a significant energy gradient, and thus the ground state $\ket{0^n}$ is the only local minimum of $\vH$.

\clearpage

\section{All local minima are global in \textsf{BQP}-hard Hamiltonians (Proof of Theorem~\ref{thm:no-suboptimal-local-minima-informal2})}
\label{sec:universal-quantum-computation}

A main result of our work is that the task of finding a local minimum for $\vH_{\circuit}$ under thermal perturbation is universal for quantum computation and hence classically hard. 
As we have seen in the main text and Appendix~\ref{sec:main-thermal}, this main result follows from Theorem~\ref{thm:no-suboptimal-local-minima-informal2}, which we prove in this appendix.

We start by defining $\vH_\circuit$ in detail.
Given a 2D $n$-qubit circuit $\vU_{\circuit}=\vU_T\cdots\vU_2 \vU_1$ with $T = 2 t_0 + L = \mathrm{poly}(n)$ gates as constructed in Fig. 1 of Ref.~\cite{OliveiraTerhal}, where the first and last $t_0$ gates are identity gates and each gate of the 2D circuit $\vU_{\circuit}$ is geometrically adjacent to the subsequent gate. We consider a geometrically local Hamiltonian on a 2D lattice with $n + T$ qubits defined as follows.

\begin{definition}[Modified circuit-to-Hamiltonian construction]
\label{def:circuit-H}
Consider a 2D circuit $$\vU_{\circuit}=\vU_T\cdots\vU_2 \vU_1$$ on $n$ qubits with $T = 2t_0 + L$ gates, where the first and last $t_0 = c L^2$ gates are identity gates with $c = \mathcal{O}(1)$, and each consecutive gates are geometrically adjacent.
We define a geometrically-local Hamiltonian $\vH_{\circuit}$ on a 2D lattice with $n+T$ qubits as follows,
\begin{equation}
    \vH_{\circuit} :=  \vH_{\clock} +  \vH_{\inp} + \vH_{\prop} \quad \text{acting on}\quad (\BC^2)^{\otimes n}\otimes (\BC^2)^{\otimes T}, \label{eq:H_clock_in_prop}
\end{equation}
where each individual term is given by
\begin{align*}
    \vH_{\clock} &:= J_{\clock}\sum_{t=1}^{T-1} f_t  \vI \otimes \ketbrat{01}_{t,t+1}, \\
    \vH_{\inp} &:= J_{\inp}\sum_{j=1}^n g_j \ketbrat{1}_j  \otimes \ketbrat{10}_{t_{j}-1,t_j}, \\
    \vH_{\prop} &:= \frac12J_{\prop}  \sum_{t=1}^T \vH_{\prop}(t),\\
    \vH_{\prop}(1) &:=  \vI - h_1(\vU_1 \otimes \ketbra{10}{00}_{1,2}  + \vU_1^\dag \otimes\ketbra{00}{10}_{1,2}), \\
    \vH_{\prop}(t) &:=  \vI - h_t (\vU_t \otimes \ketbra{110}{100}_{t-1,t,t+1} + \vU_t^\dag \otimes \ketbra{100}{110}_{t-1,t,t+1}) \qquad \text{for each}\quad 1<t<T,\\
    \vH_{\prop}(T) &:= \vI - h_T (\vU_T \otimes \ketbra{11}{10}_{T-1,T} + \vU_T^\dag \otimes \ketbra{10}{11}_{T-1,T}).
\end{align*}
The $T$ qubits correspond to the $T$ geometrically-local gates and are placed next to each gate to ensure $\vH_{\circuit}$ is geometrically local.
The couplings are chosen as
\begin{equation}
    J_{\clock}=1, \quad f_t = (T-t)/T, \quad g_j = 1/\xi_{t_j-1}, \quad h_{t}=\sqrt{t(T-t+1)}.
\end{equation}
We will set the other parameters $J_{\inp}, J_{\prop}$ later. The time $t_j$ is the first time qubit $j$ is acted on.
\end{definition}

We will show later in Appendix~\ref{sec:Hcirc-low-energy} that $\vH_{\circuit}$ has a unique ground state given by
\begin{equation}
    \ket{\eta_{\bm0}} = \sum_{t=0}^T \sqrt{\xi_t}\big(\vU_t \cdots \vU_1 \ket{0^n}\big) \otimes \ket{0^t 1^{T-t}} 
    \qquad
    \text{where} \quad
    \xi_t := \frac{1}{2^T}  \binom{T}{t}.
\end{equation}
Note this state encodes the computational history of the circuit $\vU_{\circuit}$.
By choosing $t_0=L^2$, we ensure that each time in the interesting part of the computational history (i.e., the intermediate $L$ gates) can be observed with $\Omega(1/T)$ probability as we will show later in \cref{prop:idling}.
This also implies $g_j = \CO(T)$.

We now state a detailed version of Theorem~\ref{thm:no-suboptimal-local-minima-informal2} based on the definition of $\vH_{\circuit}$ in the following.

\begin{theorem}[All local minima are global in $\vH_{\circuit}$] \label{thm:Hcircuit}
Let $\vP_G$ be the ground-space projector for the Hamiltonian $\vH_{\circuit}$ in Eq.~\eqref{eq:H_clock_in_prop}.
For any failure probability $0<\delta<1$,
there is a parameter choice $J_{\inp}, J_{\prop} = \poly(n, T, \delta^{-1})$ and a choice of $m$ two-qubit jump operators 
\begin{equation}
\label{eq:jump-ops-for-HC}
    S_0=\{\vA^a\}_{a=1}^m := \{\vI\otimes \vX_t, \vI\otimes \vZ_t\}_{t=1}^T \cup \{\vX_j \otimes \ketbrat{0}_{t_j} \}_{j=1}^n
\end{equation}
with $m = 2T+n$ satisfying the following:

For a sufficiently small $\epsilon = 1/ \poly(n, T,\delta^{-1})$, any $\epsilon$-approximate local minimum $\vrho$ of $\vH_{\circuit}$ under thermal perturbations with sufficiently large $\beta = \poly(n, T, \delta^{-1})$, $\tau = \poly(n, T, \delta^{-1})$, and system-bath interactions generated by $S_0$ is an exact global minimum with probability $\Tr(\vP_G(\vH_{\circuit}) \vrho) \geq 1 - \delta$.
\end{theorem}
We remind the reader that the thermal Lindbladians that generate the perturbations are defined in Eq.~\eqref{eq:dissipative-Lind}. The transition weight $\gamma_\beta(\omega)$ is chosen to be Glauber dynamics as defined in Eq.~\eqref{eq:glauber-dyn}, with energy cut-off $\Lambda_0=1$ as a convenient choice so that $\|\omega\gamma_\beta(\omega)\|_{\infty}\le 1$. We do not expect our result to change with other reasonable choices of $\gamma_\beta(\omega)$.

\begin{remark}
Our $\vH_{\circuit}$ is similar to previous circuit-Hamiltonian constructions (see e.g., \cite{OliveiraTerhal, AharonovAQCUniversal, KSV02}), but there are some significant differences.
One key change is that $\vH_{\prop}$ is no longer frustration-free, and its couplings $h_t$ are not uniform;
consequently, this revised $\vH_{\prop}$ has better spectral properties that enable us to lower bound its Bohr-frequency gap.
Furthermore, $\vH_{\clock}$ is given non-uniform couplings $f_t$ so that any local excitation has an incentive to move rightwards (e.g. $\ket{0011}\to\ket{0001}$), ensuring $\vH_{\clock}$ has no local minima except its ground states.
These modifications allow us to prove that all excited states of $\vH_\circuit$ have significant negative gradients, so that they will all flow to the ground state under thermal perturbations.
\end{remark}

\subsection{Characterizing low energy states of $\vH_{\circuit}$}
\label{sec:Hcirc-low-energy}

We will start by characterizing the low energy states of the circuit Hamiltonian.
We define the following sequence of Hamiltonians
\begin{align*}
    \vH_{\rom1} &= \vH_{\clock} \\
    \vH_{\rom2} &= \vH_{\clock} + \vH_{\prop}\\
    \vH_{\rom3} &= \vH_{\clock} + \vH_{\prop} + \vH_{\inp} = \vH_{\circuit}
\end{align*} 
with the ground space projectors $\vP_j$ such that
\begin{align}
    \vP_{\rom1} \supset \vP_{\rom2} \supset \vP_{\rom3}.
\end{align}
Equivalently, we have 
$\vP_{\rom1} \vP_{\rom2} = \vP_{\rom2}$ and $\vP_{\rom1} \vP_{\rom3} = \vP_{\rom2} \vP_{\rom3}=\vP_{\rom3}$.
Our approach to calculating the gradient for $\vH_{\rom3}= \vH_{\circuit}$ is \textit{perturbative}: we start with the simple Hamiltonian $\vH_{\rom1}$ and gradually add perturbations (which will split the spectrum, Figure~\ref{fig:levelsplitting}). Remarkably, the gradient is \textit{stable} as long as the perturbation is weak enough. That is, it suffices to analyze the gradient of the simpler, unperturbed Hamiltonians on suitable subspaces.

Now we describe explicitly the ground subspaces $\vP_{\rom1}$, $\vP_{\rom2}$ and $\vP_{\rom3}$. Let  
\begin{equation}
    \ket{C_t} = \ket{1^t 0^{T-t}}\quad \text{for each}\quad t=0,1,\ldots,T
\end{equation}
and 
\begin{equation}
\ket{\eta_{\xv,t}} = \big(\vU_t \cdots \vU_1 \ket{\vect{x}}\big) \otimes \ket{C_t}\quad \text{for each}\quad \vect{x}\in \{0,1\}^n\quad \text{and}\quad 0\le t\le T.
\end{equation}
The set of $\ket{\eta_{\xv,t}}$ forms an orthonormal basis for the ground space of $\vH_{\rom1}=\vH_{\clock}$, with energy 0 and a spectral gap of $J_{\clock}/T$.
The ground space projector is
\begin{equation}
    \vP_{\rom1} = \sum_{\xv\in \{0,1\}^n}\sum_{t=0}^T  \ketbrat{\eta_{\xv,t}} \quad \text{and}\quad \braket{ \eta_{\yv,t'}|\eta_{\xv,t}} = \delta_{\yv\xv}\delta_{tt'}.
\end{equation}

Observe that $[\vH_{\prop}, \vP_{\rom1}] = 0$, so the ground states of $\vH_{\rom2} = \vH_{\clock} +  \vH_{\prop}$ are given as the ground states of 
\begin{align}
\vP_{\rom1} \vH_{\prop} \vP_{\rom1} = \frac{J_{\prop}}{2}\sum_{t=1}^T [
    \vI
    - h_t (\vU_t \otimes \ketbra{C_{t}}{C_{t-1}} + \vU_t^\dag \otimes \ketbra{C_{t-1}}{C_t} )].    
\end{align}
Furthermore, observe the orthogonality relations
\begin{equation}
\label{eq:Hprop-blockdiag}
    \braket{\eta_{\xv,t}| \vH_{\prop}|\eta_{\yv,t'}}= 0
    \qquad \text{ when}\quad \xv\neq \yv \quad \text{for each}\quad t,t'.
\end{equation}
That is, $\vP_{\rom1} \vH_{\prop} \vP_{\rom1}$ is block diagonal with blocks labeled by $\xv$.
Moreover, for any $\vect{x}$, in the basis of $\ket{\eta_{\vect{x},0}}, \ket{\eta_{\vect{x},1}}, \ldots, \ket{\eta_{\vect{x},T}}$, we can explicitly write down the effective $(T+1)\times (T+1)$ Hamiltonian 
\begin{equation}
\vP_{\rom1} \vH_{\prop} \vP_{\rom1} = \frac{J_{\prop}}{2}
\begin{pmatrix}
T & -h_1 &   & & \\
     -h_1 & T & -h_2 & & \\
  	    & -h_2 & T & -h_3&  \\
	    &	& -h_3 & \ddots &\ddots \\
	    &	&  & \ddots & T & -h_T \\
	    &	&	&	& -h_T & T
\end{pmatrix}
= J_{\prop} \Big(\frac{T}{2}\vI - \vL_x\Big),
\end{equation}
where $\vL_x$ is the matrix representation of the spin-$T/2$ angular momentum operator whose spectrum is well known. In particular, the unique ground state is 
\begin{equation}
    \ket{\eta_\xv} := \sum_{t=0}^T \sqrt{\xi_t} \ket{\eta_{\xv,t}}\quad \text{with energy}\quad 0 \quad \text{and spectral gap}\quad J_{\prop}.
\end{equation}
We will call $\ket{\eta_\xv}$ the \emph{history state} with respect to input $\ket{\xv}$.
The ground space projector of $\vH_{\rom2}$ is then given as
\begin{equation}
    \vP_{\rom2} = \sum_{\xv\in \{0,1\}^n} \ketbrat{\eta_{\xv}}.
\end{equation}

Finally, note $\ket{\eta_{\bm{0}}}$ for $\bm{0} = (0,0,\cdots,0)$ is the unique ground state of $\vH_{\circuit}=\vH_{\rom3}$ and so
\begin{equation}
    \vP_{\rom3} = \ketbrat{\eta_{\bm{0}}}.
\end{equation}
This is because $\vH_{\inp}$ is positive semi-definite, and $\ket{\eta_{\bm{0}}}$ is the only state in $\vP_{\rom2}$ with zero eigenvalue with respect to $\vH_{\inp}$.

\begin{figure}[t]
	\centering
	\includegraphics[width=\textwidth]{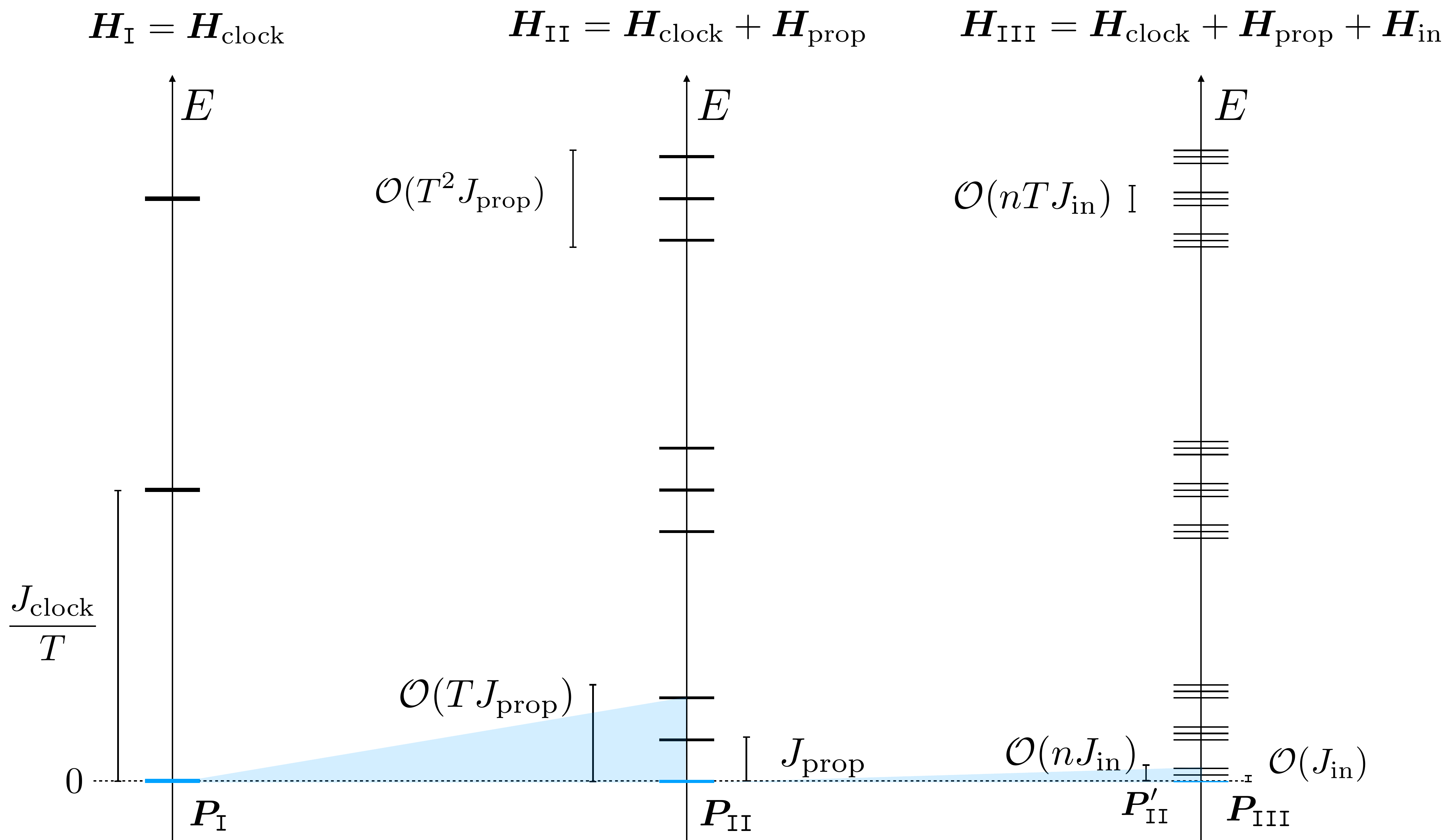}
	\caption{ The degenerate levels of $\vH_{\clock}$ split under perturbations $\vH_{\prop}$ and $\vH_{\inp}$. In particular, the ground state splitting is tracked in blue shades. The careful choice of energy scales ensures that the levels can be identified with the original degenerate blocks.
	}
	\label{fig:levelsplitting}
\end{figure}

\subsection{Proof of Theorem~\ref{thm:Hcircuit}}
To prove Theorem~\ref{thm:Hcircuit}, we show that all excited states in $\vI-\vP_{\rom3}$ have significant gradient relative to $\vH_{\rom3}$.
Our analysis for the gradient will be carried out in three subspaces
\begin{align}
    \vI - \vP_{\rom3} = \undersetbrace{\text{studying }\vH_{\rom1}}{(\vI-\vP_{\rom1})} + \undersetbrace{\text{studying }\vH_{\rom2}\vP_{\rom1}}{\vP_{\rom1}(\vI-\vP_{\rom2})} + \undersetbrace{\text{studying }\vH_{\rom3}\vP_{\rom2}}{\vP_{\rom2}(\vI-\vP_{\rom3})}
\end{align}
Let $\CL_j:=\CL^{\beta,\tau,\vH_j}$ be the thermal Lindbladian with uniform weights as in Eq.~\eqref{eq:thermal-L-unif}, defined with respect to $\vH_j$ and the jump operators in Eq.~\eqref{eq:jump-ops-for-HC} .

\paragraph{Case 1: Gradients for $\vI-\vP_{\rom1}$ from $\vH_{\rom1}$.}

We first show excited states of $\vH_{\rom1}$ have good energy gradient:
\begin{equation}
\label{eq:L1-grad}
    -\CL_{\rom1}^\dag[\vH_{\rom1}] \succeq r_1 (\vI - \vP_{\rom1}) - \epsilon_{1a}\vI.
\end{equation}
Because $\vH_{\rom1}=\vH_{\clock}$ is a commuting Hamiltonian, the global gradient can be lower bounded by checking the local gradient from individual local jumps. We carry out this computation in \cref{sec:gradient-clock}, where we show
$r_1 = \Omega(1/T\ln\beta)$ and $\epsilon_{1a}=\CO(T^{7/4}/\tau^{1/4} + T/\beta + T(1+\beta)\ln \tau / \tau )$ in \cref{lem:Hclock-grad}.

We then apply \cref{thm:mono_gradient} (with $\vH=\vH_{\rom1}$ and $\vH'=\vH_{\rom3}$) to show excited states of $\vH_{\rom1}$ have large gradient with respect to $\vH_{\rom3}$.
In other words,
\begin{equation} \label{eq:L3-outside-P1}
    -\CL_{\rom3}^\dag[\vH_{\rom3}] \succeq r_1(\vI-\vP_{\rom1}) - \epsilon_1\vI.
\end{equation}
To do this, we only need to check that the conditions of \cref{thm:mono_gradient} are satisfied.
Note $\vP_{\rom1}$ also projects onto eigenstates of $\vH_{\rom3}$ since $[\vP_{\rom1}, \vH_{\rom3}]=0.$
Note $\vH_{\rom1}$ has a discrete spectrum with a minimum Bohr-frequency gap of at least $\Delta_\nu \ge 1/T$.
We can choose sufficiently small $J_{\prop}, J_{\inp}$ such that  $\|\vV_{\rom1}\|:=\|\vH_{\prop}+\vH_{\inp}\| \le J_{\prop} T^2 + J_{\inp} n g_{\max} \ll \Delta_\nu(\vH_{\rom1}) = 1/T$, where $g_{max}:=\max_{1\le j\le n}g_j = \CO(T)$. And plugging in $\Lambda_0=1$ and other parameters into the error bound \eqref{eq:mono_err_bound}
\begin{equation}
    \epsilon_1 = \epsilon_{1a} + |S_0|  \CO\L(\frac{1}{\tau} +\frac{\norm{\vH_{\rom1}}^{3/4}}{\tau^{1/4}}+\frac{1}{\tau^{1/3}}  + \frac{1}{\sqrt{\Delta_\nu\tau}} + \frac{\e^{-\beta\Delta_{\nu}/4}}{\beta} + \Big(1+\frac{1+r_1}{\Delta_\nu}\Big) \norm{\vV_{\rom1}} \R).
\end{equation}
Noting that $|S_0| , \|\vH_{\rom1}\|, g_{\max}=\CO(T)$, we can make $\epsilon_1/r_1 \le \delta/6$ by choosing appropriate powers
\begin{equation} \label{eq:tbJ-cond-1}
\tau\ge \tOmega(T^{11}/\delta^4), \quad
\beta \ge \tOmega(T^2/\delta),  \quad J_{\prop} \le \tCO(\delta/T^5),  \quad \text{and} \quad J_{\inp} \le \tCO(\delta/nT^4).
\end{equation}

\paragraph{Case 2: Gradients for $\vP_{\rom1}(\vI-\vP_{\rom2})$ from $\vH_{\rom2}$.}

We next restrict our attention to the action of $\CL_{\rom2}^\dag[\vH_{\rom2}]$ inside the $\vP_{\rom1}$ subspace, which conveniently is also an eigensubspace of both $\vH_{\rom2}$ and $\vH_{\rom3}$ since $[\vP_{\rom1}, \vH_{\rom2}]=[\vP_{\rom1}, \vH_{\rom3}]=0$.
Explicit computation in \cref{sec:gradient-prop} shows that
\begin{equation}
\label{eq:L2-inside-P1}
    -\vP_{\rom1} \CL_{\rom2}^\dag[\vH_{\rom2}] \vP_{\rom1} \succeq r_2 \vP_{\rom1} (\vI-\vP_{\rom2}) - \epsilon_{2a} \vI,
\end{equation}
with the bounds from \cref{lem:L2-inside-P1} promising
\begin{equation}
     r_2 = \Omega(\frac{J_{\prop}}{T\ln\beta}) \quad \text{and}\quad \epsilon_{2a} = |S_0|\cdot \CO\Big(\frac{1}{\tau} +\frac{\norm{\vH_{\rom2}}^{3/4}}{\tau^{1/4}} + \frac{1}{\beta}+ \frac{1}{\sqrt{\tau J_{\prop}}}\Big).
\end{equation}

We then invoke Corollary~\ref{cor:monotonicity_subspace} with $\vQ=\vQ'=\vP_{\rom1}$, $\vH = \vH_{\rom2}$ and $\vH'=\vH_{\rom3}$ to show monotonicity of energy gradient on a subspace under perturbation
\begin{equation} 
    -\vP_{\rom1}\CL_{\rom3}^\dag[\vH_{\rom3}]\vP_{\rom1} \succeq r_2 \vP_{\rom1}(\vI - \vP_{\rom2}') - \epsilon_{2b} \vI,
\end{equation}
where $\vP_{\rom2}'$ is the perturbed eigensubspace of $\vH_{\rom3}$ that is identified with $\vP_{\rom2}$. To justify the application of \cref{cor:monotonicity_subspace}, we note
in $\vH_{\rom2}$, the eigensubspace $\vP_{\rom1}$ has an excitation gap of $\Delta_{\vQ}\ge 1/T - 2\|\vH_{\prop}\| = \Omega(1/T)$, where extra $2\|\vH_{\prop}\|$ term is due to shifts in eigenvalues of $\vH_{\rom1}$ bounded by Weyl's inequality (see \cref{fact:Weyl}).
The perturbation on $\vH_{\rom2}$ has strength $\|\vV_{\rom2}\|:=\|\vH_{\inp}\|\le J_{\inp} n g_{\max}$, and  $\Delta_\nu(\vH_{\rom2}|_{\vP_{\rom1}}) = J_{\prop}$.
Then noting $\Delta_\nu \ll \Delta_{\vQ}$, we keep the dominant terms in the error bound \eqref{eq:mono-subspace-err} and get
\begin{align}
    \epsilon_{2b} \le \epsilon_{2a} + |S_0| \cdot \CO\Big( \frac{\|\vV_{\rom2}\|\|\vH_{\rom2}\|}{\Delta_\nu \Delta_{\vQ}} + r_2\frac{\|\vV_{\rom2}\|}{\Delta_\nu}\Big) .
\end{align}
Furthermore, the eigensubspace $\vP_{\rom1}$ in $\vH_{\rom3}$ is separated by a spectral gap of $1/T-2\|\vV_{\rom1}\|=\Omega(1/T)$ from the other eigenstates, we may apply \cref{cor:PLP_and_L} to show
\begin{equation}  \label{eq:L3-outside-P2}
    -\CL_{\rom3}^\dag[\vH_{\rom3}] \succeq r_2 \vP_{\rom1}(\vI - \vP_{\rom2}') - \epsilon_2 \vI 
    \quad \text{where}\quad \epsilon_{2} = \epsilon_{2b} + |S_0| \CO(\frac{1}{\beta}+\frac{1}{\tau} + \sqrt{\frac{T}{\tau}}).
\end{equation}
Since $|S_0|,~ g_{\max},~ \|\vH_{\rom2}\|,~ \|\vH_{\rom3}\| = \CO(T)$, we can make $\epsilon_2/r_2\le \delta/6$ by choosing
\begin{equation} \label{eq:tbJ-cond-2}
\tau\ge \tOmega\Big(\frac{T^{11}}{J_{\prop}^2\delta^4}\Big), \quad
\beta\ge \tOmega\Big(\frac{T^2}{J_{\prop}\delta}\Big), \quad\text{and}\quad
J_{\inp} \le \tCO\Big(\frac{J_{\prop}^2 \delta}{nT^5}\Big).
\end{equation}

\paragraph{Case 3: Gradients for $\vP_{\rom2}(\vI - \vP_{\rom3})$ from $\vH_{\rom3}$.}

Now, we restrict our attention to $\vP_{\rom2}'$, the perturbed eigensubspace in $\vH_{\rom3}$ that correspond to $\vP_{\rom2}$.
We can show by explicit computation (deferred to \cref{sec:gradient-in}) that
\begin{equation}
\label{eq:L3-inside-P2}
    -\vP_{\rom2}' \CL_{\rom3}^\dag[\vH_{\rom3}] \vP_{\rom2}' \succeq r_3 \vP_{\rom2}' (\vI-\vP_{\rom3}) - \epsilon_{3a}\vI. 
\end{equation}
This computation shows that all valid history states $\ket{\eta_{\xv}}$ except for $\xv=\bm{0}$ have nonzero gradient with respect to $\CL_{\rom3}$.
The derivation uses a more fine-grained version of subspace gradient monotonicity (\cref{cor:off_diag_mono}) since the standard version yields insufficient bounds. Roughly, we need to capture the fact that off-diagonal perturbations induce only \textit{second-order} perturbation on the eigenvalues. The final calculated bounds in Eqs.~\eqref{eq:r3} and \eqref{eq:error-gradient-in} give us
\begin{equation}
r_3 = \Omega\Big(\frac{J_{\inp}}{T^2\ln\beta}\Big)
\quad \text{and} \quad
\epsilon_{3a} \le 
T \CO\bigg(
    \frac{T^{3/4}}{\tau^{1/4}}  +\frac{1}{\beta} + \frac{1}{\sqrt{J_{\inp}\tau}} + \e^{-\beta J_{\inp}} + n \frac{(n T J_{\inp})^2}{J_{\prop}}
\bigg).
\end{equation}

Using the fact that $\vP_{\rom2}'$ is separated by energy of at least 
$J_{\prop} - 2\|\vH_{\inp}\|$ from the other eigenstates in $\vH_{\rom3}$, we can apply \cref{cor:PLP_and_L} to get
\begin{equation} \label{eq:L3-outside-P3}
    -\CL_{\rom3}^\dag [\vH_{\rom3}] \succeq r_3 \vP_{\rom2}' (\vI-\vP_{\rom3}) - \epsilon_{3} \vI \quad \text{where}\quad \epsilon_{3} = \epsilon_{3a} + \labs{S_0}\CO\Big(\frac{1}{\beta}+\frac{1}{\tau} + \frac{1}{\sqrt{J_{\prop}\tau}}\Big).
\end{equation}
We may ensure $\epsilon_{3}/r_3 \le \delta/6$ by choosing
\begin{equation}
\label{eq:tbJ-cond-3}
    \tau \ge \tOmega\Big(\frac{T^{15}}{J_{\inp}^4\delta^4} \Big), \quad
    \beta \ge \tOmega\Big(\frac{T^3 }{J_{\inp} \delta}\Big), 
    \quad\text{and} \quad
    J_{\inp} \le \tCO\Big( \frac{J_{\prop} \delta}{n^3 T^5}\Big).
\end{equation}

\paragraph{Altogether.}
Based on the conditions in Eqs.~\eqref{eq:tbJ-cond-1} \eqref{eq:tbJ-cond-2} \eqref{eq:tbJ-cond-3} and the fact that $T= \Omega(n)$, a consistent choice of parameters that satisfies all the bounds and ensures $\epsilon_j/r_j \le \delta/6$ are
\begin{equation}
    \tau = \tilde\Theta\Big(\frac{T^{79}}{\delta^{16}}\Big),\quad
    \beta=\tilde\Theta\Big(\frac{T^{19}}{\delta^4}\Big), \quad
    J_{\prop} = \tilde\Theta\Big(\frac{\delta}{T^5}\Big), \quad \text{and} \quad
    J_{\inp} = \tilde\Theta\Big(\frac{\delta^3}{T^{16}}\Big).
\end{equation}
Then combining Eqs.~\eqref{eq:L3-outside-P1}, \eqref{eq:L3-outside-P2}, and \eqref{eq:L3-outside-P3} implies that
\begin{align}
    \vI - \vP_{\rom1} &\preceq \frac{\epsilon_1}{r_1} \vI - \frac{1}{r_1} \CL^\dag_{\rom3}[\vH_{\rom3}], \nonumber \\
    \vP_{\rom1} - \vP_{\rom2}' &\preceq \frac{\epsilon_2}{r_2} \vI - \frac{1}{r_2} \CL^\dag_{\rom3}[\vH_{\rom3}], \nonumber \\
    \vP_{\rom2}' - \vP_{\rom3} &\preceq \frac{\epsilon_3}{r_3} \vI - \frac{1}{r_3} \CL^\dag_{\rom3}[\vH_{\rom3}].
\end{align}
Note we have used $\vP_{\rom1} \vP_{\rom2}'=\vP_{\rom2}'$ and $\vP_{\rom2}'\vP_{\rom3}=\vP_{\rom3}$.
Recall that $\CL_{\rom3} = \sum_{a=1}^m \CL^{\beta, \tau, \vH_\circuit}_{a}$.
Adding all three inequalities together and normalizing suitably by the number of jumps $m=\labs{S_0}$, we have
\begin{equation}
\vI - \vP_{\rom3} \preceq -\left(\sum_{j=1}^3\frac{m}{r_j} \right) \left(\frac{1}{m}\sum_{a=1}^m \CL^{\dag \beta, \tau, \vH_\circuit}_{a}[\vH_{\rom3}] \right) + \left(\sum_{j=1}^3 \frac{\epsilon_j}{r_j} \right) \vI.
\end{equation}
The above provides the desired negative gradient condition on the full Hamiltonian $\vH_{\rom3} = \vH_{\circuit}$.
From Lemma~\ref{lem:nece-QLM}, any $\epsilon$-approximate local minimum $\vrho$ of $\vH_{\circuit}$ under thermal perturbation satisfies
\begin{equation}
1 - \Tr(\vP_{\rom3} \vrho) \leq \sum_{j=1}^3\frac{\epsilon_j + m\epsilon}{r_j} \le \frac{\delta}{2} + \epsilon m \sum_{j=1}^3\frac{1}{r_j}.
\end{equation}
By choosing $\epsilon \le \frac{\delta}{2m} \big(\sum_{j=1}^3 1/r_j\big)^{-1} = 1/\poly(n,T,\delta^{-1})$, we guarantee that $1 - \Tr(\vP_{\rom3} \vrho)\le \delta$.
This concludes our proof of \cref{thm:Hcircuit}.

\subsection{Explicit calculations for energy gradients}

In this section of the appendix, we provide the missing calculations supporting the claims that were asserted in Eqs.~\eqref{eq:L1-grad}, \eqref{eq:L2-inside-P1}, and \eqref{eq:L3-inside-P2} in the above proof of Theorem~\ref{thm:Hcircuit}.

\subsubsection{Gradient from $\vH_{\clock}$}
\label{sec:gradient-clock}
Note that $\vH_{\rom1}=\vH_{\clock}$ is a commuting Hamiltonian
\begin{align}
    \vH_{\rom1}=\vH_{\clock} = J_{\clock} \sum_{t=1}^{T-1} f_t \vh_{t,t+1}\quad &\text{where}\quad
    \quad \vh_{t,t+1} = \vI \otimes \ketbrat{01}_{t,t+1} 
\end{align}
where we set $f_t = (T-t)/T$ and $J_{\clock}=1$.
We start by computing the gradient from a single jump operator, using the simplification from \cref{lem:commuting-decomp}:

\begin{lemma}
\label{lem:H1-local-gradient}
Let the jump operator $\vA^t = \vI \otimes \vX_t$ for each $t\in[T]$.
For all $t = 2,\ldots, T$, we have
\begin{equation}
-\CD_{t}^{\dag\beta,\tau,\vH_{\rom1}}[\vH_{\rom1}] \succeq r_1\vh_{t-1,t} - \epsilon_0 \vI
\end{equation}
where $r_1=\Omega(\frac{1}{T\ln\beta})$ and $\epsilon_0 =\CO(\sqrt{T/\tau} + [(1+\beta)\ln\tau]/\tau + 1/\beta)$.
\end{lemma}

\begin{proof}
By \cref{lem:commuting-decomp}, we have $\CD_{t}^{\dag\beta,\tau,\vH_{\rom1}}[\vH_{\rom1}] = \CD_{t}^{\dag\beta,\tau,\vH_{\ni t}}[\vH_{\ni t}]$.
We then proceed in two cases.

\textbf{Case 1: $2 \le t \le T-1$.}
In this case, the relevant part of $\vH_{\rom1}$ that does not commute with $\vA^{t}$ is
\begin{equation}
    \vH_{\ni t} = f_{t-1}\ketbrat{01}_{t-1,t} + f_{t}\ketbrat{01}_{t,t+1}.
\end{equation}
Observe that $\vH_{\ni t}$ has three degenerate eigenspaces with corresponding energies as follows:
\begin{align}
\vP^0_t &= \sum_{\sv\in \{000,100,110,111\}}\ketbrat{\sv}_{t-1,t,t+1},\quad  &E_0 &= 0 , \nonumber\\
\vP^L_t &= \sum_{\sv\in \{010,011\}}\ketbrat{\sv}_{t-1,t,t+1},\quad &E_L &= f_{t-1},  \nonumber\\ 
\vP^R_t &= \sum_{\sv\in \{001,101\}} \ketbrat{\sv}_{t-1,t,t+1},\quad &E_R &= f_{t}.
\end{align}
The possible negative Bohr frequencies and the associated jumps are 
\begin{align}
    \vA^{t}_{\nu_1} &= \vP_t^0 \vX_{t} \vP_t^L = \ketbra{000}{010}_{t-1,t,t+1}, &\nu_1&=-f_{t-1}, \nonumber\\
    \vA^{t}_{\nu_2} &= \vP_t^0 \vX_{t} \vP_t^R = \ketbra{111}{101}_{t-1,t,t+1}, &\nu_2 &= -f_{t}, \nonumber \\
    \vA^{t}_{\nu_3} &= \vP_t^R \vX_{t} \vP_t^L = \ketbra{001}{011}_{t-1,t,t+1}, &\nu_3 &=f_{t}-f_{t-1}=-1/T. 
\end{align}

Furthermore, observe that the Bohr frequencies are exactly integer multiples of $1/T$, so we can lower bound the Bohr-frequency gap $\Delta_\nu(\vH_{\ni t}) \ge 1/T$.
Then by \cref{lem:recover-Davies} and \ref{lem:Davies-error-bound}, we can replace $\CD_{t}^{\dag\beta,\tau,\vH_{\ni t}}$ with $\CD_{t}^{\dag\beta,\infty,\vH_{\ni t}}$ up to an $\CO(\sqrt{T/\tau} + [(1+\beta)\ln \tau]/\tau)$ error.

Letting $\theta_j = \nu_j \gamma_\beta(\nu_j)$ for $j=1,2,3$ (recall $\gamma_\beta$ is given in Eq.~\eqref{eq:glauber-dyn} with $\Lambda_0=1$), we have
\begin{align}
\CD_{t}^{\dag\beta,\infty,\vH_{\ni t}}[\vH_{\ni t}] &= 
     \sum_{\nu\in B(\vH_{\ni t})} \nu \gamma_\beta(\nu) \vA^{a\dag}_\nu \vA^a_\nu  \nonumber \\
&= \big(\theta_1 \ketbrat{010} + \theta_2 \ketbrat{101} + \theta_3\ketbrat{011}\big)_{t-1,t,t+1} + \CO(1/\beta),
\end{align}
where the last error term is due to heating transitions (positive Bohr frequencies), which incur errors of at most $\|\omega\gamma_\beta(\omega) \indicator(\omega>0)\|_{\infty} = \CO(1/\beta)$.
Note that $\theta_1,\theta_2,\theta_3 < 0$, and furthermore we have 
\begin{align}
\min \{ |\theta_1|, |\theta_3|\} \ge  \min_{\omega\in[-1,-1/T]} |\omega| \gamma_\beta(\omega) =:r_1  = \Omega\big(\frac{1}{T\ln\beta}\big).
\end{align}
Hence,
\begin{equation}
-\CD_{t}^{\dag\beta,\infty,\vH_{\ni t}}[\vH_{\ni t}] \succeq r_1 \big(\ketbrat{010} + \ketbrat{011}\big)_{t-1,t,t+1} - \CO(1/\beta) \vI.
\end{equation}
Note the first term combines to make $\vh_{t-1,t}$.
We then return to the finite-$\tau$ Lindbladian up to the aforementioned error:
\begin{equation}
-\CD_{t}^{\dag\beta, \tau,\vH_{\ni t}}[\vH_{\ni t}] \succeq r_1\vh_{t-1,t} - \CO\Big(\sqrt{\frac{T}{\tau}} + \frac{(1+\beta)\ln\tau}{\tau} + \frac{1}{\beta}\Big) \vI
\end{equation}

\textbf{Case 2: $t=T$.}
The relevant part of $\vH_{\rom1}$ in this case is
\begin{equation}
    \vH_{\ni T} = f_{T-1}\ketbrat{01}_{T-1,T},
\end{equation}
which has two eigenspaces.
There is only one negative Bohr frequency with a corresponding jump operator filtered at $\nu$
\begin{equation}
    \vA^{T}_\nu = (\vI-\ketbrat{01}_{T-1,T})\vX_T \ketbrat{01}_{T-1,T} = \ketbra{00}{01}_{T-1,T}\quad \text{where}\quad \nu = -f_{T-1} = -\frac{1}{T}.
\end{equation}
Then,
\begin{equation}
\textstyle
\CD_{t}^{\dag\beta,\infty,\vH_{\ni T}}[\vH_{\ni T}] = -\frac{1}{T} \gamma_\beta(-\frac{1}{T}) \ketbrat{01}_{T-1,T}  + \CO(\frac{1}{\beta})\vI.
\end{equation}
Note that $\frac{1}{T} \gamma_\beta(-\frac{1}{T})\ge r_1$. Applying \cref{lem:recover-Davies} and \ref{lem:Davies-error-bound} to return to the finite $\tau$ expression, 
\begin{equation}
-\CD_{t}^{\dag\beta,\tau,\vH_{\ni T}}[\vH_{\ni T}] \succeq r_1 \vh_{T-1,T} -  \CO\Big(\sqrt{\frac{T}{\tau}} + \frac{(1+\beta)\ln\tau}{\tau} + \frac{1}{\beta}\Big) \vI,
\end{equation}
which is the advertised result.
\end{proof}

We are now ready to prove Eq.~\eqref{eq:L1-grad}, which we state as the following lemma:
\begin{lemma}
\label{lem:Hclock-grad}
Assume $1\le T\le \tau$. We have
\begin{equation}
    -\CL_{\rom1}^\dag[\vH_{\rom1}] \succeq r_1 (\vI - \vP_{\rom1})-\epsilon_{1a} \vI
\end{equation}
where
\begin{equation}
r_1=\Omega\Big(\frac{1}{T\ln\beta}\Big)
\quad \text{and} \quad
    \epsilon_{1a} = \CO\Big(\frac{T^{7/4}}{\tau^{1/4}} + \frac{T}{\beta}  + \frac{T(1+\beta)\ln\tau}{\tau}\Big).
\end{equation}
\end{lemma}
\begin{proof}
Note by linearity, we have
\begin{equation}
\CL_{\rom1}^\dag[\vH_{\rom1}] = \sum_{a\in S} \CL_a^{\dag \beta,\tau,\vH_{\rom1}} [\vH_{\rom1}].
\end{equation}
Let $S_{\rom1} = \{\vI\otimes \vX_t: 2\le t \le T\}$ be a subset of the jump operators.
Then
\begin{equation}
-\CL_{\rom1}^\dag[\vH_{\rom1}] \succeq -  \sum_{a\in S_{\rom1}} \CD_a^{\dag \beta,\tau,\vH_{\rom1}} [\vH_{\rom1}] - \CO\bigg(\labs{S_0} \Big(\frac{\|\vH_{\rom1}\|^{3/4}}{\tau^{1/4}} + \frac{1}{\tau} + \frac{1}{\beta}\Big)\bigg) \vI,
\end{equation}
where the error contribution from neglecting the Lamb-shift term and the other jump operators in $S_0\setminus S_{\rom 1}$ are bounded by 
\cref{prop:norm-lamb-shift-part} and \cref{lem:gradient_upperbound}.

Applying \cref{lem:H1-local-gradient} to the sum on the right hand side above, we get
\begin{equation}
-\sum_{a\in S_{\rom 1}} \CD_a^{\dag \beta,\tau,\vH_{\rom1}} [\vH_{\rom1}]  \succeq r_1 \sum_{t=1}^{T-1} \vh_{t,t+1} - T\epsilon_0\vI.
\end{equation}
It is not difficult to see that
\begin{equation}
    \sum_{t=1}^{T-1} \vh_{t,t+1} \succeq \vI-\vP_{\rom1},
\end{equation}
that is, the smallest excitation has energy $1$. Hence,
\begin{equation}
-\CL_{\rom1}^\dag[\vH_{\rom1}] \succeq  r_1 (\vI-\vP_{\rom1}) - \epsilon_{1a} \vI,
\end{equation}
where
\begin{align}
    \epsilon_{1a} &= \CO\bigg(\labs{S_0} \Big(\frac{\|\vH_{\rom1}\|^{3/4}}{\tau^{1/4}} + \frac{1}{\tau} + \frac{1}{\beta}\Big)\bigg) + \CO\Big(\frac{T^{3/2}}{\tau^{1/2}} + \frac{T(1+\beta)\ln\tau}{\tau} + \frac{T}{\beta}\Big) \nonumber \\
    &= \CO\Big(\frac{T^{7/4}}{\tau^{1/4}} + \frac{T}{\tau} + \frac{T}{\beta} + \frac{T^{3/2}}{\tau^{1/2}} + \frac{T(1+\beta)\ln\tau}{\tau} \Big).
\end{align}
The last equality uses that $\labs{S_0}, \|\vH_{\rom1}\| = \CO(T)$.
Since $1\le T \le \tau$, we have $T^{7/4}/\tau^{1/4} \ge T/\tau$ and $T^{7/4}/\tau^{1/4} \ge T^{3/2}/\tau^{1/2}$, so we drop the latter two terms for the final error estimate in the lemma statement.
\end{proof}

\subsubsection{Gradient from $\vH_{\prop}$}
\label{sec:gradient-prop}
In this subsection, we prove 
\begin{align}
-\vP_{\rom1} \CL_{\rom2}^\dag[\vH_{\rom2}] \vP_{\rom1} \succeq r_2 \vP_{\rom1} (\vI-\vP_{\rom2}) - \epsilon_{2a}\vI.
\end{align}

Denote $\ket{t}:=\ket{\eta_{\vect{x},t}}$ in what follows. Let $\vL_+$ be the raising operator whose only non-trivial action is
\begin{align}
    \vL_+\ket{t} &= \sqrt{(t+1)(T-t)}\ket{t+1} \quad\text{for each}\quad 0\le t\le T-1 \nonumber\\
    \text{and}\quad \vL_- &=: \vL_+^\dag.
\end{align}
Furthermore, let $\vL_x = \frac12 (\vL_+ + \vL_-)$, $\vL_y = \frac{1}{2\ri} (\vL_+ - \vL_-)$, $\vL_z=\sum_{t=0}^T (t-T/2)\ketbrat{t}$.
These operators form a set of angular momentum operators as $[\vL_a, \vL_b] = \ri\epsilon_{abc} \vL_c$ for $a,b,c\in\{x,y,z\}$. 
As noted earlier, we have
\begin{equation}
    \vP_{\rom1} \vH_{\prop} \vP_{\rom1} = J_{\prop} (\frac{T}{2}-\vL_x).
\end{equation}
The eigenstates are known to be
\begin{align}
    \ket{v_k} = \e^{\ri\pi \vL_y/2} \ket{k}\quad \text{with eigenvalues}\quad \lambda_k = k J_{\prop} \quad \text{for }\quad k = 0, 1,\ldots, T.
\end{align}
This integer spectrum means the minimum Bohr-frequency gap in the subspace $\vP_{\rom1}$ is $\Delta_\nu (\vH_{\prop}|_{\vP_{\rom1}})=J_{\prop}$.

Next, we give jump operators with nontrivial gradient on any excited state of $\vP_{\rom1}\vH_{\prop}\vP_{\rom1}$. These will be the 1-local jumps acting on the clock register
\begin{align}
    \vI\otimes \vZ_\ell \quad \text{for each}\quad 1\le\ell \le T
\end{align}
which nicely respects the block-diagonal structure of $\vH_{\prop}$ such that
\begin{equation}
\braket{\eta_{\xv,t}|\vI\otimes \vZ_\ell|\eta_{\yv,t'}}=  0 \quad \text{if}\quad \yv\ne \xv.\label{eq:xZy} 
\end{equation}
Thus, fixing $\vect{x}$, we merely need to consider effective jump operators
\begin{align}
    \vP_{\rom1} (\vI\otimes \vZ_\ell) \vP_{\rom1} \equiv \vsigma_\ell\quad \text{such that} \quad \vsigma_{\ell}\ket{t} = (-1)^{\indicator_{t\ge \ell}} \ket{t}.
    \label{eq:Zell}
\end{align}

\begin{lemma}[Good transition rates]
\label{lem:Hprop-grad}

For the operators $\vsigma_{\ell}$ in Eq.~\eqref{eq:Zell} and any $0\le k<T$, we have that
\begin{equation}
    \max_{\ell \in [T]}\left|\braket{v_{k} | \vsigma_\ell | v_{k+1}}\right| \ge \frac{1}{\sqrt{T}}.
\end{equation}
\end{lemma}
\begin{proof}
Fix any $ 0\le k<T$.
Observe that
\begin{equation} \label{eq:vk_Lz_vk1}
    \braket{v_k|\vL_z|v_{k+1}} = \braket{k|\e^{-\ri\pi \vL_y/2} \vL_z \e^{\ri\pi \vL_y/2}|k+1} = \braket{k|\vL_x|k+1} = \frac12 \sqrt{(k+1)(T-k)}.
\end{equation}
Then, 
\begin{align}
\max_{t }\left|\braket{v_{k} | \vsigma_t | v_{k+1}}\right| &\ge \frac{1}{T} \sum_{t=1}^T \left|\braket{v_k|\vsigma_t|v_{k+1}}\right| \tag{maximum is at least the mean}\\
&\ge \frac{2}{T} \left|\braket{v_{k}|\vL_z|v_{k+1}}\right|\tag{triangle inequality for $\vL_z = -\frac12 \sum_{t=1}^T \vsigma_t$}\\
&\ge \frac{1}{\sqrt{T}} \tag{Since $\sqrt{(k+1)(T-k)}\ge \sqrt{T}$ when $k<T$}
\end{align}
as advertised.
\end{proof}

Now that we understand the connectivity between the $t$ labels, we may restore the $2^n$ many labels $\xv$
\begin{equation}
    \ket{v_{k}} \rightarrow \ket{v_{k,\xv}} \quad \text{such that}\quad \braket{ v_{k',\yv}|v_{k,\xv} } = \delta_{kk'}\delta_{\xv\yv}.
\end{equation}
Fortunately, we do not need to address the explicit labels $\xv$ due to the orthogonality properties in Eq.~\eqref{eq:xZy}.
We may now calculate the gradient operator $\vP_{\rom1} \CL_{\rom2}^\dag[\vH_{\rom2}] \vP_{\rom1}$.

\begin{lemma} 
\label{lem:L2-inside-P1}
Consider the thermal Lindbladian $\CL_{\rom2} = \sum_{a\in S_0} \CL_{a}^{\beta,\tau,\vH_{\rom2}}$.
Then,
    \begin{align}
        -\vP_{\rom1} \CL_{\rom2}^\dag[\vH_{\rom2}] \vP_{\rom1} \succeq r_2 \vP_{\rom1} (\vI-\vP_{\rom2}) - \epsilon_2 \vI
    \end{align}
    for
    \begin{align}
        r_2 = \Omega(\frac{J_{\prop}}{T\ln\beta}) \quad \text{and}\quad \epsilon_2 = |S_0| \cdot \CO\L(\frac{1}{\tau} +\frac{\norm{\vH_{\rom2}}^{3/4}}{\tau^{1/4}} + \frac{1}{\beta}+ \frac{1}{\sqrt{\tau J_{\prop}}}\R).
    \end{align}
\end{lemma}

\begin{proof}

Observe that $\vP_{\rom1}$ projects onto a low-energy subspace of $\vH_{\rom2}$ with an excitation gap of at least $J_{\clock}/T- 2\|\vH_{\prop}\|$ from Weyl's inequality.
Furthermore, $\vH_{\rom2}$ restricted to $\vP_{\rom1}$ has eigenvalues that are integer multiples of $J_{\prop}$, so the Bohr-frequency gap in the subspace is $\Delta_\nu(\vH_{\rom2}|_{\vP_{\rom1}})= J_{\prop}$.
Assuming $J_{\prop}/2 < J_{\clock}/T - 2\|\vH_\prop\|$, we may apply \cref{lem:subspace-gradient-expr} with $\vH=\vH_{\rom2}$, $\vQ=\vP_{\rom1}$ to get
\begin{align}
\vP_{\rom1}\CL_{\rom2}^\dag[\vH_{\rom2}] \vP_{\rom1}
    &\stackrel{\verr}{\approx} \sum_{a\in S_0} \sum_{\nu\in B(\vH_{\rom2}|_{\vP_{\rom1}})} \vP_{\rom1}\vA^{a\dag}_{\nu}\vP_{\rom1}\vA^{a}_{\nu}\vP_{\rom1} \int_{-\infty}^{0} \gamma_{\beta}(\omega) \omega  \labs{\hat{f_{\mu}}(\omega-\nu)}^2 \rd \omega \nonumber\\
    &\preceq \sum_{a\in S_{\rom2}} \sum_{\nu\in B(\vH_{\rom2}|_{\vP_{\rom1}})} \vP_{\rom1}\vA^{a\dag}_{\nu}\vP_{\rom1}\vA^{a}_{\nu}\vP_{\rom1} \int_{-\infty}^{0} \gamma_{\beta}(\omega) \omega  \labs{\hat{f_{\mu}}(\omega-\nu)}^2 \rd \omega
    \label{eq:L2-in-P1-reduced}
\end{align}
where $\mu=\Delta_\nu/2$ and 
\begin{align}
  \verr = \labs{S_0} \CO\L(\frac{\norm{\vH}^{3/4}}{\tau^{1/4}} + \frac{1}{\tau} + \frac{1}{\beta} + \frac{1}{\sqrt{\Delta_\nu \tau}}\R).  
\end{align}
The second line uses the negativity of the half-integral to reduce to the following subset of jump operators from Eq.~\eqref{eq:jump-ops-for-HC}
\begin{equation}
    S_{\rom2} = \{\vI\otimes \vZ_\ell : \ell\in [T]\}.
\end{equation}
Let us now explicitly display the matrix elements of the above jump operators in the $\ket{v_{k,\xv}}$ basis:
\begin{align*}
    \vP_{\rom1}(\vI\otimes \vZ_\ell)\vP_{\rom1}&= 
    \sum_{k,k',\xv,\yv } \ket{v_{k',\yv}}\bra{v_{k',\yv}}\vI\otimes \vZ_\ell\ket{v_{k,\xv}} \bra{v_{k,\xv}} = \sum_{k,k',\xv} \ket{v_{k',\xv}}\bra{v_{k',\xv}}\vI\otimes \vZ_\ell\ket{v_{k,\xv}} \bra{v_{k,\xv}},
\end{align*}
where we applied Eq.~\eqref{eq:xZy} to drop the sum on $\yv$.
Thus we can rewrite the RHS of Eq.~\eqref{eq:L2-in-P1-reduced} as
\begin{align}
    (cont.)
    &= \sum_{\ell,k',k,\xv} \ketbrat{v_{k,\xv}} \vsigma_\ell \ketbrat{v_{k',\xv}}\vsigma_\ell
    \ketbrat{v_{k,\xv}}   \int_{-\infty}^{0} \gamma_{\beta}(\omega) \omega   \labs{\hat{f_{\mu}}(\omega-J_{\prop}(k'-k) )}^2 \rd \omega \nonumber \\
    & \preceq  -\Omega\left(\frac{J_{\prop}}{\ln\beta}\right) \sum_{k,\xv} \ket{v_{k,\xv}}\bra{v_{k,\xv}} \max_{\ell\in[T]}\labs{\bra{v_{k,\xv}}\vsigma_{\ell}\ket{v_{k-1,\xv}}}^2  \nonumber \\
    & \preceq -\Omega\left(\frac{J_{\prop}}{T\ln\beta}\right) \sum_{k \ge 1,\xv} \ket{v_{k,\xv}} \bra{v_{k,\xv}} \tag{applying Lemma~\ref{lem:Hprop-grad}}\\
    & \preceq  -\Omega\left(\frac{J_{\prop}}{T\ln\beta}\right) \Big(\vP_{\rom1} - \sum_{\xv} \ket{v_{0,\xv}} \bra{v_{0,\xv}}\Big) = -\Omega(\frac{J_{\prop}}{T\ln\beta}) \cdot \vP_{\rom1}(\vI-\vP_{\rom2})
\end{align}
The first line uses the orthogonality condition $\braket{v_{k',\yv}|v_{k,\xv} } = \delta_{kk'} \delta_{\xv\yv}$, and the fact that the identical $\nu$ labels on $\vA_\nu^{a\dag}$ and $\vA_\nu^a$ enforce that transitions from $k'$ need to be to the same $k$ on both sides.
The second line uses the negativity of the half-integral to focus on cooling transitions (which includes $k\rightarrow k-1$) and evaluates the integral (which concentrates near $\omega \approx J_{\prop}$). Lastly, we combine the above with the error bound on $\verr$ to conclude the proof.
\end{proof}

\subsubsection{Gradient from $\vH_{\inp}$}
\label{sec:gradient-in}
The goal of this subsection is to prove 
\begin{equation}
        -\vP_{\rom2}' \CL_{\rom3}^\dag[\vH_{\rom3}] \vP_{\rom2}' \succeq r_3 \vP_{\rom2}' (\vI-\vP_{\rom3})  - \epsilon_{3a} \vI.
\end{equation}
Here $\vP_{\rom2}'$ is the projector onto perturbed low-energy eigenstates of $\vH_{\rom3}=\vH_{\rom2} + \vH_{\inp}$, corresponding to 
\begin{equation}
\vP_{\rom2} = \sum_{\xv \in \{0,1\}^n} \ketbrat{\eta_\xv},
\end{equation}
where
\begin{equation}
    \ket{\eta_\xv} = \sum_{t=0}^T \sqrt{\xi_t} \ket{\eta_{\xv,t}} = \sum_{t=0}^T \sqrt{\xi_t} U_t\cdots U_1 \ket{\xv}\otimes \ket{C_t}\quad \text{and}\quad
    \qquad \xi_t = \frac{1}{2^{T}}\binom{T}{t}.
\end{equation}
Recall our \cref{def:circuit-H} where given a circuit with $L$ computational gates, we pad it in the beginning and the end with $t_0$ identity gates to make a total $T=2t_0 + L$ gates.
We can understand $\xi_t$ as the probability from a symmetric binomial distribution $\text{Binom}(T,\frac12)$, which has substantial weight near the center where the interesting computation takes place.

\begin{proposition}[Lower bound $\xi_t$ in the center]
    \label{prop:idling}
Suppose $T = 2 t_0 + L$ and $t_0 = c L^2$ are positive integers. Then we have
\begin{equation}
    \xi_{t}\ge \frac{\e^{-c/4}}{T+1} \qquad \text{for each}\quad t \in [t_0, T-t_0].
\end{equation}
\end{proposition}
\begin{proof}
As a property of the binomial distribution $\text{Binom}(T,\frac12)$, we have $\xi_t \ge \xi_{t_0}$ for all $t\in [t_0, T-t_0]$.
Observe that
\begin{equation}
    \binom{T}{t}^{-1} = (T+1)\int_0^1 x^{t} (1-x)^{T-t} \rd x \le (T+1) \Big(\frac{t}{T}\Big)^{t}\Big(1-\frac{t}{T}\Big)^{T-t} ,
\end{equation}
where the inequality comes from the fact that $\arg \max_{x\in[0,1]} x^t(1-x)^{T-t} = t/T$.
Then
\begin{align}
\xi_{t_0} &= \frac{1}{2^T}\binom{T}{t_0} \ge \frac{f(L)}{T+1} \nonumber \\
\text{where}\quad
f(L) &=
    \frac{1}{2^T}\Big(\frac{T}{t_0}\Big)^{t_0} \Big(\frac{T}{T-t_0}\Big)^{T-t_0} =
     \Big(1+\frac{1}{2cL}\Big)^{cL^2} \Big(1-\frac{1}{2(cL+1)}\Big)^{cL^2+L} .
\end{align}
The last equality is obtained after plugging in $T=2t_0+L$, $t_0=cL^2$ and simplifying.
We can use the first-derivative test to check that $f(L)$ is monotonically decreasing, and so $f(L) \ge \lim_{L\to\infty} f(L) = \e^{-c/4}$.
Hence, $\xi_t \ge \xi_{t_0} \ge f(L)/(T+1) \ge \e^{-c/4}/(T+1)$.
\end{proof}

Using the fact that $U_{t_j-1}\cdots U_{1}$ acts trivially on the $j$-th qubit (by definition of $t_j$), we see that $\vP_{\rom2} \vH_{\inp} \vP_{\rom2}$ is diagonal in the $\ket{\eta_\xv}$ basis:
\begin{align*}
    \braket{\eta_\xv|\vH_{\inp} |\eta_\yv} &= J_{\inp} \sum_{t,t'=0}^T \sqrt{\xi_t \xi_{t'}}\braket{\eta_{\xv,t}| \Big(\sum_{j=1}^n g_j \ketbrat{1}_j\otimes \ketbrat{C_{t_j-1}}\Big) | \eta_{\yv, t'}}  \\
    &= \delta_{\xv,\yv} \cdot J_{\inp}\sum_{j=1}^n x_j g_{j} \xi_{t_j-1}\label{eq:xHiny}.
\end{align*}
Since $g_j=1/\xi_{t_j-1}$ (see \cref{def:circuit-H}), then the above implies that that $\ket{\eta_\xv}$ are eigenstates of $\vP_{\rom2} \vH_{\inp} \vP_{\rom2}$ with eigenvalue $J_{\inp} \cdot \wt(\xv)$, where $\wt(\xv)$ is the Hamming weight of bit string $\xv$. While $\ket{\eta_\xv}$ are eigenstates of $\vP_{\rom2}\vH_{\inp}\vP_{\rom2}$, unfortunately, only $\ket{\eta_{\bm0}}$ is an eigenstate of $\vH_{\inp}$; this will require an additional perturbation step to handle this off-block-diagonal effect. 

Let $\vH_{\rom3} =\tilde{\vH}_{\rom3} + \vV_{\rom3} $ where
\begin{align}
    \tilde{\vH}_{\rom3} = \vH_{\rom2} + \vP_{\rom2} \vH_{\inp} \vP_{\rom2} + \vP_{\rom2}^\perp \vH_{\inp} \vP_{\rom2}^\perp \quad \text{and}\quad 
    \vV_{\rom3} = \vP_{\rom2} \vH_{\inp} \vP_{\rom2}^\perp +  \vP_{\rom2}^\perp \vH_{\inp}\vP_{\rom2}.
\end{align}
We will also denote $\tilde{\CL}_{\rom3}$ as the thermal Lindbladian with respect to $\tilde{\vH}_{\rom3}$.

We start by studying the gradient on $\tilde{\vH}_{\rom3}$ and 
showing all its excited states have negative energy gradients. 
It suffices to focus on the states in $\vP_{\rom2}$, which is a low-energy subspace of eigenstates of $\tilde\vH_{\rom3}$ with excitation gap of at least $J_{\prop}-2\|\vH_{\inp}\|$ from Weyl's inequality.
The effective Hamiltonian in this subspace has a simple form of decoupled qubits, which we write as
\begin{equation}
    \tilde{\vH}_\eff := \tilde\vH_{\rom3}|_{\vP_{\rom2}} = \vP_{\rom2} \vH_{\inp} \vP_{\rom2} = J_{\inp} \sum_{\xv} \wt(\xv) \ketbrat{\eta_{\xv}} \equiv J_{\inp}  \sum_{j=1}^n (\vI-\vZ_j^\eff)/2,
\end{equation}
where $\vZ_j^\eff$ is the Pauli Z operator of a virtual qubit defined as $\vZ_j^\eff\ket{\eta_\xv} = (-1)^{x_j}\ket{\eta_\xv}$.

Observe that $\tilde\vH_\eff$ has eigenvalues that are integer multiples of $J_{\inp}$, and thus its Bohr-frequency gap is $\Delta_\nu(\tilde\vH_\eff)=J_{\inp}$.
Hence, assuming $J_{\inp}/2 < J_{\prop} - 2\|\vH_{\inp}\|$, we can apply \cref{lem:subspace-gradient-expr} to see that the gradient operator sandwiched by $\vP_{\rom2}$ can be understood by fully restricting to the subspace:
\begin{equation}
\vP_{\rom2} \tilde{\CL}_{\rom3}^\dag[\tilde{\vH}_{\rom3}] \vP_{\rom2} 
\stackrel{\tilde{E}_{a}}{\approx} \sum_{a\in S_0} \sum_{\nu \in B(\tilde{\vH}_\eff)}  \vP_{\rom2} \vA_\nu^{a\dag} \vP_{\rom2} \vA_{\nu}^{a} \vP_{\rom2} \int_{-\infty}^0 \gamma_\beta(\omega) \omega |\hat{f}_\mu(\omega-\nu)|^2 \rd\omega 
\label{eq:L3tilde-in-P2-a}
\end{equation}
where $\mu=J_{\inp}/2$ and 
\begin{align}
    \tilde{\verr}_{a} = |S_0| \CO\L(\frac{\|\tilde\vH_{\rom3}\|^{3/4}}{\tau^{1/4}} + \frac{1}{\tau} + \frac{1}{\beta} + \frac{1}{\sqrt{J_{\inp} \tau}}\R).
\end{align}
To show that this has good gradients for all states in $\vP_{\rom2}(\vI-\vP_{\rom3})$, it is sufficient to consider the following subset of the jump operators from Eq.~\eqref{eq:jump-ops-for-HC}: 
\begin{equation}
 S_{\rom3} = \L\{\vX_j\otimes \ketbrat{0}_{t_j}\R\}_{j=1}^n.
\end{equation}
These jump operators from $S_{\rom3}$ effectively flip the individual virtual qubits by
\begin{equation}
\textstyle
\braket{\eta_{\yv}| (\vX_j \otimes \ketbrat{0}_{t_j}) |\eta_\xv} =  \braket{\yv| \vX_j | \xv}  \sum_{t=0}^{t_j-1} \xi_t  =: \braket{\yv| \vX_j | \xv}  \sqrt{\alpha_j}
\end{equation}
where we have denoted $\alpha_j = (\sum_{t<t_j} \xi_t)^2$.
Let 
\begin{align}
\vX^\eff_j = \vP_{\rom2}(\vX_j\otimes \ketbrat{0}_{t_j})\vP_{\rom2}/\sqrt{\alpha_j}\quad \text{such that}\quad \norm{\vX^\eff_j} = 1 \quad \text{for each}\quad j=1,\ldots, n.     
\end{align}
Note $\vX^\eff_j$ is effectively the Pauli X operator for the $j$-th virtual qubit.
Furthermore, note since $t_j\in [t_0,T-t_0]$ by our circuit construction in \cref{def:circuit-H}, we have $\alpha_j \ge \xi_{t_0}^2 \ge \Omega(1/T^2)$ by \cref{prop:idling}.

Next, we replace $\tilde\vH_{\rom3}$ by $\tilde\vH_\eff$ so we only need to talk about the virtual qubits.
We first restore the RHS of Eq.~\eqref{eq:L3tilde-in-P2-a} to the thermal Lindbladian form by undoing the approximation in \cref{lem:subspace-gradient-expr}, which incurs another error bounded by $\tilde{\verr}_a$:
\begin{equation}
    \vP_{\rom2} \tilde{\CL}_{\rom3}^\dag[\tilde{\vH}_{\rom3}]\vP_{\rom2} 
    \stackrel{2\tilde{E}_a}{\approx} \CL^{\dag\beta,\tau,\vH_{\eff}}[\vH_\eff].
\end{equation}
We then focus on the subset $S_{\rom3}$ of jump operators and write
\begin{equation}
\label{eq:Ltilde-global}
\textstyle
\CL^{\dag\beta,\tau,\vH_{\eff}}[\vH_\eff]
    \stackrel{\tilde{\verr}_{b}}{\approx} \sum_{j=1}^n \alpha_j \tilde{\CL}_j^\dag[\tilde\vH_\eff],
\end{equation}
where we denoted $\tilde{\CL_j} := \CL_j^{\beta,\tau,\tilde\vH_\eff}$ to be the thermal Lindbladian associated with effective jump operator $\vX^\eff_j$ and pulled out the normalization factor $\alpha_j$.
The error 
\begin{align}
\tilde{\verr}_{b} = |S_0|\CO\L(\frac{\|\tilde\vH_\eff\|^{3/4}}{\tau^{1/4}} + \frac{1}{\tau} + \frac{1}{\beta}\R)    
\end{align}
comes from neglecting the other jump operators $S_0\setminus S_{\rom3}$, and is bounded using \cref{lem:gradient_upperbound}.

Since the effective operators on different virtual qubits commute, we can treat them independently.
More formally, by \cref{lem:commuting-decomp} we have $\tilde\CL_j^\dag[\tilde\vH_\eff] = \CL_j^{\dag\beta,\tau,\vh_j}[\vh_j]$, where $\vh_j=J_{\inp}(\vI-\vZ_\eff^j)/2$.
To bound the global gradient in Eq.~\eqref{eq:Ltilde-global}, we first consider cooling a single qubit.

\begin{lemma}[Cooling a qubit] On a qubit, consider the thermal Lindbladian $\CL = \CL_a^{\beta,\tau,\vH}$ with the Hamiltonian $\vH = J_{\inp} (\vI-\vZ)/2$ and one jump operator $\vA = \vX$.
Then,
\begin{equation}
    -\CL^{\dag}[\vH] \succeq r_{\inp} (\vI-\vZ) - \epsilon_{\inp} \vI
\end{equation}
where
\begin{equation}
r_{\inp} = \Omega\Big(\frac{J_{\inp}}{\ln\beta}\Big) 
\quad \text{and}\quad 
\epsilon_{\inp} = \CO\Big(\frac{J_{\inp}^{3/4}}{\tau^{1/4}} + \frac{1}{\tau} + \frac{1}{\sqrt{\tau J_{\inp}}} + \e^{-\beta J_{\inp}}\Big).
\end{equation}
\end{lemma}

\begin{proof}
Again, we invoke a series of approximations
\begin{align}
   \CL^\dag[\vH] 
   &\approx \CD^\dag[\vH]\tag{Proposition~\ref{prop:norm-lamb-shift-part}}\\
   &\approx \int_{-\infty}^{\infty} \gamma_{\beta}(\omega)\omega  \vA(\omega)^\dag\vA(\omega) \rd \omega \tag{Lemma~\ref{lemma:expr_energy_gradient}}\\
   &\approx\int_{-\infty}^{\infty} \gamma_{\beta}(\omega) \omega  \vS(\omega)^\dag\vS(\omega) \rd \omega
   \tag{secular approximation ($\mu = J_{\inp}/2$): Corollary~\ref{cor:AA-SS}}\\
   &=\sum_{\nu\in B(\vH)} \vA^{\dag}_{\nu}\vA_{\nu} \int_{-\infty}^{\infty} \gamma_{\beta}(\omega) \omega  \labs{\hat{f_{\mu}}(\omega-\nu)}^2 \rd \omega \tag{different blocks $\nu\ne \nu'$ decohere}\\
   & \preceq - \Omega(J_{\inp}/\ln\beta) \ket{1}\bra{1} + \CO(\e^{-\beta J_{\inp}}) \ket{0}\bra{0}. 
\end{align}
The last line uses the transition matrix elements for the two Bohr frequencies $\nu = \pm J_{\inp}$:
\begin{align}
    \vA_{+1} = \ket{1}\bra{0}, 
    \qquad\text{and}\qquad
    \vA_{-1} = \ket{0}\bra{1}.
\end{align}
Combine the error bound to conclude the proof.
\end{proof}

Summing up over the contributions from the individual qubits, the global gradient satisfies
\begin{align}
-\vP_{\rom2} \tilde{\CL}_{\rom3}^\dag[\tilde{\vH}_{\rom3}] \vP_{\rom2}  &\succeq \sum_{j=1}^n \alpha_j \vP_{\rom2} \L[r_{\inp} (\vI-\vZ_j^\eff)  - \epsilon_{\inp}\vI\R] \vP_{\rom2}  - (2\tilde{\verr}_a + \tilde{\verr}_b)\vI \nonumber \\
    &\succeq r_3 \vP_{\rom2}(\vI-\vP_{\rom3}) - \tilde{\epsilon}_3\vI
\end{align}
where we used the fact that $\sum_{j=1}^n (\vI-\vZ_j^\eff) \succeq \vP_{\rom3}$, and denoted $r_3 = r_{\inp}\min_j \alpha_j$, and $\tilde{\epsilon}_3 = 2\tilde{\verr}_a + \tilde{\verr}_b + \epsilon_{\inp}\sum_{j=1}^n \alpha_j$.
Since $\alpha_j = \Omega(1/T^2)$ and $\alpha_j\le 1$, we have that
\begin{equation}
\label{eq:r3}
r_3 = \Omega\Big (\frac{J_{\inp}}{T^2\ln\beta}\Big) 
\quad \text{and} \quad
\tilde{\epsilon}_3 \le 2\tilde{\verr}_a + \tilde{\verr}_b + n\epsilon_{\inp} = |S_0| \CO\Big(
    \frac{\|\tilde\vH_{\rom3}\|^{3/4}}{\tau^{1/4}} + \frac{1}{\tau} + \frac{1}{\beta} + \frac{1}{\sqrt{J_{\inp}\tau}} + \e^{-\beta J_{\inp}}
    \Big).
\end{equation}

Lastly, to obtain the gradient for the final Hamiltonian $\vH_{\rom3}$,  we need to add the off-block-diagonal perturbation and show that the gradient persists on the subspace.
Directly applying subspace monotonicity (\cref{cor:monotonicity_subspace}) yields loose bounds; we will need to invoke a finer-grained subspace monotonicity (\cref{cor:off_diag_mono}) that exploits the fact that $\vV$ is off-block-diagonal and contribute to eigenvalue change at \textit{second order} $\CO(\norm{\vV}^2)$. We apply \cref{cor:off_diag_mono} with $\vH=\tilde\vH_{\rom3}$ and $\vH' = \vH_{\rom3}$, and parameters
\begin{align*}
    \vQ &= \vP_{\rom2} \tag{low-energy eigensubspace of $\vH=\tilde{\vH}_{\rom3}$}\\
    \Delta_{\vQ} &= J_{\prop} - 2\norm{\vH_{\inp}} \tag{excitation gap}\\
    \vV &= \vP_{\rom2} \vH_{\inp} (\vI-\vP_{\rom2}) +  (\vI-\vP_{\rom2})\vH_{\inp}\vP_{\rom2}\tag{off-block-diagonal perturbation}\\
    \Delta_\nu &= J_{\inp} \tag{subspace Bohr-frequency gap}
\end{align*}
Therefore, by Corollary~\ref{cor:off_diag_mono},
\begin{align}
    -\vP_{\rom2}' \CL_{\rom3}^\dag[\vH_{\rom3}] \vP_{\rom2}' \succeq r_3 \vP_{\rom2}' (\vI-\vP_{\rom3})  - \epsilon_{3a} \vI
\end{align}
where
\begin{align}
    \epsilon_{3a} &\le \tilde{\epsilon_3} + \labs{S_0}\cdot 
    \CO\bigg(
        \frac{1}{\tau} +\frac{\|\tilde\vH_{\rom3}\|^{3/4}}{\tau^{1/4}}+\frac{\Lambda_0^{2/3}}{\tau^{1/3}}+\frac{\Lambda_0}{\sqrt{\Delta_{\nu}\tau}} +\frac{\Lambda_0}{\sqrt{\Delta_{\vQ}\tau}} + \frac{\e^{-\beta\Delta_{\vQ}/4}}{\beta}  \nonumber \\
        &\qquad\qquad\qquad +\frac{\norm{\vV}^2}{\Delta_{\vQ}} +  \norm{\vH_{\vQ}} \cdot \Big( \frac{\norm{\vH_{\vQ}}\norm{\vV}}{\Delta_{\vQ}\Delta_\nu}+\frac{\norm{\vV}^2}{\Delta_{\vQ}\Delta_\nu}\Big) + r_3 \Big(\frac{\norm{\vV}}{\Delta_{\vQ}} + \frac{\norm{\vV}^2}{\Delta_\nu\Delta_{\vQ}}\Big)
    \bigg).
\end{align}

Noting that $|S_0|, \|\tilde\vH_{\rom3}\| = \CO(T)$, $\|\vV\| = \CO(n T J_{\inp})$, $\|\vH_{\vQ}\| = n J_{\inp}$, and
$J_{\inp} \ll J_{\prop}$, we simplify the error bound above by keeping the dominant term to get
\begin{equation}
\label{eq:error-gradient-in}
\epsilon_{3a} \le 
T \CO\bigg(
    \frac{T^{3/4}}{\tau^{1/4}}  +\frac{1}{\beta} + \frac{1}{\sqrt{J_{\inp}\tau}} + \e^{-\beta J_{\inp}} + n \frac{(n T J_{\inp})^2}{J_{\prop}}  
\bigg)
\end{equation}
as advertised earlier in Eq.~\eqref{eq:L3-inside-P2}.

\section{Operator Fourier Transform}
\label{sec:OFT}
Recall that the exact form of thermal Lindbladians in Appendix~\ref{sec:thermo-lindblad-detail} involves the operator Fourier transform (OFT)~\cite{Chen2023quantumthermal} for a set of jump operators $\vA^a$.
In this appendix, we provide the key properties of OFT which are used in the proofs of many statements in Appendices~\ref{sec:characterize-neg-grad-condition} and \ref{sec:monotone_gradient}. 
For any operator $\vA$ (the jump operator), and Hermitian operator $\vH$ (the Hamiltonian), and a weight function $f$, the \textit{operator Fourier Transform} (OFT) is an integral over time-evolution of the operator $\vA$,
\begin{align}
    \hat{\vA}_{\hat{f}}(\omega) &:=  \frac{1}{\sqrt{2\pi}}\int_{-\infty}^{\infty} \e^{\ri \vH t} \vA \e^{-\ri \vH t} \e^{-\ri \omega t} f(t)\rd t.
\end{align}
Often, we will also write $\hat{\vA}(\omega)$ when we choose $f$ to be the default \emph{normalized window function} 
\begin{align}
 f_{\tau}(t) = \frac{1}{\sqrt{\tau}} \cdot \begin{cases}
 1 &\text{if} \quad t\in [-\tau/2,\tau/2]\\
 0 &\text{else}.
\end{cases} \label{eq:bump}
\end{align}
It is usually helpful to consider the energy eigenspaces $\vH = \sum_i E_i \ket{\psi_i}\bra{\psi_i} = \sum_{E\in \text{Spec}(\vH)} E\vP_{E}$ and and write $\vA$ as the following decomposition
\begin{equation}
    \vA =  \sum_{E_2,E_1 \in \text{Spec}(\vH)} \vP_{E_2} \vA \vP_{E_1} = \sum_{\nu\in B(\vH)} \vA_\nu 
    \quad \text{where}\quad 
    \vA_\nu :=  \sum_{E_2-E_1=\nu} \vP_{E_2} \vA \vP_{E_1}.
\end{equation}
Formally, these energy differences $\nu \in B(\vH):= \{ E_i - E_j \, | \, E_i, E_j \in \mathrm{Spec}(\vH) \}$ are called the \textit{Bohr frequencies} (Figure~\ref{fig:energybohr}), and $\vA_\nu$ collects the matrix elements $\bra{\psi_i}\vA\ket{\psi_j}$ that changes the energy by $\nu =  E_i - E_j$. 
The Bohr frequencies are natural for the Heisenberg-picture evolution of $\vA$ since
\begin{equation}
\label{eqn:Aoperator}
\e^{\ri \vH t} \vA \e^{-\ri \vH t} = \sum_{\nu \in B(\vH)} \e^{\ri \nu t}  \vA_\nu.
\end{equation}
Then, executing the Fourier integral yields the frequency domain representation
\begin{align}
     \hat{\vA}_f(\omega)&=\sum_{\nu \in B} \vA_{\nu} \hat{f}(\omega-\nu) \quad \text{where}\quad \hat{f}(\omega) := \frac{1}{\sqrt{2\pi}}\int_{-\infty}^{\infty}f(t)\e^{-\ri\omega t} \rd t,
\end{align}
which contains a collection of Bohr frequencies $\nu$ near $\omega$. Conceptually, we can think of the operator Fourier transform $\hat{\vA}_f(\omega)$ as the \textit{smooth} probe of $\vA_{\nu}$ with exact Bohr frequency $\nu$, which generally requires resolving arbitrarily close eigenvalues.

\subsection{Useful properties}
We instantiate some useful properties of the Operator Fourier Transform.

\begin{proposition}[{Operator Parseval's identity~\cite{Chen2023quantumthermal}}]\label{prop:iso_operator_FT}
	Consider a set of operators $\{\vA^a\}_{a}$ and its operator Fourier transform with weight $f\in \mathcal{L}_2(\mathbb{R})$ and Hamiltonian $\vH$. Then, we have a certain symmetry
	$(\vA^a_{f}(\omega))^{\dag} = \vA^{a\dag}_{f^*}(-\omega)$ and 
	certain Parseval-type identity
	\begin{align}
		\sum_{a\in S} \int_{-\infty}^{\infty}  \hat{\vA}^a_{f}(\omega)^{\dag} \hat{\vA}^a_{f}(\omega) \rd \omega
		&= \sum_{a\in S} \int_{-\infty}^{\infty} \e^{\ri \vH t}\vA^{a\dag}\vA^a\e^{-\ri \vH t}\labs{f(t)}^2\rd t
		\preceq  \norm{\sum_{a\in S}\vA^{a\dag}\vA^a} \norm{f}_2^2\cdot \vI \label{eq:contFourierIdentity}\\
		\sum_{a\in S} \int_{-\infty}^{\infty}  \hat{\vA}^a_{f}(\omega)\hat{\vA}^a_{f}(\omega)^{\dag} \rd \omega
		&= \sum_{a\in S} \int_{-\infty}^{\infty} \e^{\ri \vH t}\vA^a\vA^{a\dag}\e^{-\ri \vH t}\labs{f(t)}^2\rd t
		\preceq  \norm{\sum_{a\in S}\vA^{a}\vA^{a\dag}}\norm{f}_2^2 \cdot \vI.	\label{eq:contFourierIdentity2}
	\end{align}
\end{proposition}
Intuitively, the above tells us that the Fourier Transforms $\hat{\vA}^a_f(\omega)$ with different frequencies $\omega$ are ``orthogonal'' to each other and that the average of squares of strengths is bounded (reminiscent of a probability). Without using this norm sum constraint from Fourier transforms, one easily gets loose bounds. An alternative view of the above (as a natural purification) will prove useful for manipulating norms of expression involving $\vA^a_{f}(\omega)$.  
\begin{corollary}[Purification of Operator Fourier Transform] \label{cor:iso_OFT}
In the prevailing notation, the abstract operator Fourier Transform has a norm bound,
\begin{align}
    \norm{\sum_{a\in S} \int_{-\infty}^{\infty} \hat{\vA}^a_{f}(\omega) \otimes \ket{a} \otimes\ket{\omega} \rd \omega} \le \norm{f}_2 \sqrt{\norm{\sum_{a\in S}\vA^{a\dag}\vA^a}},
\end{align}
where the continuous basis vectors satisfy the normalization
\begin{align}
\braket{\omega'|\omega} = \delta(\omega'-\omega).
\end{align}
\end{corollary}

\begin{proof} We multiply the conjugate
    \begin{align}
    &\L(\sum_{a'\in S} \int_{-\infty}^{\infty} \hat{\vA}^{a'}_{f}(\omega')^{\dag} \otimes \bra{a'} \otimes\bra{\omega'} \rd \omega'\R) \cdot \L(\sum_{a\in S} \int_{-\infty}^{\infty} \hat{\vA}^a_{f}(\omega) \otimes \ket{a} \otimes\ket{\omega} \rd \omega \R) \nonumber\\
    &= \sum_{a\in S} \int_{-\infty}^{\infty}  \hat{\vA}^a_{f}(\omega)^{\dag} \hat{\vA}^a_{f}(\omega) \rd \omega.
\end{align}
Take the operator norm and use Proposition~\ref{prop:iso_operator_FT} to conclude the proof.
\end{proof}
The above ``purification'' trick applies to other quantities with a liberal choice of summation indices, whether they are $a, \omega,$ or $\nu$.
\begin{lemma}[Norm inequalities from operator purification]
\label{lem:Op_purification_bounds}
For any operator $\vO$ and any set of operators $\vA_i, \vA'_j$ acting on the same Hilbert space, we have that
\begin{align}
\textstyle
    \left\|\sum_{i,j} \vA_i^\dag \vO \vA'_j G_{ij}\right\| &\le \|\vG\| \norm{\vO} \sqrt{\left\| \sum_i \vA_i^\dag \vA_i \right\|\left\| \sum_j \vA_j^{'\dag} \vA'_j \right\|} \\
    \left\|\sum_{i,j} \vA_i^\dag \vA'_j G_{ij}\right\| &\le \|\vG\| \sqrt{\left\| \sum_i \vA_i^\dag \vA_i \right\|\left\| \sum_j \vA_j^{'\dag} \vA'_j \right\|}
\end{align}
\end{lemma}
\begin{proof}
    By homogeneity, it suffices to set normalization to be $\norm{\sum_{j}\vA_j^{'\dag}\vA'_j} =\norm{\sum_{i}\vA_i^{\dag}\vA_i} = \norm{\vO}=\norm{\vG}=1$. Introduce purifications
    \begin{gather}
    \vG':= \vI\otimes \sum_{i,j} G_{ij}\ket{i}\bra{j}, \qquad
        \vO':= \vO\otimes \vI,\\
        \vV':= \sum_j \vA'_{j} \otimes \ket{j}, \qquad
        \vV:= \sum_i \vA_{i} \otimes \ket{i}, 
    \end{gather}
    which are all bounded by $\norm{\vG'}, \norm{\vO'}, \norm{\vV'},\norm{\vV} \le 1$. Then, 
    \begin{align}
        \bigg\|\sum_{i,j} \vA_i^\dag \vO \vA'_j G_{ij}\bigg\| & = \norm{ \vV' \vO' \vG' \vV} \le 1\\
         \bigg\|\sum_{i,j} \vA_i^\dag \vA'_j G_{ij}\bigg\| & = \norm{ \vV' \vG' \vV} \le 1.
    \end{align}
    Rescale to obtain the advertised result.
\end{proof}

\subsection{Secular approximation}\label{sec:secular}
Due to the energy-time uncertainty principle, the energies $\omega$ that are accessed by finite-time quantum algorithms always inherit an uncertainty.
Indeed, when we choose our weight function $f_{\tau}(t)$ in Eq.~\eqref{eq:bump} for our operator Fourier transform in the frequency domain, we have
\begin{align}
    \hat{\vA}_f(\omega)&=\sum_{\nu \in B} \vA_{\nu} \hat{f}(\omega-\nu) \quad \text{where}\quad \hat{f}(\omega) = \frac{\e^{\ri \omega \tau/2} - \e^{-\ri \omega \tau/2}}{\ri \omega \sqrt{2\pi\tau}} \quad \text{when}\quad f(t)=f_{\tau}(t).
\end{align}
Note with this choice, $\hat{f}(\omega)$ has a heavy tail $\sim1/\omega$, which is reminiscent of unamplified phase estimation. Therefore, even when restricting to jumps with $\omega<0$, there is a decent chance that $\hat\vA_f(\omega)$ mistakenly activates a heating transition ($\nu >0$) instead of a cooling transition $(\nu <0)$, unintentionally heating up the system instead of cooling it. 

To control the resulting error, in this section, we introduce the \textit{secular approximation}~\cite{Chen2023quantumthermal} of the Fourier transformed operators $\vA_{f}(\omega)$. The secular approximation applies truncation to the Fourier-transformed operators in the frequency domain by truncating Bohr frequencies $E \in B$ that deviate substantially from the frequency label $\omega$. Truncation at energy difference $\mu$ can be achieved by setting a step function and defining the following secular-approximated operators as follows
\begin{align}\label{eq:genSecDef}
	\hat{\vS}_{f,\mu}(\omega) := \sum_{\nu \in B} \vA_{\nu} \hat{f}(\omega-\nu)\cdot \hat{s}_{\mu}(\omega-\nu)\quad \text{where}\quad \hat{s}_{\mu}(\omega):=\indicator(\labs{\omega} < \mu).
\end{align}
We often drop subscript $f,\mu$ for simplicity. The truncation error is an operator whose norm can be bounded by the following (as a variant of Corollary~\ref{cor:iso_OFT}). 
\begin{corollary}[Secular approximation]\label{cor:S-A} In the prevailing notation,
    \begin{align}
        \norm{\sum_{a\in S} \int_{-\infty}^{\infty} (\hat{\vS}^a_{f}(\omega)-\hat{\vA}^a_{f}(\omega)) \otimes \ket{a} \otimes\ket{\omega} \rd \omega} \le \norm{\hat{f}\cdot (1-\hat{s}_{\mu})}_2 \sqrt{\norm{\sum_{a\in S}\vA^{a\dag}\vA^a}}. 
    \end{align}
\end{corollary}
The error is controlled by the 2-norm for the truncated tail
\begin{equation}
    \norm{\hat{f}\cdot (1-\hat{s}_{\mu})}_2^2 =  \int_{\labs{\omega} \ge \mu } |\hat{f}(\omega)|^2 \rd \omega
    \label{eq:tailbound} 
\end{equation}
noting that conveniently, the Fourier transform preserves the 2-norm of functions.
For our bump function~\eqref{eq:bump} in particular, we can integrate the tail-bound 
\begin{equation}
    \norm{\hat{f}_{\tau}\cdot (1-\hat{s}_{\mu})}_2^2 \le \frac{4}{\pi\mu\tau} \label{eq:secular_tailbound}.
\end{equation}
This conveniently leads to bounds on other quantities involving operator Fourier Transforms $\hat{\vA}^a(\omega)$. 
\begin{corollary}[Error from secular approximation]\label{cor:AA-SS} For any real function $\theta: \BR \rightarrow \BR$, 
    \begin{align}
        \lnorm{\sum_{a\in S} \int_{-\infty}^{\infty} \theta(\omega') \hat{\vA}_f^{a}(\omega')^{\dag}\hat{\vA}_f^a(\omega') \rd \omega'-\sum_{a\in S} \int_{-\infty}^{\infty} \theta(\omega') \hat{\vS}_{f,\mu}^{a}(\omega')^{\dag}\hat{\vS}_{f,\mu}^a(\omega') \rd \omega'} \nonumber\\
        \le 2\norm{\theta}_{\infty} \norm{\hat{f}\cdot (1-\hat{s}_{\mu})}_2 \norm{f}_2 \sqrt{\norm{\sum_{a\in S}\vA^{a\dag}\vA^a}}.
    \end{align}
\end{corollary}
\begin{proof}
It suffices to set normalization $\norm{\sum_{a\in S}\vA^{a\dag}\vA^a} = 1$. Introduce purifications
    \begin{align}
        \vF&:= \vI\otimes \vI\otimes \int_{-\infty}^{\infty} \theta(\omega) \ket{\omega}\bra{\omega} \rd \omega, \\
        \vV&:=\L(\sum_{a\in S}\int_{-\infty}^{\infty} \hat{\vA}^a_{f}(\omega) \otimes \ket{a} \otimes\ket{\omega} \rd \omega \R), \\
        \vV'&:=\L(\sum_{a\in S} \int_{-\infty}^{\infty} \hat{\vS}^a_{f}(\omega) \otimes \ket{a} \otimes\ket{\omega} \rd \omega \R).
    \end{align}
Then, by a telescoping sum,
\begin{align}
    (LHS) = \norm{\vV^{\dag} \vF\vV- \vV^{'\dag} \vF\vV'} &\le \norm{(\vV^{\dag}-\vV^{'\dag})\vF\vV} +\norm{\vV^{'\dag}\vF(\vV-\vV')}.
\end{align}
We conclude the proof using the individual bounds on the operator norm
\begin{align}
     \norm{\vV} &\le \norm{f}_2\quad \tag{Corollary~\ref{cor:iso_OFT}}\\
     \norm{\vV'}&\le \norm{\hat{f}\cdot \hat{s}_{\mu}}_2 \le \norm{f}_2\quad \tag{Corollary~\ref{cor:iso_OFT}}\\
    \norm{\vF} &\le \norm{\theta}_{\infty}\\
    \norm{\vV'-\vV}&\le \norm{\hat{f}\cdot (1-\hat{s}_{\mu})}_2 \quad \tag{Corollary~\ref{cor:S-A}}.
\end{align}
This concludes the proof.
\end{proof}

\section{Proving monotonicity of energy gradient under level splitting}
\label{sec:monotone_gradient}

It will often be helpful to understand how the energy gradients of a Hamiltonian $\vH$ change when a perturbation $\vV$ is added to yield $\vH'=\vH+\vV$.
This allows us to characterize the energy gradient of $\vH'$ by analyzing the unperturbed spectrum of $\vH$, which is usually much simpler than that of $\vH'$.
Indeed, this is an important part of our proof strategy for showing our key result that $\vH_C$ has no suboptimal local minima (\cref{thm:Hcircuit} in \cref{sec:universal-quantum-computation}). 

The relationship we can prove, which was previously stated in \cref{thm:mono_gradient} in \cref{sec:characterize-neg-grad-condition}, takes the form of \textit{monotonicity}. As the name implies, the result only holds in one direction; it fails when the $\vH'$ and $\vH$ are switched. It is imperative that $\vH$ have a highly degenerate spectrum with the \textit{Bohr-frequency gap}
\begin{align}
    \Delta_{\nu}(\vH): = \min_{\nu_1 \ne \nu_2 \in B(\vH)} \labs{\nu_1-\nu_2}, 
\end{align}
which sets an energy scale for which $\vV$ is perturbative. Note that the Bohr-frequency gap is upper bounded by the gap in the spectrum\footnote{The Bohr-frequency gap can be much smaller than the eigenvalue gap. For example, consider the energies $\{-0.99, 0 , 1\}$, which has an eigenvalue gap of $0.99$ and a Bohr-frequency gap of $0.01$.
}
\begin{align}
    \Delta_{\nu}\le \Delta_E \quad \text{where} \quad \Delta_{E} = \min_{E_1 \ne E_2 \in \text{Spec}(\vH)} \labs{E_1-E_2}.
\end{align}

We now state a more general version of \cref{thm:mono_gradient}, which we prove in the remainder of this appendix.

\begin{theorem}[Monotonicity of gradient under level splitting, expanded version]
\label{thm:mono_gradient_full}
Consider a Hamiltonian $\vH$ with a highly degenerate spectrum and Bohr-frequency gap $\Delta_{\nu} := \min_{\nu_1\neq \nu_2\in B(\vH)} |\nu_1-\nu_2|$, and a perturbed Hamiltonian $\vH'=\vH+\vV$.
Suppose the perturbation is weaker than the Bohr-frequency gap,  $\norm{\vV} \le \frac18\Delta_{\nu}$.
For any $\beta, \tau >0$, let $\CL = \sum_{a\in S} \CL^{\beta,\tau, \vH}_a, \CL' = \sum_{a\in S} \CL^{\beta,\tau, \vH'}_a$ be thermal Lindbladians with jumps $\{\vA^a\}_{a\in S}$, where $\norm{\vA^a}\le 1$ and the transition weight $\gamma_{\beta}(\omega)$ is given by Eq.~\eqref{eq:glauber-dyn}.
Let $\delta_\lambda=\max_j|\lambda_j(\vH)-\lambda_j(\vH')|$, 
where $\lambda_j(\vX)$ is the $j$-th largest eigenvalue of $\vX$, and let
$\theta_{\max}=\max_{\nu\in B(\vH)}|\nu\gamma_\beta(\nu) \indicator(\nu\le \Delta_\nu/2)|$.
For any two operators $\vO$ and $\vO'$, where $ [\vO', \vH']=0$, we have the monotone property that
\begin{align}
    -\CL^\dag[\vH] \succeq r\vO -\epsilon \vI \quad \text{implies}\quad -\CL^{'\dag}[\vH'] \succeq r\vO'- \epsilon' \vI
\end{align}
where
\begin{equation}
    \epsilon' \le \epsilon+ \labs{S} \cdot \CO\L(\frac{1}{\tau} +\frac{\norm{\vH}^{3/4}}{\tau^{1/4}}+\frac{\Lambda_0^{2/3}}{\tau^{1/3}}+\frac{\Lambda_0}{\sqrt{\Delta_{\nu}\tau}}  + \frac{\e^{-\beta\Delta_{\nu}/4}}{\beta} +\delta_\lambda + \theta_{\max} \frac{\norm{\vV}}{\Delta_{\nu}} + r \|\vO-\vO'\|\R).
\end{equation}
For the special case of $\vO = \vI-\vP$ and $\vO' = \vI-\vP'$, where $\vP$ projects onto an eigensubspace of $\vH$, and $\vP'$ projects onto the corresponding perturbed eigensubspace in $\vH'$, then we have the following simpler error bound
\begin{align}
    \epsilon' \le \epsilon+ \labs{S} \cdot \CO\L(\frac{1}{\tau} +\frac{\norm{\vH}^{3/4}}{\tau^{1/4}}+\frac{\Lambda_0^{2/3}}{\tau^{1/3}}+\frac{\Lambda_0}{\sqrt{\Delta_{\nu}\tau}}  + \frac{\e^{-\beta\Delta_{\nu}/4}}{\beta} +\big(1+\frac{\Lambda_0+r}{\Delta_\nu}\big)\norm{\vV}\R).
\end{align}
\end{theorem}

The above result is non-trivial because na\"ive perturbation theory fails: the Lindbladian depends sensitively on the perturbation $\vV$ (as it uses a long Hamiltonian simulation time $\tau\lVert \vV \rVert \gg 1$). In fact, it drastically fails if the energy spectrum of $\vH$ has a (nearly) continuous spectrum (as the opposite of the premise of gapped degenerate subspaces). We can understand the energy scale associated with the minimum Bohr-frequency gap $\Delta_{\nu}$ as the meaningful quantity for which $\vV$ is a perturbation
\begin{align}
     \frac{1}{\tau} \ll \norm{\vV} \ll \Delta_{\nu}.
\end{align}
Otherwise, the $1/\tau$ energy resolution is too small compared to the intended perturbation $\norm{\vV}$.

The proof of \cref{thm:mono_gradient_full} will be quite involved. Technically, we heavily utilize the manipulations using the operator Fourier Transform (Appendix~\ref{sec:OFT}). The key subroutines of the proof are discussed separately as follows.
First, in Appendix~\ref{sec:expressing_gradient}, we will simplify the intimidating expression for energy gradient $\CL^{\dag}[\vH]$.
Secondly, in Appendix~\ref{sec:Mono_rate}, we isolate the key non-perturbative argument, which roughly says level splitting only improves the gradient.
We then provide some results from perturbation theory in \cref{sec:perturb-theory}.
The altogether proof is presented in Appendix~\ref{sec:altogether}, with minor supporting calculations in \cref{sec:supplement_calc}.
We also prove two corollaries of \cref{thm:mono_gradient_full} that apply to subspace gradients in Appendix~\ref{sec:monotone_grad_subspace}.

Since this appendix only consider thermal Lindbladians, in what follows we will drop the superscripts $\beta,\tau,\vH$, i.e. $\CL\equiv \CL^{\beta,\tau,\vH}$, $\CL'\equiv \CL^{\beta,\tau,\vH'}$, $\CD \equiv \CD^{\beta,\tau,\vH}$, etc.

\subsection{Expressing the energy gradient}
\label{sec:expressing_gradient}
The thermal Lindbladian is quite cumbersome to manipulate. Nicely, the energy gradient operator associated with the dissipative part $\CD^{\dag}[\vH]$ permits a much simpler approximate form up to a controllable error. Combining with error bounds on the Lamb-shift term $[\vH_{LS},\vH]$ (Proposition~\ref{prop:norm-lamb-shift-part}) allows us to approximate the full gradient operator $\CL^{\dag}[\vH]$.
\begin{lemma}[Expression for energy gradient]\label{lemma:expr_energy_gradient}
Consider the operator Fourier Transforms $\hat{\vA}^a(\omega)$ weighted by the bump function $f_{\tau}$ in Eq.~\eqref{eq:bump} with Hamiltonian $\vH$. Then, for any Fourier transform pairs $\gamma(\omega)$ and $c(t)$, the energy gradient associated with the purely dissipative Lindbladian
\begin{align}
        \CD^{\dag}[\vH] &=\sum_{a\in S} \int_{-\infty}^{\infty} \gamma(\omega) \left( 
\hat{\vA}^{a}(\omega)^{\dag} \vH \hat{\vA}^a(\omega) -\frac{1}{2} \{\hat{\vA}^{a}(\omega)^{\dag}\hat{\vA}^a(\omega),\vH\} \right)\rd \omega
\end{align}
can be well approximated as a simpler form using
\begin{align}
    \norm{\CD^{\dag}[\vH] - \sum_{a\in S} \int_{-\infty}^{\infty} \gamma(\omega)\omega
 \hat{\vA}^{a}(\omega)^{\dag}\hat{\vA}^a(\omega) \rd \omega} \le\frac{2}{\sqrt{2\pi}\tau} \cdot \norm{c}_1\cdot  \norm{\sum_{a\in S}\vA^{a\dag}\vA^a}.
\end{align}
\end{lemma}
Intuitively, the following expression is the simplest proxy one can write down to capture the rate of energy change
\begin{align}
   \sum_{a\in S} \int_{-\infty}^{\infty} \omega\times \gamma(\omega) \hat{\vA}^{a}(\omega)^{\dag}\hat{\vA}^a(\omega) \rd \omega \sim \text{(energy difference)} \times \text{(rate)}.
\end{align}
Indeed, the Bohr frequency $\omega$ is essentially the energy difference after jump operator $\vA^a(\omega)$; but because of the energy uncertainty in the operator Fourier transform (i.e., because $\hat{f}(\omega)$ is not a delta function), this interpretation must be corrected by an error scaling as the energy resolution $\sim 1/\tau$.
The starting point of the calculation is a certain integration-by-part trick that relates the Hamiltonian operator $\vH$ to the scalar $\omega$.
\begin{proposition}[Integration by parts]\label{prop:byParts}
In the setting of Lemma~\ref{lemma:expr_energy_gradient},
\begin{align}
    [\vH, \hat{\vA}(\omega)] = \omega \hat{\vA}(\omega) + \frac{1}{\ri \sqrt{2\pi \tau}} \left(\vA(\tau/2)\e^{-\ri \omega \tau/2} - \vA(-\tau/2)\e^{\ri \omega \tau/2}\right).
\end{align}
\end{proposition}
\begin{proof}
Integration of the derivative can be expanded by the product rule 
\begin{align}
    \frac{1}{\sqrt{2\pi \tau}} \left(\vA(\tau/2)\e^{-\ri \omega \tau/2} - \vA(-\tau/2)\e^{\ri \omega \tau/2}\right) &= \frac{1}{\sqrt{2\pi \tau}} \int_{-\tau/2}^{\tau/2} \frac{\rd}{\rd t} \left(\vA(t)\e^{-\ri \omega t}\right) \rd t \nonumber\\
    & = \frac{1}{\sqrt{2\pi \tau}} \int_{-\tau/2}^{\tau/2} \left( \ri [\vH,\vA(t)] \e^{-\ri \omega t} - \ri \omega \vA(t)\e^{-\ri \omega t} \right) \rd t \nonumber\\
    & = \ri [\vH, \hat{\vA}(\omega)] - \ri \omega \hat{\vA}(\omega).
\end{align}
Rearrange to conclude the proof.
\end{proof}

Observe that taking the infinite time limit $\tau \rightarrow \infty$ (i.e., perfect energy resolution) in the above proposition recovers the relation for the true Bohr frequencies $\nu$
\begin{align}
    [\vH,\vA_{\nu}] = \nu \vA_{\nu} \quad \text{for each}\quad \nu \in B(\vH).
\end{align}
At finite $\tau$, the above leads to simple bounds on the correction term.
We now present the proof of Lemma~\ref{lemma:expr_energy_gradient}.
\begin{proof}[Proof of Lemma~\ref{lemma:expr_energy_gradient}]
We calculate
    \begin{align}
    \CD^{\dag}[\vH] &=\sum_{a\in S}\int_{-\infty}^{\infty} \gamma(\omega) \left( 
\hat{\vA}^{a}(\omega)^{\dag} \vH \hat{\vA}^a(\omega) -\frac{1}{2} \{\hat{\vA}^{a}(\omega)^{\dag}\hat{\vA}^a(\omega),\vH\} \right)\rd \omega \nonumber\\
& = \sum_{a\in S}\int_{-\infty}^{\infty} \gamma(\omega) \frac{1}{2}\left( 
\hat{\vA}^{a}(\omega)^{\dag}[ \vH,\hat{\vA}^a(\omega)] -[\vH,\hat{\vA}^{a}(\omega)^{\dag}]\hat{\vA}^a(\omega) \right)\rd \omega \nonumber\\
& =\sum_{a\in S} \int_{-\infty}^{\infty} \gamma(\omega)\omega
\hat{\vA}^{a}(\omega)^{\dag}\hat{\vA}^a(\omega) \rd \omega + \vE,
\end{align}
where the error term $\vE$ is given by Proposition~\ref{prop:byParts} as
\begin{align}
    \vE &:= \frac{-\ri}{2\sqrt{2\pi \tau}}\sum_{a\in S}\int_{-\infty}^{\infty} \gamma(\omega)\hat{\vA}^{a}(\omega)^{\dag} \left(\vA^a(\tau/2)\e^{-\ri \omega \tau/2} - \vA^a(-\tau/2)\e^{\ri \omega \tau/2}\right)\rd\omega \nonumber\\
    &+\frac{\ri}{2\sqrt{2\pi \tau}}\sum_{a\in S}\int_{-\infty}^{\infty} \gamma(\omega)\left(\vA^{a}(\tau/2)^\dag\e^{\ri \omega \tau/2} - \vA^{a}(-\tau/2)^\dag\e^{-\ri \omega \tau/2}\right) \hat{\vA}^{a}(\omega) \rd\omega.
\end{align}
To bound this error term, let us calculate one of the four individual terms as an example
\begin{align}
    &\sum_{a\in S} \int_{-\infty}^{\infty} \gamma(\omega)\vA^{a}(\tau/2)^\dag\hat{\vA}^{a}(\omega)\e^{\ri \omega \tau/2} \rd \omega \\
    &=\sum_{a\in S} \frac{1}{2\pi \sqrt{\tau}}\int_{-\infty}^{\infty} 
 \int_{-\infty}^{\infty} c(t_1)\e^{- \ri \omega t_1}\rd t_1 \int_{-\tau/2}^{\tau/2}  \e^{\ri \omega \tau/2} \vA^{a}(\tau/2)^\dag\vA^{a}(t_2) \e^{-\ri \omega t_2} \rd t_2 \rd \omega\quad \tag{Fourier Transforms}\\
 &= \frac{1}{\sqrt{\tau}}\sum_{a\in S} \int_{-\tau/2}^{\tau/2} c(\tau/2 -t_2)\vA^{a}(\tau/2)^\dag\vA^{a}(t_2)\rd t_2. \quad \tag{using $\int_{-\infty}^{\infty} \e^{-\ri \omega t} \rd \omega = 2\pi \delta(t)$}
\end{align}
We can bound the operator norm of this term by $\|c\|_1/\sqrt{\tau}$ after applying the triangle inequality
and using the fact that
\begin{align}
    &\norm{\sum_{a\in S} \vA^{a}(\tau/2)^\dag\vA^{a}(t_2)} = \norm{ \sum_{a\in S} \vA^{a\dag}(\tau/2)\vA^{a}(t_2)} \nonumber\\
    &\le \norm{\sum_{a\in S} \vA^{a\dag} \otimes \bra{a}}\cdot   \norm{\sum_{a\in S} \vA^{a} \otimes \ket{a}} = \norm{\sum_{a\in S} \vA^{a\dag}\vA^a}.
\end{align}
Repeat a similar argument for the other three terms to conclude the proof. 
\end{proof}

\subsection{Monotonicty of rates}
\label{sec:Mono_rate}

Thermal Lindbladians generally depend sensitively on the Hamiltonian as it uses Hamiltonian simulation for a long time $\tau$
\begin{equation}
    \e^{\ri \vH \tau} \quad \text{for} \quad \tau \gg 1.
\end{equation}
Therefore, even adding a small perturbation to the Hamiltonian $\vH ' = \vH + \vV$ may have a non-perturbative effect on $\CL$ since
\begin{align}
    \tau \norm{\vV} \gg 1\quad \text{implies}\quad \normp{\CL - \CL'}{1-1} \gg 0. 
\end{align}

In other words, at a large $\tau$, it is not obvious at all why the Lindbladians $\CL, \CL'$ are related. Indeed, the original Davies' generator ($\tau\rightarrow \infty$) is unstable against arbitrarily small perturbations to the Hamiltonian; whenever energy degeneracy is broken, the Lindbladian can change substantially.

Nevertheless, what we can show as a compromise is that the rate only \textit{increases} if the perturbation only introduces \textit{level splitting}; this amounts to the assumption that the original Hamiltonian has highly degenerate subspaces with a certain Bohr frequencies gap $\Delta_{\nu}$ as another large energy scale
\begin{align}
         \frac{1}{\tau} \ll \norm{\vV} \ll \Delta_{\nu}.
\end{align}
Intuitively, level splitting causes \textit{decoherence} (and only decoherence) in the Bohr frequencies; for large $\tau$, the Lindblidan can indeed tell the transitions $\omega, \omega'$ apart if the Bohr frequencies are sufficiently different. Fortunately, even though decoherence can change the Lindbladian by a lot, we establish certain \textit{monotonicity} of transition rates. That is, $\CL'$ must have as good transition rates as $\CL$. 
A good example of $\vO$ would be an energy subspace projector. However, the argument works for general $\vO$, which makes it more flexible to use.

\begin{lemma}[Decoherence increases the rates]\label{lemma:mono_rate}
For any set of operators $\{\vA^a\}_{a\in S}$, suppose there exists an operator $\vO$ such that
\begin{align}
    \sum_{a\in S} \vA^{a\dag} \vA^a \succeq \vO \quad \text{where} \quad [\vO,\vH] =0.
\end{align}
Then, the operator Fourier Transforms $\hat{\vA}^a(\omega)$ (subscript $f$ omitted) for some normalized weight $\int_{-\infty}^{\infty} |f(t)|^2\rd t =1$ and Hamiltonian $\vH$ satisfies
\begin{align}
    \sum_{a\in S} \int_{-\infty}^{\infty}\hat{\vA}^a(\omega)^{\dag} \hat{\vA}^a(\omega) \rd \omega\succeq \vO.
\end{align}
\end{lemma}
\begin{proof}
By Proposition~\ref{prop:iso_operator_FT}, 
    \begin{align}
        \sum_{a\in S} \int_{-\infty}^{\infty}\hat{\vA}^a(\omega)^{\dag} \hat{\vA}^a(\omega) \rd \omega - \vO &= \int_{-\infty}^{\infty} \e^{\ri \vH t}\L(\sum_{a\in S} \vA^{a\dag}\vA^a\R)\e^{-\ri \vH t} |f(t)|^2\rd t - \vO \nonumber \\
        &= \int_{-\infty}^{\infty} \e^{\ri \vH t}\L(\sum_{a\in S} \vA^{a\dag}\vA^a-\vO\R)\e^{-\ri \vH t} |f(t)|^2\rd t \nonumber\\
        &=\int_{-\infty}^{\infty} \e^{\ri \vH t}\vX^{\dag} \cdot \vX \e^{-\ri \vH t} |f(t)|^2\rd t \nonumber\\
        & \succeq 0.
    \end{align}
The second equality uses that $\int_{-\infty}^{\infty} |f(t)|^2\rd t =1$ and that $\e^{\ri \vH t} \vO \e^{-\ri \vH t} =\vO$. The last line establishes PSD order using the assumption that there exists operator $\vX$ such that $\sum_{a\in S} \vA^{a\dag} \vA^a-\vO = \vX^{\dag} \vX$. Together, we establish the desired statement.
\end{proof}

\subsection{Perturbation theory of eigenstates and eigenvalues}
\label{sec:perturb-theory}

We state a few useful facts about perturbed eigenspace and eigenvalues that would be useful in the proofs.

\begin{proposition}[Davis-Kahan $\sin\Theta$ theorem (see also Theorem VII.3.1 of \cite{Bhatia})]
\label{fact:Davis-Kahan}
Let $\vH$ and $\tilde{\vH}$ be two equal-sized Hermitian matrices. Let $\vP$ be the projector onto eigenstates of $\vH$ with eigenvalue in an interval $[a,b]$.
Let $\tilde{\vP}^\perp$ be the projector onto eigenstates of $\tilde{\vH}$ with eigenvalues outside the interval $[a-\delta, b+\delta]$.
Then
\begin{equation}
\|\vP\tilde{\vP}^\perp\| \le \|\vH-\tilde{\vH}\| / \delta.
\end{equation}
Here $\|\cdot\|$ is the spectral norm (or any unitarily invariant norm).
\end{proposition}
Furthermore, the following fact bounds errors on perturbed eigenvalues: 
\begin{proposition}[Weyl's inequality]
\label{fact:Weyl}
For any two equal-sized Hermitian matrices $\vH$ and $\tilde{\vH}$, we have $|\lambda_j(\vH)-\lambda_j(\tilde{\vH})|\le \|\vH-\tilde{\vH}\|$ for all $j$, where $\lambda_j(\vX)$ is the $j$-th largest eigenvalue of matrix $\vX$.
\end{proposition}

Together, these facts imply that
\begin{lemma}
\label{lem:PP'}
    Let $\vH$ and $\tilde{\vH}=\vH+\vV$ be Hamiltonians.
    Let $\vP$ be the projector onto eigenstates of $\vH$ with eigenvalues in some interval $[a,b]$, which are separated from the other eigenvalues by a gap of at least $\Delta$.
    If $\|\vV\|\le \Delta/4$, then 
    \begin{enumerate}
        \item There exists a spectral projector $\tilde{\vP}$ onto eigenstates of $\tilde{\vH}$ with eigenvalues in $[a-\Delta/4, b+\Delta/4]$, which are separated from the other eigenvalues by a gap of at least $\Delta/2$.
        \item $\|\vP-\tilde{\vP}\| \le 8\|\vV\|/\Delta$.
    \end{enumerate}
\end{lemma}
\begin{proof}
The existence of $\tilde{\vP}$ (item 1) holds because of \cref{fact:Weyl} above.
Then observe that
\begin{equation}
    \|\vP-\tilde{\vP}\| = \|\vP-\vP\tilde{\vP} + \vP\tilde{\vP}-\tilde{\vP}\| \le \|\vP\tilde{\vP}^\perp\| + \|\vP^\perp \tilde{\vP}\| \le 8\|\vV\|/\Delta,
\end{equation}
where the last inequality is obtained by applying \cref{fact:Davis-Kahan} with $\delta=\Delta/4$ to bound $\|\vP\tilde{\vP}^\perp\|$ and $\|\vP^\perp\tilde{\vP}\|$.
\end{proof}

\begin{lemma}[Off-block-diagonal perturbation]\label{lem:off_diag_perturbation}
Consider a block diagonal Hermitian matrix $\vD = \vD_1+ \vD_2$, where the two blocks correspond to orthogonal subspace projectors $\vP_1$ and $\vI - \vP_1 = \vP_2$ and are separated by eigenvalue gap at least $\Delta$. Add an off-block-diagonal Hermitian perturbation $\vV = \vV_{12}+ \vV_{21}$ such that $\norm{\vV}\le \Delta/4$.
Then, there is an anti-Hermitian operator $\vB$ and an absolute constant $C_0$ such that
\begin{gather}
    \vD +\vV = \e^{-\vB}\vD \e^{\vB} + (\vD + \vV  - \e^{-\vB}\vD \e^{\vB}) \nonumber \\
\text{where } \qquad
\norm{\vB} \le C_0 \frac{\norm{\vV}}{\Delta},  \qquad \text{and} \qquad
\norm{\e^{\vB} (\vD + \vV) \e^{-\vB} -\vD} \le C_0 \frac{\norm{\vV}^2}{\Delta}.
\end{gather}
This implies
the sorted eigenvalues of $\vD$ are perturbed by $C_0 \norm{\vV}^2/\Delta$.
\end{lemma}

We remark that the scaling with respect to $\norm{\vV}$ is consistent with perturbation theory: the \textit{angle} change is first-order $\sim\frac{\norm{\vV}}{\Delta}$, and the \textit{eigenvalue} change is second-order $\sim\frac{\norm{\vV}^2}{\Delta}$. Note that for diagonal perturbation, the eigenvalue change is only bounded by $\sim\norm{\vV}$. 

\begin{proof}
Observe that
\begin{equation}
    \e^{\vB} (\vD + \vV) \e^{-\vB} = \vD + (\vV +  [\vB,\vD]) + [\vB,\vV]+ \sum_{k=2}^\infty \frac{1}{k!} \ad_{\vB}^k (\vD+\vV), 
\end{equation}
Let us choose $\vB$ to cancel the first order term in $\vV$, i.e.,
\begin{equation}
\label{eq:VBD-condition}
    \vV = - [\vB,\vD].
\end{equation}
We can solve for $\vB$ by working in the eigenbasis of $\vD=\sum_{i} D_i\ketbrat{\psi_i}$.
Then denoting $O_{ij} = \braket{\psi_i|\vO|\psi_j}$, we can rewrite Eq.~\eqref{eq:VBD-condition} as
\begin{equation}
V_{ij} = B_{ij}(D_i - D_j)
\qquad
\text{or} \qquad
B_{ij} = \frac{V_{ij}}{D_i - D_j}.
\end{equation}
Note $B_{ij}=V_{ij}=0$ whenever $|D_i-D_j|\ge \Delta$ by assumption. Hence, we can solve for $\vB$ using the Heisenberg picture Fourier transform:
 \begin{align}
     B_{ij} & = \frac{1}{\sqrt{2\pi}} \int_{-\infty}^{\infty} f(t) \e^{\ri (D_i - D_j) t} V_{ij}   \rd t  = V_{ij} \hat{f}(D_j-D_i)\quad \text{for each}\quad i,j\nonumber\\
     \text{or} \quad \vB& = \frac{1}{\sqrt{2\pi}} \int_{-\infty}^{\infty} f(t) \undersetbrace{=:\vV(t)}{\e^{\ri \vD t} \vV \e^{-\ri \vD t}}  \rd t ,
 \end{align}
where we choose the function $f(t)$ whose Fourier transform matches the reciprocal at sufficiently large values,
\begin{equation}
     \hat{f}(\omega) = \frac{1}{-\omega } \quad \text{when} \quad \labs{\omega} \ge \Delta,
\end{equation}
but remain ``nice'' near $\omega=0$. One example is to use a smooth bump function
\begin{align}
    -\frac{1}{\omega}\cdot b(\frac{\omega}{\Delta})\quad \text{where}\quad b(x) = \begin{cases}
    1 \quad &\text{if}\quad \labs{x} \ge 1\\
    \CO( x^2 ) \quad &\text{if}\quad \labs{x} \approx 0.
    \end{cases}
\end{align}
For concreteness, we take
\begin{align}
    b(x) = \begin{cases}
        1- \exp(\frac{1}{1-\frac{1}{x^2}}) \quad &\text{if} \quad \labs{x}< 1\\
        1 &\text{else}
    \end{cases}
\end{align}
Then, taking triangle inequality and using the unitary invariance of the operator norm,
\begin{align}
\norm{\vB} &\le \frac{1}{\sqrt{2\pi}}\norm{f}_1\cdot \norm{\vV} \le C_0 \frac{\norm{\vV}}{\Delta}. 
\end{align}
The last inequality bounds the Fourier transform by change-of-variable $x =\omega/\Delta$ and leaves a constant $C_0$ that depends on the inverse Fourier transform of the ``dimensionless function'' $b(x)/x$.
This bound on $\vB$ then allows us to control the higher-order errors
\begin{equation*}
    \e^{\vB} (\vD + \vV) \e^{-\vB} = \vD + \vV + \int_{0}^1\e^{\vB s}[\vB,\vV]\e^{-\vB s} \rd s +  [\vB,\vD]+ 
 \int_{0}^1 \e^{\vB s}[\vB,[\vB,\vD]] \e^{-\vB s} (1-s)\rd s .
\end{equation*}
Substitute $[\vB,\vD]= - \vV$ and rearranging, we get
\begin{equation}
    \e^{\vB} (\vD + \vV) \e^{-\vB} - \vD =   \int_{0}^1\e^{\vB s}[\vB,\vV]\e^{-\vB s} s\rd s.
\end{equation}
Apply the triangle inequality, we get
\begin{equation}
\|\e^{\vB} (\vD + \vV) \e^{-\vB} - \vD \|
\le \|[\vB,\vV]\| \int_0^1 s \rd s \le C_0 \frac{\|\vV\|^2}{\Delta}.
\end{equation}
Finally, to obtain sorted eigenvalues of $\vD$, use the fact that $\e^{-\vB}$ is unitary and apply Weyl's inequality for $\e^{-\vB} \vD\e^{\vB}$.
\end{proof}

\vspace{-10pt}
\begin{figure}[hbt]
	\centering
	\includegraphics[width=0.7\textwidth]{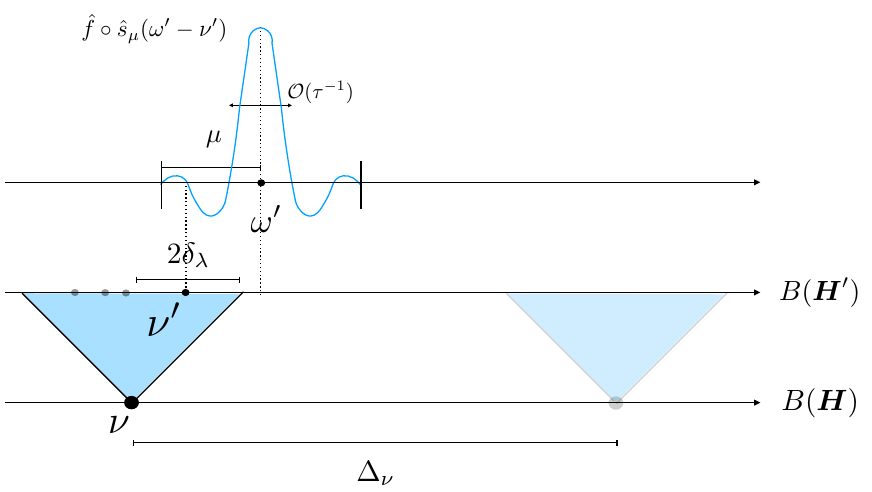}
\caption{ The energy scales $\frac{1}{\tau} \ll \norm{\vV} \ll \Delta_{\nu}$ in one plot. The Hamiltonian perturbation $\vV$ causes eigenvalues of $\vH$ to change by at most $\delta_\lambda \le \|\vV\|$, which splits the Bohr frequencies such that $|\nu-\nu'| \le 2\delta_\lambda$.
 Here $\mu$ is the cut-off frequency for the secular approximation $\hat{\vS}^a_{f,\mu}(\omega') = \sum_{\nu'} \vA_{\nu'}^a \hat{f}(\omega'-\nu') \hat{s}_{\mu}(\omega'-\nu')$. As long as we choose $\mu < \frac{\Delta_{\nu}-4\delta_\lambda}{2}$ small enough, the secular approximation $\hat{\vS}^a(\omega')$ can only contain Bohr frequencies $\vA_{\nu'}$ from at most one block $\nu$, i.e., different blocks decohere. 
	}
	\label{fig:energyscales}
\end{figure}

\subsection{Proof of \cref{thm:mono_gradient_full}}
\label{sec:altogether}

We combine the above ingredients for the proof of Theorem~\ref{thm:mono_gradient_full}.
In what follows, let $E_j$, $\vP_{E_j}$, $\nu$ be the eigenvalues, eigenspace projectors, and Bohr frequencies of the unperturbed Hamiltonian $\vH$; furthermore let  $E_j'$, $\vP_{E_j'}$ $\nu'$, be their counterpart for the perturbed Hamiltonian $\vH'$.

It will be helpful to display the structure of the Bohr frequencies and the energy eigenspaces under perturbation by
\begin{align}
\vA^a &= \sum_{\nu \in B(\vH)} \sum_{E_1-E_2=\nu} \vP_{E_1}\vA^a \vP_{E_2} = \sum_{\nu \in B(\vH)}\vA^a_{\nu} \nonumber \\
&=\sum_{\nu' \in B(\vH')} \sum_{E'_1-E'_2=\nu'} \vP_{E'_1}\vA^a \vP_{E'_2}= \sum_{\nu' \in B(\vH')} \vA^a_{\nu'} = \sum_{\nu\in B(\vH)} \vA^a_{\approx \nu} 
\end{align}
where we defined
\begin{equation}
\vA^a_{\approx \nu} :=
\sum_{\nu'\in B(\vH'),~\nu'\approx \nu} \vA^a_{\nu'}
\qquad
\text{and} \quad
\nu'\approx\nu \iff \labs{\nu'-\nu} \le 2\delta_\lambda.
\label{eq:def_A_approx_nu}
\end{equation}
In other words, the perturbed set of Bohr frequencies can be identified with the original degenerate blocks according to eigenvalue perturbation%
\footnote{Note that we always have $\delta_\lambda \le \|\vV\|$ by \cref{fact:Weyl}, but we keep $\delta_\lambda$ as an separate parameter which helps yield better bounds when $\delta_\lambda \ll \|\vV\|$, such as the case when the $\vV$ is an off-diagonal perturbation (see \cref{lem:off_diag_perturbation}).}
bounded by $\delta_\lambda$, under the assumption that the perturbation is weaker than the Bohr frequency differences $\Delta_{\nu} > 4 \delta_\lambda$.
This structure is crucial for proving the monotonicity of gradients; see Figure~\ref{fig:energyscales}.
For later use, we also define $\hat\vA^a_{\approx\nu}(\omega')$ to be the operator Fourier transform of $\vA^a_{\approx\nu}$ with respect to the perturbed Hamiltonian $\vH'$, and consider its secular approximation $\hat\vS^a_{\approx\nu}(\omega')$ at truncation scale $\mu$, i.e.,
\begin{equation}
\label{eq:AwSw_approx_nu}
\hat\vA^a_{\approx\nu}(\omega') = \sum_{\nu'\approx \nu} \vA^a_{\nu'} \hat{f}_\tau(\omega'-\nu'),
    \qquad
\hat\vS^a_{\approx\nu}(\omega') = \sum_{\nu'\approx \nu} \vA^a_{\nu'} \hat{f}_\tau(\omega'-\nu') \indicator(|\omega'-\nu'|< \mu).
\end{equation}

In what follows, we will denote $\theta(\omega)=\gamma(\omega)\omega$.
It is worth recalling the following bounds,
\begin{align}
\norm{\vA^a},~ \norm{f_\tau}_2,~
\frac{\norm{c_{\beta}}_1}{\sqrt{2\pi}} \le 1,~ \qquad
\norm{\theta}_{\infty} = \CO(\Lambda_0), \qquad
\Lambda_0 = \Theta(1).
\end{align}

Our strategy for proving \cref{thm:mono_gradient} is to rewrite the energy gradients $\CL^{\dag}[\vH]$ and $\CL^{'\dag}[\vH']$ in a form amenable to Lemma~\ref{lemma:mono_rate} between a set of operators and their Fourier Transforms.

\textbf{Step 1.}
For the perturbed Hamiltonian $\vH'$, we apply a sequence of approximations to establish 
\begin{equation}
\label{eq:LHprime-rewritten}
    \Bigg\|{\sum_{a\in S} \sum_{\nu \in B(\vH)} \theta_-(\nu) \int_{-\infty}^{\infty}\hat{\vA}^a_{\approx \nu}(\omega')^{\dag} \hat{\vA}^a_{\approx \nu}(\omega') \rd \omega' -\CL^{'\dag}[\vH']}\Bigg\| \le \epsilon_A
\end{equation}
for some $\epsilon_A >0$ and a function $\theta_-$ to be soon specified. 
Recall we write $\vA\stackrel{\verr}{\approx}\vB$ if $\|\vA-\vB\|\le \verr$.
\begin{align}
    \CL^{'\dag}[\vH']& \stackrel{\verr_1}{\approx} \CD^{'\dag}[\vH'] \tag{bounds on the Lamb-shift: Prop.~\ref{prop:norm-lamb-shift-part}}\\
    &\stackrel{\verr_2}{\approx} \sum_{a\in S} \int_{-\infty}^{\infty} \theta(\omega') \hat{\vA}^{a}(\omega')^{\dag}\hat{\vA}^a(\omega') \rd \omega'\tag{simplify: Lemma~\ref{lemma:expr_energy_gradient}}\\ &\stackrel{\verr_3}{\approx} \sum_{a\in S} \int_{-\infty}^{\infty} \theta(\omega') \hat{\vS}^{a}(\omega')^{\dag}\hat{\vS}^a(\omega') \rd \omega' \quad \tag{secular approximation: Corollary~\ref{cor:AA-SS}}\\
    & = \sum_{a\in S} \sum_{\nu \in B(\vH)} \int_{-\infty}^{\infty} \theta(\omega')\hat{\vS}^a_{\approx \nu}(\omega')^{\dag} \hat{\vS}^a_{\approx \nu}(\omega') \rd \omega' \tag{different blocks $\nu$s decohere}\\
    & \stackrel{\verr_4}{\approx} \sum_{a\in S} \sum_{\nu \in B(\vH)} \int_{-\infty}^{\infty} \theta(\omega')\hat{\vA}^a_{\approx \nu}(\omega')^{\dag} \hat{\vA}^a_{\approx \nu}(\omega') \rd \omega' \tag{undoing secular approximation: Corollary~\ref{cor:AA-SS}}\\
    & \stackrel{\verr_5}{\approx} \sum_{a\in S} \sum_{\nu \in B(\vH)} \theta(\nu) \int_{-\infty}^{\infty}\hat{\vA}^a_{\approx \nu}(\omega')^{\dag} \hat{\vA}^a_{\approx \nu}(\omega') \rd \omega' \tag{rounding $\theta$}\\
    & \stackrel{\verr_6}{\approx} \sum_{a\in S} \sum_{\nu \in B(\vH)} \theta_-(\nu) \int_{-\infty}^{\infty}\hat{\vA}^a_{\approx \nu}(\omega')^{\dag} \hat{\vA}^a_{\approx \nu}(\omega') \rd \omega' \tag{dropping positive values of $\theta$}.
\end{align}
The approximations $\verr_1,\verr_2$ are bounded by
\begin{align}
    {\verr_1} &\le \CO\L( \frac{\norm{\vH}^{3/4}}{\tau^{1/4}}\norm{c_{\beta}}_1 \norm{\sum_{a\in S} \vA^{a\dag}\vA^a} \R) = \labs{S}\CO(\frac{\norm{\vH}^{3/4}}{\tau^{1/4}})\\
    {\verr_2} &\le \frac{2\norm{\sum_{a\in S} \vA^{a\dag}\vA^a}}{\sqrt{2\pi}\tau} \norm{c_{\beta}}_1
    = \labs{S}\CO(\frac{1}{\tau}).
\end{align}

In the approximations $\verr_3$ and $\verr_4$, we choose the secular approximation parameter $\mu$ such that
\begin{equation}
\label{eq:mu-choice}
    \mu < (\Delta_{\nu}-4\delta_\lambda)/2
\end{equation}
with associated errors given by \cref{cor:AA-SS} as
\begin{align}
{\verr_3}&\le 2\norm{\theta}_{\infty} \norm{\sum_{a\in S} \vA^{a\dag}\vA^a }\norm{\hat{f}_{\tau}\cdot (1-\hat{s}_{\mu})}_2 \norm{f_{\tau}}_2 = \labs{S}\CO(\Lambda_0\sqrt{\frac{1}{\mu\tau}}), \\
{\verr_4}&\le 2\norm{\theta}_{\infty} \norm{\sum_{a, \nu}  \vA^{a\dag}_{\approx \nu}\vA^a_{\approx \nu}} \norm{\hat{f}_{\tau}\cdot (1-\hat{s}_{\mu})}_2\norm{f_{\tau}}_2 = \labs{S}\CO(\Lambda_0\sqrt{\frac{1}{\mu\tau}}),
\end{align}
where we applied Eq.~\eqref{eq:secular_tailbound} to bound $\norm{\hat{f}_{\tau}\cdot (1-\hat{s}_{\mu})}_2 \le \sqrt{4/{\pi \mu\tau}}$, and used Proposition~\ref{prop:AAAA} to bound the spectral norm of the sum of jump operators in the second line.

To justify the equality on the fourth line (different blocks $\nu$s decohere), observe the choice of the parameter $\mu$ in Eq.~\eqref{eq:mu-choice} implies
\begin{align}
\hat{\vS}^a(\omega') = \sum_{\nu \in B(\vH)} \hat{\vS}^a_{\approx \nu}(\omega') \indicator(\labs{\omega'- \nu}< \mu+2\delta_\lambda)
= \sum_{\nu \in B(\vH)} \hat{\vS}^a_{\approx \nu}(\omega') \indicator(\labs{\omega'- \nu}< \Delta_\nu/2),
\label{eq:mu_S}
\end{align}
where $\hat\vS^a_{\approx \nu}(\omega')$ is given in Eq.~\eqref{eq:AwSw_approx_nu}.
This ensures that for any given $\omega'$, $\hat{\vS}^a(\omega')$ can activate at most \textit{one} block of transitions with Bohr frequencies closest to $\nu$ (see Figure~\ref{fig:energyscales}).
Consequently,
\begin{equation}
    \hat{\vS}^{a\dag}(\omega')\hat{\vS}^a(\omega') = \sum_{\nu \in B(\vH)} \hat{\vS}^{a\dag}_{\approx \nu}(\omega') \hat{\vS}^a_{\approx \nu}(\omega').
\end{equation}

Next, for the approximation $\verr_5$, we define the following ``rounded'' function $\bar\theta(\omega')$ where an input $\omega'$ close to $\nu \in B(\vH)$ is assigned the same value $\theta(\nu)$, with uniqueness of $\nu$ guaranteed by Eq.~\eqref{eq:mu_S},
\begin{align}
\label{eq:theta-round}
    \bar{\theta}(\omega') := \begin{cases}
    \theta(\nu)\quad &\text{if}\quad \labs{\omega' - \nu} \le \mu + 2 \delta_\lambda\quad \text{for}\quad \nu \in B(\vH)\\
    \theta(\omega')\quad &\text{else}
    \end{cases}.
\end{align}
This lets us formally pull $\theta(\omega')$ out of the integral. Of course, this rounding introduces an error scaling with the energy spread multiplied with the derivative
\begin{align}
    \norm{\bar{\theta} - \theta}_{\infty} \le (2\mu + 4 \delta_\lambda) \cdot \norm{{\rd\theta}/{\rd\omega}}_{\infty}.
\end{align}
Roughly, this error quantifies how the energy gradient (i.e., $\theta(\omega)=\gamma_{\beta}(\omega)\omega$) changes due to perturbation in Bohr frequency. Thus,
\begin{equation}
    {\verr_5}\le \norm{\theta-\bar{\theta}}_{\infty} \bigg\|{\sum_{a, \nu}  \vA^{a\dag}_{\approx \nu}\vA^a_{\approx \nu}} \bigg\|\norm{\hat{f}_{\tau}\cdot \hat{s}_{\mu}}_2^2 = \labs{S}\CO\L(2\mu + 4 \delta_\lambda\R)
\end{equation}
where we applied Propositions \ref{prop:bound_derivative} and \ref{prop:AAAA} (deferred to \cref{sec:supplement_calc}), and used the fact that $\|\hat{f}_\tau \hat{s}_\mu\|_2\le \|\hat{f}_\tau\|_2 \le 1$.

Finally, in the last approximation $\verr_6$, we define the truncated weight 
\begin{equation}
    \theta_-(\nu) = 
    \theta(\nu) \indicator(\nu \le \Delta_{\nu}/2).
\end{equation}
This truncation has the property that $\theta_-(\nu) \le 0$ for each $\nu \in B(\vH)$, which ensures the last line is negative semidefinite.
Thus,
\begin{gather}
    {\verr_6} \le \norm{\theta - \theta_-}_{\infty} \norm{\sum_{a, \nu}  \vA^{a\dag}_{\approx\nu}\vA^a_{\approx\nu}}\norm{f_{\tau}}^2_2  \le \labs{S}\norm{\theta - \theta_-}_{\infty} \nonumber\\
    \text{where} \qquad  \norm{\theta - \theta_-}_{\infty} \le \max_{\omega\ge \Delta_\nu/2} \omega \gamma_\beta(\omega)  \le \frac{ \e^{-\beta\Delta_{\nu}/4} }{\beta},
\end{gather}
using the tail bound in Eq.~\eqref{eq:gamma-tailbound}. 
Altogether,
\begin{align}
    \epsilon_A &={\verr_1+\verr_2+\verr_3+\verr_4+\verr_5+\verr_6} \nonumber \\
    &\le \CO\L(\labs{S}\L(\frac{1}{\tau} +\frac{\norm{\vH}^{3/4}}{\tau^{1/4}}
    + \mu + \frac{\Lambda_0}{\sqrt{\mu\tau}}
    + \delta_\lambda + \frac{\e^{-\beta\Delta_{\nu}/4}}{\beta}\R)\R).
\end{align}
We then choose $\mu = \min(\Lambda_0^{2/3}/\tau^{1/3}, (\Delta_\nu-4\delta_\lambda)/4)$ so as to optimize the error $\CO(\mu+\Lambda_0/\sqrt{\mu \tau})$ while subject to the constraint that $\mu<(\Delta_{\nu} -4\delta_\lambda)/2$. 
This choice implies $\Lambda_0/\sqrt{\mu\tau} \le \Lambda_0^{2/3}/\tau^{1/3} + 2\Lambda_0/\sqrt{(\Delta_{\nu} -4\delta_\lambda)\tau} \le \CO(\Lambda_0^{2/3}/\tau^{1/3} + \Lambda_0/\sqrt{\Delta_\nu\tau})$, where we used $\delta_\lambda\le \|\vV\|\le \Delta_\nu/8$ which is a combination of \cref{fact:Weyl} and the assumption in the theorem statement. 
This yields the following error bound
\begin{equation}
    \epsilon_A \le \CO\L(\labs{S}\L(\frac{1}{\tau} +\frac{\norm{\vH}^{3/4}}{\tau^{1/4}}
    + \frac{\Lambda_0^{2/3}}{\tau^{1/3}} + \frac{\Lambda_0}{\sqrt{\Delta_{\nu} \tau}} 
    + \delta_\lambda + \frac{\e^{-\beta\Delta_{\nu}/4}}{\beta}\R)\R)\label{eq:epsA_bound}.
\end{equation}

\textbf{Step 2.}
For the original Hamiltonian $\vH$, we may repeat the above argument with trivial perturbation ($\vV = 0$) to get
\begin{equation}
\label{eq:LH-rewritten}
    \norm{\sum_{a\in S} \sum_{\nu \in B(\vH)} \theta_-(\nu)
    (\vA^a_{\approx \nu})^{\dag} \vA_{\approx \nu} -\CL^{\dag}[\vH]} \le \epsilon_B
\end{equation}
for some $\epsilon_B >0$.
More detailedly, we have
\begin{align}
\CL^{\dag}[\vH] &\stackrel{\verr_7}{\approx}  \sum_{a\in S} \sum_{\nu \in B(\vH)} \int_{-\infty}^{\infty}\theta_-(\nu)
    \hat{\vA}^a_{\nu}(\omega)^{\dag} \hat{\vA}^a_{\nu}(\omega) \rd \omega \tag{setting $\vV = 0$ from above} \\
    &=\sum_{a\in S} \sum_{\nu \in B(\vH)} \theta_-(\nu)
    (\vA^a_{\nu})^{\dag} \vA^a_{\nu} \tag{Proposition~\ref{prop:iso_operator_FT} and $[\vA_{\nu}^{a\dag}\vA^a_{\nu},\vH]=0$ for each $\nu\in B(\vH)$}\\
    &\stackrel{\verr_8}{\approx} \sum_{a\in S} \sum_{\nu \in B(\vH)} \theta_-(\nu)
    (\vA^a_{\approx \nu})^{\dag} \vA^a_{\approx \nu}
    \label{eq:subspace_perturb}.
\end{align}
The second line is operator Parseval's identity, where the time evolution simplifies due to commutativity $[\vA_{\nu}^{a\dag}\vA^a_{\nu},\vH]=0$.
The error $\verr_7$ can be bounded by the same bounds for $\epsilon_A$ in Eq.~\eqref{eq:epsA_bound} by setting $\vV =0$ (i.e., $\vA^a_{\approx\nu} \rightarrow \vA^a_{\nu}$ ).

The last line is a brute-force rewriting of $\vA^a_{\nu}$ into $\vA^a_{\approx \nu}$, which acts on eigenstates of $\vH'$ instead of $\vH$, with error $E_8$ bounded by perturbation theory.
This rewriting allows us to prove \cref{thm:mono_gradient_full} by directly applying Lemma~\ref{lemma:mono_rate} between the following set of operators and their Fourier transforms 
\begin{align*}
 \{\sqrt{\labs{\theta_-(\nu)}}\vA^a_{\approx \nu}\}_{a,\nu} \quad \text{and}\quad \{\sqrt{\labs{\theta_-(\nu)}}\hat{\vA}^a_{\approx \nu}(\omega')\}_{a,\nu,\omega'}\quad \text{for the perturbed Hamiltonian}\quad \vH'.
\end{align*}
We give an explicit error bound on $\verr_8$ in Proposition~\ref{prop:thetaAAAA} (deferred to \cref{sec:supplement_calc}), which yields
\begin{align}
    {\verr_8} = \CO\L( \bigg\|\sum_{a\in S} \vA^{a\dag}\vA^a\bigg\| \theta_{\max} \frac{\norm{\vV}}{\Delta_{\nu}}\R) = \CO\L( \labs{S} \theta_{\max}\frac{\norm{\vV}}{\Delta_{\nu}}\R),
\end{align}
where we used the bound $\max_{\nu\in B(\vH)} |\theta_-(\nu)| =\theta_{\max}$ provided in the theorem statement.
Collect the errors to bound $\epsilon_B = \verr_7 + \verr_8 \le \epsilon_A + \verr_8$.

\textbf{Step 3.}
Now we may finish the proof of \cref{thm:mono_gradient_full} by applying Lemma~\ref{lemma:mono_rate}. 
First, Eq.~\eqref{eq:LH-rewritten} implies
\vspace{-5pt}
\begin{align}
&-\sum_{a\in S} \sum_{\nu \in B(\vH)} \theta_-(\nu)
    (\vA^a_{\approx \nu})^{\dag} \vA_{\approx \nu} + \epsilon_B \vI 
    \succeq -\CL^{\dag}[\vH] \nonumber \\
    &\qquad \qquad 
    \succeq r\vO - \epsilon \vI  \tag{by assumption}\\
    &\qquad \qquad 
    \succeq r\vO'  - (r\|\vO'-\vO\| + \epsilon) \vI .
\end{align}
Similarly, Eq.~\eqref{eq:LHprime-rewritten} implies
\begin{align}
    -\CL^{'\dag}[\vH'] + \epsilon_A \vI &\succeq -\sum_{a\in S} \sum_{\nu \in B(\vH)} \theta_-(\nu) \int_{-\infty}^{\infty}
    \hat{\vA}^a_{\approx \nu}(\omega')^{\dag} \hat{\vA}^a_{\approx \nu}(\omega')   \rd \omega' \nonumber \\
    &\succeq r\vO' - (\epsilon + \epsilon_B + r \norm{\vO-\vO'}) \vI.
\end{align}
Note in the last step we used the assumption that $[\vO',\vH']=0$ and applied \cref{lemma:mono_rate}.
Hence, we have shown that $-\CL^{'\dag}[\vH'] \succeq r\vO' - \epsilon'\vI$, where $\epsilon' = \epsilon + \epsilon_A + \epsilon_B   + r\|\vO-\vO'\|$ can be bounded by
\begin{align}
\label{eq:mono-eps-bound-full}
    \epsilon' \le \epsilon + \CO\bigg(
\labs{S}\Big(\frac{1}{\tau} +\frac{\norm{\vH}^{3/4}}{\tau^{1/4}}+\frac{\Lambda_0^{2/3}}{\tau^{1/3}}  + \frac{\Lambda_0}{\sqrt{\Delta_{\nu} \tau}} 
+ \frac{\e^{-\beta\Delta_{\nu}/4}}{\beta}
+ \delta_\lambda + 
\theta_{\max}\frac{\|\vV\|}{\Delta_\nu} + r\|\vO-\vO'\| \Big)
\bigg).
\end{align}

\paragraph{A simpler bound.}
We now consider the special case of $\vO=\vI-\vP$, $\vO=\vI-\vP'$ to derive a simpler bound as in the theorem statement.
Note we have $\|\vO-\vO'\|= \|\vP-\vP'\| \le 8\|\vV\|/\Delta_E \le 8\|\vV\|/\Delta_\nu$, using \cref{lem:PP'} and the fact the spectral gap is lower bounded by the Bohr-frequency gap, $\Delta_E\ge \Delta_\nu$.
Furthermore generally $\theta_{\max} = \|\theta_-\|_\infty = \CO(\Lambda_0)$ for our choice of $\gamma_\beta(\omega)$ in \cref{eq:glauber-dyn}.
And we always have $\delta_\lambda \le \|\vV\|$ by \cref{fact:Weyl}.
Plugging these into \cref{eq:mono-eps-bound-full}, we have
\begin{equation}
\label{eq:mono-eps-bound-simple}
\epsilon' \le \epsilon + 
\CO\bigg(
\labs{S}\Big(\frac{1}{\tau} +\frac{\norm{\vH}^{3/4}}{\tau^{1/4}} + \frac{\Lambda_0^{2/3}}{\tau^{1/3}} + \frac{\Lambda_0}{\sqrt{\Delta_{\nu} \tau}} + \frac{\e^{-\beta\Delta_{\nu}/4}}{\beta} + \norm{\vV} + (\Lambda_0+r)\frac{\|\vV\|}{\Delta_\nu} \Big)
\bigg).
\end{equation}
This concludes our proof of \cref{thm:mono_gradient_full}.

\subsection{Supplementary calculations}
\label{sec:supplement_calc}
In this section, we provide some missing calculations that prove some propositions used in the proof in the previous section.

\begin{proposition}[Bounds on the derivative]\label{prop:bound_derivative}
There exists an absolute constant $C$ such that for any $\beta,\Lambda_0$,
    \begin{equation}
          \norm{\frac{\rd}{\rd\omega} \theta(\omega)}_{\infty} = \CO\L( \norm{\frac{\rd}{\rd\omega} \L( \frac{\e^{-\omega^2/2\Lambda_0^2}}{1+\e^{\beta \omega}}\omega \R)}_{\infty}\R) \le C.
    \end{equation}
\end{proposition}
\begin{proof}
By the product rule, 
\begin{align}
         \labs{\frac{\rd}{\rd\omega}\L( \frac{\omega\e^{-\omega^2/2\Lambda_0^2}}{1+\e^{\beta \omega}} \R)} & = \labs{\frac{\e^{-\omega^2/2\Lambda_0^2}-\e^{-\omega^2/2\Lambda_0^2}\omega^2/\Lambda_0^2 }{1+\e^{\beta \omega}} -  \frac{ \e^{-\omega^2/2\Lambda_0^2} \beta\omega\e^{\beta \omega} }{(1+\e^{\beta \omega})^2}}  \le (const.)
\end{align}
using change of variable $x = \beta \omega$ and $y = \omega/ \Lambda_0 $ to obtain the absolute constant bound.
\end{proof}

\begin{proposition}\label{prop:AAAA}
In the prevailing notation,
    \begin{equation}
    \Bigg\|{\sum_{a\in S} \sum_{\nu \in B(\vH)} \vA^{a\dag}_{\nu}\vA^a_{\nu}}\Bigg\|, 
    \quad \Bigg\|{\sum_{a\in S}\sum_{ \nu\in B(\vH)} \vA^{a\dag}_{\approx \nu}\vA^a_{\approx \nu}}\Bigg\| \le \Bigg\|{\sum_{a\in S} \vA^{a\dag}\vA^a}\Bigg\|.
    \end{equation}
\end{proposition}
\begin{proof}
We focus on $\vA^{a\dag}_{\approx \nu}$ since the case of $\vA^{a\dag}_{\nu}$ is a special case. Resolve the identity by nearby energy subspaces
\begin{equation}
    \vI = \sum_{E} \vP_{E} = \sum_{\bar{E}} \undersetbrace{=:\vP_{\approx \bE}}{\sum_{E\approx \bE} \vP_{E}}\quad \text{such that}\quad     \vA^{a}_{\approx \nu} = \sum_{\bE_2 - \bE_1 = \nu }  \vP_{\approx \bE_2} \vA^{a} \vP_{\approx \bE_1}.
\end{equation}
Now, we calculate
\begin{align}
    \Bigg\|{\sum_{a\in S}\sum_{ \nu\in B(\vH)}\vA^{a\dag}_{\approx \nu}\vA^a_{\approx \nu}}\Bigg\| &= \Bigg\|{\sum_{a\in S,\bE_2, \bE_1} \vP_{\approx \bE_1} \vA^{a\dag} \vP_{\approx \bE_2} \vA^{a} \vP_{\approx \bE_1}}\Bigg\|  = \Bigg\|{\sum_{a\in S,\bE_1} \vP_{\approx \bE_1} \vA^{a\dag} \vA^{a} \vP_{\approx \bE_1}}\Bigg\| \nonumber \\
    &= \max_{\bE} \bigg\|{\vP_{\approx \bE_1} \sum_{a\in S}\vA^{a\dag} \vA^{a} \vP_{\approx \bE_1}}\bigg\| \le \bigg\|{\sum_{a\in S} \vA^{a\dag}\vA^a}\bigg\|.
\end{align}
The last line uses that the operator norm of block-diagonal matrices equals the maximum among the blocks. 
\end{proof}
\begin{proposition}[Jumps with perturbed Hamiltonian]\label{prop:thetaAAAA}
In the prevailing notation, and for any function $h(\nu)$, we have that
\begin{align}
&\norm{\sum_{a\in S} \sum_{\nu\in B(\vH)} h(\nu) (\vA^a_{\nu})^{\dag} \vA^a_{\nu} -\sum_{a\in S} \sum_{\nu\in B(\vH)} h(\nu)
    (\vA^a_{\approx \nu})^{\dag} \vA^a_{\approx \nu}} \nonumber \\
    &\qquad\qquad\qquad\qquad\qquad\qquad\qquad
    \le \CO\Bigg( \Bigg\|{\sum_{a\in S} \vA^{a\dag}\vA^a}\Bigg\| \cdot \max_{\nu \in B(\vH)}\labs{h(\nu)} \cdot \frac{\norm{\vV}}{\Delta_{\nu}} \Bigg).     
\end{align}    
\end{proposition}

\begin{proof}
It suffices to set $\norm{\sum_{a\in S} \vA^{a\dag}\vA^a} =1$. Consider the operator Fourier transform with a smooth bump weight $\norm{g}_2 =1$, 
\begin{align}
    \vA^a_{g,\vH}(\omega) \quad \text{for}\quad \hat{g}(\omega)\propto  \begin{cases}
    0 \quad &\text{if}\quad \labs{\omega} \ge \frac{\Delta_{\nu}}{2}\\
    \CO(1) \quad &\text{else}.
    \end{cases}    
\end{align}
which automatically decohere different Bohr frequency blocks (i.e., no need to apply secular approximation). Also, extend the function locally 
\begin{align}
    h(\omega):= \begin{cases}
    h(\nu) \quad &\text{if}\quad \labs{\omega-\nu} \le 2\norm{\vV}
    \end{cases}  .  
\end{align}
Then,
\begin{align}
    \sum_{a\in S} \int_{-\infty}^{\infty} h(\omega)\vA^a_{g,\vH}(\omega)^{\dag} \vA_{g,\vH}^a(\omega)\rd \omega &=\sum_{a\in S} \sum_{\nu,\nu' \in B(\vH)}\int_{-\infty}^{\infty} h(\omega)(\vA^a_{\nu'})^{\dag} \vA^a_{\nu}\hat{g}^*(\omega-\nu')\hat{g}(\omega-\nu) \rd \omega \nonumber\\
    &=\sum_{a\in S} \sum_{\nu,\in B(\vH)}  h(\nu)(\vA^a_{\nu})^{\dag} \vA^a_{\nu}\int_{-\infty}^{\infty}\labs{\hat{g}(\omega-\nu)}^2 \rd \omega \nonumber\\
    &=\sum_{a\in S} \sum_{\nu\in B(\vH)} h(\nu)(\vA^a_{\nu})^{\dag} \vA^a_{\nu}.
\end{align}

Now, we add the perturbation $\vH + \vV$. The insight is that we can introduce an artificial Hamiltonian 
\begin{equation}
    \bar{\vH} := \sum_{E\in \text{spec}(\vH)} E\sum_{E' \approx E }\vP_{E'} \quad \text{such that}\quad \norm{\bar{\vH} - \vH} \le 2\norm{\vV} \label{eq:barH-H} 
\end{equation}
with exactly the same spectrum of the original Hamiltonian $\vH$, but with the basis according to the perturbed Hamiltonian $\vH'$. Then, the same argument with the artificial Hamiltonian implies
\begin{align} 
    \sum_{a\in S} \int_{-\infty}^{\infty} h(\omega)\vA^a_{g,\bar{\vH}}(\omega)^{\dag} \vA_{g,\bar{\vH}}^a(\omega)\rd \omega &=\sum_{a\in S} \sum_{\nu\in B(\bar{\vH})} h(\nu)(\vA^a_{\approx\nu})^{\dag} \vA^a_{\approx\nu}.
\end{align}
Lastly, we may bound the difference by the purification
\begin{align}
     &\norm{\sum_{a\in S} \int_{-\infty}^{\infty} \hat{\vA}^a_{g,\vH}(\omega) \otimes \ket{a} \otimes\ket{\omega} \rd \omega - \sum_{a\in S} \int_{-\infty}^{\infty} \hat{\vA}^a_{g,\bar{\vH}}(\omega) \otimes \ket{a} \otimes\ket{\omega} \rd \omega} \nonumber \\
     &=\norm{\sum_{a\in S} \int_{-\infty}^{\infty} \vA^a_{\vH}(t) \otimes \ket{a} \otimes g(t)\ket{t} \rd t - \sum_{a\in S} \int_{-\infty}^{\infty} \vA^a_{\bar{\vH}}(t) \otimes \ket{a} \otimes g(t)\ket{t} \rd t}\tag{Fourier Transform is unitary}\\
     &\le \sqrt{ \int_{-\infty}^{\infty} \labs{2\norm{\e^{\ri \vH t} - \e^{\ri \bar{\vH} t}} g(t)}^2 \rd t } = \CO(\frac{\norm{\vV}}{\Delta_{\nu}}).
\end{align}
The factor of 2 is due to left and right Hamiltonian evolution
\begin{align}
    \e^{\ri \vH t}\vA^a \e^{-\ri \vH t} -  \e^{\ri \bar{\vH} t}\vA^a \e^{-\ri \bar{\vH} t} &= \e^{\ri \vH t}\vA^a (\e^{-\ri \vH t} -  \e^{-\ri \bar{\vH} t}) +
(\e^{\ri \vH t} - \e^{\ri \bar{\vH} t})\vA^a \e^{-\ri \bar{\vH} t}.
\end{align}
To evaluate the integral, we use that $\norm{\e^{\ri \vH t} - \e^{\ri \bar{\vH} t}}\le \norm{\vH-\bar{\vH}}t \le 2 \norm{\vV}t$ and that $g(t)$ is rapidly decaying for large $\labs{t}\ge \frac{1}{\Delta_{\nu}}$. To conclude the proof, use the purification tricks (\cref{lem:Op_purification_bounds}).
\end{proof}

\subsection{Monotonicity of gradient on a subspace}
\label{sec:monotone_grad_subspace}
For our proof that \textsf{BQP}-hard Hamiltonians has no suboptimal local minima in \cref{sec:universal-quantum-computation}, we will need the following refinements of \cref{thm:mono_gradient_full} where the gradient operator acts on a low-energy subspace with an excitation gap.
Intuitively, we care only about the Bohr-frequency gap restricted to the low-energy subspace $\vH\vQ$ instead of the full Hilbert space; the gradient on that subspace should not be sensitive to the excited states above the excitation gap. 

\begin{corollary}[Monotonicity of gradient on a subspace; \cref{cor:monotonicity_subspace} restated]
\label{cor:mono_subspace_restated}
Consider a Hamiltonian $\vH = \sum_{\bar{E}} \bar{E}\vP_{\Bar{E}}$ and its perturbation $\vH':=\vH+\vV$. Let $\vP$ be the ground space projector for $\vH$ and $\vP'$ be the corresponding perturbed eigensubspace of $\vH'$.
Let $\vQ$ be a low-energy eigensubspace projector of $\vH$ (i.e., $\vQ  = \sum_{E \le E_{\vQ}} \vP_E$ for $E_{\vQ} \in \text{Spec}(\vH)$) with excitation gap $\Delta_{\vQ}$.
Assume $\frac{\|\vV\|\norm{\vH}}{\Delta_{\vQ}} \le \frac{1}{144}\Delta_\nu$ where $\Delta_\nu := \min_{\nu_1\neq \nu_2\in B(\vH|_{\vQ})} |\nu_1-\nu_2|$ is the Bohr-frequency gap of $\vH$ within the subspace $\vQ$.  
For any $\beta, \tau >0$, let $\CL = \sum_{a\in S} \CL^{\beta,\tau, \vH}_a, \CL' = \sum_{a\in S} \CL^{\beta,\tau, \vH'}_a$ be thermal Lindbladians with jumps $\{\vA^a\}_{a\in S}$, where $\norm{\vA^a}\le 1$ and the transition weight $\gamma_{\beta}(\omega)$ is given by Eq.~\eqref{eq:glauber-dyn}.
    Then we have the monotone property that
\begin{align}
    - \vQ \CL^{\dag}[\vH]\vQ \succeq r \vQ(\vI-\vP) - \epsilon \vI
    \quad \text{implies} \quad - \vQ' \CL'^{\dag}[\vH']\vQ' \succeq r \vQ'(\vI-\vP') - \epsilon' \vI
\end{align}
where $\vQ'$ projects onto the perturbed eigensubspace of $\vH'$ identified with $\vQ$, and
\begin{align}
    \epsilon' \le \epsilon + \labs{S}\cdot \CO\bigg(\frac{1}{\tau} +\frac{\norm{\vH}^{3/4}}{\tau^{1/4}}+\frac{\Lambda_0^{2/3}}{\tau^{1/3}}+\frac{\Lambda_0}{\sqrt{\Delta_{\nu}\tau}} +\frac{\Lambda_0}{\sqrt{\Delta_{\vQ}\tau}} + \frac{\e^{-\beta\Delta_{\nu}/4}}{\beta} + \frac{\e^{-\beta\Delta_{\vQ}/4}}{\beta}  \nonumber \\
    + \L(1+\frac{\Lambda_0}{\Delta_\nu}\R)\frac{\norm{\vV}\norm{\vH}}{\Delta_{\vQ}} + r \Big(\frac{\norm{\vV}}{\Delta_{\vQ}}+\frac{\norm{\vV}}{\Delta_{\nu}}\Big)\bigg).
\end{align}
\end{corollary}

\begin{proof}
The idea is that $\vQ' \CL^{'\dag}[\vH']\vQ'$ essentially depends only on the low energy subspace $\vQ'$ and the corresponding restricted transition $\vQ'\vA^a \vQ'$.
\begin{align}
    \vQ' \CL^{'\dag}[\vH']\vQ' &\stackrel{\verr_1}{\approx} \sum_{a\in S} \int_{-\infty}^{\infty} \theta(\omega') \vQ'\hat{\vA}^{a}(\omega')^{\dag}\hat{\vA}^a(\omega')\vQ' \rd \omega'\tag{Prop.~\ref{prop:norm-lamb-shift-part} and Lemma~\ref{lemma:expr_energy_gradient}}\\
    &\stackrel{\verr_2}{\approx}  \sum_{a\in S} \int_{-\infty}^{\infty} \theta(\omega') \vQ'\hat{\vA}^{a}(\omega')^{\dag}\vQ' \hat{\vA}^a(\omega')\vQ' \rd \omega'\tag{excitation gap $\Delta_{\vQ'}$ of $\vQ'$}\\
    &=  \sum_{a\in S} \int_{-\infty}^{\infty} \theta(\omega') \hat{\vR}'^{a}(\omega')^{\dag} \hat{\vR}'^a(\omega') \rd \omega'\tag{set $\vR'^a := \vQ'\vA^a\vQ'$ and use that $[\vQ',\vH'] =0$}\\
    &\stackrel{\verr_3}{\approx} \CL^{\dag\beta,\tau,\vH'_{\vQ'}}_{\{\vR'^a\}}[\vH'_{\vQ'}]\tag{unsimplify: Prop.~\ref{prop:norm-lamb-shift-part} and Lemma~\ref{lemma:expr_energy_gradient}}\\
    &\stackrel{\verr_4}{\approx} \CL^{\dag\beta,\tau,\vH'_{\vQ'}}_{\{\vR^a\}}[\vH'_{\vQ'}]\tag{change the jumps to $\vR^a = \vQ \vA^a \vQ$ }.
\end{align}
The approximation $\verr_2$ inserts the low-energy projector $\vQ'$.  
To do so, we resolves the identity by $\vI = \vQ' + (\vI - \vQ')$ and uses that 
\begin{align}
    &\sum_{a\in S} \int_{-\infty}^{\infty} \theta(\omega') \vQ'\hat{\vA}^{a}(\omega')^{\dag}(\vI-\vQ')\hat{\vA}^a(\omega')\vQ' \rd \omega'\tag{excitation gap of $\vQ'$}\\
    &\stackrel{\verr_{21}}{\approx} \sum_{a\in S} \int_{-\infty}^{\infty} \theta(\omega') \hat{\vS}^{a}(\omega')^{\dag} \hat{\vS}^a(\omega') \rd \omega' \tag{secular approximation for $(\vI-\vQ')\vA^a\vQ'$}\\
    &\stackrel{\verr_{22}}{\approx} 0\tag{ $\mu \ll \Delta_{\vQ}$, $\Delta_{\vQ}\beta \gg 1$. }
\end{align}
That is, we need the excitation gap to be large so that $(\vI-\vQ')\vA^a\vQ'$ have a vanishing contribution to the gradient. These error combines $\verr_2 = \verr_{21}+\verr_{22}$ where
\begin{align}
    {\verr_{21}} &\le 2\norm{\sum_{a\in S} \vA^{a\dag}\vA^a} \norm{\theta}_{\infty} \norm{f_{\tau}\cdot (1-\hat{s}_{\mu})}_2 \norm{f_{\tau}}_2 =  \labs{S} \CO(\frac{\Lambda_0}{\sqrt{\mu\tau}})\\
    {\verr_{22}}&\le \norm{\sum_{a\in S} \vA^{a\dag}\vA^a}\max_{\omega' \ge \Delta_{\vQ'} - \mu}{\labs{\theta(\omega')}} \norm{f_{\tau}\cdot \hat{s}_{\mu}}_2^2 = \labs{S} \frac{\e^{-\beta (\Delta_{\vQ'} - \mu) /2}}{\beta}.
\end{align}
Thus, we consider a safe choice of $\mu =\Delta_{\vQ'}/2$. Also, since the perturbation is small by assumption, $\norm{\vV} \le \frac{1}{144} \Delta_{\nu} \frac{\Delta_{\vQ}}{\norm{\vH}} \le \frac{1}{144} \Delta_{\vQ}$, the excitation gap remains large $\Delta_{\vQ'} \ge \Delta_{\vQ} - 2\norm{\vV} \ge \Delta_{\vQ} - \frac{1}{72} \Delta_{\vQ} \ge \Delta_{\vQ}/2$. 
The third line (after approximation $E_2$) formally disposes of the excited state of $\vH'$ above $E'_{\vQ'}$ by defining a modified Hamiltonian 
\begin{align}
    \vH'_{\vQ'} := \vH'\vQ' + E'_{\vQ'} (\vI - \vQ')
\end{align}
for the Fourier transform $\hat{\vR'^a}(\omega')$. This Hamiltonian is merely a proof artifact, and one can also set the energy of the excited subspace $\vI-\vQ'$ to be infinity. The error $\verr_3$ is merely the error to put it back to the form of an energy gradient.

The approximation $\verr_4$ changes the jumps with norm bounded by 
\begin{align}
    \norm{\sum_{a\in S} \vR'^a \otimes \ket{a} - \sum_{a\in S} \vR^a \otimes \ket{a}} \le 2\sqrt{\norm{\sum_{a\in S} \vA^{a\dag}\vA^a}} \norm{\vQ-\vQ'} = \labs{S}\CO(\frac{\norm{\vV}}{\Delta_{\vQ}})
\end{align}
using subspace perturbation bounds $\norm{\vQ-\vQ'} \le 8\frac{\norm{\vV}}{\Delta_{\vQ}}$ (Lemma~\ref{lem:PP'}). Since the (suitably normalized) gradient operator $\frac{1}{2\labs{S}\norm{\vO}}\cdot \CL^{\dag}[\vO]$ can be block-encoded using $\CO(1)$ block-encodings of the jumps (\cref{thm:LCUSim_main}, \cref{prop:blockLS}, \cref{prop:LOL}), perturbation to the jumps propagates to the gradient operator by
\begin{equation}
    {\verr_4} = \CO\L(\bigg\|\sum_a \vR'^a \ket{a} - \sum_a \vR^a \ket{a}\bigg\| \cdot \norm{\vH'_{\vQ'}}\R) = \labs{S}\CO\Big(\norm{\vH'_{\vQ'}}\frac{\norm{\vV}}{\Delta_{\vQ}}\Big).
\end{equation}

To summarize, the above gives the bound
\begin{align}
    \norm{\vQ' \CL^{'\dag}[\vH']\vQ' -\CL^{\dag\beta,\tau,\vH'_{\vQ'}}_{\{\vR^a\}}[\vH'_{\vQ'}]}
    &\le   \labs{S}\cdot \CO\L(\frac{1}{\tau} + \frac{\norm{\vH}^{3/4}}{\tau^{1/4}} + \frac{\Lambda_0}{\sqrt{\Delta_{\vQ}\tau}}+ \frac{\e^{-\beta \Delta_{\vQ}/4 }}{\beta}+\frac{\norm{\vV}}{\Delta_{\vQ}}\norm{\vH'_{\vQ'}}\R)
    \label{eq:QLQ}
\end{align}
using $\norm{\vH'} \le \norm{\vH}+ \norm{\vV} \le 2\norm{\vH}$ and $\Delta_{\vQ'} \ge \Delta_{\vQ}/2$. Similarly, the 
    \begin{align}
        \vQ \CL^{\dag}[\vH]\vQ &\stackrel{\verr_5}{\approx} \CL^{\dag\beta,\tau,\vH_{\vQ}}_{\{\vR^a\}}[\vH_{\vQ}]\quad \text{for}\quad \vH_{\vQ} := \vH\vQ + E_{\vQ} (\vI - \vQ)\label{eq:modifiedHQ}
    \end{align}
    with ${\verr_5}$ also bounded by the RHS of Eq.~\eqref{eq:QLQ} but with $\norm{\vV} \rightarrow 0$.
    
Now, we may use the monotonicity of gradient (\cref{thm:mono_gradient_full}) for Hamiltonian pairs $\vH_{\vQ}$ and $\vH'_{\vQ'}$, jumps $\{\vR^a\}_{a\in S}$, with the characteristic Bohr-frequency gap $\Delta_\nu =\min_{\nu_1\neq \nu_2\in B(\vH_{\vQ})} |\nu_1-\nu_2|$.
The modified Hamiltonian perturbation $\|\vH_{\vQ} - \vH'_{\vQ'}\|$ is bounded by the sum of the following two  errors:
\begin{equation}
    \norm{\vH\vQ  - \vH' \vQ'}  = \norm{\vH\vQ - \vH \vQ' +  \vH\vQ' - \vH' \vQ' } \le 8 \frac{\norm{\vV}}{\Delta_{\vQ}}\norm{\vH} +\norm{\vV} \le 9\frac{\norm{\vV}}{\Delta_{\vQ}}\norm{\vH},
    \label{eq:Hq-Hqprime}
\end{equation}
\begin{align}
    \norm{ E_{\vQ}(\vI - \vQ)-E'_{\vQ'}(\vI - \vQ')} &= \norm{ E_{\vQ}(\vI - \vQ)-E_{\vQ}(\vI - \vQ') + E_{\vQ}(\vI - \vQ') - E'_{\vQ'}(\vI - \vQ')} \nonumber\\
    &\le 8 E_{\vQ} \frac{\norm{\vV}}{\Delta_{\vQ}} + \norm{\vV}\le 9 \frac{\norm{\vV}}{\Delta_{\vQ}}\norm{\vH}.
    \label{eq:Eq-Eqprime}
\end{align}
So $\|{\vH_{\vQ} - \vH'_{\vQ'}}\|\le 18 \|\vV\|\|\vH\|/\Delta_{\vQ}$, which is less than $\Delta_\nu/8$ by the assumption in the corollary statement so we may apply \cref{thm:mono_gradient_full}.
Note that error due to perturbing the $\vQ(\vI-\vP)$ can be bounded directly by 
\begin{equation}
    \norm{ \vQ'(\vI-\vP') - \vQ(\vI-\vP) } \le  \norm{\vQ'-\vQ}+ \norm{\vP-\vP'} =  \CO\Big(\frac{\norm{\vV}}{\Delta_{\vQ}}+\frac{\norm{\vV}}{\Delta_{\nu}}\Big).
    \label{eq:QP-diff}
\end{equation}
Apply \cref{thm:mono_gradient_full} with $\vH=\vH_{\vQ}$, $\vH'=\vH'_{\vQ'}$, $\vO=\vQ(\vI-\vP)$, and $\vO'=\vQ'(\vI-\vP')$, $\delta_\lambda \le \|\vH_{\vQ}-\vH_{\vQ'}\|$, and $\theta_{\max}=\CO(\Lambda_0)$, along with the additional approximation error in Eq.~\eqref{eq:QLQ}, we obtain the result as advertised in the corollary statement.
\end{proof}

\begin{corollary}[Monotonicity of gradient on a subspace under off-block-diagonal perturbation; \cref{cor:off_diag_mono} restated]
In the setting of Corollary~\ref{cor:mono_subspace_restated}, instead assume $\frac{\norm{\vV}}{\Delta_\nu},\frac{\norm{\vV}}{\Delta_{\vQ}} \le (const.)$, and that the perturbation is off-block-diagonal, i.e.,  $\vQ\vV\vQ=(\vI-\vQ)\vV(\vI-\vQ) = 0$. Then,
\begin{align}
    - \vQ \CL^{\dag}[\vH]\vQ \succeq r \vQ(\vI-\vP) - \epsilon \vI
    \quad \text{implies} \quad - \vQ' \CL'^{\dag}[\vH']\vQ' \succeq r \vQ'(\vI-\vP') - \epsilon' \vI
\end{align}
where
\begin{align}
    \epsilon' &\le \epsilon + \labs{S}\cdot \CO\bigg(\frac{1}{\tau} +\frac{\norm{\vH}^{3/4}}{\tau^{1/4}}+\frac{\Lambda_0^{2/3}}{\tau^{1/3}}+\frac{\Lambda_0}{\sqrt{\Delta_{\nu}\tau}} +\frac{\Lambda_0}{\sqrt{\Delta_{\vQ}\tau}} + \frac{\e^{-\beta\Delta_{\nu}/4}}{\beta} + \frac{\e^{-\beta\Delta_{\vQ}/4}}{\beta}   \nonumber\\
    &\qquad\qquad\qquad\qquad +\frac{\norm{\vV}^2}{\Delta_{\vQ}} +  \norm{\vH_{\vQ}} \cdot \Big( \frac{\norm{\vH_{\vQ}}\norm{\vV}}{\Delta_{\vQ}\Delta_\nu}+\frac{\norm{\vV}^2}{\Delta_{\vQ}\Delta_\nu}\Big) 
    + r \Big(\frac{\norm{\vV}}{\Delta_{\vQ}}+\frac{\norm{\vV}^2}{\Delta_{\vQ}\Delta_{\nu}}\Big)\bigg).
\label{eq:mono-subspace-offdiag-err}
\end{align}
\end{corollary}
\begin{proof}
When the perturbation $\vV$ is off-diagonal in the energy eigenbasis of $\vH$, we can use tighter bounds on the changes in eigenvalues and eigensubspaces of $\vH'$ from \cref{lem:off_diag_perturbation}, which implies
\begin{align}
    \delta_\lambda = \max_j |E_j-E_j'| &= \CO\Big( \frac{\norm{\vV}^2}{\Delta_{\vQ}}\Big) \tag{originally $\CO(\norm{\vV})$}\\
     \norm{\vH_{\vQ}-\vH'_{\vQ'}} & = \CO\L(\norm{\vH_{\vQ}} \frac{\norm{\vV}}{\Delta_{\vQ}} + \frac{\norm{\vV}^2}{\Delta_{\vQ}}\R)\tag{originally $\CO(\norm{\vV})$}.\\
\norm{\vP - \vP'} & \le \CO\Big(\frac{\norm{\vV}}{\Delta_{\vQ}} + \frac{\norm{\vV}^2}{\Delta_{\vQ}\Delta_{\nu}}\Big). \tag{originally $\CO(\frac{\norm{\vV}}{\Delta_{\nu}})$}
\end{align}
The second line and third line are due to rotation and then subspace perturbation (Lemma~\ref{lem:PP'}).
Essentially, this is because (1) all subspace rotations are small ($\frac{\norm{\vV}}{\Delta_{\vQ}}\ll 1$) and (2) the energy perturbation is \textit{smaller} than the level spacing ($\frac{\norm{\vV}^2}{\Delta_{\vQ}} \ll \Delta_{\nu}$).
We then follow the same argument in the proof of \cref{cor:mono_subspace_restated} above with the improved bounds for Eqs.~\eqref{eq:Hq-Hqprime}, \eqref{eq:Eq-Eqprime}, and \eqref{eq:QP-diff}.
We can also use an improve bound of $\theta_{\max} \le \|\vH_{\vQ}\|$.
Together these improvements yield the better error bound on $\epsilon'$ as advertised.
\end{proof}

\subsection{Example where perturbation kills energy gradient}
We present an example where despite $\norm{\vV} \ll \Delta_{\nu}$, the gradient is lost due to the perturbation. Therefore, the resulting change in gradient is not multiplicative $(1- \frac{\norm{\vV}}{\Delta_{\nu}})$, but merely additive. If the gradient is polynomially small, we need $\frac{\norm{\vV}}{\Delta_{\nu}}$ to be also polynomially small to secure the gradient.  
\begin{proposition}[Perturbation kills the gradient]
    Let $\vH = \vZ = \ket{1}\bra{1} - \ket{-1}\bra{-1}$, $\vA = \vZ + \epsilon \vX$, and $\vV = \epsilon \vX$. Then, for the $\beta, \tau \rightarrow \infty$ heat bath Lindbladian, we have that
    \begin{align}
        \CL^{\dag}[\vH] \preceq -  2 \epsilon^2 (\vI - \vP) \quad \text{but}\quad  \CL'^{\dag}[\vH']=0.
    \end{align}
\end{proposition}
\begin{proof}
\begin{equation}
\CL^{\dag}[\vH] = \epsilon^2 \L(\ket{1}\bra{-1}\vH \ket{-1}\bra{1}- \frac{1}{2} \ket{1}\bra{1}\vH- \frac{1}{2} \vH\ket{1}\bra{1} \R) = -2\epsilon^2 \ket{1}\bra{1}.
\end{equation}
But $\CL'^{\dag}[\vH'] =0$ because $[\vH',\vA]=0$, i.e., $\vA$ is diagonal in the new energy basis.
\end{proof}

\bibliography{ref}
\bibliographystyle{unsrt}

\end{document}